\documentclass{amsart}
\long\def\sidebyside#1#2{%
 \hbox to\textwidth{\vtop{\hsize=.6\textwidth%
 \advance\hsize by -.32\columnsep
\parindent=0pt
\centering

 #1\vskip1sp}\hskip\columnsep\vtop{\hsize=.6\textwidth%
 \advance\hsize by -.32\columnsep
\parindent=0pt
\centering
#2

}\hfill}}

\usepackage{amsthm, amssymb, amsfonts}
\usepackage{psfrag}
\usepackage{graphicx}
\normalsize
\theoremstyle{definition}
\newtheorem{Thm}{Theorem}
\newtheorem{corollary}[Thm]{Corollary}
\newtheorem{lemma}[Thm]{Lemma}
\newtheorem{definition}[Thm]{Definition}
\newtheorem{theorem}[Thm]{Theorem}
\newtheorem{remark}[Thm]{Remark}
\newtheorem{example}[Thm]{Example}

\newcommand {\Ab}{\mathbf{A}}  % AA insted of Ab produces error msg.
\newcommand {\Bb}{\mathbf{B}}
\newcommand {\Cb}{\mathbb{C}}
\newcommand {\Cbf}{\mathbf{C}}

\newcommand {\op}{\mathbf{\oplus}} % op has spacing nicer than oplus
\newcommand {\om}{\mathbf{\ominus}} % om has spacing nicer than oplus
\newcommand {\od}{\mathbf{\otimes}}   % od has spacing nicer than oplus
\newcommand {\asubM}{\|\ab\|_{\lower.1ex \hbox {\scriptsize {M}}}}
\newcommand {\hsubM}{\|\hb\|_{\lower.1ex \hbox {\scriptsize {M}}}}
\newcommand {\hsubMs}{\|\hb\|_{\lower.1ex \hbox {\scriptsize {M}}}^2}
\newcommand {\hasubM}{\|\hb_{\ab}\|_{\lower.1ex \hbox {\scriptsize {M}}}}
\newcommand {\hbsubM}{\|\hb_{\bb}\|_{\lower.1ex \hbox {\scriptsize {M}}}}
\newcommand {\hcsubM}{\|\hb_{\cb}\|_{\lower.1ex \hbox {\scriptsize {M}}}}
\newcommand {\AsubM}{\|\Ab\|_{\lower.1ex \hbox {\scriptsize {M}}}}
\newcommand {\AsubMs}{\|\Ab\|_{\lower.1ex \hbox {\scriptsize {M}}}^2}
\newcommand {\AosubM}{\|A_1\|_{\lower.1ex \hbox {\scriptsize {M}}}}
\newcommand {\AtsubM}{\|A_2\|_{\lower.1ex \hbox {\scriptsize {M}}}}
\newcommand {\BsubM}{\|\Bb\|_{\lower.1ex \hbox {\scriptsize {M}}}}
\newcommand {\CsubM}{\|\Cb\|_{\lower.1ex \hbox {\scriptsize {M}}}}
\newcommand {\CsubMf}{\|\Cbf\|_{\lower.1ex \hbox {\scriptsize {M}}}}
\newcommand {\CsubMs}{\|\Cb\|_{\lower.1ex \hbox {\scriptsize {M}}}^2}
\newcommand {\asupM}{\|\ab\|^{\lower.1ex \hbox {\scriptsize {M}}}}
\newcommand {\aspM}{\ab^{\lower.1ex \hbox {\scriptsize {M}}}}
\newcommand {\asbM}{\ab_{\lower.1ex \hbox {\scriptsize {M}}}}
\newcommand {\AsupM}{\|\Ab\|^{\lower.1ex \hbox {\scriptsize {M}}}}
\newcommand {\BsupM}{\|\Bb\|^{\lower.1ex \hbox {\scriptsize {M}}}}
\newcommand {\CsupM}{\|\Cb\|^{\lower.0ex \hbox {\scriptsize {M}}}}
\newcommand {\CsupMf}{\|\Cbf\|^{\lower.0ex \hbox {\scriptsize {M}}}}
\newcommand {\asubP}{\|\ab\|_{\lower.1ex \hbox {\scriptsize {P}}}}
\newcommand {\asbP}{\ab_{\lower.1ex \hbox {\scriptsize {P}}}}
\newcommand {\AsubP}{\|\Ab\|_{\lower.1ex \hbox {\scriptsize {P}}}}
\newcommand {\BsubP}{\|\Bb\|_{\lower.1ex \hbox {\scriptsize {P}}}}
\newcommand {\CsubP}{\|\Cb\|_{\lower.1ex \hbox {\scriptsize {P}}}}
\newcommand {\phue}{\lower.3ex \hbox {\scriptsize {UE}}}
\newcommand {\pheu}{\lower.3ex \hbox {\scriptsize {EU}}}
\newcommand {\phum}{\lower.3ex \hbox {\scriptsize {UM}}}
\newcommand {\phmu}{\lower.3ex \hbox {\scriptsize {MU}}}
\newcommand {\phme}{\lower.3ex \hbox {\scriptsize {ME}}}
\newcommand {\phem}{\lower.3ex \hbox {\scriptsize {EM}}}
\newcommand {\ph}{\phantom{K}}
\newcommand {\lowerkluma}{\lower1.5ex \hbox {\phantom{K}}}
\newcommand {\lowerklumb}{\lower3.5ex \hbox {\phantom{K}}}
\newcommand {\lowerklumc}{\lower7.0ex \hbox {\phantom{K}}}
\newcommand {\pp}{\lower.6ex \hbox {\footnotesize {$P_{1} P_{2}$}}}

\newcommand {\mA}{m_{^A}^{\phantom{O}}\!}

\newcommand {\mO}{m_{^O}^{\phantom{O}}\!}

\newcommand {\mP}{m_{^P}^{\phantom{O}}\!} 
\newcommand {\mPpm}{m_{^{P_{\pm}}}^{\phantom{O}}\!} 
 
\newcommand {\mQ}{m_{^Q}^{\phantom{O}}\!}

 \newcommand {\lowAA}{\lower.3ex \hbox {\scriptsize {$\Ab$}}}
 \newcommand {\lowBB}{\lower.3ex \hbox {\scriptsize {$\Bb$}}}
 \newcommand {\lowCC}{\lower.3ex \hbox {\scriptsize {$\Cb$}}}
 \newcommand {\lowA}{\lower.3ex \hbox {\scriptsize {$A$}}}
 \newcommand {\lowB}{\lower.3ex \hbox {\scriptsize {$B$}}}
 \newcommand {\lowC}{\lower.3ex \hbox {\scriptsize {$C$}}}
 \newcommand {\lowE}{\lower.3ex \hbox {\tiny       {$\rm E$}}}
 \newcommand {\lowM}{\lower.3ex \hbox {\tiny       {$\rm M$}}}
 \newcommand {\lowtM}{\lower.01ex \hbox {\tiny       {$\rm M$}}}
 \newcommand {\lowtE}{\lower.01ex \hbox {\tiny       {$\rm E$}}}
 \newcommand {\lowtU}{\lower.01ex \hbox {\tiny       {$\rm U$}}}

 \newcommand {\lowEC}{\lower.3ex \hbox {\scriptsize {$EC$}}}
 \newcommand {\lowER}{\lower.3ex \hbox {\scriptsize {$ER$}}}
 \newcommand {\lowU}{\lower.3ex \hbox {\tiny       {\rm U}}}
 \newcommand {\lowDU}{\lower.3ex \hbox {\scriptsize {\rm DU}}}
 \newcommand {\lowCU}{\lower.3ex \hbox {\scriptsize {\rm CU}}}
 \newcommand {\lowDM}{\lower.3ex \hbox {\scriptsize {\rm DM}}}
 \newcommand {\lowCM}{\lower.3ex \hbox {\scriptsize {\rm CM}}}
 \newcommand {\lowo}{\lower.3ex \hbox {\scriptsize {\rm 0}}}
 \newcommand {\lowf}{\lower.3ex \hbox {\scriptsize {\rm f}}}

\newcommand {\gammaR}{\gamma_{_{R}}^{\phantom{1}}} 
\newcommand {\gammaRs}{\gamma_{_{R}}^{2}}

\newcommand {\lowmbpa}{\lower.6ex \hbox {\footnotesize {$\ome\bb\ope\ab$}}}
 
\newcommand {\uvc}{\displaystyle\frac{\lower.6ex \hbox {$\ub\ccdot\vb$}}{c^2}}
\newcommand {\uvs}{\displaystyle\frac{\lower.6ex \hbox {$\ub\ccdot\vb$}}{s^2}}
\newcommand {\unpuvc}{ \lower.6ex \hbox {$1 + \uvc$} }
\newcommand {\unpuvs}{ \lower.6ex \hbox {$1 + \uvs$} }
\newcommand {\uvcbar}{\displaystyle\frac{\lower.6ex \hbox
            {$\ubar\ccdot\vb$}}{c^2}}
\newcommand {\unpuvcbar}{ \lower.6ex \hbox {$1 + \uvcbar$} }
\newcommand {\vwc}{\displaystyle\frac{\lower.6ex\hbox{$\vb\ccdot\wb$}}{c^2}}
\newcommand {\unpvwc}{ \lower.6ex \hbox {$1 + \vwc$} }
\newcommand {\subE}{\!\lower.1ex \hbox {\tiny E}}
\newcommand {\subG}{\!\lower.1ex \hbox {\tiny G}}
\newcommand {\subH}{\!\lower.1ex \hbox {\tiny H}}
\newcommand {\subEs}{\!\lower.1ex \hbox {\tiny {E,S}}}
\newcommand {\subEt}{\!\lower.1ex \hbox {\tiny {E,2}}}
\newcommand {\subEC}{\!\lower.1ex \hbox {\tiny EC}}
\newcommand {\subU}{\!\lower.1ex \hbox {\tiny U}}
\newcommand {\subM}{\!\lower.1ex \hbox {\tiny M}}
\newcommand {\subC}{\!\lower.1ex \hbox {\tiny C}}
\newcommand {\subbE}{\!\lower.01ex \hbox {\tiny E}}
\newcommand {\subbU}{\!\lower.01ex \hbox {\tiny U}}
\newcommand {\subbC}{\!\lower.01ex \hbox {\tiny C}}
\newcommand {\subbM}{\!\lower.01ex \hbox {\tiny M}}
\newcommand {\subbG}{\!\lower.01ex \hbox {\tiny G}}
\newcommand {\subbH}{\!\lower.01ex \hbox {\tiny H}}
\newcommand {\ope}{\op_{_{\subE}}\!\,}

\newcommand {\ome}{\om_{_{\subE}}\!\,}

\newcommand {\ccdot}{\mathbf{\cdot }} % cd has spacing nicer than cdot
\newcommand {\ab}{\mathbf{a}}
\newcommand {\bb}{\mathbf{b}}
\newcommand {\cb}{\mathbf{c}}

\newcommand {\hb}{\mathbf{h}}

\newcommand {\ub}{\mathbf{u}}
\newcommand {\vb}{\mathbf{v}}

\newcommand {\wb}{\mathbf{w}}

\newcommand {\zerb}{\mathbf{0}}

\newcommand {\ubar}{\bar{\ub}}

\newcommand {\AAb}{\mathbb{A}}
\newcommand {\AAbt}{\mathbb{A}_3}

\newcommand {\Rb}{\mathbb{R}}

\newcommand {\Rn}{\Rb^n}

\newcommand {\Rstwou}{{\Rb}_{s=1}^2}
\newcommand {\Rstwo}{{\Rb}_{s}^2}

\newcommand {\Rcn}{{\Rb}_{c}^{n}}

\newcommand {\Rsn}{{\Rb}_{s}^{n}}

\newcommand {\Rct}{{\Rb}_{c}^{3}}

\newcommand {\Rt}{\Rb^3}

\newcommand {\Rtwo}{\Rb^2}

\newcommand {\gub}{\gamma_{\ub}^{\phantom{1}}}
\newcommand {\gvb}{\gamma_{\vb}^{\phantom{1}}}

\newcommand {\rmspan}{{\rm Span}}

\newcommand {\gammaaa}{\gamma_{11}^{\phantom{O}}}
\newcommand {\gammaab}{\gamma_{12}^{\phantom{O}}}
\newcommand {\gammaac}{\gamma_{13}^{\phantom{O}}}

\newcommand {\gammabc}{\gamma_{23}^{\phantom{O}}}

\newcommand {\gammaij}{\gamma_{ij}^{\phantom{O}}}
\newcommand {\timess}{\!\times\!}

\newcommand {\half}{\textstyle\frac{1}{2}}
\newcommand {\inn}{\hspace{-0.1cm}\in\hspace{-0.1cm}}

\newcommand {\subEA}{\!\lower.1ex \hbox {\tiny EA}}

 \newcommand {\gAa}{\gamma_{_{A_1}}^{\phantom{1}}}
 \newcommand {\gAb}{\gamma_{_{A_2}}^{\phantom{1}}}
 \newcommand {\gAc}{\gamma_{_{A_3}}^{\phantom{O}}}

\newcommand {\pprime}{{\prime\prime}}

\begin{document}
%%%%%%%%%%%%%%%%%%%%%%%%%%%%%%%%%%%%%%%%%%%%%%%%%%%%%%%%%%%%%%%%%%%%%%%%%%%%%
\begin{center}
% TITLE
\huge{
On the Study of \\
Hyperbolic Triangles and Circles \\
by Hyperbolic Barycentric Coordinates \\
in \\
Relativistic Hyperbolic Geometry \\
     }
\end{center}
%\vspace{-0.8cm}
\begin{center}
Abraham A. Ungar\\
Department of Mathematics\\
North Dakota State University\\
Fargo, ND 58105, USA\\
Email: Abraham.Ungar@ndsu.edu\\
%%%%%%%%%%%%%%%%%%%%%%%%%%%%%%%%%%%%%%%%%%%%%%%%%%%%%%%%%%%%%%%%%%%%%%%%%%%%%
\end{center}

%%%%%%%%%%%%%%%%%%%%%%%%%%%%%%%%%%%%%%%%%%%%%%%%%%%%%%%%%%%%%%%%%%%%%%%%%%%%%
%SECTION NUMBER 1
%-----------------------------------------------
\begin{quotation}
{\bf Abstract}
Barycentric coordinates are commonly used in Euclidean geometry.
Following the adaptation of barycentric coordinates for use in hyperbolic geometry
in recently published books on analytic hyperbolic geometry,
known and novel results concerning triangles and circles in the hyperbolic geometry
of Lobachevsky and Bolyai are discovered.
Among the novel results are the hyperbolic counterparts of important
theorems in Euclidean geometry. These are:
(1) the Inscribed Gyroangle Theorem,
(ii) the Gyrotangent-Gyrosecant Theorem,
(iii) the Intersecting Gyrosecants Theorem, and
(iv) the Intersecting Gyrochord Theorem.
Here in gyrolanguage, the language of analytic hyperbolic geometry,
we prefix a gyro to any term that describes a
concept in Euclidean geometry and in associative algebra
to mean the analogous concept in hyperbolic geometry and nonassociative algebra.
Outstanding examples are {\it gyrogroups} and {\it gyrovector spaces}, and
Einstein addition being both {\it gyrocommutative} and {\it gyroassociative}.
The prefix ``gyro'' stems from ``gyration'', which is the mathematical abstraction
of the special relativistic effect known as ``Thomas precession''.
\end{quotation}
%-----------------------------------------------
%%%%%%%%%%%%%%%%%%%%%%%%%%%%%%%%%%%%%%%%%%%%%%%%%%%%%%%%%%%%%%%%%%%%%%%%%%%%%

%SECTION NUMBER 1
\section{Introduction} \label{secint}

A barycenter in astronomy is the point between two objects where they balance
each other. It is the center of gravity where two or more celestial bodies
orbit each other.
In 1827 M\"obius published a book whose title,
{\it Der Barycentrische Calcul},
translates as
{\it The Barycentric Calculus}. The word {\it barycentric} means
center of gravity, but the book is entirely geometrical and, hence, called
by Jeremy Gray \cite{gray93},
{\it M\"obius's Geometrical Mechanics}.
The 1827 M\"obius book is best remembered for introducing a new system
of coordinates, the {\it barycentric coordinates}.\index{barycentric coordinates}
The historical contribution of M\"obius' barycentric coordinates to vector analysis
is described in \cite[pp.~48--50]{crowe94}.

Commonly used as a tool in the study of Euclidean geometry, barycentric coordinates
have been adapted for use as a tool in the study of the hyperbolic geometry of
Lobachevsky and Bolyai as well, in several recently published books
\cite{mybook02,mybook03,mybook06,mybook05}.

Relativistic hyperbolic geometry is a model of analytic hyperbolic geometry
in which Einstein addition plays the role of vector addition.
Einstein addition is a binary operation in the ball of vector spaces, which
is neither commutative nor associative. However,
Einstein addition is both gyrocommutative and gyroassociative, giving rise to
gyrogroups and gyrovector spaces. The latter, in turn,
form the algebraic setting for relativistic hyperbolic geometry, just as
vector spaces form the algebraic setting for the standard model of
Euclidean geometry.

Relativistic hyperbolic geometry admits the notion of
relativistic hyperbolic barycentric coordinates,
just as Euclidean geometry admits the notion of Euclidean barycentric coordinates.
Relativistic hyperbolic barycentric coordinates and
classical Euclidean barycentric coordinates
share remarkable analogies. In particular, they are both covariant. Indeed,
Relativistic barycentric coordinate representations are covariant with respect to the
Lorentz coordinate transformation group, just as
classical, Euclidean barycentric coordinate representations are covariant with respect to the
Galilean coordinate transformation group.
The remarkable analogies suggest that hyperbolic barycentric coordinates
can prove useful in the study of hyperbolic geometry,
just as Euclidean barycentric coordinates
prove useful in the study of Euclidean geometry.

Indeed, following the adaptation of Euclidean barycentric coordinates for use in
hyperbolic geometry, where they are called {\it gyrobarycentric coordinates},
we employ here the technique of gyrobarycentric coordinates
to rediscover and discover known and new results in hyperbolic geometry.
An introduction to hyperbolic barycentric coordinates and their application in
hyperbolic geometry is found in \cite{ungar13s}.
Some familiarity with relativistic hyperbolic geometry as studied in \cite{ungar13s}
is assumed. Relativistic hyperbolic geometry is studied extensively
in \cite{mybook02,mybook03,mybook06,mybook05};
see also
\cite{mybook01,mybook04,walterrev2002,rassiasrev2008,rassiasrev2010}
and
\cite{quasi91,service,ahlfors,ungarrassias07,incenter08,ungar12s}.

Among the novel results in hyperbolic geometry that are discovered here are
the following outstanding results:

\begin{enumerate}
\item  % 1
The Inscribed Gyroangle Theorem,
which is the hyperbolic counterpart of the well-known
Inscribed Angle Theorem in Euclidean geometry
(Sects.~\ref{dknke2}\,--\,\ref{dknke2s}).
\item  % 2
The Gyrotangent-Gyrosecant Theorem,
which is the hyperbolic counterpart of the well-known
Tangent-Secant Theorem in Euclidean geometry (Sect.~\ref{slila}).
\item  % 3
The Intersecting Gyrosecants Theorem,
which is the hyperbolic counterpart of the well-known
Intersecting Secants Theorem in Euclidean geometry (Sect.~\ref{slila2}).
\item  % 4
The Intersecting Gyrochords Theorem,
which is the hyperbolic counterpart of the well-known
Intersecting Chords Theorem in Euclidean geometry (Sect.~\ref{slila8}).
\end{enumerate}

The prefix ``gyro'' that we extensively use stems from the term
``gyration'' \cite{ungarsmale12},
which is the mathematical abstraction
of the special relativistic effect known as ``Thomas precession''.

The use of nonassociative algebra and
hyperbolic trigonometry (gyrotrigonometry) involves
staightforward, but complicates calculations.
Hence, computer algebra, like Mathematica, for algebraic manipulations
is an indispensable tool in this work.

%SECTION NUMBER 2
\section{Einstein Addition} \label{secein02}

Our journey into the fascinating world of relativistic hyperbolic geometry
begins in Einstein addition and passes through important novel theorems
that capture remarkable analogies between Euclidean and hyperbolic geometry.
Einstein addition, in turn, is the binary operation that stems from
Einstein's composition law of relativistically
admissible velocities that he introduced in his
1905 paper \cite{einstein05} \cite[p.~141]{einsteinfive}
that founded the special theory of relativity.

Let $c$ be an arbitrarily fixed positive constant and let
$\Rn=(\Rn,+,\ccdot)$ be the Euclidean $n$-space,
$n=1,2,3,\ldots,$
equipped with the common vector addition, +, and inner product, $\ccdot$.
The home of all $n$-dimensional Einsteinian velocities is the $c$-ball
\begin{equation} \label{eqcball}
\Rcn    = \{\vb\in\Rn: \|\vb\| < c \}
\end{equation}
It is the open ball of radius $c$, centered at the
origin of $\Rn$, consisting of all vectors $\vb$
in $\Rn$ with norm smaller than $c$.

Einstein addition and scalar multiplication play in the ball $\Rcn$
the role that vector addition and scalar multiplication play in the
Euclidean $n$-space $\Rn$.

% DEFINITION NUMBER 1
\begin{definition}\label{defeinsadd}
Einstein addition is a binary operation, $\op$,
in the $c$-ball $\Rcn$ given by the equation,
{\rm
\cite{mybook01},
\cite[Eq.~2.9.2]{urbantkebookeng},\cite[p.~55]{moller52},\cite{fock},
}
\begin{equation} \label{eq01}
{\ub}\op{\vb}=\frac{1}{\unpuvc}
\left\{ {\ub}+ \frac{1}{\gub}\vb+\frac{1}{c^{2}}\frac{\gamma _{{\ub}}}{%
1+\gamma _{{\ub}}}( {\ub}\ccdot{\vb}) {\ub} \right\}
\end{equation} 
for all $\ub,\vb\in\Rcn   $,
where $\gub$ is the Lorentz gamma factor given by the equation
\begin{equation} \label{v72gs}
\gvb = \frac{1}{\sqrt{1-\displaystyle\frac{\|\vb\|^2}{c^2}}}
\end{equation}
where $\ub\ccdot\vb$ and $\|\vb\|$
are the inner product and the norm
in the ball, which the ball $\Rcn  $ inherits from its space $\Rn$.
\end{definition}

A frequently used identity that follows immediately from \eqref{v72gs} is
\begin{equation} \label{rugh1ds}
\frac{\vb^2}{c^2} =
\frac{\|\vb\|^2}{c^2} = \frac{\gamma_\vb^2 - 1}{\gamma_\vb^2}
\end{equation}

A nonempty set with a binary operation is called a {\it groupoid} so that,
accordingly, the pair $(\Rcn,\op)$ is an
{\it Einstein groupoid}.

In the Newtonian limit of large $c$, $c\rightarrow\infty$, the ball $\Rcn   $
expands to the whole of its space $\Rn$, as we see from \eqref{eqcball},
and Einstein addition $\op$ in $\Rcn   $
reduces to the ordinary vector addition $+$ in $\Rn$,
as we see from \eqref{eq01} and \eqref{v72gs}.

In applications to velocity spaces, $\Rn=\Rt$ is the Euclidean 3-space,
which is the space of all classical, Newtonian velocities, and
$\Rcn  =\Rct\subset\Rt$ is the $c$-ball of $\Rt$ of all relativistically
admissible, Einsteinian velocities.
The constant $c$ represents in special relativity the
vacuum speed of light.
Since we are interested in geometry, we allow $n$ to be
any positive integer and, sometimes, replace $c$ by $s$.

We naturally use the abbreviation
$\ub\om\vb=\ub\op(-\vb)$ for Einstein subtraction, so that,
for instance, $\vb\om\vb = \zerb$,
$\om\vb = \zerb\om\vb=-\vb$.
Einstein addition and subtraction satisfy the equations
\begin{equation} \label{eq01a}
\om(\ub\op\vb) = \om\ub\om\vb
\end{equation}
and
\begin{equation} \label{eq01b}
\om\ub\op(\ub\op\vb) = \vb
\end{equation}
for all $\ub,\vb$ in the ball $\Rcn$,
in full analogy with vector addition and subtraction in $\Rn$. Identity
\eqref{eq01a} is called the {\it gyroautomorphic inverse property} of
Einstein addition, and Identity
\eqref{eq01b} is called the {\it left cancellation law} of Einstein addition.
We may note that
Einstein addition does not obey the naive right counterpart of the
left cancellation law \eqref{eq01b} since, in general,
\begin{equation} \label{eq01c}
(\ub\op\vb)\om\vb \ne \ub
\end{equation}
However, this seemingly lack of a {\it right cancellation law} of
Einstein addition is repaired, for instance, in \cite[Sect.~1.9]{mybook05}.

Finally, as demonstrated in \cite{ungar13s} and in
\cite{mybook01,mybook02,mybook03,mybook04,mybook06,mybook05},
Einstein addition admits scalar multiplication, giving
rise to Einstein gyrovector spaces. These, in turn, form the
algebraic setting for relativistic model of hyperbolic geometry, just as
vector spaces form the
algebraic setting for the standard model of Euclidean geometry.
A brief description of the road from Einstein addition to
gyrogroups and gyrovector spaces, necessary for a fruitful reading of this chapter,
is found in \cite{ungar13s}.

%SECTION NUMBER 3
\section{Gyrocircles}\label{dkrymv}
\index{gyrocircle}

Assuming familiarity with Einstein gyrovector spaces, as studied in \cite{ungar13s},
and particularly with the concepts of
{\it gyrobarycentric independence}, {\it gyroflats} and {\it hyperbolic span}
in \cite[Def.~13]{ungar13s}, the gyrocircle definition follows.

%%%%%%%%%%%%%%%%%%%%%%%%%%%%%%%%%%%%%%%%%%%%%%%%%%%%%%%%%%%%%%%%%%%%
% DEFINITON NUMBER 8.1
\begin{definition}\label{dfgyrocircle}
{\bf (Gyrocircle).}
\index{gyrocircle), def.}
{\it
Let $S=\{A_1,A_2,A_3\}$ be a gyrobarycentrically independent set
in an Einstein gyrovector space $(\Rsn,\op,\od)$, $n\ge 2$, and let
\begin{equation} \label{dambdy}
\AAbt = A_1\op\rmspan\{\om A_1\op A_2,\om A_1\op A_3\} \subset\Rn
\,.
\end{equation}
The locus of a point in $\AAbt\cap\Rsn$ which is at a
constant gyrodistance $r$
from a fixed point $O\inn\AAbt\cap\Rsn$ is a gyrocircle $C(r,O)$ with
gyrocenter $O$ and gyroradius $r$.
}
\end{definition}
%%%%%%%%%%%%%%%%%%%%%%%%%%%%%%%%%%%%%%%%%%%%%%%%%%%%%%%%%%%%%%%%%%%%

%%%%%%%%%%%%%%%%%%%%%%%%%%%%%%%%%%%%%%%%%%%%%%%%%%%%%%%%%%%%%%%%%%%%
% FIGURE 1
  
%%%%%%%%%%%%%%%%%%%%%%%%%%%%%%%%%%%%%%%%%%%%%%%%%%%%%%%%%%%%%%%%%%%%%%
%%%%% The hyperbolic semi-circle theorem            %%%%%%%%%%%%%%%%%%
%\begin{figure}[htbp]
\begin{figure}[t]  % try to put this figure on the top of the page
 \centering         % center the figure
%
%\psfrag{a1}{$A_1$}
%
%\includegraphics[width=9cm]{/home/ungar/dir_amy/dir_papers/dir_mybook01/dir_figs/fig266a.eps}
 \includegraphics[width=9cm]{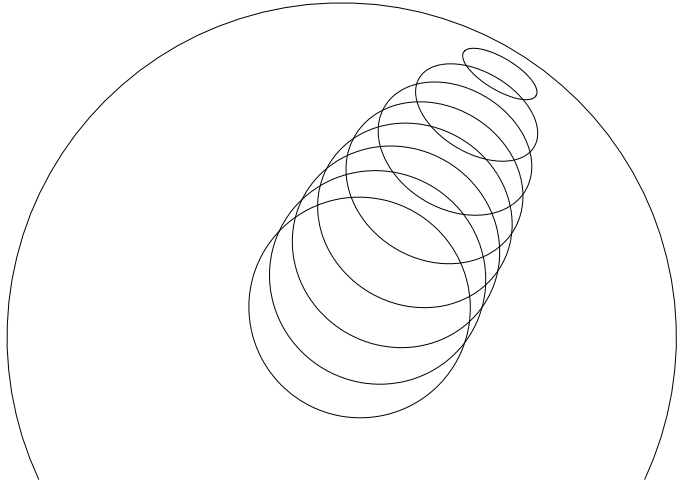}
\caption[Gyrocircles]{
A sequence of gyrocircles with gyroradius $\frac{1}{3}$ in an Einstein gyrovector plane $\Rstwou$
with gyrocenters approaching the boundary of the open unit disc $\Rstwou$
is shown.
The center of the disc is conformal. Hence, the
gyrocircle with gyrocenter at the center of the disc coincides with a Euclidean circle.
The Euclidean circle is increasingly flattened
as its gyrocenter approaches the boundary of the disc.
\label{fig266am}}
\end{figure}
%%%%%%%%%%%%%%%%%%%%%%%%%%%%%%%%%%%%%%%%%%%%%%%%%%%%%%%%%%%%%%%%%%%%%%
 % Gyrocircles in (+)E 
%   Fig.~\ref{fig266am} 
%%%%%%%%%%%%%%%%%%%%%%%%%%%%%%%%%%%%%%%%%%%%%%%%%%%%%%%%%%%%%%%%%%%n%

The gyrocircle $C(r,O)$ with gyroradius $r$, $0<r<s$, and gyrocenter $O\in\Rstwo$
in an Einstein gyrovector plane $(\Rstwo,\op,\od)$ is the set of all points
$P\in\Rstwo$ such that $\|\om P\op O\|=r$.
It is given by the equation
\begin{equation} \label{ekufn1}
C(r,O,\theta) = O~\op \begin{pmatrix} r\cos\theta \\ r\sin\theta \end{pmatrix}
\,,
\end{equation}
$0\le\theta < 2\pi$. Indeed, by the left cancellation law we have
\begin{equation} \label{ekufn2}
\| \om O \op C(r,O,\theta) \| = 
\left\| \begin{pmatrix} r\cos\theta \\ r\sin\theta \end{pmatrix} \right\| = r
\,.
\end{equation}

A sequence of gyrocircles of gyroradius $\frac{1}{3}$ in an Einstein gyrovector plane $\Rstwou$
with gyrocenters approaching the boundary of the open unit disc $\Rstwou$
is shown in Fig.~\ref{fig266am}.
The center of the disc in Fig.~\ref{fig266am} is conformal,
as explained in \cite[Sect.~6.2]{mybook05}.
Accordingly, a gyrocircle with gyrocenter at the center of the disc is identical
to a Euclidean circle. This Euclidean circle is increasingly flattened in the Euclidean sense
when the gyrocircle gyrocenter approaches the boundary of the disc.

%SECTION NUMBER 4
\section{Gyrotriangle Circumgyrocenter}\label{dknce1}
\index{circumgyrocenter}
%%%%%%%%%%%%%%%%%%%%%%%%%%%%%%%%%%%%%%%%%%%%%%%%%%%%%%%%%%%%%%%%%%%%

% DEFINITON NUMBER 8.2
\index{circumgyrocenter, gyrotriangle}
\index{circumgyroradius, gyrotriangle}
\index{circumgyrocircle, gyrotriangle}
\begin{definition}\label{defhsnmeg2}
{\bf (Gyrotriangle Circumgyrocircle, Circumgyrocenter, Circumgyroradius).}
{\it
Let $A_1A_2A_3$ be a gyrotriangle
in an Einstein gyrovector space $(\Rsn,\op,\od)$, $n\ge2$,
and let $\AAbt$ be the set
\begin{equation} \label{dambdx}
\AAbt = A_1\op\rmspan\{\om A_1\op A_2,\om A_1\op A_3\} \subset\Rn
\,.
\end{equation}
The circumgyrocenter of the gyrotriangle is the point $O$,
$O \in \AAbt \cap \Rsn$,
equigyrodistant from the three gyrotriangle vertices.
The gyrodistance from $O$ to each vertex $A_k$, $k=1,2,3$,
of the gyrotriangle
is the gyrotriangle circumgyroradius, and the gyrocircle with
gyrocenter $O$ and gyroradius $r$ in $\AAbt\cap\Rsn$ is the
gyrotriangle circumgyrocircle.
}
\end{definition}
%%%%%%%%%%%%%%%%%%%%%%%%%%%%%%%%%%%%%%%%%%%%%%%%%%%%%%%%%%%%%%%%%%%%

%%%%%%%%%%%%%%%%%%%%%%%%%%%%%%%%%%%%%%%%%%%%%%%%%%%%%%%%%%%%%%%%%%%%
% FIGURE 2
 
%%%%%%%%%%%%%%%%%%%%%%%%%%%%%%%%%%%%%%%%%%%%%%%%%%%%%%%%%%%%%%%%%%%%%%
%%%%% The Einstein gyroparallelogram                %%%%%%%%%%%%%%%%%%
%\begin{figure}[htbp]
\begin{figure}[t]  % try to put this figure on the top of the page
              % [h] tries to place the figure here
              % [b] tries to place the figure on the bottom of the page
              % [t] tries to place the figure on the top of the page
              % [P] tries to place the figure floatingly on the page
 \centering         % center the figure
\psfrag{O}[]{$\phantom{O}$}
\psfrag{A1}[]{$A_1$}
\psfrag{A2}[]{$A_2$}
\psfrag{A3}[]{$A_3$}
\psfrag{O}[]{$O$}
\psfrag{pa1}{$a_{23}$}
\psfrag{pa2}[]{$a_{31}=a_{13}$}
\psfrag{pa3}{$a_{12}$}
\psfrag{text1}[]{$\ab_{23}=\om A_2 \op A_3$}
\psfrag{text2}[]{$\ab_{31}=\om A_3 \op A_1$}
\psfrag{text3}[]{$\ab_{12}=\om A_1 \op A_2$}
\psfrag{al1}[]{$\alpha_1$}
\psfrag{al2}[]{$\alpha_2$}
\psfrag{al3}[]{$\alpha_3$}
\psfrag{formula01}{$a_{12}=\|\ab_{12}\|=\|\om A_1 \op A_2\|$}
\psfrag{formula02}{$a_{13}=\|\ab_{31}\|=\|\om A_3 \op A_1\|$}
\psfrag{formula03}{$a_{23}=\|\ab_{23}\|=\|\om A_2 \op A_3\|$}
\psfrag{formula04}[]{$\cos\alpha_1=\frac{\om A_1\op A_2}{\|\om A_1\op A_2\|}
\ccdot \frac{\om A_1\op A_3}{\|\om A_1\op A_3\|}$}
\psfrag{formula05}[]{$\cos\alpha_2=\frac{\om A_2\op A_1}{\|\om A_2\op A_1\|}
\ccdot \frac{\om A_2\op A_3}{\|\om A_2\op A_3\|}$}
\psfrag{formula06}[]{$\cos\alpha_3=\frac{\om A_3\op A_1}{\|\om A_3\op A_1\|}
\ccdot \frac{\om A_3\op A_2}{\|\om A_3\op A_2\|}$}
%\psfrag{fig241}{$\hspace{0.8cm}
 \psfrag{fig241}{$\hspace{1.8cm}
O = \frac{
\sum_{k=1}^{3} m_k \gamma_{_{A_k}}^{\phantom{1}} A_k
}{
\sum_{k=1}^{3} m_k \gamma_{_{A_k}}^{\phantom{1}}
}
$}
%%%%%%%%%%%%%%%%%%%%%%%%%%%%%%%%%%%%%%%%%%%%%%%%%%%%%%%%%%%%%%%%%%%%%%
%\psfrag{----chets1}[]{\lower-1.2ex \hbox {$\blacktriangleright$}}
%\psfrag{----chets2}[]{\lower-1.26ex \hbox {$\blacktriangleright$}}
%\psfrag{----chets3}[]{\lower-1.2ex \hbox {$\blacktriangleright$}}
%%%%%%%%%%%%%%%%%%%%%%%%%%%%%%%%%%%%%%%%%%%%%%%%%%%%%%%%%%%%%%%%%%%%%%
%\includegraphics[width=10cm]{/home/ungar/dir_amy/dir_papers/dir_mybook01/dir_figs/fig241en2a.eps}
 \includegraphics[width=10cm]{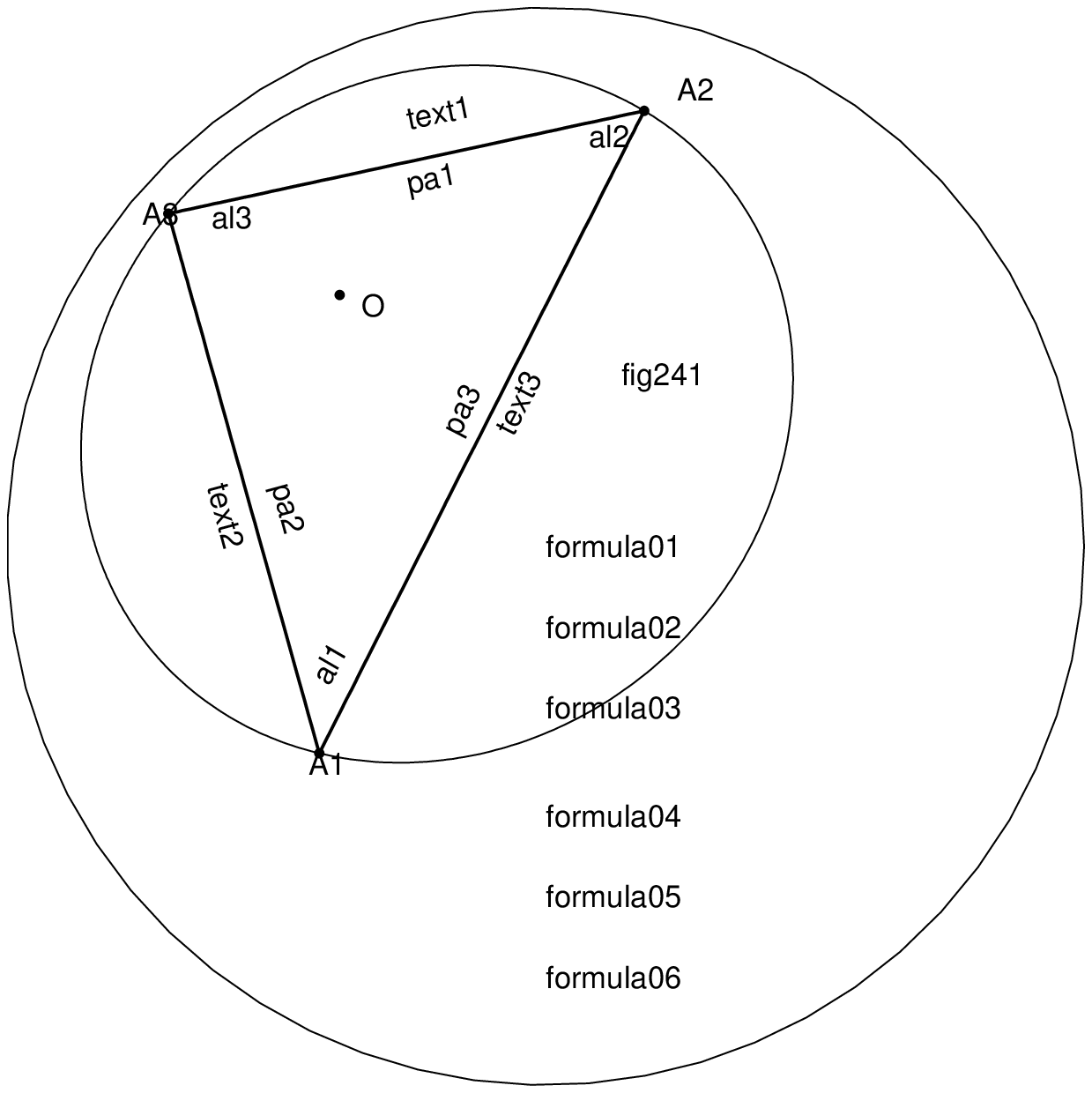}
\caption[The gyrotriangle circumgyrocircle, circumgyrocenter]{
The circumgyrocircle, and the
circumgyrocenter $O$, of gyrotriangle $A_1A_2A_3$ in an
Einstein gyrovector space $(\Rsn,\op,\od)$, $n=2$, is shown along with its associated
index notation.
Here $\|\om A_1 \op O\|$ $=$ $\|\om A_2 \op O\|$ $=$ $\|\om A_3 \op O\|$,
where $O$ is the gyrotriangle circumgyrocenter, given by its
gyrobarycentric representation \eqref{eicksfnein},
with respect to the gyrobarycentrically independent set $S=\{A_1,A_2,A_3\}$.
The Euclidean counterpart of this figure is shown in
Fig.~\ref{fig241euc1m}, p.~\pageref{fig241euc1m}.
\label{fig241en2am}}
\end{figure}
%%%%%%%%%%%%%%%%%%%%%%%%%%%%%%%%%%%%%%%%%%%%%%%%%%%%%%%%%%%%%%%%%%%%%%
 % The gyrotriangle circumgyrocenter in (+)E
                     % and circumgyrocircle.
%                       Fig.~\ref{fig241en2am}      Fig. 7.1
%%%%%%%%%%%%%%%%%%%%%%%%%%%%%%%%%%%%%%%%%%%%%%%%%%%%%%%%%%%%%%%%%%%%

It should be noted that not every gyrotriangle in $\Rsn$ possesses
a circumgyrocenter.

Let $A_1A_2A_3$ be a gyrotriangle in an Einstein gyrovector space
$(\Rsn,\op,\od)$ that possesses a circumgyrocircle, and let
$O \in \AAbt\cap\Rsn$
be the circumgyrocenter of the gyrotriangle,
as shown in Fig.~\ref{fig241en2am}.
Then, $O$ possesses a gyrobarycentric representation,
\index{circumgyrocenter}
\begin{equation} \label{hdsjb01}
O = \frac{
 m_1 \gamma_{_{A_1}}^{\phantom{1}} A_1 +
 m_2 \gamma_{_{A_2}}^{\phantom{1}} A_2 +
 m_3 \gamma_{_{A_3}}^{\phantom{1}} A_3
}{
m_1 \gamma_{_{A_1}}^{\phantom{1}} +
m_2 \gamma_{_{A_2}}^{\phantom{1}} +
m_3 \gamma_{_{A_3}}^{\phantom{1}}
}
\,,
\end{equation}
with respect to the gyrobarycentrically independent set $S=\{A_1,A_2,A_3\}$.
The gyrobarycentric coordinates $m_1,m_2$ and $m_3$ are to be determined
in \eqref{hdsjb07} below, in terms of gamma factors of the gyrotriangle sides and,
alternatively in \eqref{hdsjb11}, in terms of the gyrotriangle gyroangles.

Following the
{\it Gyrobarycentric representation Gyrocovariance Theorem},
\cite[Theorem 4.6, pp.~90-91]{mybook05},
with a left gyrotranslation by $X=\om A_1$,
and using the index notation
\begin{equation} \label{indexnotation}
\ab_{ij} = \om A_i \op A_j \,,
\hspace{1.2cm}
a_{ij}=\|\ab_{ij}\|\,,
\hspace{1.2cm}
\gamma_{ij}^{\phantom{O}} = \gamma_{\ab_{ij}}^{\phantom{O}} =\gamma_{a_{ij}}^{\phantom{O}}
\,,
\end{equation}
noting that $a_{ij}=a_{ji}$, $\gamma_{ij}^{\phantom{O}}=\gamma_{ji}^{\phantom{O}}$,
$\ab_{ii}=\zerb$, $a_{ii}=0$ and $\gamma_{ii}^{\phantom{O}}=1$.
we have
% %%%%%%%%%%%%%%%%%%%%%%%%%%%%%%%%%%%%%%%%%%%%%%%%%%%%%%%%
\begin{equation} \label{hdsjb02}
 \begin{split}
 \gamma_{_{\om A_1 \op O}}^{\phantom{1}} &= \frac{
m_1 \gamma_{_{\om A_1 \op A_1}}^{\phantom{1}} +
m_2 \gamma_{_{\om A_1 \op A_2}}^{\phantom{1}} +
m_3 \gamma_{_{\om A_1 \op A_3}}^{\phantom{1}}
}{\mO}
 \\[8pt]  &= \frac{
m_1 +
m_2 \gamma_{12}^{\phantom{1}} +
m_3 \gamma_{13}^{\phantom{1}}
}{\mO}
\,,
 \end{split}
\end{equation}
% %%%%%%%%%%%%%%%%%%%%%%%%%%%%%%%%%%%%%%%%%%%%%%%%%%%%%%%%
where by \cite[Eq.~(4.27), p.~90]{mybook05} and the
{\it Gyrobarycentric representation Gyrocovariance Theorem},
\cite[Theorem 4.6, pp.~90-91]{mybook05},
the circumgyrocenter gyrobarycentric representation constant $\mO>0$
with respect to the set of the gyrotriangle vertices
is given by the equation
\index{circumgyrocenter constant}
\begin{equation} \label{yjuhdg}
m_O^2 = m_1^2 + m_2^2 + m_3^2 + 2(
m_1m_2\gammaab + m_1m_3\gammaac +  m_2m_3\gammabc)
\,.
\end{equation}

Similarly, by the
{\it Gyrobarycentric representation Gyrocovariance Theorem},
\cite[Theorem 4.6, pp.~90-91]{mybook05},
with left gyrotranslations
by $X=\om A_1$, by $X=\om A_2$, and by $X=\om A_3$, we have, respectively,
%%%%%%%%%%%%%%%%%%%%%%%%%%%%%%%%%%%%%%%%%%%%%%%%%%%%%%%%%%%%%%%%%%%%
\begin{equation} \label{hdsjb03}
\begin{split}
\gamma_{_{\om A_1 \op O}}^{\phantom{1}} &= \frac{
m_1 +
m_2 \gamma_{12}^{\phantom{1}} +
m_3 \gamma_{13}^{\phantom{1}}
}{\mO}
\\[1pt]
\gamma_{_{\om A_2 \op O}}^{\phantom{1}} &= \frac{
m_1 \gamma_{12}^{\phantom{1}} +
m_2 +
m_3 \gamma_{23}^{\phantom{1}}
}{\mO}
\\[1pt]
\gamma_{_{\om A_3 \op O}}^{\phantom{1}} &= \frac{
m_1 \gamma_{13}^{\phantom{1}} +
m_2 \gamma_{23}^{\phantom{1}} +
m_3
}{\mO}
\,.
\end{split}
\end{equation}
%%%%%%%%%%%%%%%%%%%%%%%%%%%%%%%%%%%%%%%%%%%%%%%%%%%%%%%%%%%%%%%%%%%%

The condition that the circumgyrocenter $O$ is equigyrodistant from its gyrotriangle vertices
$A_1,A_2$, and $A_3$ implies
\begin{equation} \label{hdsjb04}
\gamma_{_{\om A_1 \op O}}^{\phantom{1}} =
\gamma_{_{\om A_2 \op O}}^{\phantom{1}} =
\gamma_{_{\om A_3 \op O}}^{\phantom{1}}
\,.
\end{equation}

Equations \eqref{hdsjb03} and \eqref{hdsjb04}, along with the normalization
condition $m_1+m_2+m_3=1$, yield the following system of three equations for
the three unknowns $m_1,m_2$, and $m_3$,
%%%%%%%%%%%%%%%%%%%%%%%%%%%%%%%%%%%%%%%%%%%%%%%%%%%%%%%%%%%%%%%%%%%%
\begin{equation} \label{hdsjb05}
\begin{split}
m_1+m_2+m_3 &=1 \\
m_1 + m_2 \gamma_{12}^{\phantom{1}} +  m_3 \gamma_{13}^{\phantom{1}}
&=
m_1 \gamma_{13}^{\phantom{1}} +  m_2 \gamma_{23}^{\phantom{1}} + m_3 \\
 m_1 \gamma_{12}^{\phantom{1}} + m_2 + m_3 \gamma_{23}^{\phantom{1}}
&=
m_1 \gamma_{13}^{\phantom{1}} +  m_2 \gamma_{23}^{\phantom{1}} + m_3
\,,
\end{split}
\end{equation}
%%%%%%%%%%%%%%%%%%%%%%%%%%%%%%%%%%%%%%%%%%%%%%%%%%%%%%%%%%%%%%%%%%%%
which can be written as the matrix equation,
%%%%%%%%%%%%%%%%%%%%%%%%%%%%%%%%%%%%%%%%%%%%%%%%%%%%%%%%%%%%%%%%%%%%
\begin{equation} \label{hdsjb06}
\begin{pmatrix} 1 & 1  & 1 \\[4pt]
1- \gamma_{13}^{\phantom{1}} & \gamma_{12}^{\phantom{1}} - \gamma_{23}^{\phantom{1}}
& \gamma_{13}^{\phantom{1}} - 1 \\[4pt]
\gamma_{12}^{\phantom{1}} - \gamma_{13}^{\phantom{1}} & 1 - \gamma_{23}^{\phantom{1}}
& \gamma_{23}^{\phantom{1}} - 1
\end{pmatrix}
\begin{pmatrix} m_1 \\[4pt] m_2 \\[4pt] m_3 \end{pmatrix}
=
\begin{pmatrix} 1 \\[4pt] 0 \\[4pt] 0 \end{pmatrix}
\,.
\end{equation}
%%%%%%%%%%%%%%%%%%%%%%%%%%%%%%%%%%%%%%%%%%%%%%%%%%%%%%%%%%%%%%%%%%%%

Solving \eqref{hdsjb06} for the unknowns $m_1,m_2$, and $m_3$, we have
%%%%%%%%%%%%%%%%%%%%%%%%%%%%%%%%%%%%%%%%%%%%%%%%%%%%%%%%%%%%%%%%%%%%
\begin{equation} \label{hdsjb07}
\begin{split}
m_1 &= \frac{1}{D}
(\phantom{-} \gamma_{12}^{\phantom{1}} + \gamma_{13}^{\phantom{1}} - \gamma_{23}^{\phantom{1}} -1)
(\gamma_{23}^{\phantom{1}} -1)
\\[6pt]
m_2 &= \frac{1}{D}
(\phantom{-} \gamma_{12}^{\phantom{1}} - \gamma_{13}^{\phantom{1}} + \gamma_{23}^{\phantom{1}} -1)
(\gamma_{13}^{\phantom{1}} -1)
\\[6pt]
m_3 &= \frac{1}{D}
(         -  \gamma_{12}^{\phantom{1}} + \gamma_{13}^{\phantom{1}} + \gamma_{23}^{\phantom{1}} -1)
(\gamma_{12}^{\phantom{1}} -1)
\,,
\end{split}
\end{equation}
%%%%%%%%%%%%%%%%%%%%%%%%%%%%%%%%%%%%%%%%%%%%%%%%%%%%%%%%%%%%%%%%%%%%
where $D$ is the determinant of the $3\timess3$ matrix in \eqref{hdsjb06},
\begin{equation} \label{eufjsd}
\begin{split}
 D &= 2( \gamma_{12}^{\phantom{1}} \gamma_{13}^{\phantom{1}}
      + \gamma_{12}^{\phantom{1}} \gamma_{23}^{\phantom{1}}
      + \gamma_{13}^{\phantom{1}} \gamma_{23}^{\phantom{1}})
\\ & \phantom{-}
      - (\gamma_{12}^2-1) - (\gamma_{13}^2-1) - (\gamma_{23}^2-1)
      - 2(\gamma_{12}^{\phantom{1}} + \gamma_{13}^{\phantom{1}} + \gamma_{23}^{\phantom{1}})
\\ &=
1+2\gammaab\gammaac\gammabc-\gamma_{12}^2-\gamma_{13}^2-\gamma_{23}^2
 -2(\gammaab-1)(\gammaac-1)(\gammabc-1)
\,.
\end{split}
\end{equation}
% By MATH stam138
% Solution confirmed in MATH stam140

Gyrotrigonometric substitutions into the extreme right-hand side of \eqref{eufjsd}
from \cite[Sect.~7.12]{mybook05}
give the gyrotrigonometric representation of $D$,
\begin{equation} \label{hadomf}
\begin{split}
D &= \frac{16F^2}{\sin^2\alpha_1 \sin^2\alpha_2 \sin^2\alpha_3}
-
\frac{16F}{\sin^2\alpha_1 \sin^2\alpha_2 \sin^2\alpha_3} \sin^2\tfrac{\delta}{2}
\\ &=
\frac{16F}{\sin^2\alpha_1 \sin^2\alpha_2 \sin^2\alpha_3}
(F-\sin^2\tfrac{\delta}{2})
\\ &=
\frac{16F}{\sin^2\alpha_1 \sin^2\alpha_2 \sin^2\alpha_3}
\sin\tfrac{\delta}{2}
\\ & \times \{
\sin(\alpha_1+\tfrac{\delta}{2})
\sin(\alpha_2+\tfrac{\delta}{2})
\sin(\alpha_3+\tfrac{\delta}{2})
-\sin\tfrac{\delta}{2}
\}
\,,
\end{split}
\end{equation}
%MATHEMATICA stam283 zero1
where $\delta=\pi-\alpha_1-\alpha_2-\alpha_3$ is the defect of gyrotriangle $A_1A_2A_3$,
and where
\begin{equation} \label{grfde}
F=F(\alpha_1,\alpha_2,\alpha_3) =
\sin\tfrac{\delta}{2}
\sin(\alpha_1+\tfrac{\delta}{2})
\sin(\alpha_2+\tfrac{\delta}{2})
\sin(\alpha_3+\tfrac{\delta}{2})
\,.
\end{equation}

The extreme right-hand side of \ref{hadomf} is the product of two factors
the first of which is positive.
The second factor is positive if and only if the
gyrotriangle circumgyrocenter $O\in\Rn$ lies inside the ball $\Rsn$, as we will see in
Theorem \ref{thmtivhvn}, p.~\pageref{thmtivhvn}.
This factor vanishes if and only if $O$ lies on the boundary of the ball $\Rsn$,
and it is negative if and only if $O$ lies outside the closure of the ball $\Rsn$,
as we will see in Theorem \ref{thmtivhvn}.

The circumgyrocenter $O$ of gyrotriangle $A_1A_2A_3$ is given by \eqref{hdsjb01}
where the gyrobarycentric coordinates $m_1,m_2$, and $m_3$
are given by \eqref{hdsjb07}.
Since in gyrobarycentric coordinates only ratios of coordinates are
relevant, the gyrobarycentric coordinates, $m_1,m_2$, and $m_3$ in
\eqref{hdsjb07} can be simplified by removing their common nonzero factor $1/D$.
\index{circumgyrocenter}

Gyrobarycentric coordinates, $m_1,m_2$, and $m_3$, of the circumgyrocenter $O$
of gyrotriangle $A_1A_2A_3$ are thus given by the equations
%%%%%%%%%%%%%%%%%%%%%%%%%%%%%%%%%%%%%%%%%%%%%%%%%%%%%%%%%%%%%%%%%%%%
\begin{equation} \label{hdsjb07bb}
\begin{split}
m_1^\prime &=
(\phantom{-} \gamma_{12}^{\phantom{1}} + \gamma_{13}^{\phantom{1}} - \gamma_{23}^{\phantom{1}} -1)
(\gamma_{23}^{\phantom{1}} -1)
\\[4pt]
m_2^\prime &=
(\phantom{-} \gamma_{12}^{\phantom{1}} - \gamma_{13}^{\phantom{1}} + \gamma_{23}^{\phantom{1}} -1)
(\gamma_{13}^{\phantom{1}} -1)
\\[4pt]
m_3^\prime &=
(         -  \gamma_{12}^{\phantom{1}} + \gamma_{13}^{\phantom{1}} + \gamma_{23}^{\phantom{1}} -1)
(\gamma_{12}^{\phantom{1}} -1)
\,.
\end{split}
\end{equation}
%%%%%%%%%%%%%%%%%%%%%%%%%%%%%%%%%%%%%%%%%%%%%%%%%%%%%%%%%%%%%%%%%%%%

Hence, by \eqref{yjuhdg} along with the gyrobarycentric coordinates
in \eqref{hdsjb07bb}, we have
\begin{equation} \label{gkdnseu}
\begin{split}
m_O^2 &= (1+2\gammaab\gammaac\gammabc-\gamma_{12}^2-\gamma_{13}^2-\gamma_{23}^2)
\\ &\times
\{1+2\gammaab\gammaac\gammabc-\gamma_{12}^2-\gamma_{13}^2-\gamma_{23}^2
-2(\gammaab-1)(\gammaac-1)(\gammabc-1)\}
\\ & \hspace{-0.6cm} =
 \{ (\gammaab+\gammaac+\gammabc-1)^2
-2 (\gamma_{12}^2+\gamma_{13}^2+\gamma_{23}^2-1) \}
\\ &\times
(1+2\gammaab\gammaac\gammabc - \gamma_{12}^2 - \gamma_{13}^2 - \gamma_{23}^2)
.
\end{split}
\end{equation}
% Calculated in MATHEMATICA stam133 zeromos
% Corroborated numerically for posneg in test0437 (posv3).

In order to emphasize comparative patterns that the
gyrotriangle circumgyrocircle and the
gyrotetrahedron circumgyrosphere possess, we present the first equation
in \eqref{gkdnseu} in the determinantal form
\begin{equation} \label{drekc}
m_O^2 = D_3 (D_3 - H_3)
\,,
\end{equation}
where $D_3$ is the determinant
\begin{equation} \label{dethkcf}
D_3 ~=~ \left|
\begin{matrix}
1 & \gammaab & \gammaac  \\[6pt]
\gammaab & 1 & \gammabc  \\[6pt]
\gammaac & \gammabc & 1
\end{matrix}
\right|
\end{equation}
and where
\begin{equation} \label{dethkdg}
H_3 = 2(\gammaab-1)(\gammaac-1)(\gammabc-1)
\,.
\end{equation}

The gyrotriangle
$A_1A_2A_3$ in Fig.~\ref{fig241en2am} possesses a circumgyrocenter
if and only if $m_O^2>0$.

The factor $D_3$ of $m_O^2$ in \eqref{gkdnseu} and in \eqref{drekc},
\begin{equation} \label{mizsk}
D_3 =1+2\gammaab\gammaac\gammabc-\gamma_{12}^2-\gamma_{13}^2-\gamma_{23}^2
= \left|
\begin{matrix} 
1 & \gammaab & \gammaac  \\[6pt]
\gammaab & 1 & \gammabc  \\[6pt]
\gammaac & \gammabc & 1
\end{matrix}
\right|
~>~0
\,,
\end{equation}
is positive for any
gyrotriangle $A_1A_2A_3$ in an Einstein gyrovector space.
Hence, as we see from \eqref{gkdnseu}, $m_O^2>0$ if and only if
the points $A_1,~A_2$, and $A_3$ obey the {\it circumgyrocircle existence condition}
\index{circumgyrocircle existence condition}
\begin{subequations} \label{rjksd}
\begin{equation} \label{rjksda}
(\gammaab+\gammaac+\gammabc-1)^2 ~>~ 2 (\gamma_{12}^2+\gamma_{13}^2+\gamma_{23}^2-1)
\end{equation}
% MATLAB test0437 posneg1
or, equivalently, the circumgyrocircle existence existence condition
\begin{equation} \label{rjksdb}
4(\gammaab-1)(\gammaac-1) ~>~ (\gammaab+\gammaac-\gammabc-1)^2
\,,
\end{equation}
%MATHEMATICA stam285a "efes"  (stam378 for gyrotetrahedron)
or, equivalently, the circumgyrocircle existence condition
\begin{equation} \label{rjksdc}
D_3~>~H_3
\,.
\end{equation}
\end{subequations}

Gamma factors of gyrotriangle side gyrolengths
are related to its gyroangles
by the $AAA$ to $SSS$ {\it Conversion Law} \cite[Theorem 6.5, p.~137]{mybook05}
%%%%%%%%%%%%%%%%%%%%%%%%%%%%%%%%%%%%%%%%%%%%%%%%%%%%%%%%%%%%%%%%%%%%
\begin{equation} \label{jsf04st}
\begin{split}
\gamma_{23}^{\phantom{O}} &= \frac{\cos\alpha_1+\cos\alpha_2\cos\alpha_3}{\sin\alpha_2\sin\alpha_3}\\[4pt]
\gamma_{13}^{\phantom{O}} &= \frac{\cos\alpha_2+\cos\alpha_1\cos\alpha_3}{\sin\alpha_1\sin\alpha_3}\\[4pt]
\gamma_{12}^{\phantom{O}} &= \frac{\cos\alpha_3+\cos\alpha_1\cos\alpha_2}{\sin\alpha_1\sin\alpha_2}
\,.
\end{split}
\end{equation}
%%%%%%%%%%%%%%%%%%%%%%%%%%%%%%%%%%%%%%%%%%%%%%%%%%%%%%%%%%%%%%%%%%%%

Substituting these from \eqref{jsf04st}
into \eqref{hdsjb07bb} we obtain
%%%%%%%%%%%%%%%%%%%%%%%%%%%%%%%%%%%%%%%%%%%%%%%%%%%%%%%%%%%%%%%%%%%%
\begin{equation} \label{hdsjb09}
\begin{split}
m_1^\prime &= F^\prime \sin( \frac{         -  \alpha_1 + \alpha_2 + \alpha_3}{2} ) \sin\alpha_1
\\[4pt]
m_2^\prime &= F^\prime \sin( \frac{\phantom{-} \alpha_1 - \alpha_2 + \alpha_3}{2} ) \sin\alpha_2
\\[4pt]
m_3^\prime &= F^\prime \sin( \frac{\phantom{-} \alpha_1 + \alpha_2 - \alpha_3}{2} ) \sin\alpha_3
\,,
\end{split}
\end{equation}
%%%%%%%%%%%%%%%%%%%%%%%%%%%%%%%%%%%%%%%%%%%%%%%%%%%%%%%%%%%%%%%%%%%%
where the common factor $F^\prime = F^\prime(\alpha_1,\alpha_2,\alpha_3)$
in \eqref{hdsjb09} is given by the equation
\begin{equation} \label{hdsjb10}
F^\prime=2^3\frac{
\cos^2( \frac{\alpha_1 + \alpha_2 + \alpha_3}{2} )
\cos(\frac{-\alpha_1+\alpha_2+\alpha_3}{2} )
\cos(\frac{ \alpha_1-\alpha_2+\alpha_3}{2} )
\cos(\frac{ \alpha_1+\alpha_2-\alpha_3}{2} )
}{
\sin\alpha_1 \sin\alpha_2 \sin\alpha_3
}
\,.
\end{equation}
% MATHEMATICA stam102a, stam102b, stam102c.
% MATLAB fig241en2am

Since in gyrobarycentric coordinates only ratios of coordinates are
relevant, the gyrobarycentric coordinates,
$m_1^\prime,m_2^\prime$, and $m_3^\prime$ in
\eqref{hdsjb09} can be simplified by removing a common nonzero factor.
Hence, convenient gyrobarycentric coordinates,
$m_1^{\prime\prime},m_2^{\prime\prime}$, and $m_3^{\prime\prime}$, of the circumgyrocenter $O$
of gyrotriangle $A_1A_2A_3$, expressed in terms of the gyrotriangle gyroangles
are given by the equations
\index{circumgyrocenter}
%%%%%%%%%%%%%%%%%%%%%%%%%%%%%%%%%%%%%%%%%%%%%%%%%%%%%%%%%%%%%%%%%%%%
\begin{equation} \label{hdsjb11}
\begin{split}
m_1^{\prime\prime} &=-\sin( \frac{         -  \alpha_1 + \alpha_2 + \alpha_3}{2} ) \sin\alpha_1
=\cos(\alpha_1+\tfrac{\delta}{2}) \sin\alpha_1
\\[4pt]
m_2^{\prime\prime} &=-\sin( \frac{\phantom{-} \alpha_1 - \alpha_2 + \alpha_3}{2} ) \sin\alpha_2
=\cos(\alpha_2+\tfrac{\delta}{2}) \sin\alpha_2
\\[4pt]
m_3^{\prime\prime} &=-\sin( \frac{\phantom{-} \alpha_1 + \alpha_2 - \alpha_3}{2} ) \sin\alpha_3
=\cos(\alpha_3+\tfrac{\delta}{2}) \sin\alpha_3
\,,
\end{split}
\end{equation}
% MATHEMATICA stam272
%%%%%%%%%%%%%%%%%%%%%%%%%%%%%%%%%%%%%%%%%%%%%%%%%%%%%%%%%%%%%%%%%%%%
where $\delta=\pi-\alpha_1-\alpha_2-\alpha_3$ is the defect of gyrotriangle $A_1A_2A_3$.

The circumgyrocenter $O$,
\eqref{hdsjb01}, lies in the interior of its gyrotriangle $A_1A_2A_3$ if and only if
its gyrobarycentric coordinates are all positive or all negative.
Hence, we see from the
gyrobarycentric coordinates \eqref{hdsjb11} of $O$
that
\begin{enumerate}
\item
the circumgyrocenter $O$ lies in the
interior of its gyrotriangle $A_1A_2A_3$ if and only if the largest gyroangle
of the gyrotriangle has measure less than the sum of the measures of the other
two gyroangles.
This result is known in hyperbolic geometry; see, for instance,
\cite[p.~132]{kelly81}, where the result is proved synthetically.
Similarly, we also see from the gyrobarycentric coordinates
\eqref{hdsjb11} of $O$ in \eqref{hdsjb01} that
\item
the circumgyrocenter $O$ lies in the interior of its gyrotriangle $A_1A_2A_3$
if and only if all the three gyroangles
$\alpha_1+\delta/2$, $\alpha_2+\delta/2$ and $\alpha_3+\delta/2$
are acute.
\end{enumerate}

Expressing Inequality \eqref{rjksd} gyrotrigonometrically,
by means of \eqref{jsf04st}, it can be shown
from \eqref{gkdnseu} that $m_O^2>0$ if and only if
\begin{equation} \label{hurdms}
\begin{split}
\cos\frac{ 3\alpha_1-\alpha_2-\alpha_3 } {2}
+
\cos\frac{-\alpha_1+3\alpha_2-\alpha_3 } {2}
&+
\cos\frac{-\alpha_1 - \alpha_2+3\alpha_3} {2}
\\[4pt]
&>
3\cos\frac{\alpha_1+\alpha_2+\alpha_3}{2}
\end{split}
\end{equation}
% MATLAB fig269.m (condi and trigcondi).
or, equivalently, if and only if
\begin{equation} \label{hurdmst}
\sin(2\alpha_1+\tfrac{\delta}{2}) +
\sin(2\alpha_2+\tfrac{\delta}{2}) +
\sin(2\alpha_3+\tfrac{\delta}{2}) > 3\sin\tfrac{\delta}{2}
\end{equation}
% MATHEMATICA stam272
% MATHEMATICA stam276

Inequality \eqref{hurdmst} is an elegant condition for the existence of
a circumgyrocenter.
In a different approach, we will discover below in \eqref{hurdmsts} the
circumgyrocenter existence condition\index{circumgyrocenter existence condition}
in a different elegant form.

Gyrotrigonometric substitutions into the second equation in \eqref{gkdnseu}
from \cite[Sect.~7.12]{mybook05}
yield the following gyrotrigonometric expression for the
constant $\mO>0$ of the circumgyrocenter gyrobarycentric representation \eqref{hdsjb01}:
\begin{equation} \label{hdshsn}
\begin{split}
m_O^2 &= \frac{4^2F^2}{\sin^2\alpha_1 \sin^2\alpha_2 \sin^2\alpha_3}
\left\{
\frac{4^2F^2}{\sin^2\alpha_1 \sin^2\alpha_2 \sin^2\alpha_3}
-\frac{4^2F\sin^2\tfrac{\delta}{2}}{\sin^2\alpha_1 \sin^2\alpha_2 \sin^2\alpha_3}
\right\}
\\[6pt] &=
\frac{4^4F^3}{\sin^4\alpha_1 \sin^4\alpha_2 \sin^4\alpha_3}
(F-\sin^2\tfrac{\delta}{2})
\\[6pt] &=
\frac{4^4F^3\sin\tfrac{\delta}{2}}
{\sin^4\alpha_1 \sin^4\alpha_2 \sin^4\alpha_3}
\{
\sin(\alpha_1+\tfrac{\delta}{2})
\sin(\alpha_2+\tfrac{\delta}{2})
\sin(\alpha_3+\tfrac{\delta}{2})
-\sin\tfrac{\delta}{2}
\}
\,.
\end{split}
\end{equation}

Hence, $m_O^2>0$ if and only if
$F>\sin^2(\delta/2)$ or, equivalently,
\begin{equation} \label{hurdmsts}
\sin(\alpha_1+\tfrac{\delta}{2})
\sin(\alpha_2+\tfrac{\delta}{2})
\sin(\alpha_3+\tfrac{\delta}{2})
>\sin\tfrac{\delta}{2}
\,.
\end{equation}
%MATHEMATICA stam272 ey6
%MATLAB fig314E.m condition

One may demonstrate directly that the two
circumgyrocenter existence conditions
\eqref{hurdmst} and \eqref{hurdmsts}
are equivalent.

Formalizing the main results of this section, we have the following theorem:

%%%%%%%%%%%%%%%%%%%%%%%%%%%%%%%%%%%%%%%%%%%%%%%%%%%%%%%%%%%%%%%%%%%%
% THEOREM NUMBER 8.3
\begin{theorem}\label{thmtivhvn}
{\bf (Circumgyrocenter Theorem).}\index{circumgyrocenter theorem}
Let $S=\{ A_1,A_2,A_3\}$ be a gyrobarycentrically independent set of three points
in an Einstein gyrovector space $(\Rsn,\op,\od)$.
The circumgyrocenter\index{circumgyrocenter}
$O\inn\Rn$ of gyrotriangle $A_1A_2A_3$, shown in Fig.~\ref{fig241en2am},
possesses the gyrobarycentric representation
\begin{equation} \label{eicksfnein}
O = \frac{
 m_1 \gamma_{_{A_1}}^{\phantom{1}} A_1 +
 m_2 \gamma_{_{A_2}}^{\phantom{1}} A_2 +
 m_3 \gamma_{_{A_3}}^{\phantom{1}} A_3
}{
m_1 \gamma_{_{A_1}}^{\phantom{1}} +
m_2 \gamma_{_{A_2}}^{\phantom{1}} +
m_3 \gamma_{_{A_3}}^{\phantom{1}}
}
\end{equation}
with respect to the set $S=\{A_1,A_2,A_3\}$, with
gyrobarycentric coordinates $(m_1:m_2:m_3)$ given by
%%%%%%%%%%%%%%%%%%%%%%%%%%%%%%%%%%%%%%%%%%%%%%%%%%%%%%%%%%%%%%%%%%%%
\begin{equation} \label{hfjdv07bb}
\begin{split}
m_1 &=
(\phantom{-} \gamma_{12}^{\phantom{1}} + \gamma_{13}^{\phantom{1}} - \gamma_{23}^{\phantom{1}} -1)
(\gamma_{23}^{\phantom{1}} -1)
\\[4pt]
m_2 &=
(\phantom{-} \gamma_{12}^{\phantom{1}} - \gamma_{13}^{\phantom{1}} + \gamma_{23}^{\phantom{1}} -1)
(\gamma_{13}^{\phantom{1}} -1)
\\[4pt]
m_3 &=
(         -  \gamma_{12}^{\phantom{1}} + \gamma_{13}^{\phantom{1}} + \gamma_{23}^{\phantom{1}} -1)
(\gamma_{12}^{\phantom{1}} -1)
\end{split}
\end{equation}
%%%%%%%%%%%%%%%%%%%%%%%%%%%%%%%%%%%%%%%%%%%%%%%%%%%%%%%%%%%%%%%%%%%%
or, equivalently, by the gyrotrigonometric gyrobarycentric coordinates
%%%%%%%%%%%%%%%%%%%%%%%%%%%%%%%%%%%%%%%%%%%%%%%%%%%%%%%%%%%%%%%%%%%%
\begin{equation} \label{fjvdnw1}
\begin{split}
m_1 &=   \cos(\alpha_1+\tfrac{\delta}{2}) \sin\alpha_1
\\[4pt]
m_2 &=   \cos(\alpha_2+\tfrac{\delta}{2}) \sin\alpha_2
\\[4pt]
m_3 &=   \cos(\alpha_3+\tfrac{\delta}{2}) \sin\alpha_3
\,.
\end{split}
\end{equation}
%%%%%%%%%%%%%%%%%%%%%%%%%%%%%%%%%%%%%%%%%%%%%%%%%%%%%%%%%%%%%%%%%%%%

The circumgyrocenter gyrobarycentric representation constant $\mO$ with
respect to the set $S=\{A_1,A_2,A_3\}$
is an elegant product of two factors, given by the equation
\index{circumgyrocenter constant}
\begin{equation} \label{gkdnsfh}
\begin{split}
m_O^2 = &\{ (\gammaab+\gammaac+\gammabc-1)^2
-2 (\gamma_{12}^2+\gamma_{13}^2+\gamma_{23}^2-1) \}
\\[4pt] & \times
(1+2\gammaab\gammaac\gammabc - \gamma_{12}^2 - \gamma_{13}^2 - \gamma_{23}^2)
=D_3(D_3-H_3)
\,,
\end{split}
\end{equation}
% Calculated in MATHEMATICA stam133.
% Corroborated numerically for posneg in test0437 (posv3, see m0s1).
where $D_3$ and $H_3$ are given by \eqref{dethkcf} and \eqref{dethkdg}.

The circumgyrocenter\index{circumgyrocenter} lies in the ball, $O\in\Rsn$,
if and only if $m_O^2>0$ or, equivalently, if and only if
one of the following mutually equivalent inequalities, each of which is
a circumgyrocircle existence condition,\index{circumgyrocircle existence condition} is satisfied:
\begin{equation} \label{hurdmsw}
\sin(2\alpha_1+\tfrac{\delta}{2}) +
\sin(2\alpha_2+\tfrac{\delta}{2}) +
\sin(2\alpha_3+\tfrac{\delta}{2}) > 3\sin\tfrac{\delta}{2}
% MATHEMATICA stam272
% MATHEMATICA stam276
\end{equation}
\begin{equation} \label{hurdmstc}
\sin(\alpha_1+\tfrac{\delta}{2})
\sin(\alpha_2+\tfrac{\delta}{2})
\sin(\alpha_3+\tfrac{\delta}{2})
>\sin\tfrac{\delta}{2}
\end{equation}
\begin{equation} \label{hurdvk1}
(\gammaab+\gammaac+\gammabc-1)^2 ~>~ 2 (\gamma_{12}^2+\gamma_{13}^2+\gamma_{23}^2-1)
\end{equation}
\begin{equation} \label{hurdvk2}
4(\gammaab-1)(\gammaac-1) ~>~ (\gammaab+\gammaac-\gammabc-1)^2
\,.
\end{equation}

\end{theorem}
%%%%%%%%%%%%%%%%%%%%%%%%%%%%%%%%%%%%%%%%%%%%%%%%%%%%%%%%%%%%%%%%%%%%

%SECTION NUMBER 5
\section{Triangle Circumcenter}\label{dkncd2}
\index{circumgyrocenter}
%%%%%%%%%%%%%%%%%%%%%%%%%%%%%%%%%%%%%%%%%%%%%%%%%%%%%%%%%%%%%%%%%%%%

In this section the gyrotriangle circumgyrocenter in
Fig.~\ref{fig241en2am}, p.~\pageref{fig241en2am},
will be translated into its Euclidean counterpart in
Fig.~\ref{fig241euc1m}.

%%%%%%%%%%%%%%%%%%%%%%%%%%%%%%%%%%%%%%%%%%%%%%%%%%%%%%%%%%%%%%%%%%%%
% FIGURE 3
 
%%%%%%%%%%%%%%%%%%%%%%%%%%%%%%%%%%%%%%%%%%%%%%%%%%%%%%%%%%%%%%%%%%%%%%
%%%%% The Euclidean circumcenter                    %%%%%%%%%%%%%%%%%%
%\begin{figure}[htbp]
\begin{figure}[t]  % try to put this figure on the top of the page
              % [h] tries to place the figure here
              % [b] tries to place the figure on the bottom of the page
              % [t] tries to place the figure on the top of the page
              % [P] tries to place the figure floatingly on the page
 \centering         % center the figure
\psfrag{O}[]{$\phantom{O}$}
\psfrag{A1}[]{$A_1$}
\psfrag{A2}[]{$A_2$}
\psfrag{A3}[]{$A_3$}
\psfrag{O}[]{$O$}
\psfrag{pa1}{$a_{23}$}
\psfrag{pa2}[]{$a_{31}=a_{13}$}
\psfrag{pa3}{$a_{12}$}
\psfrag{text1}[]{$\ab_{23}= -  A_2  +  A_3$}
\psfrag{text2}[]{$\ab_{31}= -  A_3  +  A_1$}
\psfrag{text3}[]{$\ab_{12}= -  A_1  +  A_2$}
\psfrag{al1}[]{$\alpha_1$}
\psfrag{al2}[]{$\alpha_2$}
\psfrag{al3}[]{$\alpha_3$}
\psfrag{formula01}{$a_{12}=\|\ab_{12}\|=\| -  A_1  +  A_2\|$}
\psfrag{formula02}{$a_{13}=\|\ab_{31}\|=\| -  A_3  +  A_1\|$}
\psfrag{formula03}{$a_{23}=\|\ab_{23}\|=\| -  A_2  +  A_3\|$}
\psfrag{formula04}[]{$\cos\alpha_1=\frac{ -  A_1 +  A_2}{\| -  A_1 +  A_2\|}
\ccdot \frac{ -  A_1 +  A_3}{\| -  A_1 +  A_3\|}$}
\psfrag{formula05}[]{$\cos\alpha_2=\frac{ -  A_2 +  A_1}{\| -  A_2 +  A_1\|}
\ccdot \frac{ -  A_2 +  A_3}{\| -  A_2 +  A_3\|}$}
\psfrag{formula06}[]{$\cos\alpha_3=\frac{ -  A_3 +  A_1}{\| -  A_3 +  A_1\|}
\ccdot \frac{ -  A_3 +  A_2}{\| -  A_3 +  A_2\|}$}
\psfrag{fig241}{$\hspace{0.8cm}
O = \frac{
\sum_{k=1}^{3} m_k A_k
}{
\sum_{k=1}^{3} m_k
}
$}
%%%%%%%%%%%%%%%%%%%%%%%%%%%%%%%%%%%%%%%%%%%%%%%%%%%%%%%%%%%%%%%%%%%%%%
%\psfrag{----chets1}[]{\lower-1.2ex \hbox {$\blacktriangleright$}}
%\psfrag{----chets2}[]{\lower-1.26ex \hbox {$\blacktriangleright$}}
%\psfrag{----chets3}[]{\lower-1.2ex \hbox {$\blacktriangleright$}}
%%%%%%%%%%%%%%%%%%%%%%%%%%%%%%%%%%%%%%%%%%%%%%%%%%%%%%%%%%%%%%%%%%%%%%
%\includegraphics[width=10cm]{/home/ungar/dir_amy/dir_papers/dir_mybook01/dir_figs/fig241euc1.eps}
 \includegraphics[width=10cm]{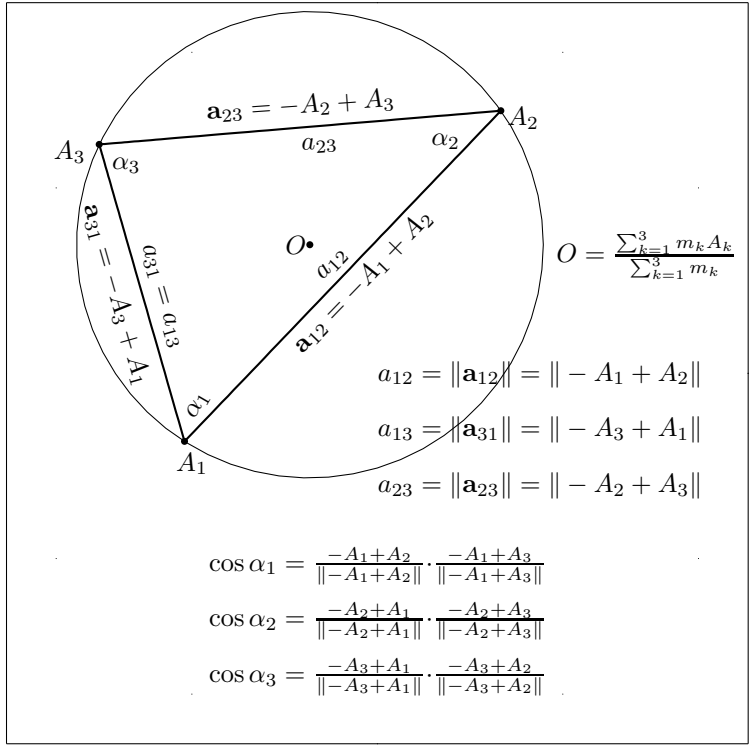}
\caption[A triangle circumcircle]{
The circumcircle, and
the circumcenter $O$, of triangle $A_1A_2A_3$ in a
Euclidean vector space $\Rn$, $n=2$, is shown along with its standard
notation.
Here $R=\| -  A_1  +  O\|$ $=$ $\| -  A_2  +  O\|$ $=$ $\| -  A_3  +  O\|$,
where $R$ is the triangle circumradius\index{circumradius}\index{circumcenter}
and $O$ is the triangle circumcenter, given by its
barycentric coordinate representation
\eqref{rukdis}
with respect to the barycentrically independent set $S=\{A_1,A_2,A_3\}$.
The hyperbolic counterpart of this figure is shown in
Fig.~\ref{fig241en2am}, p.~\pageref{fig241en2am}.
%See also Fig.~\ref{fig241euc2m}, p.~\pageref{fig241euc2m}.
\label{fig241euc1m}}
\end{figure}
%%%%%%%%%%%%%%%%%%%%%%%%%%%%%%%%%%%%%%%%%%%%%%%%%%%%%%%%%%%%%%%%%%%%%%
 % The triangle circumcenter in (+)E
%                       Fig.~\ref{fig241euc1m}
%%%%%%%%%%%%%%%%%%%%%%%%%%%%%%%%%%%%%%%%%%%%%%%%%%%%%%%%%%%%%%%%%%%%

The gyrobarycentric representation \eqref{eicksfnein}
with gyrotrigonometric gyrobarycentric coordinates $(m_1:m_2:m_3)$
given by \eqref{fjvdnw1} remains invariant in form under the
Euclidean limit $s\rightarrow\infty$, so that it is valid in
Euclidean geometry as well, where $\delta=0$.
Hence, in the transition from hyperbolic geometry, where $\delta>0$,
to Euclidean geometry, where $\delta=0$,
the gyrobarycentric coordinates \eqref{fjvdnw1} reduce to the
barycentric coordinates
%%%%%%%%%%%%%%%%%%%%%%%%%%%%%%%%%%%%%%%%%%%%%%%%%%%%%%%%%%%%%%%%%%%%
\begin{equation} \label{hkdfn3}
\begin{split}
m_1 &= \cos\alpha_1 \sin\alpha_1 = \half \sin2\alpha_1
\\[6pt]
m_2 &= \cos\alpha_2 \sin\alpha_2 = \half \sin2\alpha_2
\hspace{1.2cm} (Euclidean~Geometry)
\\[6pt]
m_3 &= \cos\alpha_3 \sin\alpha_3 = \half \sin2\alpha_3
\,.
\end{split}
\end{equation}
%%%%%%%%%%%%%%%%%%%%%%%%%%%%%%%%%%%%%%%%%%%%%%%%%%%%%%%%%%%%%%%%%%%%

Hence, finally, a trigonometric barycentric representation of the circumcenter $O$
of triangle $A_1A_2A_3$ in $\Rn$, Fig.~\ref{fig241euc1m},
with respect to the barycentrically independent set $S=\{A_1,A_2,A_3\}\subset\Rn$
is given by Result \eqref{rukdis} of the following corollary of Theorem \ref{thmtivhvn},
which recovers a well-known result in Euclidean geometry \cite{kimberlingweb}.

% COROLLARY NUMBER 8.4
\begin{corollary}\label{vdmfcb}
Let $\alpha_k$, $k=1,2,3$ and $O$ be the angles and circumcenter
of a triangle $A_1A_2A_3$ in a Euclidean space $\Rn$.
Then,
\begin{equation} \label{rukdis}
O = \frac{
\sin2\alpha_1 A_1 + \sin2\alpha_2 A_2 + \sin2\alpha_3 A_3
}{
\sin2\alpha_1     + \sin2\alpha_2     + \sin2\alpha_3
}
%\hspace{1.2cm} (Euclidean~Geometry)
\,.
% MATLAB fig241euc1  zero1
\end{equation}
\end{corollary}

Theorem \ref{thmtivhvn} and its Corollary \ref{vdmfcb} form an elegant example
that illustrates the result that
\begin{enumerate}
\item \label{avigd1}
gyrotrigonometric gyrobarycentric coordinates of a point in
an Einstein gyrovector space $\Rsn$ survive unimpaired in
Euclidean geometry, where they form
\item \label{avigd2}
trigonometric barycentric coordinates of a point in
a corresponding Euclidean vector space $\Rn$.

The converse is, however, not valid since
\item \label{avigd3}
trigonometric barycentric coordinates of a point in a Euclidean vector space $\Rn$
may embody the Euclidean condition that the
triangle angle sum in $\pi$, so that they need not survive in
hyperbolic geometry.
\end{enumerate}

%SECTION NUMBER 6
\section{Gyrotriangle Circumgyroradius}\label{drymv2}
\index{circumgyroradius}

The circumgyroradius $R$ of gyrotriangle $A_1A_2A_3$ with circumgyrocenter $O$
in an Einstein gyrovector space $(\Rsn,\op,\od)$,
shown in Figs.~\ref{fig241en2am} and \ref{fig241en5m},
is given by
\begin{equation} \label{vdknr}
R = \|\om A_1 \op O\| = \|\om A_2 \op O\| = \|\om A_3 \op O\|
\,.
\end{equation}
By \eqref{hdsjb02}, the circumgyroradius $R$ satisfies the equation
\begin{equation} \label{fskvm1}
\gamma_{_{R}}^{\phantom{1}} =
\gamma_{_{\|\om A_1 \op O\|}}^{\phantom{1}} =
\gamma_{_{\om A_1 \op O}}^{\phantom{1}} = \frac{
m_1+m_2\gammaab + m_3\gammaac}{\mO}
\,,
\end{equation}
where $m_1,m_2$ and $m_3$ are given by \eqref{hfjdv07bb}, and where $\mO$
is given by \eqref{gkdnsfh}.

Hence, following \eqref{fskvm1}, \eqref{hfjdv07bb} and \eqref{gkdnsfh},
with the notation for
$D_3$ and $H_3$ in \eqref{dethkcf}\,--\,\eqref{dethkdg},
we have
\begin{equation} \label{fskvm2}
\begin{split}
\gamma_R^2 &= \frac{
1+2\gammaab\gammaac\gammabc - \gamma_{12}^2-\gamma_{13}^2-\gamma_{23}^2
}{
1+2\gammaab\gammaac\gammabc - \gamma_{12}^2-\gamma_{13}^2-\gamma_{23}^2
- 2(\gammaab-1)(\gammaac-1)(\gammabc-1)
}
\\[6pt] &= \frac{
2\gammaab\gammaac\gammabc - (\gamma_{12}^2+\gamma_{13}^2+\gamma_{23}^2-1)
}{
(\gammaab+\gammaac+\gammabc-1)^2 - 2(\gamma_{12}^2+\gamma_{13}^2+\gamma_{23}^2-1)
}
\\[6pt] &
= \frac{D_3}{D_3-H_3}
\,.
\end{split}
\end{equation}
% MATHEMATICA stam143
% MATLAB fig314E.m zerogrs1,2; fig318enE1,zerogrs0

Interestingly, from \eqref{fskvm2} and \eqref{drekc} we obtain an elegant
relationship between
(i) the constant $\mO$ of the gyrobarycentric representation
of the circumgyrocenter $O$ of a gyrotriangle $A_1A_2A_3$
with respect to the gyrotriangle vertices, and
(ii) the gamma factor of the circumgyroradius $R$,
\begin{equation} \label{turkb}
\mO \gamma_R^{\phantom{O}} = D_3
\,.
\end{equation}

Following \eqref{fskvm2} we have, by \eqref{rugh1ds},
\begin{equation} \label{fskvm3}
R^2 = s^2 \frac{\gamma_R^2-1}{\gamma_R^2} = 2s^2 \frac{
(\gammaab-1) (\gammaac-1) (\gammabc-1)
}{
1+2\gammaab\gammaac\gammabc - \gamma_{12}^2-\gamma_{13}^2-\gamma_{23}^2
}
=s^2 \frac{H_3}{D_3}
\,.
% MATLAB fig318enE1, zerors0
\end{equation}

Hence, finally, the circumgyroradius $R$ of gyrotriangle $A_1A_2A_3$ in
Figs.~\ref{fig241en2am} and \ref{fig241en5m}
is given by
\begin{equation} \label{fskvm4}
R = \sqrt{2} s ~\sqrt{ \frac{
(\gammaab-1) (\gammaac-1) (\gammabc-1)
}{
1+2\gammaab\gammaac\gammabc - \gamma_{12}^2-\gamma_{13}^2-\gamma_{23}^2
}}
~~,
\end{equation}
%MATLAB fig241en5.m - zerors=0.
implying
\begin{equation} \label{fskvm5}
\sqrt{ \frac{
(\gammaab+1) (\gammaac+1) (\gammabc+1)}{2}} ~R
= s~ \sqrt{ \frac{
(\gamma_{12}^2-1) (\gamma_{13}^2-1) (\gamma_{23}^2-1)
}{
1+2\gammaab\gammaac\gammabc - \gamma_{12}^2-\gamma_{13}^2-\gamma_{23}^2
}}
~~.
\end{equation}

%%%%%%%%%%%%%%%%%%%%%%%%%%%%%%%%%%%%%%%%%%%%%%%%%%%%%%%%%%%%%%%%%%%%
% FIGURE 4
  
%%%%%%%%%%%%%%%%%%%%%%%%%%%%%%%%%%%%%%%%%%%%%%%%%%%%%%%%%%%%%%%%%%%%%%
%%%%% The hyperbolic semi-circle theorem            %%%%%%%%%%%%%%%%%%
%\begin{figure}[htbp]
\begin{figure}[t]  % try to put this figure on the top of the page
 \centering         % center the figure
%
%%%%%%%%%%%%%%%%%%%%%%%%%%%%%%%%%%%%%%%%%%%%%%%%%%%%%%%%%%%%%%%
\psfrag{O}[]{$\phantom{O}$}
\psfrag{A1}[]{$A_1$}
\psfrag{A2}[]{$A_2$}
\psfrag{A3}[]{$A_3$}
\psfrag{O}[]{$O$}
\psfrag{pa1}{$a_{23}$}
\psfrag{pa2}[]{$a_{31}=a_{13}$}
\psfrag{pa3}{$a_{12}$}
\psfrag{text1}[]{$\ab_{23}=\om A_2 \op A_3$}
\psfrag{text2}[]{$\ab_{31}=\om A_3 \op A_1$}
\psfrag{text3}[]{$\ab_{12}=\om A_1 \op A_2,~~~\gamma_{12}=\gamma_{\ab_{12}}$}
\psfrag{al1}[]{$\alpha_1$}
\psfrag{al2}[]{$\alpha_2$}
\psfrag{al3}[]{$\alpha_3$}
 \psfrag{formula00}[]{$R=\|\om A_1\op O\|=\|\om A_2\op O\|=\|\om A_3\op O\|$}
%\psfrag{formula01}{$a_{12}=\|\ab_{12}\|=\|\om A_1 \op A_2\|$}
%\psfrag{formula02}{$a_{13}=\|\ab_{31}\|=\|\om A_3 \op A_1\|$}
%\psfrag{formula03}{$a_{23}=\|\ab_{23}\|=\|\om A_2 \op A_3\|$}
%\psfrag{formula04}[]{$\cos\alpha_1=\frac{\om A_1\op A_2}{\|\om A_1\op A_2\|}
%\ccdot \frac{\om A_1\op A_3}{\|\om A_1\op A_3\|}$}
%\psfrag{formula05}[]{$\cos\alpha_2=\frac{\om A_2\op A_1}{\|\om A_2\op A_1\|}
%\ccdot \frac{\om A_2\op A_3}{\|\om A_2\op A_3\|}$}
%\psfrag{formula06}[]{$\cos\alpha_3=\frac{\om A_3\op A_1}{\|\om A_3\op A_1\|}
%\ccdot \frac{\om A_3\op A_2}{\|\om A_3\op A_2\|}$}
%\psfrag{fig241}{$\hspace{0.8cm}
%O = \frac{
%\sum_{k=1}^{3} m_k \gamma_{_{A_k}}^{\phantom{1}} A_k
%}{
%\sum_{k=1}^{3} m_k \gamma_{_{A_k}}^{\phantom{1}}
%}
%$}
%%%%%%%%%%%%%%%%%%%%%%%%%%%%%%%%%%%%%%%%%%%%%%%%%%%%%%%%%%%%%%%
%
%\includegraphics[width=9cm]{/home/ungar/dir_amy/dir_papers/dir_mybook01/dir_figs/fig241en5.eps}
 \includegraphics[width=9cm]{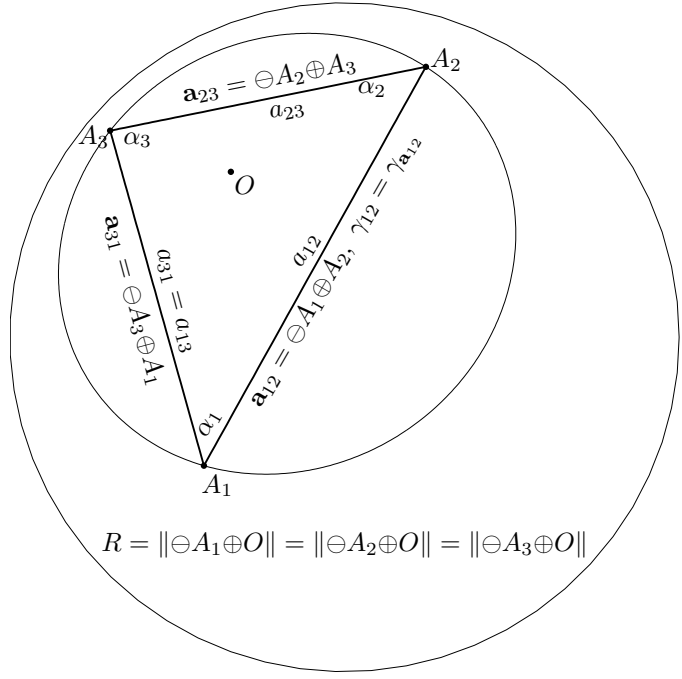}
\caption[The gyrotriangle circumgyrocircle, circumgyroradius]{
The circumgyrocircle of gyrotriangle $A_1A_2A_3$ in an Einstein gyrovector space
$(\Rsn,\op,\od)$ is shown for $n=2$. Its gyrocenter, $O$,
is the gyrotriangle circumgyrocenter,
given by its gyrobarycentric representation \eqref{eicksfnein}, p.~\pageref{eicksfnein},
and its gyroradius $R$ is the
gyrotriangle circumgyroradius, given by each of the equations
$R=\|\om A_k\op O\|$, $k=1,2,3$.
The gyrocircle is a flattened Euclidean circle, as shown in Fig.~\ref{fig266am}.
\label{fig241en5m}}
\end{figure}
%%%%%%%%%%%%%%%%%%%%%%%%%%%%%%%%%%%%%%%%%%%%%%%%%%%%%%%%%%%%%%%%%%%%%%
 % The circumgyrocircle
%   Fig.~\ref{fig241en5m}
%%%%%%%%%%%%%%%%%%%%%%%%%%%%%%%%%%%%%%%%%%%%%%%%%%%%%%%%%%%%%%%%%%%n%

Identity \eqref{fskvm5} captures a remarkable analogy between the
law of gyrosines and the law of sines.
Indeed, following \eqref{fskvm5}, the law of gyrosines \cite[Theorem 6.9, p.~140]{mybook05}
for gyrotriangle $A_1A_2A_3$ in Fig.~\ref{fig241en5m} is linked to the
circumgyroradius $R$ of the gyrotriangle by the equation
\index{law of gyrosines, circumgyroradius}
\begin{equation} \label{fskvm6}
\frac{\gammabc a_{23}}{\sin\alpha_1} =
\frac{\gammaac a_{13}}{\sin\alpha_2} =
\frac{\gammaab a_{12}}{\sin\alpha_3} =
\sqrt{ \frac{
(\gammaab+1) (\gammaac+1) (\gammabc+1)}{2}} ~R
\,,
\end{equation}
%MATLAB fig241en5.m zerolaw1=0.
called the {\it extended law of gyrosines}.
\index{law of gyrosines, extended}

Following the gamma-gyrotrigonometric identity \cite[Eq.~(7.154), p.~189]{mybook05},
the extended law of gyrosines \eqref{fskvm6} can be written as
\begin{equation} \label{fskvm6a}
\frac{\gammabc a_{23}}{\sin\alpha_1} =
\frac{\gammaac a_{13}}{\sin\alpha_2} =
\frac{\gammaab a_{12}}{\sin\alpha_3} =
2\frac{
\sin(\alpha_1+\tfrac{\delta}{2})
\sin(\alpha_2+\tfrac{\delta}{2})
\sin(\alpha_3+\tfrac{\delta}{2})
}{
\sin\alpha_1 \sin\alpha_2 \sin\alpha_3
}
~R
\,,
\end{equation}
%MATLAB fig241en5.m zerolaw2=0.
where $\delta=\pi-(\alpha_1+\alpha_2+\alpha_3)$ is the defect of
gyrotriangle $A_1A_2A_3$.

In the Euclidean limit of large $s$, $s\rightarrow\infty$, gamma factors tend to 1
and gyrotriangle defects tend to 0. Hence, in that limit,
the extended law of gyrosines \eqref{fskvm6} tends to the
well-known extended law of sines \cite[p.~87]{maor98},
\begin{equation} \label{fskvm7}
\frac{a_{23}}{\sin\alpha_1} =
\frac{a_{13}}{\sin\alpha_2} =
\frac{a_{12}}{\sin\alpha_3} = 2R
\hspace{1.2cm} {\rm (Euclidean~Geometry)}
\,.
\end{equation}

Formalizing the results in \eqref{fskvm6}\,--\,\eqref{fskvm6a} we have
the following theorem:
\index{gyrosines law, extended}
% THEOREM NUMBER 8.5
\begin{theorem}\label{thmhfkbc}
{\bf (Extended Law of Gyrosines).}\index{law of gyrosines, extended}
Let $A_1A_2A_3$ be a gyrotriangle in an Einstein gyrovector space $(\Rsn\op,\od)$
with gyroangles $\alpha_1,\alpha_2,\alpha_3$, side-gyrolengths $a_{23},a_{13},a_{12}$,
and circumgyroradius $R$, Fig.~\ref{fig241en5m}.

Then
\begin{equation} \label{fskwn6}
\frac{\gammabc a_{23}}{\sin\alpha_1} =
\frac{\gammaac a_{13}}{\sin\alpha_2} =
\frac{\gammaab a_{12}}{\sin\alpha_3} =
\sqrt{ \frac{
(\gammaab+1) (\gammaac+1) (\gammabc+1)}{2}} ~R
\end{equation}
and
\begin{equation} \label{fskwn6a}
\frac{\gammabc a_{23}}{\sin\alpha_1} =
\frac{\gammaac a_{13}}{\sin\alpha_2} =
\frac{\gammaab a_{12}}{\sin\alpha_3} =
2\frac{
\sin(\alpha_1+\tfrac{\delta}{2})
\sin(\alpha_2+\tfrac{\delta}{2})
\sin(\alpha_3+\tfrac{\delta}{2})
}{
\sin\alpha_1 \sin\alpha_2 \sin\alpha_3
}
~R
\,.
\end{equation}
%MATLAB fig241en5.m zerolaw1,2=0.
\end{theorem}
\index{law of gyrosines, extended}

Interestingly, the gyrotriangle circumgyroradius $R$ has an elegant representation
in terms of its gyrotriangle gyroangles. Indeed, expressing the gamma factors
in \eqref{fskvm3} in terms of the gyrotriangle gyroangles $\alpha_k$, $k=1,2,3$,
takes the gyrotrigonometric form
\begin{equation} \label{fudsk}
\begin{split}
\frac{R^2}{s^2} &= \frac{\cos\frac{\alpha_1+\alpha_2+\alpha_3}{2}
}{
\cos\frac{ \alpha_1-\alpha_2-\alpha_3}{2}
\cos\frac{-\alpha_1+\alpha_2-\alpha_3}{2}
\cos\frac{-\alpha_1-\alpha_2+\alpha_3}{2}
}
\\[8pt] &= \frac{\sin\tfrac{\delta}{2}}{
\sin(\alpha_1+\tfrac{\delta}{2})
\sin(\alpha_2+\tfrac{\delta}{2})
\sin(\alpha_3+\tfrac{\delta}{2})
}
\\[8pt] &= \frac{F(\alpha_1,\alpha_2,\alpha_3)
}{
\sin^2(\alpha_1+\tfrac{\delta}{2})
\sin^2(\alpha_2+\tfrac{\delta}{2})
\sin^2(\alpha_3+\tfrac{\delta}{2})
}
\end{split}
% MATHEMATICA stam169
\end{equation}
in any Einstein gyrovector space $(\Rsn,\op,\od)$, where
$F=F(\alpha_1,\alpha_2,\alpha_3)$ is given by
\eqref{grfde} \cite[Eq.~(7.144), p.~187]{mybook05}.
Equation \eqref{fudsk} can be written conveniently as
\begin{equation} \label{fudskf}
\frac{R}{s} = \frac{\sqrt{F}}{
\sin(\alpha_1+\tfrac{\delta}{2})
\sin(\alpha_2+\tfrac{\delta}{2})
\sin(\alpha_3+\tfrac{\delta}{2})
}
\end{equation}

Following \eqref{fudskf}, the first equation in \cite[Eq.~(7.143), p.~187]{mybook05}
and \cite[Eq.~(7.150), p.~188]{mybook05}
we have the equation
\begin{equation} \label{dagim4}
\sin(\alpha_3+\tfrac{\delta}{2}) = \frac{
\gammaab a_{12}} {R(\gammaab+1)}
\,,
% MATLAB fig241en5 zerosine
\end{equation}
which will prove useful in \eqref{dugit01}, p.~\pageref{dugit01}.

In the Euclidean limit, $s\rightarrow\infty$, the gyrotriangle defect
tends to 0, $\delta\rightarrow0$, so that each side of
\eqref{fudsk} tends to 0.

An important relation that results from \eqref{fudsk} is formalized in the
following Theorem:

% THEOREM NUMBER 8.6
\begin{theorem}\label{thmhdfhc}
Let $\alpha_k$, $k=1,2,3$, and $R$ be the gyroangles and circumgyroradius of
a gyrotriangle $A_1A_2A_3$ in an Einstein gyrovector space $(\Rsn\op,\od)$.
Then
\begin{equation} \label{kmdnbe}
s^2\sin\tfrac{\delta}{2} = R^2
\sin(\alpha_1+\tfrac{\delta}{2})
\sin(\alpha_2+\tfrac{\delta}{2})
\sin(\alpha_3+\tfrac{\delta}{2})
\end{equation}
and
\begin{equation} \label{kmdnbe00}
s \sqrt{F(\alpha_1,\alpha_2,\alpha_3)}
=
R  \sin(\alpha_1+\tfrac{\delta}{2})
\sin(\alpha_2+\tfrac{\delta}{2}) \sin(\alpha_3+\tfrac{\delta}{2})
\end{equation}
where $F(\alpha_1,\alpha_2,\alpha_3)$ is given by
\eqref{grfde} \cite[Eq.~(7.144), p.~187]{mybook05}.
\end{theorem}
\begin{proof}
Identity \eqref{kmdnbe} follows immediately from \eqref{fudsk}, and
Identity \eqref{kmdnbe00} follows from \eqref{kmdnbe} and the
definition of $F$.
\end{proof}

The Euclidean limit, $s\rightarrow\infty$, of the
left-hand side of each of the equations \eqref{kmdnbe} and \eqref{kmdnbe00}
of Theorem \ref{thmhdfhc}
is indeterminate, being a limit of type $\infty\ccdot0$.
In that limit, the gyrotriangle defect $\delta$ tends to zero,
so that the right-hand side of \eqref{kmdnbe} tends to
$R^2\sin\alpha_1\sin\alpha_2\sin\alpha_3$.

Theorem \ref{thmhdfhc} gives rise to useful results,
as indicated by the following theorem.

% THEOREM NUMBER 8.7
\begin{theorem}\label{thmdmsk1}
Let $A_1A_2A_3$ be a gyrotriangle that possesses a circumgyroradius $R$
in an Einstein gyrovector space $(\Rsn\op,\od)$,
Fig.~\ref{fig241en2am}, p.~\pageref{fig241en2am}.
Then, in the gyrotriangle index notation \eqref{indexnotation},
%%%%%%%%%%%%%%%%%%%%%%%%%%%%%%%%%%%%%%%%%%%%%%%%%%%%%%%%%%%%%%%%%%%%
\begin{equation} \label{dkotie01}
\begin{split}
a_{12} &= \frac{
2R \sin(\alpha_1+\tfrac{\delta}{2})
\sin(\alpha_2+\tfrac{\delta}{2})
\sin(\alpha_3+\tfrac{\delta}{2})
}{
\cos\alpha_3+\cos\alpha_1\cos\alpha_2
}
\\[8pt]
a_{13} &= \frac{
2R \sin(\alpha_1+\tfrac{\delta}{2})
\sin(\alpha_2+\tfrac{\delta}{2})
\sin(\alpha_3+\tfrac{\delta}{2})
}{
\cos\alpha_2+\cos\alpha_1\cos\alpha_3
}
\\[8pt]
a_{23} &= \frac{
2R \sin(\alpha_1+\tfrac{\delta}{2})
\sin(\alpha_2+\tfrac{\delta}{2})
\sin(\alpha_3+\tfrac{\delta}{2})
}{
\cos\alpha_1+\cos\alpha_2\cos\alpha_3
}
\end{split}
\end{equation}
%%%%%%%%%%%%%%%%%%%%%%%%%%%%%%%%%%%%%%%%%%%%%%%%%%%%%%%%%%%%%%%%%%%%
and
%%%%%%%%%%%%%%%%%%%%%%%%%%%%%%%%%%%%%%%%%%%%%%%%%%%%%%%%%%%%%%%%%%%%
\begin{equation} \label{dkotie02}
\begin{split}
\gammaab a_{12} &= \frac{
2R \sin(\alpha_1+\tfrac{\delta}{2})
\sin(\alpha_2+\tfrac{\delta}{2})
\sin(\alpha_3+\tfrac{\delta}{2})
}{
\sin\alpha_1\sin\alpha_2
}
\\[8pt]
\gammaac a_{13} &= \frac{
2R \sin(\alpha_1+\tfrac{\delta}{2})
\sin(\alpha_2+\tfrac{\delta}{2})
\sin(\alpha_3+\tfrac{\delta}{2})
}{
\sin\alpha_1\sin\alpha_3
}
\\[8pt]
\gammabc a_{23} &= \frac{
2R \sin(\alpha_1+\tfrac{\delta}{2})
\sin(\alpha_2+\tfrac{\delta}{2})
\sin(\alpha_3+\tfrac{\delta}{2})
}{
\sin\alpha_2\sin\alpha_3
}
\,,
\end{split}
%MATLAB fig241en5 zeroij and zeroaij
\end{equation}
where $a_{ij}=\|\om A_i \op A_j\|$, $1\le i,j\le3$.
%%%%%%%%%%%%%%%%%%%%%%%%%%%%%%%%%%%%%%%%%%%%%%%%%%%%%%%%%%%%%%%%%%%%
\end{theorem}
\begin{proof}
The identities in \eqref{dkotie01} and \eqref{dkotie02} follow
from Identity \eqref{kmdnbe00} of Theorem \ref{thmhdfhc} and from
the $AAA$ to $SSS$ Conversion Law \cite[Theorem 6.8, p.~140]{mybook05}.
\end{proof}

Comparing the first equation in each of \eqref{dkotie01} and \eqref{dkotie02}
yields the gyrotriangle identity
\begin{equation} \label{dkotie02d5}
\cos\alpha_3 + \cos\alpha_1\cos\alpha_2 = \gammaab \sin\alpha_1\sin\alpha_2
\hspace{1.2cm} {\rm (Hyperbolic~Geometry)}
%MATLAB fig241en5 zeroevery
\end{equation}
that any gyrotriangle $A_1A_2A_3$
in an Einstein gyrovector space $(\Rsn,\op,\od)$ obeys.

With $\delta=0$, Theorem \ref{thmdmsk1} specializes to the following
corresponding theorem in Euclidean geometry:

% THEOREM NUMBER 8.8
\begin{theorem}\label{thmdmsk1euc}
Let $A_1A_2A_3$ be a triangle with a circumradius $R$
in a Euclidean vector space $\Rn$, Fig.~\ref{fig241euc1m}, p.~\pageref{fig241euc1m}.
Then, in the triangle index notation,
%%%%%%%%%%%%%%%%%%%%%%%%%%%%%%%%%%%%%%%%%%%%%%%%%%%%%%%%%%%%%%%%%%%%
\begin{equation} \label{dkotie01euc}
\begin{split}
a_{12} &= \frac{
2R \sin\alpha_1\sin\alpha_2\sin\alpha_3
}{
\cos\alpha_3+\cos\alpha_1\cos\alpha_2
}
\\[8pt]
a_{13} &= \frac{
2R \sin\alpha_1\sin\alpha_2\sin\alpha_3
}{
\cos\alpha_2+\cos\alpha_1\cos\alpha_3
}
\\[8pt]
a_{23} &= \frac{
2R \sin\alpha_1\sin\alpha_2\sin\alpha_3
}{
\cos\alpha_1+\cos\alpha_2\cos\alpha_3
}
\end{split}
\end{equation}
% MATLAB fig241euc1, zeroaij.
% MATLAB fig276a, zero4.
%%%%%%%%%%%%%%%%%%%%%%%%%%%%%%%%%%%%%%%%%%%%%%%%%%%%%%%%%%%%%%%%%%%%
and
%%%%%%%%%%%%%%%%%%%%%%%%%%%%%%%%%%%%%%%%%%%%%%%%%%%%%%%%%%%%%%%%%%%%
\begin{equation} \label{dkotie02euc}
\begin{split}
a_{12} &= \frac{
2R \sin\alpha_1\sin\alpha_2\sin\alpha_3
}{
\sin\alpha_1\sin\alpha_2
}
\\[8pt]
a_{13} &= \frac{
2R \sin\alpha_1\sin\alpha_2\sin\alpha_3
}{
\sin\alpha_1\sin\alpha_3
}
\\[8pt]
a_{23} &= \frac{
2R \sin\alpha_1\sin\alpha_2\sin\alpha_3
}{
\sin\alpha_2\sin\alpha_3
}
\,,
\end{split}
\end{equation}
%%%%%%%%%%%%%%%%%%%%%%%%%%%%%%%%%%%%%%%%%%%%%%%%%%%%%%%%%%%%%%%%%%%%
where $a_{ij}=\|-A_i + A_j\|$, $1\le i,j\le3$.
\end{theorem}

As expected from Identities \eqref{dkotie01euc} and \eqref{dkotie02euc}
of Theorem \ref{thmdmsk1euc},
the triangle angles obey the triangle identity $\alpha_1+\alpha_2+\alpha_3=\pi$,
so that
\begin{equation} \label{dukin1}
\cos\alpha_3 + \cos\alpha_1\cos\alpha_2 = \sin\alpha_1\sin\alpha_2
\hspace{1.2cm} {\rm (Euclidean~Geometry)}
\end{equation}
for any triangle $A_1A_2A_3$ in the Euclidean space $\Rn$.
Indeed, the triangle identity \eqref{dukin1} is the
Euclidean counterpart of the gyrotriangle identity \eqref{dkotie02d5}.

%SECTION NUMBER 7
\section{Triangle Circumradius}\label{drymf2}
\index{circumradius}

The circumradius $R_{euc}$ of triangle $A_1A_2A_3$ with circumcenter $O$ in
a Euclidean space $\Rn$,
shown in Fig.~\ref{fig335eucm}, p.~\pageref{fig335eucm},
is given by
\begin{equation} \label{vdknreuc}
R_{euc} = \|-A_1 + O\| = \|-A_2 + O\| = \|-A_3 + O\|
\,.
\end{equation}
We will determine the triangle circumradius in a Euclidean space $\Rn$
as the Euclidean limit of the circumgyroradius of gyrotriangle $A_1A_2A_3$
in an Einstein gyrovector space $(\Rsn,\op,\od)$, that is,
$R_{euc}=\lim_{s\rightarrow \infty} R$.
Accordingly, let $A_1A_2A_3$ be a gyrotriangle that possesses a
circumgyrocircle in the Einstein gyrovector space, and let $R$ be
the gyrotriangle circumgyroradius.

Then, $R$ is given by \eqref{fskvm3}, which can be written equivalently as
\begin{equation} \label{adamtk01}
R^2 = \frac{2s^2(\gammabc-1)}{D}
\end{equation}
where
%%%%%%%%%%%%%%%%%%%%%%%%%%%%%%%%%%%%%%%%%%%%%%%%%%%%%%%%%%%%%%%%%%%%
\begin{equation} \label{adamtk02}
\begin{split}
D = &2\left\{1+\frac{\gammabc-1}{\gammaac-1}+\frac{\gammabc-1}{\gammaab-1}
+(\gammabc-1) \right\}
\\[8pt]
&-\left\{\frac{\gammaab-1}{\gammabc-1}+\frac{\gammabc-1}{\gammaab-1}
+\frac{(\gammabc-1)^2}{(\gammaab-1)(\gammaac-1)}\right\}
\,,
\end{split}
\end{equation}
%%%%%%%%%%%%%%%%%%%%%%%%%%%%%%%%%%%%%%%%%%%%%%%%%%%%%%%%%%%%%%%%%%%%
as one can check by straightforward algebra.

Expressing $R$ by \eqref{adamtk01}\,--\,\eqref{adamtk02},
rather than \eqref{fskvm3}, enables its
Euclidean limit to be determined manifestly.

Indeed, we have the following Lemma about Euclidean limits:
%%%%%%%%%%%%%%%%%%%%%%%%%%%%%%%%%%%%%%%%%%%%%%%%%%%%%%%%%%%%%%%%%%%%
%%%%%%%%%%%%%%%%%%%%%%%%%%%%%%%%%%%%%%%%%%%%%%%%%%%%%%%%%%%%%%%%%%%%
% LEMMA NUMBER 7.19
\begin{lemma}\label{lemwindk}
Let $A_i,A_j\in\Rsn\subset\Rn$ be two distinct points
in an Einstein gyrovector space $(\Rsn,\op,\od)$,
and let
%%%%%%%%%%%%%%%%%%%%%%%%%%%%%%%%%%%%%%%%%%%%%%%%%%%%%%%%%%%%%%%%%%%%
\begin{equation} \label{kurcud}
\begin{split}
a_{ij}^{ein} &= \|\om A_i \op A_j\|
\\[8pt]
\gamma_{ij}^{\phantom{O}} &= \gamma_{\om A_i \op A_j}^{\phantom{O}}
\\[8pt]
a_{ij} &= \|-A_i + A_j\|
\,.
\end{split}
\end{equation}
%%%%%%%%%%%%%%%%%%%%%%%%%%%%%%%%%%%%%%%%%%%%%%%%%%%%%%%%%%%%%%%%%%%%
Then,
\begin{equation} \label{tushiat}
\lim_{s\rightarrow \infty} s^2(\gammaij-1) = \half a_{ij}^2
\end{equation}
and
\begin{equation} \label{tushiats}
\lim_{s\rightarrow \infty} s^2(\gamma_{ij}^2-1) = a_{ij}^2
\,.
\end{equation}
\end{lemma}
\begin{proof}
In the Euclidean limit, $s\rightarrow\infty$, gamma factors tend to 1. Hence,
by \eqref{rugh1ds}, p.~\pageref{rugh1ds},
%%%%%%%%%%%%%%%%%%%%%%%%%%%%%%%%%%%%%%%%%%%%%%%%%%%%%%%%%%%%%%%%%%%%
\begin{equation} \label{mkdem2}
\begin{split}
\lim_{s\rightarrow \infty} s^2(\gammaij-1)
&=
\lim_{s\rightarrow \infty} s^2(\gammaij-1) \timess\half
\lim_{s\rightarrow \infty} (\gammaij+1)
\\[2pt] &=
\half \lim_{s\rightarrow \infty} s^2(\gamma_{ij}^2-1)
\\[2pt] &=
\half \lim_{s\rightarrow \infty} \gamma_{ij}^2(a_{ij}^{ein})^2
\\[2pt] &=
\half a_{ij}^2
\,,
\end{split}
\end{equation}
%%%%%%%%%%%%%%%%%%%%%%%%%%%%%%%%%%%%%%%%%%%%%%%%%%%%%%%%%%%%%%%%%%%%
as desired.
The proof of \eqref{tushiats} is similar.
\end{proof}
%%%%%%%%%%%%%%%%%%%%%%%%%%%%%%%%%%%%%%%%%%%%%%%%%%%%%%%%%%%%%%%%%%%%

Following Lemma \ref{lemwindk} we have the Euclidean limit
\begin{equation} \label{adamtk04}
\lim_{s\rightarrow \infty} D = 2\left\{1+\frac{a_{23}^2}{a_{13}^2}
+\frac{a_{23}^2}{a_{12}^2} + 0\right\}
-
\left\{\frac{a_{12}^2}{a_{13}^2} + \frac{a_{13}^2}{a_{12}^2}
+ \frac{a_{23}^4}{a_{12}^2a_{13}^2} \right\}
\,.
\end{equation}
Hence, by \eqref{adamtk01}, \eqref{mkdem2}\,--\,\eqref{adamtk04},
and straightforward algebra,
\begin{equation} \label{adamtk05}
R_{euc}^2 := \lim_{s\rightarrow \infty} R^2 = \frac{
a_{12}^2 a_{13}^2 a_{23}^2
}{
2(a_{12}^2 a_{13}^2 + a_{12}^2 a_{23}^2 + a_{13}^2 a_{23}^2)
-(a_{12}^4 + a_{13}^4 + a_{23}^4))
}
\,,
%MATLAB fig335euc zerorr
\end{equation}
where $a_{ij}=\|-A_i+A_j\|$.

Formalizing the result of this section we obtain the following theorem.
%%%%%%%%%%%%%%%%%%%%%%%%%%%%%%%%%%%%%%%%%%%%%%%%%%%%%%%%%%%%%%%%%%%%
% THEOREM NUMBER 8.--
\begin{theorem}\label{thmdksew} 
The circumradius $R$ of a triangle $A_1A_2A_3$ in a Euclidean space
$\Rn$, $n\ge2$, is given by the equation
\begin{equation} \label{adamtk06}
R^2 = \frac{
a_{12}^2 a_{13}^2 a_{23}^2
}{
2(a_{12}^2 a_{13}^2 + a_{12}^2 a_{23}^2 + a_{13}^2 a_{23}^2)
-(a_{12}^4 + a_{13}^4 + a_{23}^4))
}
\,,
%MATLAB fig335euc zerorr
\end{equation}
where
\begin{equation} \label{adamtk07}
a_{ij} = \|-A_i + A_j \|
\,,
\end{equation}
$1\le i,j \le3$, are the sidelengths of the triangle.
\end{theorem}
%%%%%%%%%%%%%%%%%%%%%%%%%%%%%%%%%%%%%%%%%%%%%%%%%%%%%%%%%%%%%%%%%%%%

%SECTION NUMBER 8
\section{The Gyrocircle Through Three Points}\label{dkngkv}

%%%%%%%%%%%%%%%%%%%%%%%%%%%%%%%%%%%%%%%%%%%%%%%%%%%%%%%%%%%%%%%%%%%%
%sidebyside                                   Gyrocircles in (+)E
% FIGURE 5 and 6  A Gyrocircle through 3 points.
 %%%  Double Figs.  %%%%%%%%%%%%%%%%%%%%%%%%%%%%%%%%%%%%%%%%%%%%%%%%%%%%%%%%%%
%%%%%%%%%%%%%%%%%%%%%%%%%%%%%%%%%%%%%%%%%%%%%%%%%%%%%%%%%%%%%%%%%%%%%%
%\begin{figure}[ht]
\begin{figure}[t]  % try to put this figure on the top of the page
              % [h] tries to place the figure here
              % [b] tries to place the figure on the bottom of the page
              % [t] tries to place the figure on the top of the page
              % [P] tries to place the figure floatingly on the page
 \sidebyside {       % center two figures
\psfrag{A1}{$A_1$}
\psfrag{A2}{$A_2$}
\psfrag{A3}{$A_3$}
\psfrag{O}{$O$}
 \includegraphics[width=0.5\textwidth]{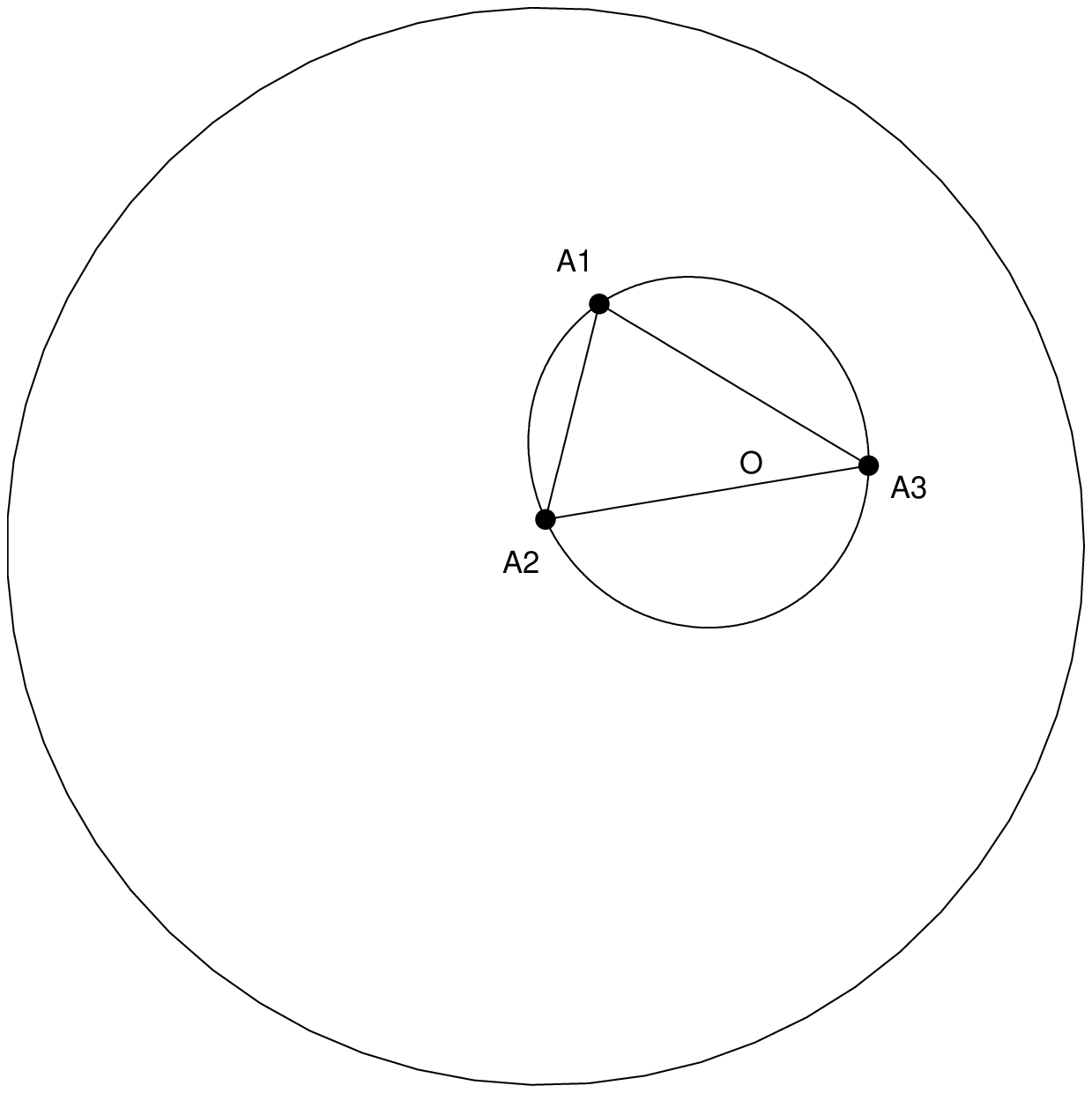}
\caption[A gyrotriangle that possesses a circumgyrocircle]{
Here $A_1,~A_2$ and $A_3$ are arbitrarily selected three points of an
Einstein gyrovector space $(\Rsn,\op,\od)$ that satisfy the
circumgyrocircle existence condition \eqref{rjksds}. Accordingly, there exists a unique
gyrocircle that passes through these points.
\label{fig269a1m}}}
  {
\psfrag{A1}{$A_1$}
\psfrag{A2}{$A_2$}
\psfrag{A3}{$A_3$}
 \includegraphics[width=0.5\textwidth]{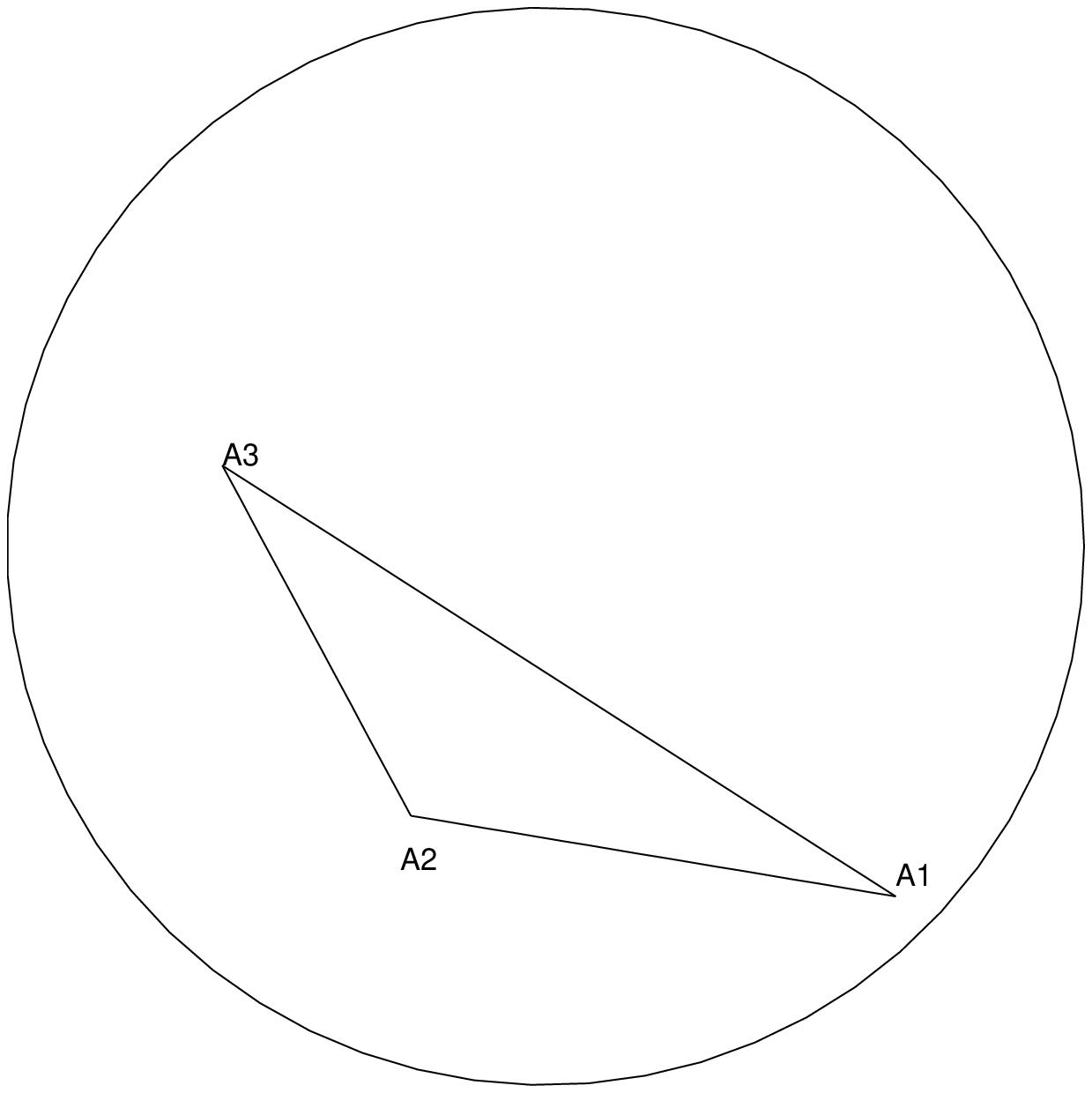}
\caption[A gyrotriangle that does not possess a circumgyrocircle]{
Here $A_1,~A_2$ and $A_3$ are arbitrarily selected three points of an
Einstein gyrovector space $(\Rsn,\op,\od)$ that do not satisfy the
circumgyrocircle existence condition \eqref{rjksds}. Accordingly, there exists no
gyrocircle that passes through these points.
\label{fig269b1m}} }
\end{figure}
%%%%%%%%%%%%%%%%%%%%%%%%%%%%%%%%%%%%%%%%%%%%%%%%%%%%%%%%%%%%%%%%%%%%%%

%   Fig.~\ref{fig269a1m} and Fig.~\ref{fig269b1m}
%%%%%%%%%%%%%%%%%%%%%%%%%%%%%%%%%%%%%%%%%%%%%%%%%%%%%%%%%%%%%%%%%%%n%

%%%%%%%%%%%%%%%%%%%%%%%%%%%%%%%%%%%%%%%%%%%%%%%%%%%%%%%%%%%%%%%%%%%%
% THEOREM NUMBER 8.9
\begin{theorem}\label{thmhdknc}
{\bf (The Gyrocircle Through Three Points).}\index{gyrocircle, three points}
Let $A_1,~A_2$ and $A_3$ be any three points that form a
gyrobarycentrically independent set in an Einstein gyrovector space
$(\Rsn,\op,\od)$, shown in Figs.~\ref{fig269a1m}\,--\,\ref{fig269b1m}.
There exists a unique gyrocircle that passes through these
points\index{circumgyrocircle existence condition}
if and only if gyrotriangle $A_1A_2A_3$
obeys the circumgyrocircle existence condition, \eqref{rjksd},
\begin{equation} \label{rjksds}
(\gammaab+\gammaac+\gammabc-1)^2 ~>~ 2 (\gamma_{12}^2+\gamma_{13}^2+\gamma_{23}^2-1)
\end{equation}
or, equivalently,
if and only if gyrotriangle $A_1A_2A_3$
obeys the gyrotrigonometric
circumgyrocircle existence condition, \eqref{hurdmst},
\begin{equation} \label{rjksht}
\sin(2\alpha_1+\tfrac{\delta}{2}) +
\sin(2\alpha_2+\tfrac{\delta}{2}) +
\sin(2\alpha_3+\tfrac{\delta}{2}) > 3\sin\tfrac{\delta}{2}
\,,
% MATHEMATICA stam272
% MATHEMATICA stam276
\end{equation}
where we use the index notation \eqref{indexnotation}
for gyrotriangle $A_1A_2A_3$.

When a gyrocircle through the three points exists,
it is the unique gyrocircle with gyrocenter $O$ given by,
\eqref{eicksfnein},
\begin{equation} \label{eicksfneins}
O = \frac{
 m_1 \gamma_{_{A_1}}^{\phantom{1}} A_1 +
 m_2 \gamma_{_{A_2}}^{\phantom{1}} A_2 +
 m_3 \gamma_{_{A_3}}^{\phantom{1}} A_3
}{
m_1 \gamma_{_{A_1}}^{\phantom{1}} +
m_2 \gamma_{_{A_2}}^{\phantom{1}} +
m_3 \gamma_{_{A_3}}^{\phantom{1}}
}
\,,
\end{equation}
where
%%%%%%%%%%%%%%%%%%%%%%%%%%%%%%%%%%%%%%%%%%%%%%%%%%%%%%%%%%%%%%%%%%%%
\begin{equation} \label{hfjdv07bbs}
\begin{split}
m_1 &=
(\phantom{-} \gamma_{12}^{\phantom{1}} + \gamma_{13}^{\phantom{1}} - \gamma_{23}^{\phantom{1}} -1)
(\gamma_{23}^{\phantom{1}} -1)
\\[4pt]
m_2 &=
(\phantom{-} \gamma_{12}^{\phantom{1}} - \gamma_{13}^{\phantom{1}} + \gamma_{23}^{\phantom{1}} -1)
(\gamma_{13}^{\phantom{1}} -1)
\\[4pt]
m_3 &=
(         -  \gamma_{12}^{\phantom{1}} + \gamma_{13}^{\phantom{1}} + \gamma_{23}^{\phantom{1}} -1)
(\gamma_{12}^{\phantom{1}} -1)
\end{split}
\end{equation}
%%%%%%%%%%%%%%%%%%%%%%%%%%%%%%%%%%%%%%%%%%%%%%%%%%%%%%%%%%%%%%%%%%%%
and with gyroradius $R$ given by, \eqref{fskvm4},
\begin{equation} \label{fskvm4s}
R = \sqrt{2} s ~\sqrt{ \frac{
(\gammaab-1) (\gammaac-1) (\gammabc-1)
}{
1+2\gammaab\gammaac\gammabc - \gamma_{12}^2-\gamma_{13}^2-\gamma_{23}^2
}}
=s\sqrt{\frac{H_3}{D_3}}
~~.
\end{equation}
%MATLAB fig241en5.m - zerors=0 and zeroaij, zerogaij.
\end{theorem}
%%%%%%%%%%%%%%%%%%%%%%%%%%%%%%%%%%%%%%%%%%%%%%%%%%%%%%%%%%%%%%%%%%%%
\begin{proof}
The gyrocircle in the Theorem, if exists, is the circumgyrocircle of
gyrotriangle $A_1A_2A_3$.
The gyrocenter $O$ of the gyrocircle is, therefore, given by
\eqref{eicksfneins}\,--\,\eqref{hfjdv07bbs}, as we see from
Theorem \ref{thmtivhvn}, p.~\pageref{thmtivhvn}; and
the gyroradius, $R$, of the gyrocircle is given by
\eqref{fskvm4}, p.~\pageref{fskvm4}.

Finally, the circumgyrocircle of gyrotriangle $A_1A_2A_3$ exists
if and only if the gyrotriangle satisfies the
circumgyrocircle existence condition\index{circumgyrocircle existence condition} \eqref{rjksds},
or, equivalently, \eqref{eicksfneins},
as explained in the paragraph of Inequality \eqref{rjksd}, p.~\pageref{rjksd},
and in \eqref{hurdms}\,--\,\eqref{hurdmst}.
\end{proof}

% EXAMPLE NUMBER 8.10
\begin{example}\label{exhkd1}
If the three points $A_1,~A_2$ and $A_3$ in Theorem \ref{thmhdknc} are not distinct,
a gyrocircle through these points is not unique. Indeed, in this case we have
\begin{equation} \label{rjksds1}
(\gammaab+\gammaac+\gammabc-1)^2 ~=~ 2 (\gamma_{12}^2+\gamma_{13}^2+\gamma_{23}^2-1)
\,,
\end{equation}
as one can readily check, thus violating the
circumgyrocircle existence condition\index{circumgyrocircle existence condition} \eqref{rjksds}.
\end{example}

% EXAMPLE NUMBER 8.11
\begin{example}\label{exhkd2}
If the three points $A_1,~A_2$ and $A_3$ in Theorem \ref{thmhdknc} are distinct
and gyrocollinear, there is no gyrocircle through these points. Hence, in this case
the circumgyrocircle existence condition\index{circumgyrocircle existence condition} \eqref{rjksds}
must be violated.
Hence, these points must satisfy the inequality
\begin{equation} \label{rjksds2}
(\gammaab+\gammaac+\gammabc-1)^2 ~\le~ 2 (\gamma_{12}^2+\gamma_{13}^2+\gamma_{23}^2-1)
\,.
\end{equation}
\end{example}

% EXAMPLE NUMBER 8.12
\begin{example}\label{exhkd3}
Let the three points $A_1,~A_2$ and $A_3$ in Theorem \ref{thmhdknc} be the vertices
of an  equilateral gyrotriangle\index{gyrotriangle,  equilateral}
with side gyrolengths $a$. Then, $\gammaab=\gammaac=\gammabc=\gamma_a^{\phantom{O}}$,
so that the
circumgyrocircle existence condition\index{circumgyrocircle existence condition} \eqref{rjksds}
reduces to
\begin{equation} \label{rjksds3}
\gamma_a^{\phantom{O}} > 1
\,,
\end{equation}
which is satisfied by any side gyrolength $a$, $0<a<s$.
Hence, by Theorem \ref{thmhdknc}, any equilateral gyrotriangle in
an Einstein gyrovector space possesses a circumgyrocircle.
Moreover, the circumgyrocenter of an equilateral gyrotriangle lies on
the interior of the gyrotriangle since the circumgyrocenter gyrobarycentric coordinates
$m_k$, $k=1,2,3$, in \eqref{hfjdv07bbs} are all positive.

More about the unique gyrocircle that passes through three given points
of an Einstein gyrovector space is studied in Sect.~\ref{slila3}.
\end{example}

A generalization of results in this section
from the gyrotriangle circumgyrocircle to the
gyrotetrahedron circumgyrosphere is presented in Chap.~10.
Remarkably, the pattern that $D_3$ and $H_3$ exhibit in this section
remains the same pattern in Chap.~10, exhibited by $D_4$ and $H_4$.

%SECTION NUMBER 9
\section{The Inscribed Gyroangle Theorem I}\label{dknke2}
\index{inscribed gyroangle}

The Inscribed Gyroangle Theorem appears in two distinct, interesting
versions, each of which reduces in the Euclidean limit of large $s$
to the well-known Inscribed Angle Theorem in Euclidean geometry.
Version I is presented in this section, and
Version II is presented in sect.~\ref{dknke2s}.

%%%%%%%%%%%%%%%%%%%%%%%%%%%%%%%%%%%%%%%%%%%%%%%%%%%%%%%%%%%%%%%%%%%%
% FIGURE 7
 
%%%%%%%%%%%%%%%%%%%%%%%%%%%%%%%%%%%%%%%%%%%%%%%%%%%%%%%%%%%%%%%%%%%%%%
%%%%% The hyperbolic semi-circle theorem            %%%%%%%%%%%%%%%%%%
%\begin{figure}[htbp]
\begin{figure}[t]  % try to put this figure on the top of the page
 \centering         % center the figure
%
%\psfrag{a1}{$A_1$}
 \psfrag{A1}{$A_1$}
 \psfrag{A2}{$A_2$}
 \psfrag{A3}{$A_3$}
 \psfrag{C}{$C$}
 \psfrag{O}{$O$}
 \psfrag{Af}{$\theta$}
 \psfrag{B1}{$\phi$}
 \psfrag{B2}{$\phi$}
 \psfrag{M12}{$M_{12}$}
 \psfrag{r1}{$R$}
 \psfrag{r2}{$R$}
 \psfrag{text1}{$\ab_{13}$}
 \psfrag{text2}{$\ab_{23}$}
 \psfrag{text3a}[]{$\half\od\ab_{12}$}
 \psfrag{text3b}{$\half\od\ab_{12}$}
 \psfrag{text4}{$\ab_{12}$}
 \psfrag{formula00}[]{$\boxed{\sin\theta=\displaystyle\frac{
2\gamma_R
}{
\sqrt{(\gammaac+1)(\gammabc+1)}
} \sin\phi}$}
 \psfrag{formula01}{$\ab_{12}:=\om A_1\op A_2$}
 \psfrag{formula02}{$\half\od\ab_{12}=\om A_1\op M_{12}$}
 \psfrag{formula03}{$\gamma_{13}=\gamma_{\ab_{13}}$}
 \psfrag{formula04}{$\gamma_{23}=\gamma_{\ab_{23}}$}
 \includegraphics[width=9cm]{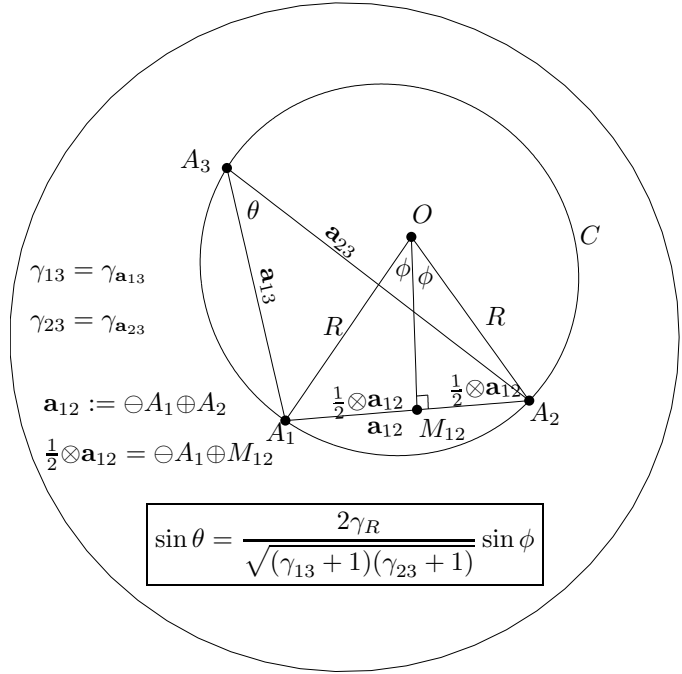}
\caption[Illustration of the Inscribed Gyroangle Theorem I]{
The Inscribed Gyroangle Theorem I. Gyroangle
$\theta=\angle A_1A_3A_2$
is inscribed in a gyrocircle $C$ of gyroradius $R$
(the circumgyroradius of gyrotriangle $A_1A_2A_3$) centered at
$O$ in an Einstein gyrovector plane $(\Rstwo,\op,\od)$, and
$\phi=\angle A_1 OM_{12} = \angle A_2 OM_{12}$, where $M_{12}$ is the
gyromidpoint of the gyrosegment $A_1A_2$.
Accordingly, $2\phi = \angle A_1OA_2$
is a gyrocentral gyroangle, and both $\theta$ and $2\phi$
subtend on the same gyroarc on the gyrocircle $C$.
The relationship between $\theta$ and $\phi$, \eqref{khdnkm}, is shown.
In the Euclidean limit of large $s$, $s\rightarrow\infty$,
gamma factors tend to 1 and, hence, the relationship between $\theta$ and $\phi$ in
Euclidean geometry reduces to $\sin\theta=\sin\phi$ or, equivalently, $\theta=\phi$
(if $A_3$ and $O$ lie on the same side of gyrosegment $A_1A_2$)
and $\theta=\pi-\phi$
(if $A_3$ and $O$ lie on opposite sides of gyrosegment $A_1A_2$;
see Fig.~\ref{fig267ciibm}).
\label{fig267cm}}
\end{figure}
%%%%%%%%%%%%%%%%%%%%%%%%%%%%%%%%%%%%%%%%%%%%%%%%%%%%%%%%%%%%%%%%%%%%%%
 % The Inscribed Gyroangle Theorem in (+)E
%                       Fig.~\ref{fig267cm}
%%%%%%%%%%%%%%%%%%%%%%%%%%%%%%%%%%%%%%%%%%%%%%%%%%%%%%%%%%%%%%%%%%%%

Fig.~\ref{fig267cm} presents a gyrotriangle $A_1A_2A_3$ and its
circumgyrocircle with gyrocenter $O$ at the
gyrotriangle circumgyrocenter,\index{circumgyrocenter}
given by \eqref{eicksfnein}, p.~\pageref{eicksfnein},
and with gyroradius $R$,
given by the gyrotriangle circumgyroradius\index{circumgyroradius}
\eqref{fskvm4}, p.~\pageref{fskvm4}.
The gamma factor $\gamma_R$ of $R$ is given by \eqref{fskvm2}, p.~\pageref{fskvm2}.

A gyrocentral gyroangle\index{gyroangle, gyrocentral}
of a gyrocircle is a gyroangle whose vertex is located at
the gyrocenter of the gyrocircle. For instance,
gyroangle $\angle A_1OA_2$ in Figs.~\ref{fig267cm}\,--\,\ref{fig267ciibm}
is gyrocentral.

An inscribed gyroangle of a gyrocircle is a gyroangle whose vertex is on
the gyrocircle and whose sides each intersects the gyrocircle at another point.
For instance,
gyroangle $\angle A_1A_3A_2$ in Figs.~\ref{fig267cm}\,--\,\ref{fig267ciibm}
is inscribed.
The inscribed gyroangle theorem gives a relation between
inscribed gyroangles and a gyrocentral gyroangle of a gyrocircle
that subtend on the same gyroarc of the gyrocircle
in an Einstein gyrovector space.

%%%%%%%%%%%%%%%%%%%%%%%%%%%%%%%%%%%%%%%%%%%%%%%%%%%%%%%%%%%%%%%%%%%%
% THEOREM NUMBER 16
\begin{theorem}\label{thmfdknb}
{\bf (The Inscribed Gyroangle Theorem I).}
Let $\theta$ be a gyroangle inscribed in a gyrocircle of gyroradius $R$, and let
$2\phi$ be the gyrocentral gyroangle of the gyrocircle
such that both $\theta$ and $2\phi$
subtend on the same gyroarc $\widehat{A_1A_2}$
on the gyrocircle, as shown in Fig.~\ref{fig267cm}.
Then, in the notation of Fig.~\ref{fig267cm} and in
\eqref{indexnotation}, p.~\pageref{indexnotation},
\begin{equation} \label{khdnkm}
\sin\theta=\displaystyle\frac{ 2\gamma_R }{
\sqrt{(\gammaac+1)(\gammabc+1)}
} \sin\phi
\,.
\end{equation}
% MATLAB fig26c7.m zero1-2-3.
\end{theorem}
\begin{proof}
Under the conditions of the theorem, as described in Fig.~\ref{fig267cm}, let
$M_{12}$ be the gyromidpoint of gyrosegment $A_1A_2$, implying
\begin{equation} \label{ryugf1}
\phi := \angle A_1 OM_{12} = \angle A_2 OM_{12} = \half\angle A_1 OA_2
\end{equation}
so that $2\phi$ is the gyrocentral gyroangle $\angle A_1OA_2$ shown in Fig.~\ref{fig267cm}.

Furthermore, let
\begin{equation} \label{ryugf2}
\ab_{12} = \om A_1 \op A_2
\end{equation}
so that \cite[p.~100]{mybook05},
\begin{equation} \label{ryugf3}
\om A_1 \op M_{12} = \half\od\ab_{12}
\end{equation}
and
\begin{equation} \label{ryugf4}
\gamma_{\half\od\ab_{12}}^{\phantom{O}} (\half\od\ab_{12}) = \displaystyle\frac{
\gammaab \ab_{12}
}{
\sqrt{2} \sqrt{1+\gammaab}
}
\,.
\end{equation}
Taking magnitudes of both sides of \eqref{ryugf4} and noting the
homogeneity property
$(V9)$ of Einstein gyrovector spaces \cite[Sect.~2.7]{mybook06},
we have
\begin{equation} \label{ryugf4h}
\gamma_{\half\od a_{12}}^{\phantom{O}} (\half\od a_{12}) = \displaystyle\frac{
\gammaab a_{12}
}{
\sqrt{2} \sqrt{1+\gammaab}
}
\,.
\end{equation}

Applying the extended law of gyrosines\index{law of gyrosines, extended}
\eqref{fskvm6}, p.~\pageref{fskvm6}, to
gyrotriangle $A_1A_2A_3$ and its circumgyroradius $R$ in  Fig.~\ref{fig267cm}, we have
\begin{equation} \label{ryugf5}
\frac{\gammaab a_{12}}{\sin\theta} =
\sqrt{ \frac{
(\gammaab+1) (\gammaac+1) (\gammabc+1)}{2}} ~R
\,,
\end{equation}
implying
\begin{equation} \label{ryugf6}
\sin\theta = \displaystyle\frac{
\sqrt{2} \gammaab a_{12}
}{
\sqrt{
(1+\gammaab) (1+\gammaac) (1+\gammabc)} ~R
}
\,.
\end{equation}

Applying the elementary gyrosine definition in gyrotrigonometry,
illustrated in \cite[Fig.~6.5, p.~147]{mybook05},
to the right gyroangled gyrotriangle $A_1M_{12}O$ in  Fig.~\ref{fig267cm}, we
obtain the first equation in \eqref{ryugf7},
\begin{equation} \label{ryugf7}
\sin\phi = \displaystyle\frac{
\gamma_{\half\od a_{12}}^{\phantom{O}} (\half\od a_{12})
}{
\gamma_R R
}
=
\displaystyle\frac{
\gammaab a_{12}
}{
\sqrt{2}\sqrt{1+\gammaab}
~\gamma_R R
}
\,.
\end{equation}
The second equation in \eqref{ryugf7} follows from \eqref{ryugf4h}.

Finally, the desired identity \eqref{khdnkm} follows by eliminating
the factor $\gammaab a_{12}$ between
\eqref{ryugf6} and \eqref{ryugf7}.
% MATLAB fig267c.m zero1-2-3.
\end{proof}

%SECTION NUMBER 10
\section{The Inscribed Gyroangle Theorem II}\label{dknke2s}
\index{inscribed gyroangle}

%%%%%%%%%%%%%%%%%%%%%%%%%%%%%%%%%%%%%%%%%%%%%%%%%%%%%%%%%%%%%%%%%%%%
% FIGURE 8
 
%%%%%%%%%%%%%%%%%%%%%%%%%%%%%%%%%%%%%%%%%%%%%%%%%%%%%%%%%%%%%%%%%%%%%%
%%%%% The hyperbolic semi-circle theorem            %%%%%%%%%%%%%%%%%%
%\begin{figure}[htbp]
\begin{figure}[t]  % try to put this figure on the top of the page
 \centering         % center the figure
%
%\psfrag{a1}{$A_1$}
 \psfrag{A1}{$A_1$}
 \psfrag{A2}{$A_2$}
 \psfrag{A3}{$A_3$}
 \psfrag{C}{$C$}
 \psfrag{O}{$O$}
 \psfrag{Te}{$\theta$}
 \psfrag{Ph}{$2\phi$}
 \psfrag{Ep}{$\epsilon$}
%\psfrag{B2}{$\phi$}
%\psfrag{M12}{$M_{12}$}
 \psfrag{r1}{$R$}
 \psfrag{r2}{$R$}
%\psfrag{text1}{$\ab_{13}$}
%\psfrag{text2}{$\ab_{23}$}
%\psfrag{text3a}[]{$\half\od\ab_{12}$}
%\psfrag{text3b}{$\half\od\ab_{12}$}
%\psfrag{text4}{$\ab_{12}$}
 \psfrag{formula00}[]{$\boxed{\theta+\half\delta_{A_1A_2A_3} = \phi+\half\delta_{A_1A_2O}}$}
 \psfrag{formula01}{$\delta_{A_1A_2A_3}$}
 \psfrag{formula02}{$\delta_{A_1A_2O}$}
 \psfrag{formula03}{$\epsilon=\angle OA_1A_2=\angle OA_2A_1$}
%\psfrag{formula04}{$\gamma_{23}=\gamma_{\ab_{23}}$}
%
%\includegraphics[width=9cm]{/home/ungar/dir_amy/dir_papers/dir_mybook01/dir_figs/fig267ciia.eps}
 \includegraphics[width=9cm]{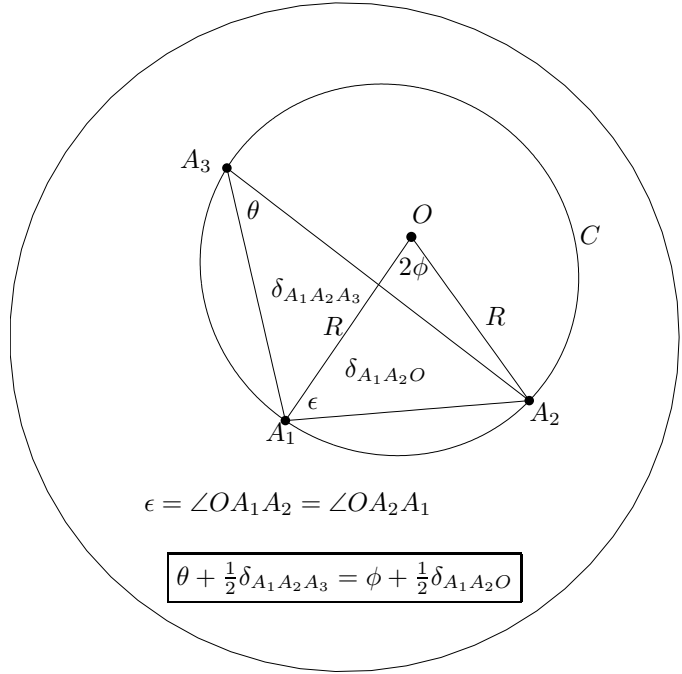}
\caption[Illustration of the Inscribed Gyroangle Theorem II, case 1]{
The Inscribed Gyroangle Theorem II. 
Unlike the Inscribed Gyroangle Theorem described in Fig.~\ref{fig267cm},
here the relation between the inscribed gyroangle $\theta$ and the
gyrocentral gyroangle $2\phi$ is expressed in terms of
gyrotriangle defects.
The latter vanish in Euclidean geometry, reducing the relation between
$\theta$ and $\phi$ to the equation $\theta=\phi$.
\label{fig267ciiam}}
\end{figure}
%%%%%%%%%%%%%%%%%%%%%%%%%%%%%%%%%%%%%%%%%%%%%%%%%%%%%%%%%%%%%%%%%%%%%%
 % The Inscribed Gyroangle Theorem II in (+)E
%                       Fig.~\ref{fig267ciiam}
%%%%%%%%%%%%%%%%%%%%%%%%%%%%%%%%%%%%%%%%%%%%%%%%%%%%%%%%%%%%%%%%%%%%

%%%%%%%%%%%%%%%%%%%%%%%%%%%%%%%%%%%%%%%%%%%%%%%%%%%%%%%%%%%%%%%%%%%%
% FIGURE 9
 
%%%%%%%%%%%%%%%%%%%%%%%%%%%%%%%%%%%%%%%%%%%%%%%%%%%%%%%%%%%%%%%%%%%%%%
%%%%% The hyperbolic semi-circle theorem            %%%%%%%%%%%%%%%%%%
%\begin{figure}[htbp]
\begin{figure}[t]  % try to put this figure on the top of the page
 \centering         % center the figure
%
%\psfrag{a1}{$A_1$}
 \psfrag{A1}{$A_1$}
 \psfrag{A2}{$A_2$}
 \psfrag{A3}{$A_3$}
 \psfrag{C}{$C$}
 \psfrag{O}{$O$}
 \psfrag{Te}{$\theta$}
 \psfrag{Ph}{$2\phi$}
 \psfrag{Ep}{$\epsilon$}
%\psfrag{B2}{$\phi$}
%\psfrag{M12}{$M_{12}$}
 \psfrag{r1}{$R$}
 \psfrag{r2}{$R$}
%\psfrag{text1}{$\ab_{13}$}
%\psfrag{text2}{$\ab_{23}$}
%\psfrag{text3a}[]{$\half\od\ab_{12}$}
%\psfrag{text3b}{$\half\od\ab_{12}$}
%\psfrag{text4}{$\ab_{12}$}
 \psfrag{formula00}[]{$\boxed{\theta+\half\delta_{A_1A_2A_3}
                        = \pi - (\phi+\half\delta_{A_1A_2O})}$}
 \psfrag{formula01}{$\delta_{A_1A_2A_3}$}
 \psfrag{formula02}{$\delta_{A_1A_2O}$}
 \psfrag{formula03}{$\epsilon=\angle OA_1A_2=\angle OA_2A_1$}
%\psfrag{formula04}{$\gamma_{23}=\gamma_{\ab_{23}}$}
%
%\includegraphics[width=9cm]{/home/ungar/dir_amy/dir_papers/dir_mybook01/dir_figs/fig267ciib.eps}
 \includegraphics[width=9cm]{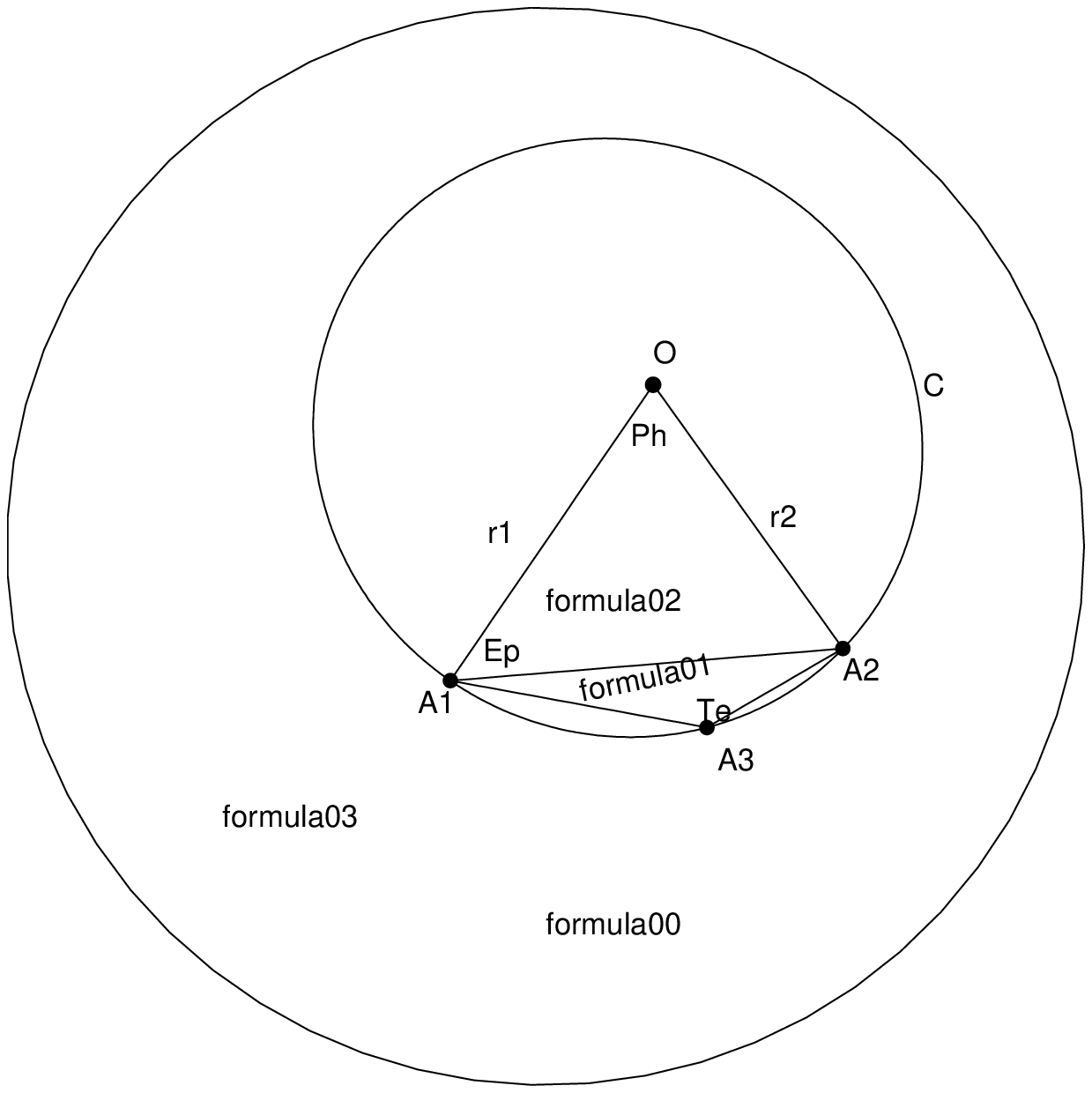}
\caption[Illustration of the Inscribed Gyroangle Theorem II, case 2]{
The Inscribed Gyroangle Theorem II. 
This figure is similar to Fig.~\ref{fig267ciiam}
except that here the points $O$ and $A_3$ lie on opposite sides of
the chord $A_1A_2$ of gyrocircle $C$ in the gyroplane through the points
$A_1,A_2$ and $A_3$.
As in Fig.~\ref{fig267ciiam}
the relation between the inscribed gyroangle $\theta$ and the
gyrocentral gyroangle $2\phi$ is expressed in terms of
gyrotriangle defects.
The latter vanish in Euclidean geometry, reducing the relation between
$\theta$ and $\phi$ to the equation $\theta=\pi-\phi$.
\label{fig267ciibm}}
\end{figure}
%%%%%%%%%%%%%%%%%%%%%%%%%%%%%%%%%%%%%%%%%%%%%%%%%%%%%%%%%%%%%%%%%%%%%%
 % The Inscribed Gyroangle Theorem II in (+)E
%                       Fig.~\ref{fig267ciibm}
%%%%%%%%%%%%%%%%%%%%%%%%%%%%%%%%%%%%%%%%%%%%%%%%%%%%%%%%%%%%%%%%%%%%

%%%%%%%%%%%%%%%%%%%%%%%%%%%%%%%%%%%%%%%%%%%%%%%%%%%%%%%%%%%%%%%%%%%%
% THEOREM NUMBER 8.14
\index{inscribed gyroangle theorem II}
\begin{theorem}\label{thmfdknb2}
{\bf (The Inscribed Gyroangle Theorem II).}
Let $\theta$ be a gyroangle inscribed in a gyrocircle with
gyrocenter $O$, and let
$2\phi$ be the gyrocentral gyroangle of the gyrocircle
such that both $\theta$ and $2\phi$
subtend on the same gyroarc $\widehat{A_1A_2}$
on the gyrocircle, as shown in Fig.~\ref{fig267ciiam}.
Furthermore, in the notation in Fig.~\ref{fig267ciiam},
let $\delta_{A_1A_2A_3}$ be the defect of gyrotriangle $A_1A_2A_3$ and, similarly,
let $\delta_{A_1A_2O}$ be the defect of gyrotriangle $A_1A_2O$.
Then,
\begin{equation} \label{kfhsd0}
\sin(\theta + \half\delta_{A_1A_2A_3}) = \sin(\phi + \half\delta_{A_1A_2O})
\,,
\end{equation}
that is, either
\begin{subequations} \label{kfhsd}
\begin{equation} \label{kfhsda}
\theta + \half\delta_{A_1A_2A_3} = \phi + \half\delta_{A_1A_2O}
\,,
% MATLAB fig267ciia.m zeroins1 zeroins2 zeroins
\end{equation}
as in Fig.~\ref{fig267ciiam}, or
\begin{equation} \label{kfhsdb}
\theta + \half\delta_{A_1A_2A_3} = \pi - (\phi + \half\delta_{A_1A_2O})
\,,
% MATLAB fig267ciib.m zeroins1 zeroins2 zeroins
\end{equation}
as in Fig.~\ref{fig267ciibm}.
\end{subequations}
\end{theorem}
\begin{proof}
In the  gyrotriangle index notation
\eqref{indexnotation}, p.~\pageref{indexnotation},
with $\alpha_3=\theta$ and $\delta=\delta_{A_1A_2A_3}$
to conform with the notation for gyrotriangle $A_1A_2A_3$ in
Figs.~\ref{fig267ciiam}\,--\,\ref{fig267ciibm},
we have by \eqref{dagim4}, p.~\pageref{dagim4},
\begin{equation} \label{dugit01}
\sin(\theta+\half\delta_{A_1A_2A_3})
= \dfrac{\gammaab a_{12}} {R(\gammaab+1)}
\,.
% MATLAB fig241en5 zerosine
\end{equation}
Here $\delta_{A_1A_2A_3}$ is the defect of gyrotriangle $A_1A_2A_3$
and $R$ is the circumgyroradius of gyrotriangle $A_1A_2A_3$ in
Figs.~\ref{fig267cm}\,--\,\ref{fig267ciiam}.
Accordingly,
\begin{equation} \label{dugit02}
R = \|\om O\op A_1\| = \|\om O\op A_2\|
\,.
\end{equation}

The expression $\gammaab a_{12}/(R(\gammaab+1))$ is expressed in \eqref{dugit01}
in terms of gyroangles of gyrotriangle $A_1A_2A_3$, shown in
Figs.~\ref{fig267ciiam}\,--\,\ref{fig267ciibm}.
We now wish to express it in terms of gyroangles of gyrotriangle $A_1A_2O$,
also shown in
Figs.~\ref{fig267ciiam}\,--\,\ref{fig267ciibm}.
Accordingly, let
%%%%%%%%%%%%%%%%%%%%%%%%%%%%%%%%%%%%%%%%%%%%%%%%%%%%%%%%%%%%%%%%%%%%
\begin{equation} \label{dugit03}
\begin{split}
\epsilon &= \angle A_1A_2O = \angle A_2A_1O
\\[4pt]
a_{12} &= \|\om A_1\op A_2\|
\\[4pt]
\gammaab &= \gamma_{a_{12}}^{\phantom{1}}
\end{split}
\end{equation}
%%%%%%%%%%%%%%%%%%%%%%%%%%%%%%%%%%%%%%%%%%%%%%%%%%%%%%%%%%%%%%%%%%%%
be parameters of gyrotriangle $A_1A_2O$
and let $\delta=\delta_{A_1A_2O}$ be the defect of the gyrotriangle,
as shown in
Figs.~\ref{fig267ciiam}\,--\,\ref{fig267ciibm}.

Then,
\begin{equation} \label{dugit04}
\delta = \delta_{A_1A_2O} = \pi - 2\phi - 2\epsilon
\,,
\end{equation}
so that
\begin{equation} \label{dugit05}
\epsilon = \frac{\pi}{2} - (\phi+\frac{\delta}{2})
\,,
\end{equation}
implying
\begin{equation} \label{dugit06}
\epsilon + \frac{\delta}{2} = \frac{\pi}{2} - \phi
\,,
\end{equation}
so that
\begin{equation} \label{dugit07}
\sin(\epsilon + \frac{\delta}{2}) = \cos\phi
\,.
\end{equation}

Applying the first equation in \cite[Eq.~(7.143), p.~187]{mybook05}
to gyrotriangle $A_1A_2O$ in Fig.~\ref{fig267ciiam}, we have
\begin{equation} \label{dugit08}
\frac{1}{s}\gammaab a_{12} = \frac{
2\sqrt{F}}{\sin^2\epsilon}
\end{equation}
where, by \eqref{grfde}, p.~\pageref{grfde}, and by \eqref{dugit07},
%%%%%%%%%%%%%%%%%%%%%%%%%%%%%%%%%%%%%%%%%%%%%%%%%%%%%%%%%%%%%%%%%%%%
\begin{equation} \label{dugit09}
\begin{split}
F &= \sin\tfrac{\delta}{2} \sin^2(\epsilon+\tfrac{\delta}{2})
\sin(2\phi+\tfrac{\delta}{2})
\\[4pt] &=
\sin\tfrac{\delta}{2} \cos^2\phi \sin(2\phi+\tfrac{\delta}{2})
\,.
\end{split}
\end{equation}
%%%%%%%%%%%%%%%%%%%%%%%%%%%%%%%%%%%%%%%%%%%%%%%%%%%%%%%%%%%%%%%%%%%%

Applying \cite[Eq.~(7.150), p.~188]{mybook05}
to gyrotriangle $A_1A_2O$
in Fig.~\ref{fig267ciiam}, we have
\begin{equation} \label{dugit10}
\gammaab+1 = \frac{
2\sin^2(\epsilon+\tfrac{\delta}{2}) }{\sin^2\epsilon}
\,.
\end{equation}
Following \eqref{dugit08}, \eqref{dugit10} and \eqref{dugit07} we have
\begin{equation} \label{dugit11}
\frac{1}{s}\frac{\gammaab a_{12}}{\gammaab+1}
=\frac{\sqrt{F}}{\sin^2(\epsilon+\frac{\delta}{2})}
=\frac{\sqrt{F}}{\cos^2\phi}
\,.
\end{equation}

We now turn to calculate $R/s$. Applying
the $AAA$ to $SSS$ {\it Conversion Law} \cite[Theorem 6.5, p.~137]{mybook05}
to gyrotriangle $A_1A_2O$
in Fig.~\ref{fig267ciiam}, we have
%%%%%%%%%%%%%%%%%%%%%%%%%%%%%%%%%%%%%%%%%%%%%%%%%%%%%%%%%%%%%%%%%%%%
\begin{equation} \label{dugit12}
\begin{split}
\gammaR &= \frac{
\cos\epsilon+\cos\epsilon\cos2\phi
}{
\sin\epsilon\sin2\phi
}
=\frac{\cos\epsilon}{\sin\epsilon}
 \frac{1+\cos2\phi}{\sin2\phi}
\\[8pt] &=
\cot\epsilon\frac{2}{\sin2\phi}\frac{1+\cos2\phi}{2}
= \cot\epsilon\frac{\cos^2\phi}{\sin\phi\cos\phi}
\\[8pt] &=
\cot\epsilon\cot\phi
\,.
\end{split}
\end{equation}
%%%%%%%%%%%%%%%%%%%%%%%%%%%%%%%%%%%%%%%%%%%%%%%%%%%%%%%%%%%%%%%%%%%%

By \eqref{rugh1ds}, p.~\pageref{rugh1ds}, by \eqref{dugit12}
and by \eqref{dugit09} we have
%%%%%%%%%%%%%%%%%%%%%%%%%%%%%%%%%%%%%%%%%%%%%%%%%%%%%%%%%%%%%%%%%%%%
\begin{equation} \label{dugit13}
\begin{split}
\frac{R^2}{s^2} &= \frac{\gammaRs-1}{\gammaRs}
= \frac{
\cot^2\epsilon\cot^2\phi-1}{\cot^2\epsilon\cot^2\phi}
\\[8pt] &=
\frac{\sin \tfrac{\delta}{2} \sin(2\phi+\tfrac{\delta}{2})}
{\cos^2\phi \sin^2(\phi+\tfrac{\delta}{2})}
= \frac{F}{\cos^4\phi\sin^2(\phi+\tfrac{\delta}{2})}
\,,
\end{split}
\end{equation}
%%%%%%%%%%%%%%%%%%%%%%%%%%%%%%%%%%%%%%%%%%%%%%%%%%%%%%%%%%%%%%%%%%%%
%MATHEMATICA stam277
%(see Prob.~\ref{profbc2n}, p.~\pageref{profbc2n})
where $\delta=\delta_{A_1A_2O}$ is the defect of gyrotriangle $A_1A_2O$
given by \eqref{dugit04}, thus obtaining the equation
\begin{equation} \label{dugit14}
\frac{R}{s} = \frac{\sqrt{F}}{\cos^2\phi\sin(\phi+\tfrac{\delta}{2})}
\,.
\end{equation}

Equations \eqref{dugit11} and \eqref{dugit14}, along with the notation
for $\delta$ in \eqref{dugit04}, imply
\begin{equation} \label{dugit15}
\frac{
\gammaab a_{12}
}{
R(\gammaab+1)
}
= \frac{
\frac{1}{s} \frac{\gammaab a_{12}}{\gammaab+1}
}{
\frac{R}{s}
}
=\frac{
\frac{\sqrt{F}}{\cos^2\phi}
}{
\frac{\sqrt{F}}{\cos^2\phi \sin(\phi+\tfrac{\delta}{2})}
}
= \sin(\phi+\tfrac{\delta}{2})
= \sin(\phi+\half\delta_{A_1A_2O}).
\end{equation}

Finally, \eqref{dugit01} and \eqref{dugit15} imply
\begin{equation} \label{dugit16}
\sin(\theta+\half\delta_{A_1A_2A_3})
=
\sin(\phi+\half\delta_{A_1A_2O})
\,,
\end{equation}
which, in turn, implies the results of the Theorem
in \eqref{kfhsd0} and in \eqref{kfhsda}\,--\,\eqref{kfhsdb}.
\end{proof}

% REMARK NUMBER 8.15
\begin{remark}
In the proof of Theorem \ref{thmfdknb2} we take for granted that
the gyroline through the points $A_1$ and $A_2$ has two sides in the
gyroplane through the points $A_1,A_2$ and $A_3$ so that
in Fig.~\ref{fig267ciiam} the points $A_3$ and $O$ lie on the same side
of the gyroline while
in Fig.~\ref{fig267ciibm} these points lie on opposite sides
of the gyroline.
The need for a careful study
of the very intuitive idea that every line in a Euclidean plane
has "two sides" was pointed out by
Millman and Parker in \cite[p.~63]{millmanparker91}, resulting in what they
call the {\it Plane Separation Axiom} (PSA).
An analogous {\it Gyroplane Separation Axiom} (GPSA)
for a gyroplane in an Einstein gyrovector space,
is studied by S\"onmez and Ungar in \cite{sonmezungar13}.
\label{remuvbhy}
\index{Axiom, plane separation (PSA)}
\index{Axiom, gyroplane separation (GPSA)}
\end{remark}

%SECTION NUMBER 11
\section{Gyrocircle Gyrotangent Gyrolines}\label{tangent46}
\index{gyrotangent, gyrocircle}

Employing the Inscribed Gyroangle Theorem II, Theorem \ref{thmfdknb2}, p.~\pageref{thmfdknb2},
we prove the following theorem.

%%%%%%%%%%%%%%%%%%%%%%%%%%%%%%%%%%%%%%%%%%%%%%%%%%%%%%%%%%%%%%%%%%%%
% FIGURE 10
 
%%%%%%%%%%%%%%%%%%%%%%%%%%%%%%%%%%%%%%%%%%%%%%%%%%%%%%%%%%%%%%%%%%%%%%
%%%%% The hyperbolic semi-circle theorem            %%%%%%%%%%%%%%%%%%
%\begin{figure}[htbp]
\begin{figure}[t]  % try to put this figure on the top of the page
 \centering         % center the figure
 \psfrag{PA}{$A$}
 \psfrag{PO}{$O$}
 \psfrag{PP}{$P$}
 \psfrag{PQ}{$Q$}
 \psfrag{PT}{$T$}
 \psfrag{PH}{$2\phi$}
 \psfrag{TE}{$\theta$}
 \includegraphics[width=9cm]{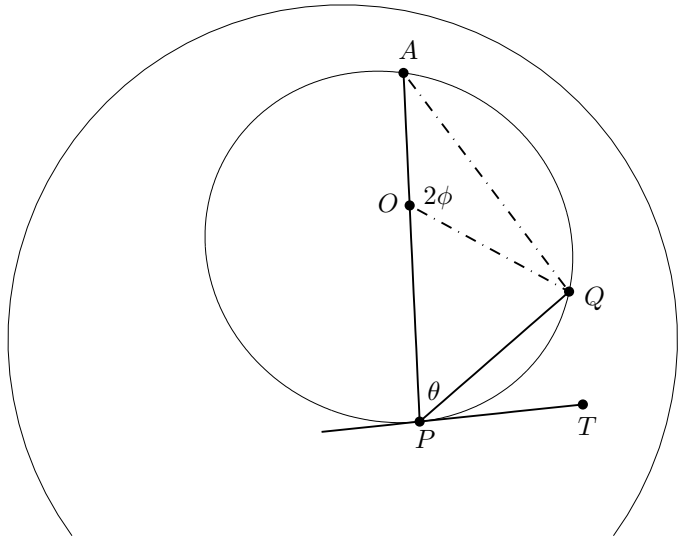}
\caption[A circumgyrocircle, circumgyroradius and a gyrotangent gyroline]{
The circumgyrocircle of gyrotriangle $APQ$ with circumgyrocenter $O$
in an Einstein gyrovector plane. The circumgyrocircle is shown along with
(i) its gyrodiameter $AP$ that extends its gyroradius $OP$, and
(ii) its gyrotangent gyroline at the gyrotangency point $P$.
When $Q$ approaches $P$, the gyroline through $P$ and $Q$
approaches the gyrotangent gyroline through $P$ and $T$.
Shown are also an inscribed gyroangle $\theta$ and a gyrocentral gyroangle $2\phi$,
which are related by the Inscribed Gyroangle Theorem II,
Theorem \ref{thmfdknb2}, p.~\pageref{thmfdknb2}.
\label{fig337m}}
\end{figure}
%%%%%%%%%%%%%%%%%%%%%%%%%%%%%%%%%%%%%%%%%%%%%%%%%%%%%%%%%%%%%%%%%%%%%%
 % A Gyrocircular gyrocevian in (+)E
%                       Fig.~\ref{fig337m}      Fig. 7.10
%%%%%%%%%%%%%%%%%%%%%%%%%%%%%%%%%%%%%%%%%%%%%%%%%%%%%%%%%%%%%%%%%%%%

%%%%%%%%%%%%%%%%%%%%%%%%%%%%%%%%%%%%%%%%%%%%%%%%%%%%%%%%%%%%%%%%%%%%
% THEOREM NUMBER 8.16
\begin{theorem}\label{hdgeb5}
%{\bf ().}
Let $C$ be a gyrocircle with gyrocenter $O$
and let $L_{^{PT}}$ be a gyrotangent gyroline of $C$ at a
tangency point $P$ in an Einstein gyrovector space $(\Rsn,\op,\od)$,
shown in Fig.~\ref{fig337m}.
Then, the gyrotangent gyroline $L_{^{PT}}$ is perpendicular to the
gyroradius $OP$ terminating at $P$.
\end{theorem}
\begin{proof}
The proof is given by the following chain of equations,
which are numbered for subsequent derivation.
Let $Q$ be a point close to $P$ on gyrocircle $C$,
shown in Fig.~\ref{fig337m}. Then,
%%%%%%%%%%%%%%%%%%%%%%%%%%%%%%%%%%%%%%%%%%%%%%%%%%%%%%%%%%%%%%%%%%%%
\begin{equation} \label{kengf7}
\begin{split}
\angle APT~~
&
\overbrace{=\!\!=\!\!=}^{(1)} \hspace{0.2cm}
\lim_{Q\rightarrow P} \angle APQ
\\&
\overbrace{=\!\!=\!\!=}^{(2)} \hspace{0.2cm}
\lim_{Q\rightarrow P} \left\{
\half \angle AOQ + \half \delta_{AOQ} - \half \delta_{APQ}
\right\}
\\&
\overbrace{=\!\!=\!\!=}^{(3)} \hspace{0.2cm}
\half \angle AOP + \half \delta_{AOP} - \half \delta_{APP}
\\&
\overbrace{=\!\!=\!\!=}^{(4)} \hspace{0.2cm}
\frac{\pi}{2}
\,.
 \end{split}
 \end{equation}
%%%%%%%%%%%%%%%%%%%%%%%%%%%%%%%%%%%%%%%%%%%%%%%%%%%%%%%%%%%%%%%%%%%%
Derivation of the numbered equalities in \eqref{kengf7} follows:
%%%%%%%%%%%%%%%%%%%%%%%%%%%%%%%%%%%%%%%%%%%%%%%%%%%%%%%%%%%%%%%%%%%%
\begin{enumerate}
\item \label{mddns1}
This limit is clear from Fig.~\ref{fig337m}.
Indeed, when $Q$ approaches $P$, the gyroline through $P$ and $Q$
approaches the gyrotangent gyroline through $P$ and $T$.
\item \label{mddns2}
Follows from (1) by the Inscribed Gyroangle Theorem II,
Theorem \ref{thmfdknb2}, p.~\pageref{thmfdknb2}, with
$\theta=\angle APQ$.
\item \label{mddns3}
Follows from (2) by obvious limits as $Q\rightarrow P$,
noting that the defect of a degenerate gyrotriangle (that is, a
gyrotriangle whose vertices are gyrocollinear) vanishes.
\item \label{mddns4}
Follows from (3) since $\angle AOP = \pi$.
\end{enumerate}
\end{proof}

%SECTION NUMBER 12
\section{Semi-Gyrocircle Gyrotriangles}\label{semigy7d8}
\index{semi-gyrocircle gyrotriangle}

%%%%%%%%%%%%%%%%%%%%%%%%%%%%%%%%%%%%%%%%%%%%%%%%%%%%%%%%%%%%%%%%%%%%
% FIGURE 11
 
%%%%%%%%%%%%%%%%%%%%%%%%%%%%%%%%%%%%%%%%%%%%%%%%%%%%%%%%%%%%%%%%%%%%%%
%%%%% The hyperbolic semi-circle theorem            %%%%%%%%%%%%%%%%%%
%\begin{figure}[htbp]
\begin{figure}[t]  % try to put this figure on the top of the page
 \centering         % center the figure
 \psfrag{A1}{$A_1$}
 \psfrag{A2}{$A_2$}
 \psfrag{A3}{$A_3$}
 \psfrag{O}{$O$}
 \psfrag{al1}{$\alpha_1$}
 \psfrag{al2}{$\alpha_2$}
 \psfrag{al3}{$\theta$}
 \psfrag{formula01}{$\theta + \half \delta_{A_1A_2A_3}= \frac{\pi}{2}$}
 \psfrag{formula02}{$\theta = \alpha_1+\alpha_2$}
 \includegraphics[width=9cm]{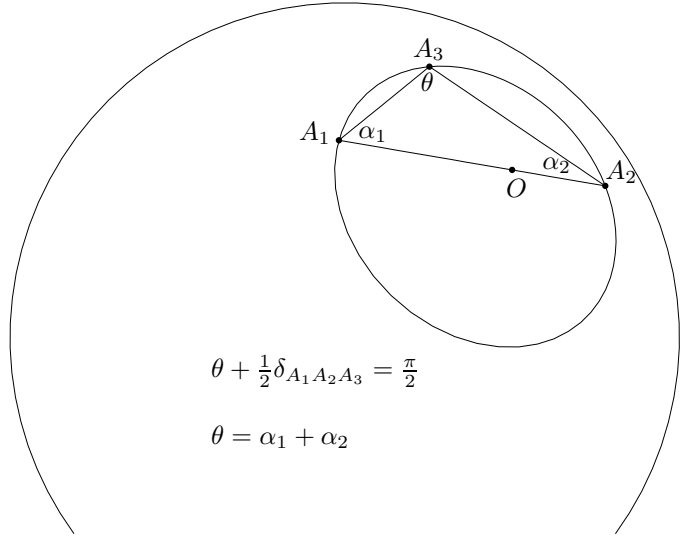}
\caption[The semi-gyrocircle gyrotriangle]{
The Semi-Gyrocircle Gyrotriangle.
\label{fig319enm}}
\end{figure}
%%%%%%%%%%%%%%%%%%%%%%%%%%%%%%%%%%%%%%%%%%%%%%%%%%%%%%%%%%%%%%%%%%%%%%
 % A Gyrocircular gyrocevian in (+)E
%                       Fig.~\ref{fig319enm}      Fig. 7.10
%%%%%%%%%%%%%%%%%%%%%%%%%%%%%%%%%%%%%%%%%%%%%%%%%%%%%%%%%%%%%%%%%%%%

In the special case when the gyrocircle gyrochord $A_1A_2$ in
Figs.~\ref{fig267ciiam}\,--\,\ref{fig267ciibm} is the gyrocircle gyrodiameter,
as shown in Fig.~\ref{fig319enm}, we have
$2\phi=\pi$ and $\delta_{A_1A_2O}=0$ so that,
by Theorem \ref{thmfdknb2}, \eqref{kfhsda},
\begin{equation} \label{kurish1}
\theta + \half \delta_{A_1A_2A_3}= \frac{\pi}{2}
\,.
\end{equation}

By definition, the gyrotriangular defect of gyrotriangle $A_1A_2A_3$
in Fig.~\ref{fig319enm} is given by the equation
\begin{equation} \label{kurish1p}
\delta_{A_1A_2A_3}=\pi-(\alpha_1+\alpha_2+\theta)
\,.
\end{equation}
Hence, by \eqref{kurish1}\,--\,\eqref{kurish1p},
\begin{equation} \label{kurish2}
\theta = \alpha_1+\alpha_2
\,.
\end{equation}

In Euclidean geometry triangle defects vanish, so that \eqref{kurish1}
reduces to the well known result according to which $\theta = \pi/2$
in Euclidean geometry.

% Chapter09.tex bk7 in cdcrc7

% SECTION NUMBER 13
\section{The Gyrotangent--Gyrosecant Theorem} \label{slila}

A gyrotangent\index{gyrotangent}
gyroline of a gyrocircle is a gyroline that intersects the gyrocircle in
exactly one point. The point of contact is called the point of
tangency.\index{tangency point}

A gyrosecant\index{gyrosecant}
gyroline of a gyrocircle is a gyroline that intersects the gyrocircle in
two different points.
The gyrosegment that links these points is a
gyrochord\index{gyrochord} of the gyrocircle.

By Theorem \ref{hdgeb5}, p.~\pageref{hdgeb5},
the gyroradius of a gyrocircle drawn to the point of tangency of a
gyrotangent gyroline of the gyrocircle is
perpendicular to the gyrotangent gyroline.
Accordingly, the gyrotangent gyrosegment $A_1P$ of the gyrocircle
in Fig.~\ref{fig304enm}, with the tangency point $A_1$ is
perpendicular to the gyroradius $OA_1$ drawn from the
gyrocircle gyrocenter $O$ to the tangency point $A_1$.
In full analogy with the well-known
tangent--secant theorem\index{tangent--secant theorem}
of Euclidean geometry we present the
gyrotangent--gyrosecant theorem\index{gyrotangent--gyrosecant theorem}
of hyperbolic geometry.

%%%%%%%%%%%%%%%%%%%%%%%%%%%%%%%%%%%%%%%%%%%%%%%%%%%%%%%%%%%%%%%%%%%%
% FIGURE 12
 
%%%%%%%%%%%%%%%%%%%%%%%%%%%%%%%%%%%%%%%%%%%%%%%%%%%%%%%%%%%%%%%%%%%%%%
%%%%% The hyperbolic semi-circle theorem            %%%%%%%%%%%%%%%%%%
%\begin{figure}[htbp]
\begin{figure}[t]  % try to put this figure on the top of the page
 \centering         % center the figure
 \psfrag{O}{$O$}
 \psfrag{P}{$P$}
 \psfrag{A1}{$A_1$}
 \psfrag{A2}{$A_2$}
 \psfrag{A3}{$A_3$}
 \psfrag{d1}{$d_1$}
 \psfrag{d2}{$d_2$}
 \psfrag{d3}{$d_3$}
 \psfrag{d23}{$d_{23}$}
 \psfrag{rr}{$R$}
 \psfrag{formula00}{$\angle OA_1P=\tfrac{\pi}{2}$}
 \psfrag{formula01}[]{$d_1=\|\om A_1\op P\|$}
 \psfrag{formula02}[]{$d_2=\|\om A_2\op P\|$}
 \psfrag{formula03}[]{$d_3=\|\om A_3\op P\|$}
 \psfrag{formula04}[]{$d_{23}=\|\om A_2\op A_3\|=d_3\om d_2$}
 \psfrag{formula05}[]{$R=\|\om A_k\op O\|,~k=1,2,3$}
 \includegraphics[width=9cm]{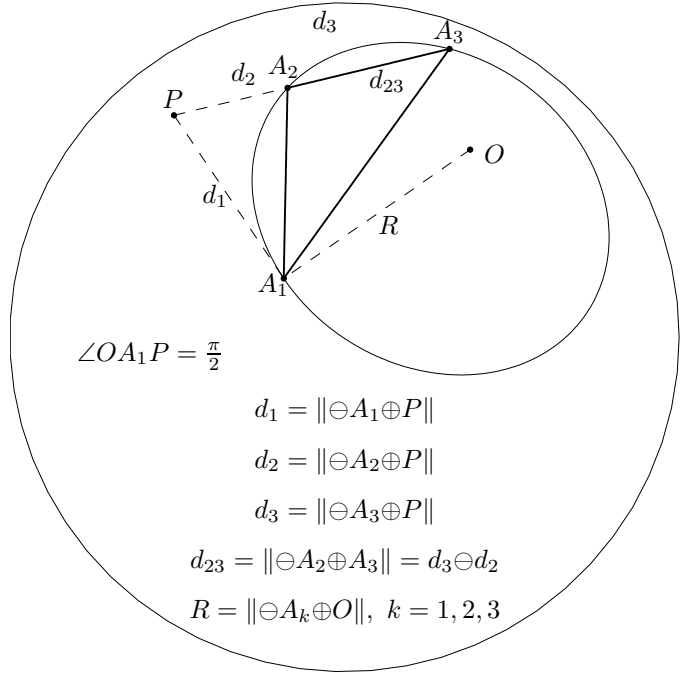}
\caption[Illustration of the Gyrotangent--Gyrosecant Theorem]{
Illustrating the Gyrotangent--Gyrosecant Theorem \ref{thmtnscnt},
a gyrotriangle $A_1A_2A_3$
in an Einstein gyrovector space $(\Rsn,\op,\od)$ is shown for $n=2$, along with
its circumgyrocircle and its circumgyrocenter $O$ and circumgyroradius $R=|OA_1|$.
The gyrotangent gyrosegment
$PA_1$ is tangent to the circumgyrocircle at the tangency point $A_1$.
Gyrosecant gyrosegment $PA_3$ intersects the circumgyrocircle at the points $A_2$ and $A_3$.
Gyrodistances $d_k$, $k=1,2,3$ between various points that illustrate
the gyrotangent--gyrosecant Theorem are also shown.
Thus, in particular, $d_{23}$ is the gyrodistance between $A_2$ and $A_3$.
Owing to the gyrotriangle equality\index{gyrotriangle equality}
$d_3=d_2\op d_{23}$, so that
$d_{23}$ equals the gyrodifference between $d_3$ and $d_2$.
\label{fig304enm}}
\end{figure}
%%%%%%%%%%%%%%%%%%%%%%%%%%%%%%%%%%%%%%%%%%%%%%%%%%%%%%%%%%%%%%%%%%%%%%
 % The gyrotriangle circumgyrocenter in (+)E
%                       Fig.~\ref{fig304enm}      Fig. 10.1
%%%%%%%%%%%%%%%%%%%%%%%%%%%%%%%%%%%%%%%%%%%%%%%%%%%%%%%%%%%%%%%%%%%%

%%%%%%%%%%%%%%%%%%%%%%%%%%%%%%%%%%%%%%%%%%%%%%%%%%%%%%%%%%%%%%%%%%%%
% THEOREM NUMBER 9.1
\index{gyrotangent--gyrosecant theorem, I}
\begin{theorem}\label{thmtnscnt}
{\bf (The Gyrotangent--Gyrosecant Theorem, I).}
Let $A_1,A_2,A_3\in\Rsn$ be three distinct points that lie on a gyrocircle
in an Einstein gyrovector space $\Rsn=(\Rsn,\op,\od)$ such that
the gyrosecant gyroline of the gyrocircle through the points $A_2,A_3$
and the gyrotangent gyroline of the gyrocircle with the tangency point $A_1$
share a point $P$ not on the gyrocircle, as shown in Fig.~\ref{fig304enm}.

Furthermore, let
%%%%%%%%%%%%%%%%%%%%%%%%%%%%%%%%%%%%%%%%%%%%%%%%%%%%%%%%%%%%%%%%%%%%
\begin{equation} \label{napofm1}
\begin{split}
d_1 &= \|\om A_1 \op P\| \\
d_2 &= \|\om A_2 \op P\| \\
d_3 &= \|\om A_3 \op P\|
\,.
\end{split}
\end{equation}
%%%%%%%%%%%%%%%%%%%%%%%%%%%%%%%%%%%%%%%%%%%%%%%%%%%%%%%%%%%%%%%%%%%%
and
\begin{equation} \label{napofm1s}
d_{23} = \|\om A_2 \op A_3\| = d_3 \om d_2
\,.
\end{equation}
Then
\begin{equation} \label{napofm2}
\gamma_{d_1}^2 d_1^2 = \frac{2}
{\gamma_{d_3\om d_2}^{\phantom{O}} +1}
\gamma_{d_2}^{\phantom{O}} d_2
\gamma_{d_3}^{\phantom{O}} d_3
\,.
\end{equation}
%MATLAB fig304en zerof2
\end{theorem}
\begin{proof}
Using the gyrotriangle index notation \eqref{indexnotation}, p.~\pageref{indexnotation},
we have, in particular,
\begin{equation} \label{napofp1}
\gammabc = \gamma_{\om A_2\op A_3}^{\phantom{O}}
         = \gamma_{\|\om A_2\op A_3\|}^{\phantom{O}}
\,.
\end{equation}
The points $P$, $A_2$ and $A_3$ are gyrocollinear, as shown in Fig.~\ref{fig304enm}.
Hence, by the gyrotriangle equality
\cite[Sect.~2.4]{mybook05},
\begin{equation} \label{napofp2}
d_2 \op \|\om A_2\op A_3\| = d_3
\,,
\end{equation}
implying
\begin{equation} \label{napofp3}
\|\om A_2\op A_3\| = d_3 \om d_2
\,,
\end{equation}
as stated in the Theorem, \eqref{napofm1s}.
Hence, by \eqref{napofp1} and \eqref{napofp3},
\begin{equation} \label{napofp4}
\gammabc = \gamma_{d_3\om d_2}^{\phantom{O}}
\,.
\end{equation}

Let $O$ be the circumgyrocenter of gyrotriangle $A_1A_2A_3$,
as shown in Fig.~\ref{fig304enm}.
Gyrosegment $PA_1$ is tangent to the circumgyrocircle of the gyrotriangle
at the tangency point $A_1$. Hence,
gyrotriangle $PA_1O$ is right gyroangled, with
$\angle PA_1O=\pi/2$, so that,
by the Einstein--Pythagoras Identity \cite[Eq.~(6.57), p.~144]{mybook05},
\index{Einstein--Pythagoras Theorem}
\begin{equation} \label{napof01}
\gamma_{\om A_1\op P}^{\phantom{O}}
\gamma_{\om A_1\op O}^{\phantom{O}}
=
\gamma_{\om P\op O}^{\phantom{O}}
\,.
\end{equation}

The circumgyrocenter $O$ of gyrotriangle $A_1A_2A_3$ in Fig.~\ref{fig304enm}
possesses the gyrobarycentric representation
\begin{equation} \label{napof02}
O = \frac{
m_{1,O} \gAa A_1 + m_{2,O} \gAb A_2 + m_{3,O} \gAc A_3
}{
m_{1,O} \gAa     + m_{2,O} \gAb     + m_{3,O} \gAc
}
\end{equation}
with respect to the gyrobarycentrically independent set $\{A_1,A_2,A_3\}$ where,
by \eqref{hfjdv07bb}, p.~\pageref{hfjdv07bb},
\begin{equation} \label{napof03}
\begin{split}
m_{1,O} &= (\phantom{-}\gammaab+\gammaac-\gammabc-1)(\gammabc-1) \\
m_{2,O} &= (\phantom{-}\gammaab-\gammaac+\gammabc-1)(\gammaac-1) \\
m_{3,O} &= (-\gammaab+\gammaac+\gammabc-1)(\gammaab-1)
\,.
\end{split}
\end{equation}
%MATHEMATICA stam282

Hence, by the Gyrobarycentric Representation Gyrocovariance Theorem
\cite[Theorem 4.6, pp.~90-91]{mybook05}
with $X=\om A_1$, we have
the following two equations \eqref{napof04} and \eqref{napof05}.
The first equation is
%%%%%%%%%%%%%%%%%%%%%%%%%%%%%%%%%%%%%%%%%%%%%%%%%%%%%%%%%%%%%%%%%%%%
\begin{equation} \label{napof04}
\begin{split}
\om A_1 \op O &= \om A_1 \op \frac{
m_{1,O} \gAa A_1 + m_{2,O} \gAb A_2 + m_{3,O} \gAc A_3
}{
m_{1,O} \gAa     + m_{2,O} \gAb     + m_{3,O} \gAc
}
\\[8pt] & \hspace{-1.2cm}=
\frac{
m_{1,O} \gamma_{\om A_1\op A_1}^{\phantom{O}} (\om A_1\op A_1) +
m_{2,O} \gamma_{\om A_1\op A_2}^{\phantom{O}} (\om A_1\op A_2) +
m_{3,O} \gamma_{\om A_1\op A_3}^{\phantom{O}} (\om A_1\op A_3)
}{
m_{1,O} \gamma_{\om A_1\op A_1}^{\phantom{O}} +
m_{2,O} \gamma_{\om A_1\op A_2}^{\phantom{O}} +
m_{3,O} \gamma_{\om A_1\op A_3}^{\phantom{O}}
}
\\[8pt] &= \frac{
m_{2,O} \gammaab\ab_{12} + m_{3,O} \gammaac\ab_{13}
}{
m_{1,O} + m_{2,O} \gammaab + m_{3,O} \gammaac
}
\,.
\end{split}
\end{equation}

The second equation is
\begin{equation} \label{napof05}
\gamma_{R}^{\phantom{O}} = \gamma_{\om A_1\op O}^{\phantom{O}} = \frac{
m_{1,O} + m_{2,O} \gammaab + m_{3,O} \gammaac
}{
\mO
}
\end{equation}
where $R=\|\om A_1\op O\|$ is the circumgyroradius of gyrotriangle $A_1A_2A_3$.
Here, the constant $\mO>0$
of the gyrobarycentric representation \eqref{napof02}
of $O$ is given by the equation
%%%%%%%%%%%%%%%%%%%%%%%%%%%%%%%%%%%%%%%%%%%%%%%%%%%%%%%%%%%%%%%%%%%%
\begin{equation} \label{napof06}
\begin{split}
m_O^2 &= m_{1,O}^2 + m_{2,O}^2 + m_{3,O}^2 \\
&+ 2(m_{1,O}m_{2,O}\gammaab + m_{1,O}m_{3,O}\gammaac + m_{2,O}m_{3,O}\gammabc)
\,,
\end{split}
\end{equation}
%%%%%%%%%%%%%%%%%%%%%%%%%%%%%%%%%%%%%%%%%%%%%%%%%%%%%%%%%%%%%%%%%%%%
as we see from the
{\it gyrobarycentric representation constant} associated with
the Gyrobarycentric Representation Gyrocovariance Theorem
\cite[Theorem 4.6, pp.~90-91]{mybook05}.
There will be no need to use the right-hand side of \eqref{napof06}
in the proof.

Let the point $P$, shown in Fig.~\ref{fig304enm}, be given by its
gyrobarycentric representation
\begin{equation} \label{napof07}
P = \frac{
m_2 \gAb A_2 + m_3 \gAc A_3
}{
m_2 \gAb + m_3 \gAc
}
\end{equation}
with respect to the gyrobarycentrically independent set $\{A_1,A_2\}$,
where the gyrobarycentric coordinates $m_2$ and $m_3$ are to be determined
in \eqref{napof18} below.
The gyrobarycentric representation \eqref{napof07} with respect to the
set $\{A_1,A_2\}$ exists since the point $P$ lies on the gyroline
that passes through the points $A_2$ and $A_3$.

By the Gyrobarycentric Representation Gyrocovariance Theorem
\cite[Theorem 4.6, pp.~90-91]{mybook05}
with $X=\om A_1$  we have
\begin{equation} \label{napof08}
\om A_1 \op P = \frac{
m_2 \gammaab\ab_{12} + m_3 \gammaac\ab_{13}
}{
m_2 \gammaab + m_3 \gammaac
}
\end{equation}
and
\begin{equation} \label{napof09}
\gamma_{\om A_1\op P}^{\phantom{O}} = \frac{
m_2\gammaab + m_3\gammaac}{\mP}
\,.
\end{equation}
Here, $\mP>0$ is the constant of the gyrobarycentric representation \eqref{napof07}
of $P$, given by the equation
\begin{equation} \label{napof10}
m_P^2 = m_2^2+m_3^2+2m_2m_3\gammabc
\,,
\end{equation}
as we see from the representation constant associated with
the Gyrobarycentric Representation Gyrocovariance Theorem
\cite[Theorem 4.6, pp.~90-91]{mybook05}.

Similarly to \eqref{napof08} and \eqref{napof09}, we have
%%%%%%%%%%%%%%%%%%%%%%%%%%%%%%%%%%%%%%%%%%%%%%%%%%%%%%%%%%%%%%%%%%%%
\begin{equation} \label{napof11}
\begin{split}
\om A_2 \op P = \frac{
m_3 \gammabc\ab_{23}
}{
m_2 + m_3 \gammabc
}
\\[8pt]
\om A_3 \op P = \frac{
m_2 \gammabc\ab_{32}
}{
m_2 \gammabc + m_3
}
\end{split}
\end{equation}
%%%%%%%%%%%%%%%%%%%%%%%%%%%%%%%%%%%%%%%%%%%%%%%%%%%%%%%%%%%%%%%%%%%%
and
%%%%%%%%%%%%%%%%%%%%%%%%%%%%%%%%%%%%%%%%%%%%%%%%%%%%%%%%%%%%%%%%%%%%
\begin{equation} \label{napof12}
\begin{split}
\gamma_{\om A_2\op P}^{\phantom{O}} = \frac{
m_2 + m_3\gammabc}{\mP}
\\[8pt]
\gamma_{\om A_3\op P}^{\phantom{O}} = \frac{
m_2\gammabc + m_3}{\mP}
\,.
\end{split}
\end{equation}
%%%%%%%%%%%%%%%%%%%%%%%%%%%%%%%%%%%%%%%%%%%%%%%%%%%%%%%%%%%%%%%%%%%%

Following \eqref{napof02} we have, by
the Gyrobarycentric Representation Gyrocovariance Theorem
\cite[Theorem 4.6, pp.~90-91]{mybook05},
%%%%%%%%%%%%%%%%%%%%%%%%%%%%%%%%%%%%%%%%%%%%%%%%%%%%%%%%%%%%%%%%%%%%
\begin{equation} \label{napof13}
\om P\op O = \frac{
m_{1,O} \gamma_{\om A_1\op P}^{\phantom{O}} (\om P\op A_1) +
m_{2,O} \gamma_{\om A_2\op P}^{\phantom{O}} (\om P\op A_2) +
m_{3,O} \gamma_{\om A_3\op P}^{\phantom{O}} (\om P\op A_3)
}{
m_{1,O} \gamma_{\om A_1\op P}^{\phantom{O}} +
m_{2,O} \gamma_{\om A_2\op P}^{\phantom{O}} +
m_{3,O} \gamma_{\om A_3\op P}^{\phantom{O}}
}
\,,
\end{equation}
%%%%%%%%%%%%%%%%%%%%%%%%%%%%%%%%%%%%%%%%%%%%%%%%%%%%%%%%%%%%%%%%%%%%
noting that
$\gamma_{\om P\op A_k}^{\phantom{O}}= \gamma_{\om A_k\op P}^{\phantom{O}}$,
$k=1,2,3$.\\[2pt]

Hence, by a gamma factor identity of the
Gyrobarycentric representation Gyrocovariance Theorem
\cite[Theorem 4.6, pp.~90-91]{mybook05},
applied to \eqref{napof13}, and by \eqref{napof09}, \eqref{napof11},
we have
%%%%%%%%%%%%%%%%%%%%%%%%%%%%%%%%%%%%%%%%%%%%%%%%%%%%%%%%%%%%%%%%%%%%
\begin{equation} \label{napof14}
\begin{split}
\gamma_{\om P\op O}^{\phantom{O}} &= \frac{
m_{1,O} \gamma_{\om A_1\op P}^{\phantom{O}} +
m_{2,O} \gamma_{\om A_2\op P}^{\phantom{O}} +
m_{3,O} \gamma_{\om A_3\op P}^{\phantom{O}}
}{
\mO
}
\\[8pt] & \hspace{-0.6cm}= \frac{
m_{1,O}(m_2\gammaab+m_3\gammaac) +
m_{2,O}(m_2        +m_3\gammabc) +
m_{3,O}(m_2\gammabc+m_3        )
}{\mO\mP}
\,.
\end{split}
\end{equation}
%%%%%%%%%%%%%%%%%%%%%%%%%%%%%%%%%%%%%%%%%%%%%%%%%%%%%%%%%%%%%%%%%%%%

Substituting \eqref{napof05}, \eqref{napof09} and \eqref{napof14}
into the Einstein--Pythagoras Identity\index{Einstein--Pythagoras Identity} \eqref{napof01},
we obtain the equation
\begin{equation} \label{napof15}
\begin{split}
& m_{1,O}(m_2\gammaab+m_3\gammaac) + m_{2,O}(m_2 +m_3\gammabc) + m_{3,O}(m_2\gammabc+m_3)
\\ &=
(m_2\gammaab+m_3\gammaac) ( m_{1,O} + m_{2,O}\gammaab + m_{3,O}\gammaac)
\,.
\end{split}
\end{equation}
The latter, in turn, can be written as
\begin{equation} \label{napof16}
\begin{split}
0 &=
m_2\{m_{2,O}(\gamma_{12}^2-1) + m_{3,O} (\gammaab\gammaac-\gammabc)\}
\\ &+
m_3\{m_{2,O}(\gammaab\gammaac-\gammabc) + m_{3,O} (\gamma_{13}^2-1)\}
\,.
\end{split}
\end{equation}
Solving \eqref{napof16} for $m_2$ and $m_3$ we obtain the equations
%%%%%%%%%%%%%%%%%%%%%%%%%%%%%%%%%%%%%%%%%%%%%%%%%%%%%%%%%%%%%%%%%%%%
\begin{equation} \label{napof17}
\begin{split}
m_2 &= K\{m_{2,O}(\gammaab\gammaac-\gammabc) + m_{3,O} (\gamma_{13}^2-1) \}
\\
m_3 &=-K\{m_{2,O}(\gamma_{12}^2-1) + m_{3,O} (\gammaab\gammaac-\gammabc) \}
\end{split}
\end{equation}
%%%%%%%%%%%%%%%%%%%%%%%%%%%%%%%%%%%%%%%%%%%%%%%%%%%%%%%%%%%%%%%%%%%%
%MATHEMATICA stam282
for any nonzero factor $K$.

Owing to the homogeneity of gyrobarycentric coordinates, the factor $K$ is irrelevant.
Hence, we select $K=1$, obtaining the equations
%%%%%%%%%%%%%%%%%%%%%%%%%%%%%%%%%%%%%%%%%%%%%%%%%%%%%%%%%%%%%%%%%%%%
\begin{equation} \label{napof18}
\begin{split}
m_2 &= m_{2,O}(\gammaab\gammaac-\gammabc) + m_{3,O} (\gamma_{13}^2-1)
\\
m_3 &=-m_{2,O}(\gamma_{12}^2-1) - m_{3,O} (\gammaab\gammaac-\gammabc)
\,,
\end{split}
\end{equation}
%MATHEMATICA stam282
%%%%%%%%%%%%%%%%%%%%%%%%%%%%%%%%%%%%%%%%%%%%%%%%%%%%%%%%%%%%%%%%%%%%
where $m_{2,O}$ and $m_{3,O}$ are given by \eqref{napof03}.

Equations \eqref{napof18} for the gyrobarycentric coordinates $m_2$ and $m_3$
of the point $P$ in \eqref{napof07} result from the
Einstein--Pythagoras Identity\index{Einstein--Pythagoras Identity} \eqref{napof01}.
Hence, they insure that gyrosegments $A_1P$ and $A_1O$
are perpendicular to each other, as desired.

By the first equation in \eqref{napofm1},
and by \eqref{rugh1ds}, p.~\pageref{rugh1ds}, we have
\begin{equation} \label{napof19}
\frac{1}{s^2} \gamma_{d_1}^2 d_1^2 = \gamma_{d_1}^2 - 1
= \gamma_{\om A_1\op P}^2 -1
\,,
\end{equation}
where the gamma factor $\gamma_{\om A_1\op P}^{\phantom{O}}$
is given by \eqref{napof09}, \eqref{napof10} and \eqref{napof18}.

Substituting  \eqref{napof09} into  \eqref{napof19}, we have
\begin{equation} \label{napof20}
\frac{1}{s^2} \gamma_{d_1}^2 d_1^2 = \frac{
(m_2\gammaab+m_3\gammaac)^2-m_P^2}{m_P^2}
\,,
\end{equation}
where $m_2$, $m_3$ and $m_P^2$ are given by \eqref{napof18} and \eqref{napof10}.

Hence, by \eqref{napof20}, \eqref{napof18} and
by straightforward algebra,
\begin{equation} \label{napof21a}
\frac{1}{s^2} \gamma_{d_1}^2 d_1^2 = \frac{2}{s^2D^2m_P^2}
(\gammaab-1) (\gammaac-1) (\gammabc-1)
\,,
%MATHEMATICA stam282
\end{equation}
where
\begin{equation} \label{napof21b}
D = \gamma_{12}^2+\gamma_{13}^2+\gamma_{23}^2 - 2\gammaab\gammaac\gammabc
\,.
\end{equation}

Similarly to \eqref{napof19}, by \eqref{napofm1} and
\eqref{rugh1ds}, p.~\pageref{rugh1ds}, we have
\begin{equation} \label{napof22}
\frac{1}{s^2} \gamma_{d_2}^2 d_2^2 = \gamma_{d_2}^2 - 1
= \gamma_{\om A_2\op P}^2 -1
\,,
\end{equation}
where the gamma factor $\gamma_{\om A_2\op P}^{\phantom{O}}$
is given by the first equation in \eqref{napof12},
and by \eqref{napof10} and \eqref{napof18}.

Substituting the first equation in \eqref{napof12} into  \eqref{napof22}, we have
\begin{equation} \label{napof23}
\frac{1}{s^2} \gamma_{d_2}^2 d_2^2 = \frac{
(m_2+m_3\gammabc)^2-m_P^2}{m_P^2}
\,,
\end{equation}
where $m_2$, $m_3$ and $m_P^2$ are given by \eqref{napof18} and \eqref{napof10}.

Hence, by \eqref{napof23}, \eqref{napof20} and straightforward algebra,
\begin{equation} \label{napof24}
\frac{1}{s^2} \gamma_{d_2}^2 d_2^2 = \frac{1}{s^2D^2m_P^2}
(\gammaab-1)^2 (\gamma_{23}^2-1)
\,.
%MATHEMATICA stam282
\end{equation}

Similarly to \eqref{napof19} and \eqref{napof22}, by \eqref{napofm1} and
\eqref{rugh1ds}, p.~\pageref{rugh1ds}, we have
\begin{equation} \label{napof25}
\frac{1}{s^2} \gamma_{d_3}^2 d_3^2 = \gamma_{d_3}^2 - 1
= \gamma_{\om A_3\op P}^2 -1
\,,
\end{equation}
where the gamma factor $\gamma_{\om A_3\op P}^{\phantom{O}}$
is given by the second equation in \eqref{napof12},
and by \eqref{napof10} and \eqref{napof18}.

Substituting the second equation in \eqref{napof12} into  \eqref{napof25}, we have
\begin{equation} \label{napof26}
\frac{1}{s^2} \gamma_{d_3}^2 d_3^2 = \frac{
(m_2\gammabc+m_3)^2-m_P^2}{m_P^2}
\,,
\end{equation}
where $m_2$, $m_3$ and $m_P^2$ are given by \eqref{napof18} and \eqref{napof10}.

Hence, by \eqref{napof26}, \eqref{napof18} and straightforward algebra,
\begin{equation} \label{napof27}
\frac{1}{s^2} \gamma_{d_3}^2 d_3^2 = \frac{1}{s^2D^2m_P^2}
(\gammaac-1)^2 (\gamma_{23}^2-1)
\,.
%MATHEMATICA stam282
\end{equation}

Finally, following \eqref{napof21a}, \eqref{napof24} and \eqref{napof27}
and straightforward algebra, noting \eqref{napofp4}, we have
\begin{equation} \label{napof28}
\frac{
(\gamma_{d_1}^2 d_1^2)^2
}{
\gamma_{d_2}^2 d_2^2 \gamma_{d_3}^2 d_3^2
}
= \left( \frac{2}{\gammabc+1} \right)^2
= \left( \frac{2}{\gamma_{d_3\om d_2}^{\phantom{O}}+1} \right)^2
\,,
\end{equation}
thus verifying the result \eqref{napofm2} of the Theorem.
\end{proof}

In order to restate Theorem \ref{thmtnscnt} in a way that emphasizes
analogies with its Euclidean counterpart, we introduce the notation
\begin{equation} \label{manof01}
|AB|:=\|\om A\op B\|=\|\om B\op A\|
\hspace{1.2cm} {\rm (Hyperbolic~Geometry)}
\end{equation}
for the gyrodistance between points $A$ and $B$
of an Einstein gyrovector space $(\Rsn,\op,\od)$.
Accordingly, $|AB|$ is the gyrolength of gyrosegment $AB$.

In order to emphasize analogies we use, ambiguously,
the same notation in the context of Euclidean geometry as well,
that is,
\begin{equation} \label{manof01euc}
|AB|:=\|-A+B\|=\|-B+A\|
\hspace{1.2cm} {\rm (Euclidean~Geometry)}
\end{equation}
is the distance between points $A$ and $B$ of a Euclidean vector space $\Rn$.
Accordingly, $|AB|$ is the length of segment $AB$.
It should always be clear from the context whether $|AB|$ represents
the gyrolength of a gyrosegment in hyperbolic geometry,
or the length of a segment in Euclidean geometry.

Using the notation in \eqref{manof01}\,--\,\eqref{manof01euc},
we restate Theorem \ref{thmtnscnt} as follows,
noting that by \eqref{napofp3} and \eqref{napofm1},
\begin{equation} \label{manofi}
d_3 \om d_2 = |PA_3|\om |PA_2|=|A_2A_3|
\,.
\end{equation}

%%%%%%%%%%%%%%%%%%%%%%%%%%%%%%%%%%%%%%%%%%%%%%%%%%%%%%%%%%%%%%%%%%%%
% THEOREM NUMBER 9.2
\index{gyrotangent--gyrosecant theorem, II}
\begin{theorem}\label{thmtnscns}
{\bf (The Gyrotangent--Gyrosecant Theorem, II).}
If a gyrotangent of a gyrocircle from an external point $P$ meets the
gyrocircle at $A_1$,
and a gyrosecant from $P$ meets the gyrocircle at $A_2$ and $A_3$,
as shown in Fig.~\ref{fig304enm}, then
\begin{equation} \label{manof02s}
\gamma_{|PA_1|}^2 |PA_1|^2 = \frac{2}{
\gamma_{|A_2A_3|}^{\phantom{O}} +1
}
\gamma_{|PA_2|}^{\phantom{O}} |PA_2|
\gamma_{|PA_3|}^{\phantom{O}} |PA_3|
\,.
\end{equation}
\end{theorem}
%%%%%%%%%%%%%%%%%%%%%%%%%%%%%%%%%%%%%%%%%%%%%%%%%%%%%%%%%%%%%%%%%%%%

In the Euclidean limit, $s\rightarrow\infty$,
gyrolengths of gyrosegments tend to lengths of corresponding segments
and gamma factors tend to 1. Hence, in that limit,
the Gyrotangent--Gyrosecant Theorem \ref{thmtnscns}
reduces to the following well-known Tangent--Secant Theorem of Euclidean geometry:

%%%%%%%%%%%%%%%%%%%%%%%%%%%%%%%%%%%%%%%%%%%%%%%%%%%%%%%%%%%%%%%%%%%%
% THEOREM NUMBER 9.3
\index{tangent--secant theorem}
\begin{theorem}\label{thmtnscnseuc}
{\bf (The Tangent--Secant Theorem).}
If a tangent of a circle from an external point $P$ meets the
circle at $A_1$,
and a secant from $P$ meets the circle at $A_2$ and $A_3$, then
\begin{equation} \label{manof02euc}
|PA_1|^2 = |PA_2| |PA_3|
\,.
\end{equation}
\end{theorem}
%%%%%%%%%%%%%%%%%%%%%%%%%%%%%%%%%%%%%%%%%%%%%%%%%%%%%%%%%%%%%%%%%%%%

% SECTION NUMBER 14
\section{The Intersecting Gyrosecants Theorem} \label{slila2}

%%%%%%%%%%%%%%%%%%%%%%%%%%%%%%%%%%%%%%%%%%%%%%%%%%%%%%%%%%%%%%%%%%%%
% FIGURE 13
 
%%%%%%%%%%%%%%%%%%%%%%%%%%%%%%%%%%%%%%%%%%%%%%%%%%%%%%%%%%%%%%%%%%%%%%
%%%%% The hyperbolic semi-circle theorem            %%%%%%%%%%%%%%%%%%
%\begin{figure}[htbp]
\begin{figure}[t]  % try to put this figure on the top of the page
 \centering         % center the figure
 \psfrag{O}{$O$}
 \psfrag{P}{$P$}
 \psfrag{A1}{$A_1$}
 \psfrag{A2}{$A_2$}
 \psfrag{A3}{$A_3$}
 \psfrag{B2}{$B_2$}
 \psfrag{B3}{$B_3$}
%\psfrag{d1}{$d_1$}
%\psfrag{d2}{$d_2$}
%\psfrag{d3}{$d_3$}
%\psfrag{d23}{$d_{23}$}
%\psfrag{rr}{$R$}
%\psfrag{formula00}{$\angle OA_1P=\tfrac{\pi}{2}$}
%\psfrag{formula01}[]{$d_1=\|\om A_1\op P\|$}
%\psfrag{formula02}[]{$d_2=\|\om A_2\op P\|$}
%\psfrag{formula03}[]{$d_3=\|\om A_3\op P\|$}
%\psfrag{formula04}[]{$d_{23}=\|\om A_2\op A_3\|=d_3\om d_2$}
%\psfrag{formula05}[]{$R=\|\om A_k\op O\|,~k=1,2,3$}
%
%\includegraphics[width=9cm]{/home/ungar/dir_amy/dir_papers/dir_mybook01/dir_figs/fig312en.eps}
 \includegraphics[width=9cm]{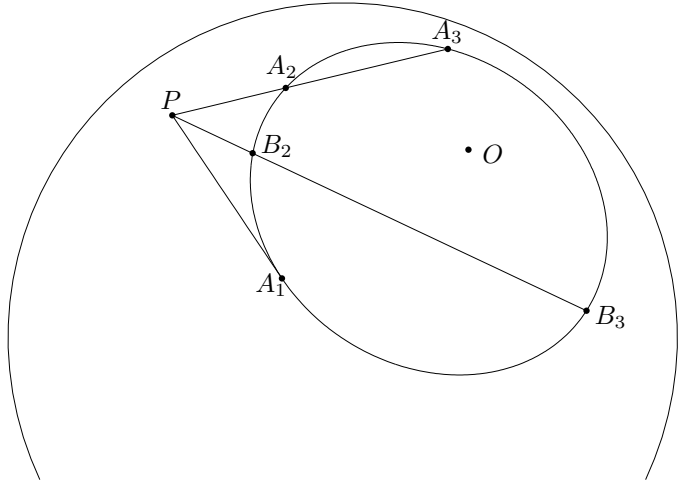}
\caption[Illustration of the Intersecting Gyrosecants Theorem]{
Illustrating the Intersecting Gyrosecants Theorem \ref{thmscscnt},
two intersecting gyrosecants $PA_3$ and $PB_3$ of a gyrocircle are shown.
They, respectively, intersect the gyrocircle at the points $A_2,A_3$
and at the points $B_2,B_3$.
\label{fig312enm}}
\end{figure}
%%%%%%%%%%%%%%%%%%%%%%%%%%%%%%%%%%%%%%%%%%%%%%%%%%%%%%%%%%%%%%%%%%%%%%
 % The Intersecting Gyrosecants Thm  in (+)E
%                       Fig.~\ref{fig312enm}      Fig. 10.2
%%%%%%%%%%%%%%%%%%%%%%%%%%%%%%%%%%%%%%%%%%%%%%%%%%%%%%%%%%%%%%%%%%%%

As an obvious corollary of the Gyrotangent--Gyrosecant Theorem \ref{thmtnscns}
we have the following theorem for intersecting gyrosecants of a gyrocircle:

%%%%%%%%%%%%%%%%%%%%%%%%%%%%%%%%%%%%%%%%%%%%%%%%%%%%%%%%%%%%%%%%%%%%
% THEOREM NUMBER 9.4
\index{intersecting gyrosecants theorem}
\index{gyrosecants, intersecting, theorem}
\begin{theorem}\label{thmscscnt}
{\bf (The Intersecting Gyrosecants Theorem).}
If two gyrosecants of a gyrocircle in an Einstein gyrovector space $(\Rsn,\op,\od)$,
drawn to the gyrocircle from an
external point $P$, meet the gyrocircle at points $A_2,A_3$ and
at points $B_2,B_3$, respectively, as shown in Fig.~\ref{fig312enm},
then
\begin{equation} \label{manof03}
\frac{
\gamma_{|PA_2|}^{\phantom{O}} |PA_2|
\gamma_{|PA_3|}^{\phantom{O}} |PA_3|}
{\gamma_{|A_2A_3|}^{\phantom{O}} +1}
=
\frac{
\gamma_{|PB_2|}^{\phantom{O}} |PB_2|
\gamma_{|PB_3|}^{\phantom{O}} |PB_3|}
{\gamma_{|B_2B_3|}^{\phantom{O}} +1}
\,.
\end{equation}
\end{theorem}
\begin{proof}
Let $PA_1$ be a gyrotangent gyrosegment of the gyrocircle drawn from $P$
and meeting the gyrocircle at $A_1$, as shown in Fig.~\ref{fig312enm}.
Then, by Theorem \ref{thmtnscns},
each of the two sides of \eqref{manof03} equals half the left-hand side
of \eqref{manof02s} thus verifying \eqref{manof03}.
\end{proof}

In the Euclidean limit, $s\rightarrow\infty$,
gyrolengths of gyrosegments tend to lengths of corresponding segments
and gamma factors tend to 1. Hence, in that limit,
the Intersecting Gyrosecants Theorem \ref{thmscscnt}
of hyperbolic geometry reduces to the
following well-known Intersecting Secants Theorem of Euclidean geometry:

%%%%%%%%%%%%%%%%%%%%%%%%%%%%%%%%%%%%%%%%%%%%%%%%%%%%%%%%%%%%%%%%%%%%
% THEOREM NUMBER 9.5
\index{intersecting secants theorem}
\index{secants, intersecting, theorem}
\begin{theorem}\label{thmscscnteuc}
{\bf (The Intersecting Secants Theorem).}
If two secants of a circle in a Euclidean vector space $\Rn$,
drawn to the circle from an
external point $P$, meet the circle at points $A_2,A_3$ and
at points $B_2,B_3$, respectively,
then
\begin{equation} \label{manof03euc}
|PA_2| |PA_3| = |PB_2| |PB_3|
\,.
\end{equation}
\end{theorem}
%%%%%%%%%%%%%%%%%%%%%%%%%%%%%%%%%%%%%%%%%%%%%%%%%%%%%%%%%%%%%%%%%%%%

% SECTION NUMBER 15
\section{Gyrocircle Gyrobarycentric Representation} \label{slila3}
\index{gyrocircle gyrobarycentric representation}

Let $A_1A_2A_3$ be a gyrotriangle that possesses a circumgyrocircle
in an Einstein gyrovector space $(\Rsn,\op,\od)$.
The circumgyrocenter $O$ of gyrotriangle $A_1A_2A_3$,
shown in Figs.~\ref{fig312enm} and \ref{fig313enm},
possesses the gyrobarycentric representation
\begin{equation} \label{napog02}
O = \frac{
m_{1,O} \gAa A_1 + m_{2,O} \gAb A_2 + m_{3,O} \gAc A_3
}{
m_{1,O} \gAa     + m_{2,O} \gAb     + m_{3,O} \gAc
}
\end{equation}
with respect to the gyrobarycentrically independent set $\{A_1,A_2,A_3\}$ of
the reference gyrotriangle vertices where,
by \eqref{hfjdv07bb}, p.~\pageref{hfjdv07bb},
\begin{equation} \label{napog03}
\begin{split}
m_{1,O} &= (\phantom{-}\gammaab+\gammaac-\gammabc-1)(\gammabc-1) \\
m_{2,O} &= (\phantom{-}\gammaab-\gammaac+\gammabc-1)(\gammaac-1) \\
m_{3,O} &= (-\gammaab+\gammaac+\gammabc-1)(\gammaab-1)
\,,
\end{split}
\end{equation}
as in \eqref{napof02} and \eqref{napof03}.

The constant $\mO>0$
of the gyrobarycentric representation \eqref{napog02}
of $O$ is given by \eqref{gkdnseu}\,--\,\eqref{drekc}, p.~\pageref{gkdnseu},
%%%%%%%%%%%%%%%%%%%%%%%%%%%%%%%%%%%%%%%%%%%%%%%%%%%%%%%%%%%%%%%%%%%%
\begin{equation} \label{napog06}
\begin{split}
m_O^2 &= m_{1,O}^2 + m_{2,O}^2 + m_{3,O}^2 \\
&+ 2(m_{1,O}m_{2,O}\gammaab + m_{1,O}m_{3,O}\gammaac + m_{2,O}m_{3,O}\gammabc)
\\&=
D_3(D_3-H_3)
\,.
\end{split}
\end{equation}
% Calculated in MATHEMATICA stam133 zeromos, (stam283)
% Corroborated numerically for posneg in test0437 (posv3).
%%%%%%%%%%%%%%%%%%%%%%%%%%%%%%%%%%%%%%%%%%%%%%%%%%%%%%%%%%%%%%%%%%%%

With $\AAbt = A_1\op\rmspan\{\om A_1\op A_2,\om A_1\op A_3\} \subset\Rn$,
let $A$ be a generic point in $\AAbt\cap\Rsn$, given by its
gyrobarycentric representation
\begin{equation} \label{napog04}
A = \frac{
m_1 \gAa A_1 + m_2 \gAb A_2 + m_3 \gAc A_3
}{
m_1 \gAa + m_2 \gAb + m_3 \gAc
}
\end{equation}
with respect to the gyrobarycentrically independent set $S=\{A_1,A_2,A_3\}$.
A relationship between
the gyrobarycentric coordinates $m_1,m_2$ and $m_3$ of $A$ is to be determined
in \eqref{muramd04} below
by the condition that the point $A$ lies on the circumgyrocircle of
gyrotriangle $A_1A_2A_3$, as shown in Fig.~\ref{fig313enm}.

%%%%%%%%%%%%%%%%%%%%%%%%%%%%%%%%%%%%%%%%%%%%%%%%%%%%%%%%%%%%%%%%%%%%
% FIGURE 14
 
%%%%%%%%%%%%%%%%%%%%%%%%%%%%%%%%%%%%%%%%%%%%%%%%%%%%%%%%%%%%%%%%%%%%%%
%%%%% The hyperbolic semi-circle theorem            %%%%%%%%%%%%%%%%%%
%\begin{figure}[htbp]
\begin{figure}[t]  % try to put this figure on the top of the page
 \centering         % center the figure
 \psfrag{A}{$A$}
 \psfrag{O}{$O$}
 \psfrag{A1}{$A_1$}
 \psfrag{A2}{$A_2$}
 \psfrag{A3}{$A_3$}
 \psfrag{A4}{$A_4$}
 \psfrag{d1}{$d_1$}
 \psfrag{d2}{$d_2$}
 \psfrag{d3}{$d_3$}
 \psfrag{d23}{$d_{23}$}
 \psfrag{rr}{$R$}
 \psfrag{formula01}[]{$d_1=\|\om A_1\op A\|$}
 \psfrag{formula02}[]{$d_2=\|\om A_2\op A\|$}
 \psfrag{formula03}[]{$d_3=\|\om A_3\op A\|$}
 \psfrag{formula04}[]{$d_{23}=\|\om A_2\op A_3\|=d_3\om d_2$}
 \psfrag{formula05}[]{$R=\|\om A_k\op O\|,~k=1,2,3$}
 \includegraphics[width=9cm]{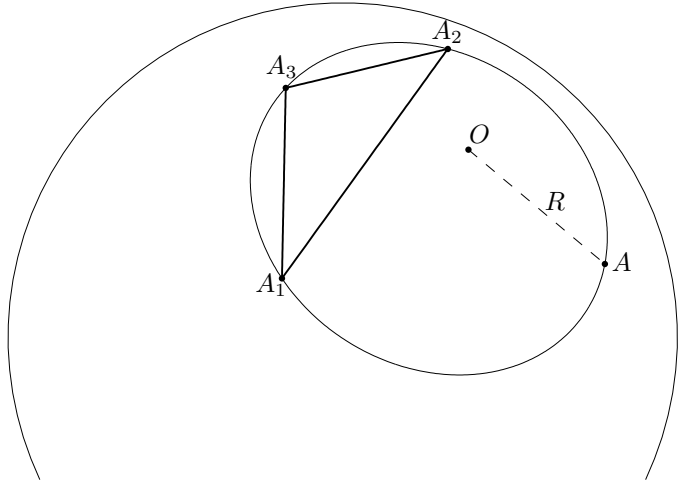}
\caption[A generic point on a gyrotriangle circumgyrocircle]{
A generic point $A$ on the
circumgyrocircle $C(A_1A_2A_3)$ of a gyrotriangle $A_1A_2A_3$
in an Einstein gyrovector space $(\Rsn,\op,\od)$.
The point $A$ lies on the circumgyrocircle $C(A_1A_2A_3)$ if and only if
it possesses the gyrobarycentric representation \eqref{napog04}
with gyrobarycentric coordinates that satisfy \eqref{muramd04}.
\label{fig313enm}}
\end{figure}
%%%%%%%%%%%%%%%%%%%%%%%%%%%%%%%%%%%%%%%%%%%%%%%%%%%%%%%%%%%%%%%%%%%%%%
 % The gyrotriangle circumgyrocenter in (+)E
%                       Fig.~\ref{fig313enm}      Fig. 10.3
%%%%%%%%%%%%%%%%%%%%%%%%%%%%%%%%%%%%%%%%%%%%%%%%%%%%%%%%%%%%%%%%%%%%

By the Gyrobarycentric Representation Gyrocovariance Theorem
\cite[Theorem 4.6, pp.~90-91]{mybook05}
with $X=\om A$ we have
%%%%%%%%%%%%%%%%%%%%%%%%%%%%%%%%%%%%%%%%%%%%%%%%%%%%%%%%%%%%%%%%%%%%
\begin{equation} \label{napof04s}
\begin{split}
\om A \op O &= \om A \op \frac{
m_{1,O} \gAa A_1 + m_{2,O} \gAb A_2 + m_{3,O} \gAc A_3
}{
m_{1,O} \gAa     + m_{2,O} \gAb     + m_{3,O} \gAc
}
\\[8pt] & \hspace{-1.2cm}=
\frac{
m_{1,O} \gamma_{\om A_1\op A}^{\phantom{O}} (\om A\op A_1) +
m_{2,O} \gamma_{\om A_2\op A}^{\phantom{O}} (\om A\op A_2) +
m_{3,O} \gamma_{\om A_3\op A}^{\phantom{O}} (\om A\op A_3)
}{
m_{1,O} \gamma_{\om A_1\op A}^{\phantom{O}} +
m_{2,O} \gamma_{\om A_2\op A}^{\phantom{O}} +
m_{3,O} \gamma_{\om A_3\op A}^{\phantom{O}}
}
\end{split}
\end{equation}
%%%%%%%%%%%%%%%%%%%%%%%%%%%%%%%%%%%%%%%%%%%%%%%%%%%%%%%%%%%%%%%%%%%%
and
\begin{equation} \label{napof05s}
\gamma_{d}^{\phantom{O}} = \gamma_{\om A\op O}^{\phantom{O}} = \frac{
m_{1,O} \gamma_{\om A_1\op A}^{\phantom{O}} +
m_{2,O} \gamma_{\om A_2\op A}^{\phantom{O}} +
m_{3,O} \gamma_{\om A_3\op A}^{\phantom{O}}
}{
\mO
}
\,,
\end{equation}
where $d=\|\om A\op O\|$ is the gyrodistance from $A$ to $O$,
and where the constant $\mO>0$ of the gyrobarycentric representation \eqref{napog02}
of $O$ is given by \eqref{napog06}.

We will now calculate the gamma factors
$\gamma_{\om A_k\op A}^{\phantom{O}}$, $k=1,2,3$, that appear in
\eqref{napof05s}.

Applying the Gyrobarycentric Representation Gyrocovariance Theorem
\cite[Theorem 4.6, pp.~90-91]{mybook05}
with $X=\om A_1$
to the gyrobarycentric representation \eqref{napog04} of $A$,
we have
%%%%%%%%%%%%%%%%%%%%%%%%%%%%%%%%%%%%%%%%%%%%%%%%%%%%%%%%%%%%%%%%%%%%
\begin{equation} \label{manod08}
\begin{split}
\om A_1 \op A &= \om A_1 \op \frac{
m_1 \gAa A_1 + m_2 \gAb A_2 + m_3 \gAc A_3
}{
m_1 \gAa     + m_2 \gAb     + m_3 \gAc
}
\\[8pt] & \hspace{-1.2cm}=
\frac{
m_1 \gamma_{\om A_1\op A_1}^{\phantom{O}} (\om A_1\op A_1) +
m_2 \gamma_{\om A_1\op A_2}^{\phantom{O}} (\om A_1\op A_2) +
m_3 \gamma_{\om A_1\op A_3}^{\phantom{O}} (\om A_1\op A_3)
}{
m_1 \gamma_{\om A_1\op A_1}^{\phantom{O}} +
m_2 \gamma_{\om A_1\op A_2}^{\phantom{O}} +
m_3 \gamma_{\om A_1\op A_3}^{\phantom{O}}
}
\\[8pt] &= \frac{
m_2 \gammaab\ab_{12} + m_3 \gammaac\ab_{13}
}{
m_1 + m_2 \gammaab + m_3 \gammaac
}
\end{split}
\end{equation}
%%%%%%%%%%%%%%%%%%%%%%%%%%%%%%%%%%%%%%%%%%%%%%%%%%%%%%%%%%%%%%%%%%%%
and
\begin{equation} \label{manod08p}
\gamma_{\om A_1\op A}^{\phantom{O}} = \frac{
m_1 + m_2 \gammaab + m_3 \gammaac
}{
\mA
}
\,,
\end{equation}
where $\mA>0$ is the constant of the gyrobarycentric representation
\eqref{napog04} of $A$, given by the equation
%%%%%%%%%%%%%%%%%%%%%%%%%%%%%%%%%%%%%%%%%%%%%%%%%%%%%%%%%%%%%%%%%%%%
\begin{subequations} \label{napof06s}
%%%%%%%%%%%%%%%%%%%%%%%%%%%%%%%%%%%%%%%%%%%%%%%%%%%%%%%%%%%%%%%%%%%%
\begin{equation} \label{napof06sa}
\begin{split}
m_A^2 &= m_1^2 + m_2^2 + m_3^2
+ 2(m_1m_2\gammaab + m_1m_3\gammaac + m_2m_3\gammabc)
\\&=
(m_1+m_2+m_3)^2
\\&
+2\{m_1m_2(\gammaab-1)+m_1m_3(\gammaac-1)+m_2m_3(\gammabc-1)\}
\,.
\end{split}
\end{equation}
%%%%%%%%%%%%%%%%%%%%%%%%%%%%%%%%%%%%%%%%%%%%%%%%%%%%%%%%%%%%%%%%%%%%

It proves useful to write \eqref{napof06sa} as
\begin{equation} \label{napof06sb}
m_A^2=(m_1+m_2+m_3)^2 + 2K(A;A_1,A_2,A_3)
\,,
\end{equation}
where
\begin{equation} \label{napof06sc}
K=K(A;A_1,A_2,A_3) =
\sum_{\substack{i,j=1\\i<j}}^3 m_im_j(\gamma_{ij}^{\phantom{O}} - 1)
\,,
\end{equation}
%%%%%%%%%%%%%%%%%%%%%%%%%%%%%%%%%%%%%%%%%%%%%%%%%%%%%%%%%%%%%%%%%%%%
\end{subequations}
%%%%%%%%%%%%%%%%%%%%%%%%%%%%%%%%%%%%%%%%%%%%%%%%%%%%%%%%%%%%%%%%%%%%
where $m_k$, $k=1,2,3$, are the gyrobarycentric coordinates of $A$
in the representation \eqref{napog04} of $A$ with respect to the
vertices of gyrotriangle $A_1A_2A_3$, and
$\gamma_{ij}^{\phantom{O}} = \|\om A_i \op A_j\|$.

Similarly to \eqref{manod08}\,--\,\eqref{manod08p}, we have
\begin{equation} \label{manod09}
\om A_2 \op A = \frac{
m_1 \gammaab\ab_{21} + m_3 \gammabc\ab_{23}
}{
m_1 \gammaab + m_2 + m_3 \gammabc
}
\end{equation}
and
\begin{equation} \label{manod09p}
\gamma_{\om A_2\op A}^{\phantom{O}} = \frac{
m_1 \gammaab + m_2 + m_3 \gammabc
}{
\mA
}
\,.
\end{equation}

Similarly to \eqref{manod08}\,--\,\eqref{manod08p}
and to \eqref{manod09}\,--\,\eqref{manod09p}, we have
\begin{equation} \label{manod10}
\om A_3 \op A = \frac{
m_1 \gammaac\ab_{31} + m_2 \gammabc\ab_{32}
}{
m_1 \gammaac + m_2 \gammabc + m_3
}
\end{equation}
and
\begin{equation} \label{manod10p}
\gamma_{\om A_3\op A}^{\phantom{O}} = \frac{
m_1 \gammaac + m_2 \gammabc + m_3
}{
\mA
}
\,.
\end{equation}

Substituting the gamma factors \eqref{manod08p}, \eqref{manod09p} and \eqref{manod10p},
as well as \eqref{napog03} and \eqref{napog06},
into \eqref{napof05s}, we obtain the elegant equation
\begin{equation} \label{muramd01}
\gamma_d^2 = \frac{D_3}{D_3-H_3}
\left( \frac{m_1+m_2+m_3}{\mA}\right)^2
\,,
%MATHEMATICA stam283a
%MATLAB fir318enE1, zerogsdao5
\end{equation}
where, as in \eqref{dethkdg}\,--\,\eqref{mizsk}, p.~\pageref{dethkdg},
%%%%%%%%%%%%%%%%%%%%%%%%%%%%%%%%%%%%%%%%%%%%%%%%%%%%%%%%%%%%%%%%%%%%
\begin{subequations} \label{muramd02}
\begin{equation} \label{muramd02a}
\begin{split}
D_3 &= 1+2\gammaab\gammaac\gammabc-\gamma_{12}^2-\gamma_{13}^2-\gamma_{23}^2
\\[4pt] &=
2\left\{(\gammaab-1) (\gammaac-1) + (\gammaab-1) (\gammabc-1) + (\gammaac-1) (\gammabc-1)
\right.
\\[4pt] & \phantom{=}
\left.
+ (\gammaab-1) (\gammaac-1) (\gammabc-1)
\right\}
\\[4pt] & \phantom{=}
- \left\{(\gammaab-1)^2+(\gammaac-1)^2+(\gammabc-1)^2\right\}
\\[4pt] &=
\left|
\begin{matrix}
1 & \gammaab & \gammaac  \\[6pt]
\gammaab & 1 & \gammabc  \\[6pt]
\gammaac & \gammabc & 1
\end{matrix}
\right|
\end{split}
%MATHEMATICA stam378 sdmat3, stam378a
\end{equation}
%%%%%%%%%%%%%%%%%%%%%%%%%%%%%%%%%%%%%%%%%%%%%%%%%%%%%%%%%%%%%%%%%%%%
and
\begin{equation} \label{muramd02b}
H_3 = 2(\gamma_{12}-1) (\gamma_{13}-1) (\gamma_{23}-1)
\,.
\end{equation}
\end{subequations}

As expected,
the denominator $D_3-H_3$ in \eqref{muramd01} is positive, by
the circumgyrocircle existence condition \eqref{rjksdc}, p.~\pageref{rjksdc}.

The point $A$ in \eqref{napog04} lies on the circumgyrocircle of
gyrotriangle $A_1A_2A_3$ (Fig.~\ref{fig313enm}) if and only if
$d=R$, or, equivalently, if and only if
\begin{equation} \label{vakin1}
\gamma_{d}^2 = \gamma_{R}^2
\,.
\end{equation}
Inserting into \eqref{vakin1} $\gamma_{d}^2$ from \eqref{muramd01}
and $\gamma_{R}^2$ from \eqref{fskvm2}, p.~\pageref{fskvm2},
we obtain the equation
\begin{equation} \label{vakin2}
\frac{D_3}{D_3-H_3}
\left( \frac{m_1+m_2+m_3}{\mA}\right)^2
=
\frac{D_3}{D_3-H_3}
\,,
\end{equation}
implying
\begin{equation} \label{muramd03}
m_A^2 = (m_1+m_2+m_3)^2
\,.
\end{equation}
The latter, in turn, is valid if and only
\begin{equation} \label{muramd04}
K(A;A_1,A_2,A_3) :=
%m_1m_2(\gammaab-1)+m_1m_3(\gammaac-1)+m_2m_3(\gammabc-1) = 0
\sum_{\substack{i,j=1\\i<j}}^3 m_im_j(\gamma_{ij}^{\phantom{O}} - 1) = 0
\,,
\end{equation}
as we see from \eqref{napof06s}.
Equation \eqref{muramd04} expresses the
{\it circumgyrocircle condition}\index{circumgyrocircle condition}
in the sense that it provides
a necessary and sufficient condition that the
point $A$ in \eqref{napog04} lies on
the circumgyrocircle $C(A_1A_2A_3)$ of gyrotriangle $A_1A_2A_3$.

A gyrotrigonometric version of the circumgyrocircle condition \eqref{muramd04}
follows from \cite[Eq.~(7.149), p.~188]{mybook05},
\begin{equation} \label{hermdn}
\begin{split}
m_1m_2\sin\alpha_3\sin(\alpha_3+\frac{\delta}{2})
&+
m_1m_3 \sin\alpha_2\sin(\alpha_2+\frac{\delta}{2})
\\&
+
m_2m_3 \sin\alpha_1\sin(\alpha_1+\frac{\delta}{2})
=0
\,.
\end{split}
\end{equation}

In order to parametrize the points of the
circumgyrocircle $C(A_1A_2A_3)$
of a gyrotriangle $A_1A_2A_3$ in an
Einstein gyrovector space $\Rsn$ by a real parameter, we assume $m_3\ne0$,
so that the circumgyrocircle condition\index{circumgyrocircle condition}
\eqref{muramd04} can be written as
\begin{equation} \label{muramd06}
\frac{m_1}{m_3} \frac{m_2}{m_3} (\gammaab-1) +
\frac{m_1}{m_3} (\gammaac-1)+ \frac{m_2}{m_3} (\gammabc-1) = 0
\,.
\end{equation}
Selecting $m_2/m_3=t$ as a parameter on the {\it extended real line},
$t\in\Rb\cup\{-\infty,\infty\}$, a
system of parametric gyrobarycentric coordinates of
the point $A$ in \eqref{napog04}
with respect to the gyrobarycentrically independent set $\{A_1,A_2,A_3\}$
that obeys the circumgyrocircle condition \eqref{muramd04} is obtained,
given by
%%%%%%%%%%%%%%%%%%%%%%%%%%%%%%%%%%%%%%%%%%%%%%%%%%%%%%%%%%%%%%%%%%%%
\begin{equation} \label{muramd07}
\begin{split}
m_1 &= -(\gammabc-1)t
\\
m_2 &= (\gammaab-1)t^2 + (\gammaac-1)t = m_3t
\\
m_3 &= (\gammaab-1)t + (\gammaac-1)
\,.
\end{split}
\end{equation}
%%%%%%%%%%%%%%%%%%%%%%%%%%%%%%%%%%%%%%%%%%%%%%%%%%%%%%%%%%%%%%%%%%%%

%%%%%%%%%%%%%%%%%%%%%%%%%%%%%%%%%%%%%%%%%%%%%%%%%%%%%%%%%%%%%%%%%%%%
% FIGURE 15
 
%%%%%%%%%%%%%%%%%%%%%%%%%%%%%%%%%%%%%%%%%%%%%%%%%%%%%%%%%%%%%%%%%%%%%%
%%%%% The hyperbolic semi-circle theorem            %%%%%%%%%%%%%%%%%%
%\begin{figure}[htbp]
\begin{figure}[t]  % try to put this figure on the top of the page
 \centering         % center the figure
 \psfrag{A}{$A$}
 \psfrag{A1}{$A_1$}
 \psfrag{A2}{$A_2$}
 \psfrag{A3}{$A_3$}
 \psfrag{O}{$O$}
 \psfrag{0}{$0$}
 \psfrag{1}{$\hspace{-0.28cm}1$}
 \psfrag{2}{$\hspace{-0.28cm}2$}
 \psfrag{3}{\lower-1.40ex \hbox {$3$}}
 \psfrag{4}{\lower-1.50ex \hbox {$4$}}
 \psfrag{5}{\lower-1.50ex \hbox {$5$}}
 \psfrag{6}{\lower-1.50ex \hbox {$6$}}
 \psfrag{7}{\lower-1.60ex \hbox {$\hspace{-0.06cm}7$}}
 \psfrag{8}{\lower-1.60ex \hbox {$\hspace{-0.10cm}8$}}
 \psfrag{0}{\lower-1.60ex \hbox {$0$}}
 \psfrag{-1}{$\hspace{-0.60cm}-1$}
 \psfrag{-2}{$\hspace{-0.60cm}-2$}
 \psfrag{-3}{$\hspace{-0.60cm}-3$}
 \psfrag{-4}{$\hspace{-0.60cm}-4$}
 \psfrag{-5}{$\hspace{-0.60cm}-5$}
 \psfrag{-6}{$\hspace{-0.60cm}-6$}
 \psfrag{-7}{$\hspace{-0.60cm}-7$}
 \psfrag{-8}{\lower-0.70ex \hbox {$\hspace{-0.55cm}-8$}}
 \includegraphics[width=9cm]{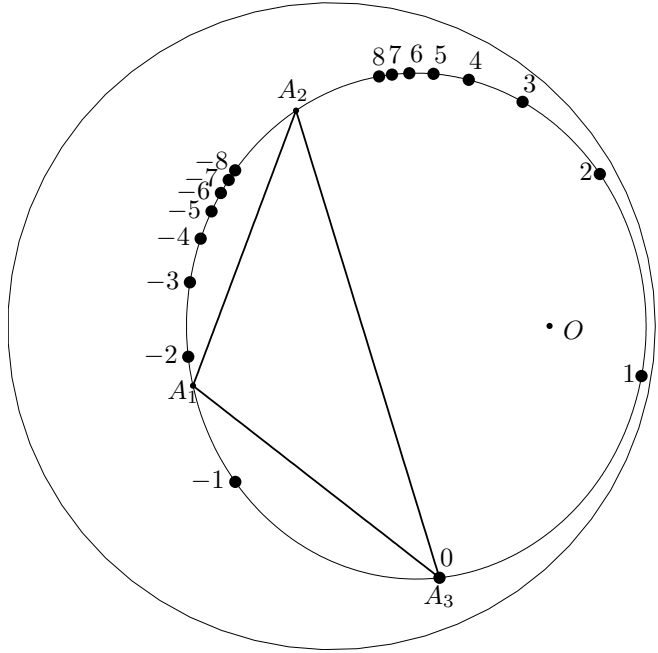}
\caption[Parametrizing the gyrotriangle circumgyrocircle, I]{
Several points on the circumgyrocircle $C(A_1A_2A_3)$ of a
gyrotriangle $A_1A_2A_3$
in an Einstein gyrovector space $(\Rsn,\op,\od)$, $n\ge2$,
determined by several values of the
gyrocircle parameter $t$
in \eqref{muramd07r}, $t=-8,-7,\ldots,-1,0,1,\ldots,7,8$, are presented.
Clearly, the parameter value $t=0$
corresponds to gyrobarycentric coordinates $m_1=m_2=0$, $m_3\ne0$,
as evidenced from \eqref{muramd07}. The latter, in turn, corresponds to the point
$A=A_3$ on the circumgyrocircle, as shown here and as evidenced from \eqref{napog04}.
When $t\rightarrow\pm\infty$,
corresponding points on the circumgyrocircle tend to $A_2$,
as indicated here and as evidenced from \eqref{muramd07}, where
$|m_1|,|m_3|<\!\!<|m_2|$ for large $|t|$.
The point $A_1$ on the circumgyrocircle corresponds to the parameter value
$t=-(\gamma_{13}-1)/(\gamma_{12}-1)$.
\label{fig316enm}}
\end{figure}
%%%%%%%%%%%%%%%%%%%%%%%%%%%%%%%%%%%%%%%%%%%%%%%%%%%%%%%%%%%%%%%%%%%%%%
 % gyrotriangle circumgyrocircle points in (+)E
%                       Fig.~\ref{fig316enm}      Fig. 10.4
%%%%%%%%%%%%%%%%%%%%%%%%%%%%%%%%%%%%%%%%%%%%%%%%%%%%%%%%%%%%%%%%%%%%

%%%%%%%%%%%%%%%%%%%%%%%%%%%%%%%%%%%%%%%%%%%%%%%%%%%%%%%%%%%%%%%%%%%%
% FIGURE 16
 
%%%%%%%%%%%%%%%%%%%%%%%%%%%%%%%%%%%%%%%%%%%%%%%%%%%%%%%%%%%%%%%%%%%%%%
%%%%% The hyperbolic semi-circle theorem            %%%%%%%%%%%%%%%%%%
%\begin{figure}[htbp]
\begin{figure}[t]  % try to put this figure on the top of the page
 \centering         % center the figure
 \psfrag{A}{$A$}
 \psfrag{A1}{$A_1$}
 \psfrag{A2}{$A_2$}
 \psfrag{A3}{$\hspace{-0.02cm}A_3$}
 \psfrag{O}{$\hspace{-0.04cm}O$}
 \psfrag{0}{$0$}
 \psfrag{formula01}{$\theta=k\tfrac{2\pi}{24}$}
 \psfrag{formula02}{$k=0,1,2,\ldots.24$}
 \psfrag{formula03}{$0\le\theta\le2\pi$}
 \psfrag{1}{\lower-1.00ex \hbox {$\hspace{0.20cm}1$}}
 \psfrag{2}{\lower-1.00ex \hbox {$\hspace{0.12cm}2$}}
 \psfrag{3}{\lower-1.00ex \hbox {$\hspace{0.12cm}3$}}
 \psfrag{4}{\lower-1.00ex \hbox {$\hspace{0.12cm}4$}}
 \psfrag{5}{\lower-1.00ex \hbox {$\hspace{0.12cm}5$}}
 \psfrag{6}{\lower-1.00ex \hbox {$\hspace{0.12cm}6$}}
 \psfrag{7}{\lower-0.00ex \hbox {$\hspace{0.08cm}7$}}
 \psfrag{8}{\lower-0.00ex \hbox {$\hspace{0.08cm}8$}}
 \psfrag{9}{\lower-0.00ex \hbox {$\hspace{0.08cm}9$}}
 \psfrag{10}{\lower1.00ex \hbox {$\hspace{0.00cm}10$}}
 \psfrag{11}{\lower2.00ex \hbox {$\hspace{0.08cm}11$}}
 \psfrag{12}{$12$}
 \psfrag{13}{$13$}
 \psfrag{14}{$14$}
 \psfrag{15}{$15$}
 \psfrag{16}{\lower-1.20ex \hbox {$\hspace{0.00cm}16$}}
 \psfrag{17}{$17$}
 \psfrag{18}{$18$}
 \psfrag{19}{\lower1.40ex \hbox {$\hspace{0.04cm}19$}}
 \psfrag{20}{\lower1.40ex \hbox {$\hspace{0.04cm}20$}}
 \psfrag{21}{\lower1.40ex \hbox {$\hspace{0.08cm}21$}}
 \psfrag{22}{\lower1.60ex \hbox {$\hspace{0.08cm}22$}}
 \psfrag{23}{\lower1.80ex \hbox {$\hspace{0.18cm}23$}}
 \psfrag{24}{\lower1.20ex \hbox {$\hspace{0.48cm}24$}}
 \psfrag{0}{\lower-1.20ex \hbox {$\hspace{0.48cm}0$}}
 \includegraphics[width=9cm]{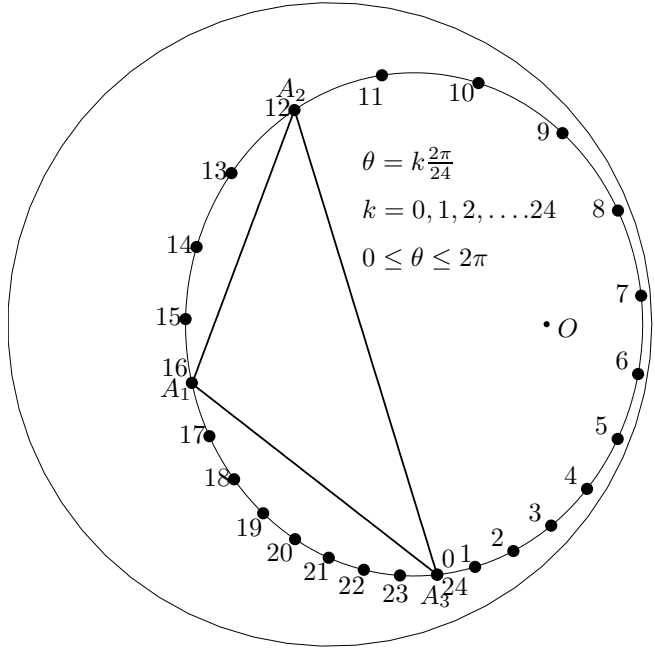}
\caption[Parametrizing the gyrotriangle circumgyrocircle, II]{
Several points on the circumgyrocircle $C(A_1A_2A_3)$
of a gyrotriangle $A_1A_2A_3$
in an Einstein gyrovector space $(\Rsn,\op,\od)$, $n\ge2$,
determined by several values of the
circumgyrocircle parameter $\theta$
in \eqref{mador02},
$\theta=k(2\pi/24)$, $k=0,1,2,\ldots,24$, are presented.
The parameter values that correspond to $k=0$ $(\theta=0)$ and
$k=24$ $(\theta=2\pi)$ represent jointly the point $A_3$ and are,
accordingly, identified.
The parameter value that corresponds to $k=12$ $(\theta=\pi)$
represents the point $A_2$.
\label{fig316aenm}}
\end{figure}
%%%%%%%%%%%%%%%%%%%%%%%%%%%%%%%%%%%%%%%%%%%%%%%%%%%%%%%%%%%%%%%%%%%%%%
 % gyrotriangle circumgyrocircle points in (+)E
%                       Fig.~\ref{fig316aenm}      Fig. 10.5
%%%%%%%%%%%%%%%%%%%%%%%%%%%%%%%%%%%%%%%%%%%%%%%%%%%%%%%%%%%%%%%%%%%%

Several points $A$ of the
circumgyrocircle $C(A_1A_2A_3)$ of a gyrotriangle $A_1A_2A_3$,
which are given by \eqref{napog04} and \eqref{muramd07}
and which
correspond to the circumgyrocircle parameter values
$t=-8,-7,\ldots,-1,0,1,\ldots,7,8$,
are shown in Fig.~\ref{fig316enm}.
The parameter values $t=-\infty$ and $t=\infty$ correspond jointly
to the point $A_2$.

The two special values, $t=\pm\infty$, of the parameter $t$ are identified
in the sense that they correspond jointly to the same point, $A_2$,
on circumgyrocircle $C(A_1A_2A_3)$,
as shown in Fig.~\ref{fig316enm}.
This identification suggests the replacement of the
parameter $t\in\Rb\cup\{-\infty,\infty\}$
by the parameter $\theta$, $-\pi\le\theta\le\pi$, according to the
bijective (one-to-one) relation
\begin{equation} \label{mador01}
t=\tan\frac{\theta}{2}
\,,
\end{equation}
for which a corresponding identification is built in naturally.
Indeed, the two values $\theta=\pm\pi$ correspond to the two values
$t=\pm\infty$ and, accordingly, $\theta=-\pi$ and $\theta=\pi$
are identified.

Inserting $t$ from \eqref{mador01} in \eqref{muramd07},
and noting that
$\tan\frac{\theta}{2}=\sin\frac{\theta}{2}/\cos\frac{\theta}{2}$,
and that gyrobarycentric coordinates are homogeneous,
we obtain the following
system of parametric gyrobarycentric coordinates of
the point $A$ in \eqref{napog04}, p.~\pageref{napog04},
with respect to the set $S$,
%%%%%%%%%%%%%%%%%%%%%%%%%%%%%%%%%%%%%%%%%%%%%%%%%%%%%%%%%%%%%%%%%%%%
\begin{equation} \label{mador02a}
\begin{split}
m_1 &= -(\gammabc-1)\sin\frac{\theta}{2} \cos\frac{\theta}{2}
\\
m_2 &= (\gammaab-1)\sin^2\frac{\theta}{2}
+ (\gammaac-1)\sin\frac{\theta}{2} \cos\frac{\theta}{2}
\\
m_3 &= (\gammaab-1)\sin\frac{\theta}{2} \cos\frac{\theta}{2}
+ (\gammaac-1)\cos^2\frac{\theta}{2}
\,.
\end{split}
\end{equation}
%%%%%%%%%%%%%%%%%%%%%%%%%%%%%%%%%%%%%%%%%%%%%%%%%%%%%%%%%%%%%%%%%%%%

We now apply to \eqref{mador02a} well-known
trigonometric half-angle identities,
noting that gyrobarycentric coordinates are homogeneous,
and translate the parameter interval by $\pi$,
from the interval $[-\pi,\pi]$ to the interval $[0,2\pi]$,
obtaining the following elegant parametric gyrobarycentric coordinates
of $A$ in \eqref{napog04}. p.~\pageref{napog04}.
%%%%%%%%%%%%%%%%%%%%%%%%%%%%%%%%%%%%%%%%%%%%%%%%%%%%%%%%%%%%%%%%%%%%
\begin{equation} \label{mador02}
\begin{split}
m_1 &= - (\gammabc-1)\sin\theta
\\
m_2 &=  \phantom{-} (\gammaac-1)\sin\theta
+ (\gammaab-1)(1-\cos\theta)
\\
m_3 &=  \phantom{-} (\gammaab-1)\sin\theta
+ (\gammaac-1)(1+\cos\theta)
\,,
\end{split}
\end{equation}
%Mathematica stam365
%Matlab fig316aen
%%%%%%%%%%%%%%%%%%%%%%%%%%%%%%%%%%%%%%%%%%%%%%%%%%%%%%%%%%%%%%%%%%%%
$0\le\theta\le2\pi$.
The two values of the parameter $\theta$,
$\theta=0$ and $\theta=2\pi$, are identified,
as indicated in Fig.~\ref{fig316aenm}, where several points $A$
on the circumference of circumgyrocircle $C(A_1A_2A_3)$,
which correspond to several values of the circumgyrocircle parameter
$\theta$, $0\le\theta\le2\pi$, are presented.

% SECTION NUMBER 16
\section{Gyrocircle Interior and Exterior Points} \label{slilc3}
\index{circle barycentric representation}

We are now in the position to present a definition followed
by two theorems that characterize the points of the circumgyrocircle of a
given gyrotriangle, as well as interior and exterior points
of the circumgyrocircle in Einstein gyrovector spaces.

%%%%%%%%%%%%%%%%%%%%%%%%%%%%%%%%%%%%%%%%%%%%%%%%%%%%%%%%%%%%%%%%%%%%
% DEFINITON NUMBER 9.6
\begin{definition}\label{defintext}
{\bf (Gyrocircle\,Interior\,and\,Exterior\,Points).}
\index{gyrocircle, interior exterior}
{\it
Let $A_1A_2A_3$ be a gyrotriangle that possesses a
circumgyrocircle $C(A_1A_2A_3)$ with circumgyrocenter $O$ and
circumgyroradius $R$
in an Einstein gyrovector space $(\Rsn,\op,\od)$.
Let $P$ be a generic point in $\AAbt\cap\Rsn$,
\begin{equation} \label{camh}
P = \frac{m_1\gAa A_1 + m_2\gAb A_2 + m_3\gAc A_3}
         {m_1\gAa + m_2\gAb + m_3\gAc}
\in\AAbt\cap\Rsn
\,,
\end{equation}
where, as in \eqref{dambdy}, p.~\pageref{dambdy},
\begin{equation} \label{danady}
\AAbt = A_1\op\rmspan\{\om A_1\op A_2,\om A_1\op A_3\} \subset\Rn
\,,
\end{equation}
and let $d=\|\om P\op O\|$ be the gyrodistance of $P$ from $O$.
Then, $P$ lies
\begin{enumerate}
\item
in the interior of $C(A_1A_2A_3)$
if $d<R$;
\item
in the exterior of $C(A_1A_2A_3)$
if $d>R$; and
\item
on $C(A_1A_2A_3)$ if $d=R$.
\end{enumerate}
}
\end{definition}

%%%%%%%%%%%%%%%%%%%%%%%%%%%%%%%%%%%%%%%%%%%%%%%%%%%%%%%%%%%%%%%%%%%%
% THEOREM NUMBER 9.7
\begin{theorem}\label{sveqni}
{\bf (Point\,to\,Circumgyrocenter\,Gyrodistance).}
Let $A_1A_2A_3$ be a gyrotriangle that possesses a
circumgyrocircle $C(A_1A_2A_3)$ with circumgyrocenter $O$ and
circumgyroradius $R$
in an Einstein gyrovector space $(\Rsn,\op,\od)$,
and let $A\inn\AAbt\cap\Rsn$ be a point given by its
gyrobarycentric representation
\begin{equation} \label{napog04s}
A = \frac{
m_1 \gAa A_1 + m_2 \gAb A_2 + m_3 \gAc A_3
}{
m_1 \gAa + m_2 \gAb + m_3 \gAc
}
\end{equation}
with respect to the gyrobarycentrically independent set $S=\{A_1,A_2,A_3\}$.
Furthermore, let
\begin{equation} \label{dkner1}
d=\|\om A\op O\|
\end{equation}
be the gyrodistance from $A$ to $O$.

Then,
\begin{equation} \label{dkner2}
d=\sqrt{1-\frac{
2s^2K
}{
M^2 \gamma_R^2 R^2
}}
\, R
\,,
%MATLAB fig318enE1, zerodao
\end{equation}
where
\begin{equation} \label{dkner3}
K =
\sum_{\substack{i,j=1\\i<j}}^3 m_im_j(\gamma_{ij}^{\phantom{O}} - 1)
\end{equation}
and
\begin{equation} \label{dkner4}
M = \sum_{k=1}^{3} m_k
\,.
\end{equation}
\end{theorem}
\begin{proof}
%%%%%%%%%%%%%%%%%%%%%%%%%%%%%%%%%%%%%%%%%%%%%%%%%%%%%%%%%%%%%%%%%%%%
Following
(i) \eqref{muramd01}, p.~\pageref{muramd01}, and
(ii) \eqref{fskvm2}, p.~\pageref{fskvm2}, and
(iii) \eqref{napof06s}, p.~\pageref{napof06s},
we have
%%%%%%%%%%%%%%%%%%%%%%%%%%%%%%%%%%%%%%%%%%%%%%%%%%%%%%%%%%%%%%%%%%%%
\begin{equation} \label{sabon01}
\begin{split}
\gamma_d^2 &= \frac{D_3}{D_3-H_3}
\frac{M^2}{m_A^2}
\\[8pt]
\gamma_R^2 &= \frac{D_3}{D_3-H_3}
\\[8pt]
m_A^2 &= M^2+2K
\end{split}
\end{equation}
where $D_3$ and $H_3$ are given by \eqref{muramd02}.

Hence, by \eqref{sabon01} and \eqref{rugh1ds}, p.~\pageref{rugh1ds},
or \eqref{fskvm3}, p.~\pageref{fskvm3},
%%%%%%%%%%%%%%%%%%%%%%%%%%%%%%%%%%%%%%%%%%%%%%%%%%%%%%%%%%%%%%%%%%%%
\begin{equation} \label{sabon02}
R^2 = s^2\frac{\gamma_R^2-1}{\gamma_R^2} = s^2\frac{H_3}{D_3}
\end{equation}
and
\begin{equation} \label{sabon02b}
d^2 = s^2 \frac{\gamma_d^2-1}{\gamma_d^2} = s^2\frac{
M^2 H_3-2K(D_3-H_3)
}{ M^2D_3 }
\,,
\end{equation}
so that, by straightforward algebra,
\begin{equation} \label{sabon03}
\frac{d^2}{R^2} = 1-\frac{2s^2K}{M^2\gamma_R^2R^2}
\,,
\end{equation}
as desired.
%%%%%%%%%%%%%%%%%%%%%%%%%%%%%%%%%%%%%%%%%%%%%%%%%%%%%%%%%%%%%%%%%%%%
\end{proof}
%%%%%%%%%%%%%%%%%%%%%%%%%%%%%%%%%%%%%%%%%%%%%%%%%%%%%%%%%%%%%%%%%%%%
% THEOREM NUMBER 9.8
\index{gyrocircle gyrobarycentric representation}
\begin{theorem}\label{thmsnvur}
{\bf (Gyrocircle Gyrobarycentric Representation).}
Let $A_1A_2A_3$ be a gyrotriangle that possesses a
circumgyrocircle $C(A_1A_2A_3)$ with circumgyrocenter $O$ and
circumgyroradius $R$
in an Einstein gyrovector space $(\Rsn,\op,\od)$,
and let $A\inn\AAbt\cap\Rsn$ be a point given by its
gyrobarycentric representation
\begin{equation} \label{napok04}
A = \frac{
m_1 \gAa A_1 + m_2 \gAb A_2 + m_3 \gAc A_3
}{
m_1 \gAa + m_2 \gAb + m_3 \gAc
}
\end{equation}
with respect to the gyrobarycentrically independent set $S=\{A_1,A_2,A_3\}$.
Furthermore, let
\begin{equation} \label{indk6}
\begin{split}
K(A;A_1,A_2,A_3) &=
m_1m_2(\gammaab-1)+m_1m_3(\gammaac-1)+m_2m_3(\gammabc-1)
\,,
\\
T(A;A_1,A_2,A_3) &=
m_1m_2\sin\alpha_3\sin(\alpha_3+\frac{\delta}{2})
+
m_1m_3 \sin\alpha_2\sin(\alpha_2+\frac{\delta}{2})
\\&
+
m_2m_3 \sin\alpha_1\sin(\alpha_1+\frac{\delta}{2})
\,,
\end{split}
\end{equation}
where $\alpha_k$, $k=1,2,3$, is the vertex gyroangle of vertex $A_k$
of gyrotriangle $A_1A_2A_3$,
be two scalars associated with the point $A$.
 \begin{enumerate}
\item \label{thmsnvur01}
The point $A$ lies on the
circumgyrocircle $C(A_1A_2A_3)$ of gyrotriangle $A_1A_2A_3$
if and only if the gyrobarycentric coordinates $m_1,m_2,m_3$ of $A$
in \eqref{napok04} satisfy the circumgyrocircle condition
\begin{subequations} \label{muramd05}
\begin{equation} \label{muramd05aa}
K(A;A_1,A_2,A_3) = 0
\end{equation}
or, equivalently,
the gyrotrigonometric circumgyrocircle condition
\begin{equation} \label{muramd05bb}
T(A;A_1,A_2,A_3) = 0
\,.
\end{equation}
\end{subequations}
\item \label{thmsnvur02}
The point $A$ lies in the interior of circumgyrocircle $C(A_1A_2A_3)$
if and only if
\begin{subequations} \label{muramd05a1}
\begin{equation} \label{muramd05a1a}
K(A;A_1,A_2,A_3) > 0
\end{equation}
or, equivalently,
\begin{equation} \label{muramd05a1b}
T(A;A_1,A_2,A_3) > 0
\,.
\end{equation}
\end{subequations}
\item \label{thmsnvur03}
The point $A$ lies in the exterior of circumgyrocircle $C(A_1A_2A_3)$
if and only if
\begin{subequations} \label{muramd05a2}
\begin{equation} \label{muramd05a2aa}
K(A;A_1,A_2,A_3) < 0
\end{equation}
or, equivalently,
\begin{equation} \label{muramd05a2bb}
T(A;A_1,A_2,A_3) < 0
\,.
\end{equation}
\end{subequations}

Moreover, the circumgyrocircle $C(A_1A_2A_3)$ is the
locus of the point $A$ in \eqref{napok04}, with
parametric gyrobarycentric coordinates $m_k$, $k=1,2,3$,
given
\item \label{thmsnvur04}
by the parametric equations
%%%%%%%%%%%%%%%%%%%%%%%%%%%%%%%%%%%%%%%%%%%%%%%%%%%%%%%%%%%%%%%%%%%%
\begin{equation} \label{muramd07r}
\begin{split}
m_1 &= -(\gammabc-1)t
\\
m_2 &= (\gammaab-1)t^2 + (\gammaac-1)t = m_3t
\\
m_3 &= (\gammaab-1)t + (\gammaac-1)
\end{split}
\end{equation}
%%%%%%%%%%%%%%%%%%%%%%%%%%%%%%%%%%%%%%%%%%%%%%%%%%%%%%%%%%%%%%%%%%%%
with the parameter $t$, $t\in\Rb\cup\{-\infty,\infty\}$,
where the two parameter values $t=-\infty$ and $t=\infty$
are identified (Fig.~\ref{fig316enm}), or, equivalently,
\item \label{thmsnvur05}
by the parametric equations
%%%%%%%%%%%%%%%%%%%%%%%%%%%%%%%%%%%%%%%%%%%%%%%%%%%%%%%%%%%%%%%%%%%%
\begin{equation} \label{mador04}
\begin{split}
m_1 &= - (\gammabc-1)\sin\theta
\\
m_2 &=  \phantom{-} (\gammaac-1)\sin\theta
+ (\gammaab-1)(1-\cos\theta)
\\
m_3 &=  \phantom{-} (\gammaab-1)\sin\theta
+ (\gammaac-1)(1+\cos\theta)
\end{split}
\end{equation}
%Mathematica stam365
%Matlab fig316aen
%%%%%%%%%%%%%%%%%%%%%%%%%%%%%%%%%%%%%%%%%%%%%%%%%%%%%%%%%%%%%%%%%%%%
with the parameter $\theta$, $0\le\theta\le2\pi$,
where the two parameter values $\theta=0$ and $\theta=2\pi$
are identified (Fig.~\ref{fig316aenm}).
\item \label{thmsnvur06}
The circumgyrocenter, $O$, of gyrotriangle $A_1A_2A_3$ is given by
\eqref{napog02}\,--\,\eqref{napog03}, p.~\pageref{napog02}, and,
gyrotrigonometrically, by \eqref{eicksfnein}\,--\,\eqref{fjvdnw1}, p.~\pageref{fjvdnw1}.
\item \label{thmsnvur07}
The circumgyroradius, $R$, of gyrotriangle $A_1A_2A_3$ is given by
\eqref{sabon02}.
 \end{enumerate}
\end{theorem}
\begin{proof}
The equivalence between
$K$ and $T$ in the Theorem is proved in \eqref{muramd04}\,--\,\eqref{hermdn}.
The proof of each item of the Theorem follows.
\begin{enumerate}
\item
%1
It follows from \eqref{dkner2} that $K=K(A;A_1,A_2,A_3)$ vanishes
if an only if $d=R$, that is, by Def.~\ref{defintext},
if an only if the point $A$ lies on circumgyrocircle $C(A_1A_2A_3)$.
\item
%2
It follows from \eqref{dkner2} that $K=K(A;A_1,A_2,A_3)>0$
if an only if $d<R$, that is, by Def.~\ref{defintext},
if an only if the point $A$ lies in the interior of circumgyrocircle $C(A_1A_2A_3)$.
\item
%3
It follows from \eqref{dkner2} that $K=K(A;A_1,A_2,A_3)<0$
if an only if $d>R$, that is, by Def.~\ref{defintext},
if an only if the point $A$ lies in the exterior of circumgyrocircle $C(A_1A_2A_3)$.
\item %4
Item \eqref{thmsnvur04} follows from \eqref{muramd07}.
\item %5
Item \eqref{thmsnvur05} follows from \eqref{mador02}.
\item %6
The proof of Item \eqref{thmsnvur06}
is given in the derivation of
\eqref{napog02}\,--\,\eqref{napog03}, p.~\pageref{napog02}, and
\eqref{eicksfnein}\,--\,\eqref{fjvdnw1}, p.~\pageref{fjvdnw1}.
\item %7
The proof of Item \eqref{thmsnvur07}
is given in the derivation of \eqref{sabon02}.
\end{enumerate}
The proof of the Theorem is thus complete.
\end{proof}

% EXAMPLE NUMBER 9.9
\begin{example}\label{exfmv1}
Let $A$ be a point in an Einstein gyrovector space $\Rsn$, $n\ge2$, given by its
gyrobarycentric representation \eqref{napok04}
with respect to a  gyrobarycentrically independent set $\{A_1,A_2,A_3\}$, with
gyrobarycentric coordinates $m_1=0$, $m_2=0$ and $m_3\ne0$.
Then, $A=A_3$, and the gyrobarycentric coordinates of $A$ satisfy
the circumgyrocircle condition \eqref{muramd05} of Theorem \ref{thmsnvur}.
Hence, by Theorem \ref{thmsnvur}, the point $A_3$ lies on the
circumgyrocircle of gyrotriangle $\{A_1,A_2,A_3\}$, as obviously expected.
\end{example}

% EXAMPLE NUMBER 9.10
\begin{example}\label{exfnw1}
Let $A$ be a point in an Einstein gyrovector space $\Rsn$, $n\ge2$, given by its
gyrobarycentric representation \eqref{napok04}
with respect to a  gyrobarycentrically independent set $\{A_1,A_2,A_3\}$,
where gyrotriangle $A_1A_2A_3$ is equilateral.
Then, $A$ lies on the circumgyrocircle of gyrotriangle $A_1A_2A_3$
if and only if the gyrobarycentric coordinates $m_1,m_2,m_3$ of $A$
in \eqref{napok04} satisfy the circumgyrocircle condition
\begin{equation} \label{muramd05b}
m_1m_2+m_1m_3+m_2m_3= 0
\,.
\end{equation}
\end{example}

% SECTION NUMBER 17
\section{Circle Barycentric Representation} \label{slilb3}
\index{circle barycentric representation}

Lemma \ref{lemwindk}, p.~\pageref{lemwindk}, enables Theorem \ref{thmsnvur} to be reduced
to its Euclidean counterpart, Theorem \ref{sveqnieuc}.
As an immediate application of Lemma \ref{lemwindk} we, therefore, note that
\begin{equation} \label{gfnert}
\lim_{s\rightarrow \infty} 2s^2K(A;A_1,A_2,A_3) =
m_1m_2a_{12}^2+m_1m_3a_{13}^2+m_2m_3a_{23}^2
\end{equation}
where $K=K(A;A_1,A_2,A_3)$ is given by \eqref{indk6},
and $a_{ij}=\|-A_i+A_j\|$.

%%%%%%%%%%%%%%%%%%%%%%%%%%%%%%%%%%%%%%%%%%%%%%%%%%%%%%%%%%%%%%%%%%%%
% DEFINITON NUMBER 9.11
\begin{definition}\label{defintexteuc}
{\bf (Circle\,Interior\,and\,Exterior\,Points).}
\index{circle, interior exterior}
{\it
Let $A_1A_2A_3$ be a triangle and let $C(A_1A_2A_3)$, $O$ and $R$ be,
respectively, the circumcircle, circumcenter and circumradius
of the triangle in a Euclidean space $\Rn$.
Let $P$ be a generic point in $\AAb_3^{euc}$,
\begin{equation} \label{camheuc}
P = \frac{m_1 A_1 + m_2 A_2 + m_3 A_3}
         {m_1 + m_2 + m_3}
\in\AAb_3^{euc}
\,,
\end{equation}
where
\begin{equation} \label{danadyeuc}
\AAb_3^{euc} = A_1+\rmspan\{-A_1 + A_2,-A_1 + A_3\} \subset\Rn
\,,
\end{equation}
and let $d=\|-P+O\|$ be the distance of $P$ from $O$.
Then, $P$ lies
\begin{enumerate}
\item
in the interior of $C(A_1A_2A_3)$
if $d<R$;
\item
in the exterior of $C(A_1A_2A_3)$
if $d>R$; and
\item
on $C(A_1A_2A_3)$ if $d=R$.
\end{enumerate}
}
\end{definition}

%%%%%%%%%%%%%%%%%%%%%%%%%%%%%%%%%%%%%%%%%%%%%%%%%%%%%%%%%%%%%%%%%%%%
% THEOREM NUMBER 9.12
\begin{theorem}\label{sveqnieuc}
{\bf (Point\,to\,Circumcenter\,Distance).}
Let $A_1A_2A_3$ be a triangle and let $C(A_1A_2A_3)$, $O$ and $R$ be,
respectively, the circumcircle, circumcenter and circumradius
of the triangle in a Euclidean space $\Rn$,
and let $A\inn\AAb_3^{euc}$ be a point given by its
barycentric representation
\begin{equation} \label{napog04seuc}
A = \frac{
m_1 A_1 + m_2 A_2 + m_3 A_3
}{
m_1 + m_2 + m_3
}
\end{equation}
with respect to the barycentrically independent set $S=\{A_1,A_2,A_3\}$.
Furthermore, let
\begin{equation} \label{dkner1euc}
d=\| - A + O\|
\end{equation}
be the distance from $A$ to $O$.

Then,
\begin{equation} \label{dkner2euc}
d=\sqrt{1-\frac{
K_{euc}
}{
M^2 R^2
}}
\, R
\,,
\end{equation}
% MATLAB fig335euc, zerod confirms "d".
where
\begin{equation} \label{dkner3euc}
K_{euc} =
\sum_{\substack{i,j=1\\i<j}}^3 m_im_j a_{ij}^2
\,,
\end{equation}
$a_{ij}=\|-A_i+A_j\|$,
and
\begin{equation} \label{dkner4euc}
M = \sum_{k=1}^{3} m_k
\,.
\end{equation}
\end{theorem}
\begin{proof}
Theorem \ref{sveqnieuc} is the Euclidean counterpart of Theorem \ref{sveqni}.
The proof of Theorem \ref{sveqnieuc} from Theorem \ref{sveqni} is immediate,
noting \eqref{gfnert}, and noting that each equation in Theorem \ref{sveqnieuc}
is the Euclidean limit, $s\rightarrow\infty$, of a corresponding equation
in Theorem \ref{sveqni}.
\end{proof}
%%%%%%%%%%%%%%%%%%%%%%%%%%%%%%%%%%%%%%%%%%%%%%%%%%%%%%%%%%%%%%%%%%%%

%%%%%%%%%%%%%%%%%%%%%%%%%%%%%%%%%%%%%%%%%%%%%%%%%%%%%%%%%%%%%%%%%%%%
% FIGURE 17
 
%%%%%%%%%%%%%%%%%%%%%%%%%%%%%%%%%%%%%%%%%%%%%%%%%%%%%%%%%%%%%%%%%%%%%%
%%%%% The hyperbolic semi-circle theorem            %%%%%%%%%%%%%%%%%%
%\begin{figure}[htbp]
\begin{figure}[t]  % try to put this figure on the top of the page
 \centering         % center the figure
 \psfrag{A}{$A$}
 \psfrag{O}{$O$}
 \psfrag{A1}{$A_1$}
 \psfrag{A2}{$A_2$}
 \psfrag{A3}{$A_3$}
 \includegraphics[width=9cm]{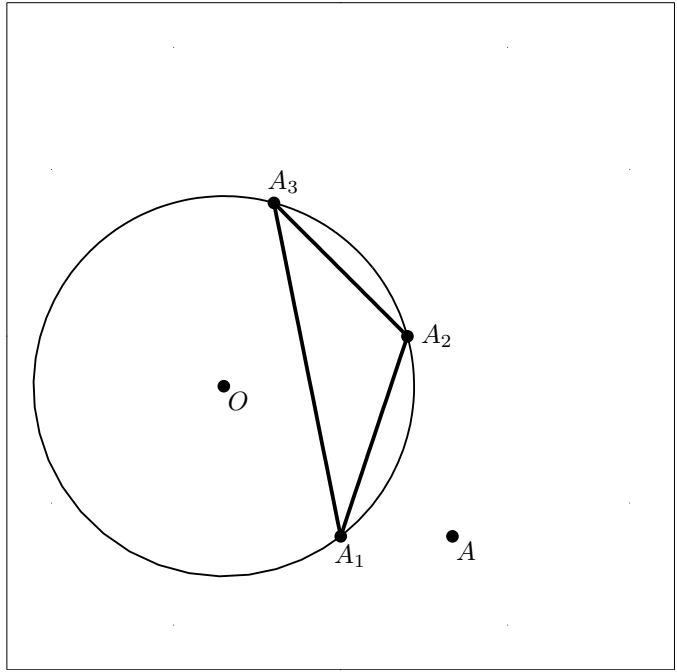}
\caption[A triangle circumcircle]{
Illustrating Theorems \ref{sveqnieuc} and \ref{thmsnvureuc}.
The circumcircle $C(A_1A_2A_3)$ of triangle $A_1A_2A_3$ in the Euclidean
plane $\Rtwo$ is the locus of the point $A$ with barycentric representation
\eqref{napok04euc}, where its parametric barycentric coordinates are
given by \eqref{muramd07reuc} or, equivalently, by \eqref{mador04euc}.
The radius of $C(A_1A_2A_3)$ is $R=\|-O+A_k\|$, $k=1,2,3$, and the distance
$d=\|-A+O\|$ from a generic point $A$ in the triangle plane to the center $O$
of $C(A_1A_2A_3)$ is given by \eqref{dkner2euc}.
The latter provides an elegant condition in Theorem \ref{thmsnvureuc} that
determines whether the point $A$ lies on the circumcircle $C(A_1A_2A_3)$,
or in its interior or exterior.
\label{fig335eucm}}
\end{figure}
%%%%%%%%%%%%%%%%%%%%%%%%%%%%%%%%%%%%%%%%%%%%%%%%%%%%%%%%%%%%%%%%%%%%%%
 % The gyrotriangle circumgyrocenter in (+)E
%                       Fig.~\ref{fig335eucm}      Fig. 10.6
%%%%%%%%%%%%%%%%%%%%%%%%%%%%%%%%%%%%%%%%%%%%%%%%%%%%%%%%%%%%%%%%%%%%

The Euclidean counterpart of the
Gyrocircle Gyrobarycentric Representation Theorem \ref{thmsnvur}
is the following elegant theorem, illustrated in Fig.~\ref{fig335eucm}.

%%%%%%%%%%%%%%%%%%%%%%%%%%%%%%%%%%%%%%%%%%%%%%%%%%%%%%%%%%%%%%%%%%%%
% THEOREM NUMBER 9.13
\index{circle barycentric representation}
\begin{theorem}\label{thmsnvureuc}
{\bf (Circle Barycentric Representation).}
Let $A$,
$A\in A_1+\rmspan\{-A_1+A_2,-A_1+A_3\}\subset\Rn$,
be a point given by its barycentric representation
\begin{equation} \label{napok04euc}
A = \frac{
m_1 A_1 + m_2 A_2 + m_3 A_3
}{
m_1 + m_2 + m_3
}
\end{equation}
with respect to the barycentrically independent set $S=\{A_1,A_2,A_3\}$
in a Euclidean space $\Rn$, $n\ge2$, and let
%%%%%%%%%%%%%%%%%%%%%%%%%%%%%%%%%%%%%%%%%%%%%%%%%%%%%%%%%%%%%%%%%%%%
\begin{equation} \label{hkdne3}
\begin{split}
K_{euc}(A;A_1A_2A_3) &= m_1m_2 a_{12}^2 + m_1m_3 a_{13}^2 + m_2m_3 a_{23}^2
\\
T_{euc}(A;A_1A_2A_3) &= m_1m_2 \sin^2\alpha_3 + m_1m_3 \sin^2\alpha_2
+ m_2m_3 \sin^2\alpha_3
\end{split}
\end{equation}
%%%%%%%%%%%%%%%%%%%%%%%%%%%%%%%%%%%%%%%%%%%%%%%%%%%%%%%%%%%%%%%%%%%%
where $\alpha_k$, $k=1,2,3$, is the vertex angle of vertex $A_k$
of triangle $A_1A_2A_3$,
be two scalars associated with $A$.
\begin{enumerate}
\item \label{eucsnvur01}
The point $A$ lies on the circumcircle $C(A_1A_2A_3)$ of triangle $A_1A_2A_3$
if and only if the barycentric coordinates $m_1,m_2,m_3$ of $A$
in \eqref{napok04euc} satisfy the circumcircle condition
\begin{subequations} \label{muramd05euc}
\begin{equation} \label{muramd05euca}
K_{euc}(A;A_1,A_2,A_3) =0
\end{equation}
or, equivalently,
the trigonometric circumcircle condition
\begin{equation} \label{muramd05eucb}
T_{euc}(A;A_1,A_2,A_3) =0
\,.
\end{equation}
%MATLAB fig318eua.m, zzero
\end{subequations}
\item \label{eucsnvur02}
The point $A$ lies in the interior of circumcircle $C(A_1A_2A_3)$
if and only if
\begin{subequations} \label{eucamd05a1}
\begin{equation} \label{eucamd05a1a}
K_{euc}(A;A_1,A_2,A_3) >0
\end{equation}
or, equivalently,
\begin{equation} \label{eucamd05a1b}
T_{euc}(A;A_1,A_2,A_3) >0
\,.
\end{equation}
\end{subequations}
\item \label{eucsnvur03}
The point $A$ lies in the exterior of circumcircle $C(A_1A_2A_3)$
if and only if
\begin{subequations} \label{eucamd05a2}
\begin{equation} \label{eucamd05a2a}
K_{euc}(A;A_1,A_2,A_3) <0
\end{equation}
or, equivalently,
\begin{equation} \label{eucamd05a2b}
T_{euc}(A;A_1,A_2,A_3) <0
\,.
\end{equation}
\end{subequations}

Moreover, the circumcircle $C(A_1A_2A_3)$ is the
locus of the point $A$ in \eqref{napok04euc}, with
parametric barycentric coordinates $m_k$, $k=1,2,3$,
given
\item \label{eucsnvur04}
by the parametric equations
%%%%%%%%%%%%%%%%%%%%%%%%%%%%%%%%%%%%%%%%%%%%%%%%%%%%%%%%%%%%%%%%%%%%
\begin{equation} \label{muramd07reuc}
\begin{split}
m_1 &= -t\sin^2\alpha_1
\\
m_2 &= t\sin^2\alpha_2 + t^2\sin^2\alpha_3 = m_3t
\\
m_3 &= \sin^2\alpha_2 + t\sin^2\alpha_3
\end{split}
%Matlab fig335aeuc
\end{equation}
%%%%%%%%%%%%%%%%%%%%%%%%%%%%%%%%%%%%%%%%%%%%%%%%%%%%%%%%%%%%%%%%%%%%
with the parameter $t$, $t\in\Rb\cup\{-\infty,\infty\}$,
where the two parameter values $t=-\infty$ and $t=\infty$
are identified, or, equivalently,
\item \label{eucsnvur05}
by the parametric equations
%%%%%%%%%%%%%%%%%%%%%%%%%%%%%%%%%%%%%%%%%%%%%%%%%%%%%%%%%%%%%%%%%%%%
\begin{equation} \label{mador04euc}
\begin{split}
m_1 &= - \sin^2\alpha_1 \sin\theta
\\
m_2 &= \sin^2\alpha_2 \sin\theta + \sin^2\alpha_3 (1-\cos\theta)
\\
m_3 &= \sin^2\alpha_3 \sin\theta + \sin^2\alpha_2 (1+\cos\theta)
\end{split}
\end{equation}
%Matlab fig335euc
%%%%%%%%%%%%%%%%%%%%%%%%%%%%%%%%%%%%%%%%%%%%%%%%%%%%%%%%%%%%%%%%%%%%
with the parameter $\theta$, $0\le\theta\le2\pi$,
where the two parameter values $\theta=0$ and $\theta=2\pi$
are identified.
\item \label{eucsnvur06}
The circumcenter, $O$, of triangle $A_1A_2A_3$ is given trigonometrically by
\eqref{rukdis}, p.~\pageref{rukdis}.
\item \label{eucsnvur07}
The circumradius, $R$, of triangle $A_1A_2A_3$ is given by
\eqref{adamtk06}, p.~\pageref{adamtk06}.
\end{enumerate}
\end{theorem}
\begin{proof}
In the Euclidean limit, $s\rightarrow\infty$, gamma factors tend to $1$.
Hence, in that limit, the
gyrobarycentric representation of $A\in\Rsn$ in \eqref{napok04}
tends to the corresponding barycentric representations of $A\in\Rn$
in \eqref{napok04euc}.

By Lemma \ref{lemwindk}, in the Euclidean limit the
circumgyrocircle condition \eqref{muramd05aa}
(where $K$ is given by \eqref{indk6})
tends to \eqref{muramd05euca}
(where $K_{euc}$ is given by \eqref{hkdne3}).

In the Euclidean limit gyrotriangle defects vanish. Hence,
in that limit,
the circumgyrocircle condition \eqref{muramd05bb}
(where $T$ is given by \eqref{indk6})
tends to \eqref{muramd05eucb}
(where $T_{euc}$ is given by \eqref{hkdne3}), as desired.

The proof of Inequalities \eqref{eucamd05a1} and \eqref{eucamd05a2}
follows, similarly, from Inequalities
\eqref{muramd05a1} and \eqref{muramd05a2}.

Similarly, the proof of \eqref{muramd07reuc}\,--\,\eqref{mador04euc}
follows from \eqref{muramd07r}\,--\,\eqref{mador04}
and from Euclidean limits that result immediately from
Lemma \ref{lemwindk}.

The proof of Item \eqref{eucsnvur06} is presented in the
derivation of \eqref{rukdis}, p.~\pageref{rukdis}, and
the proof of Item \eqref{eucsnvur07} is presented in the
derivation of \eqref{adamtk06}, p.~\pageref{adamtk06}.
\end{proof}

% SECTION NUMBER 18
\section{Gyrocircle Gyroline Intersection} \label{slila4d5}
\index{gyrocircle gyroline intersection}

%%%%%%%%%%%%%%%%%%%%%%%%%%%%%%%%%%%%%%%%%%%%%%%%%%%%%%%%%%%%%%%%%%%%
% FIGURE 18
 
%%%%%%%%%%%%%%%%%%%%%%%%%%%%%%%%%%%%%%%%%%%%%%%%%%%%%%%%%%%%%%%%%%%%%%
%%%%% The hyperbolic semi-circle theorem            %%%%%%%%%%%%%%%%%%
%\begin{figure}[htbp]
\begin{figure}[t]  % try to put this figure on the top of the page
 \centering         % center the figure
 \psfrag{A1}{$A_1$}
 \psfrag{A2}{$A_2$}
 \psfrag{A3}{$A_3$}
 \psfrag{P}{$P$}
 \psfrag{PP}{$P^\prime$}
 \psfrag{PPP}{$P^\pprime$}
 \includegraphics[width=9cm]{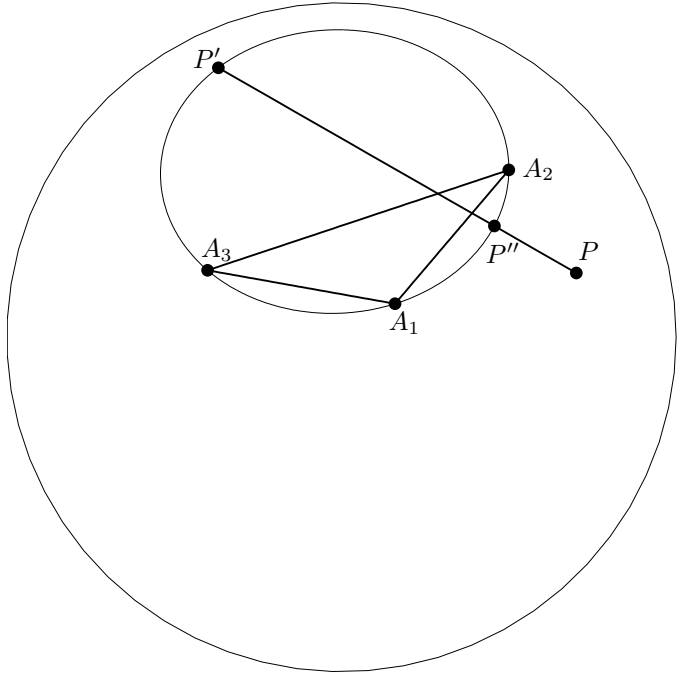}
\caption[A point in the exterior of a circumgyrocircle]{
The point $P^\prime$ lies arbitrarily on the circumgyrocircle
$C(A_1A_2A_3)$ of gyrotriangle $A_1A_2A_3$ in
an Einstein gyrovector plane $(\Rstwo,\op,\od)$, and $P,~P\ne P^\prime$,
is a point in $\Rstwo$.
Both $P$ and $P^\prime$ are given by their gyrobarycentric representation
with respect to the gyrobarycentrically independent set $S=\{ A_1,A_2,A_3\}$
of vertices of the reference gyrotriangle $A_1A_2A_3$.
The point $P^\pprime$, $P^\pprime\ne P^\prime$, that lies on the intersection of the
circumgyrocircle $C(A_1A_2A_3)$ and the gyroline
$L_{^{PP^\prime}}$ through $P$ and $P^\prime$ possesses
a gyrobarycentric representation with respect to the set $S$, with
 gyrobarycentric coordinates that are determined by those of $P$ and $P^\prime$
according to \eqref{skner02sthm}\,--\,\eqref{kfcn01thm}.
\label{fig333enm}}
\end{figure}
%%%%%%%%%%%%%%%%%%%%%%%%%%%%%%%%%%%%%%%%%%%%%%%%%%%%%%%%%%%%%%%%%%%%%%
 % The gyrotriangle circumgyrocenter in (+)E
%                       Fig.~\ref{fig333enm}      Fig. 10.7
%%%%%%%%%%%%%%%%%%%%%%%%%%%%%%%%%%%%%%%%%%%%%%%%%%%%%%%%%%%%%%%%%%%%

%%%%%%%%%%%%%%%%%%%%%%%%%%%%%%%%%%%%%%%%%%%%%%%%%%%%%%%%%%%%%%%%%%%%
% FIGURE 19
 
%%%%%%%%%%%%%%%%%%%%%%%%%%%%%%%%%%%%%%%%%%%%%%%%%%%%%%%%%%%%%%%%%%%%%%
%%%%% The hyperbolic semi-circle theorem            %%%%%%%%%%%%%%%%%%
%\begin{figure}[htbp]
\begin{figure}[t]  % try to put this figure on the top of the page
 \centering         % center the figure
 \psfrag{A1}{$A_1$}
 \psfrag{A2}{$A_2$}
 \psfrag{A3}{$A_3$}
 \psfrag{P}{$P$}
 \psfrag{PP}{$P^\prime$}
 \psfrag{PPP}{$P^\pprime$}
 \includegraphics[width=9cm]{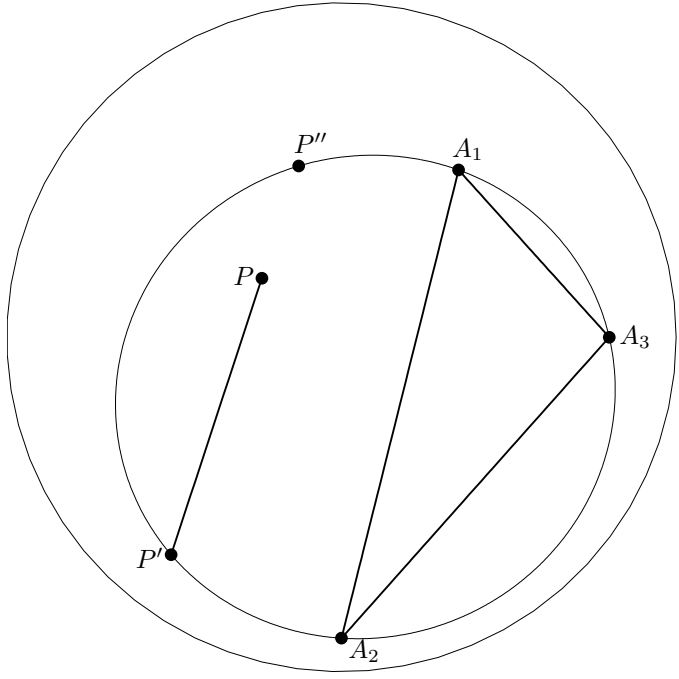}
\caption[A point in the interior of a circumgyrocircle]{
In Fig.~\ref{fig333enm} the point $P$ lies in the exterior of the
circumgyrocircle $C(A_1A_2A_3)$ of gyrotriangle $A_1A_2A_3$.
Here, in contrast, the point $P$ lies in the interior
of $C(A_1A_2A_3)$.
As in Fig.~\ref{fig333enm}, the point $P^\pprime$ lies on the
intersection of circumgyrocircle $C(A_1A_2A_3)$ and
the gyroline $L_{^{PP^\prime}}$ that passes through the points
$P$ and $P^\prime$.
The gyrobarycentric coordinates of $P^\pprime$
are determined by those of $P$ and $P^\prime$
according to \eqref{skner02sthm}\,--\,\eqref{kfcn01thm}.
Both Figs. \ref{fig333enm} and \ref{fig333aenm}
illustrate Theorem \ref{gyintersecthm}, p.~\pageref{gyintersecthm}.
\label{fig333aenm}}
\end{figure}
%%%%%%%%%%%%%%%%%%%%%%%%%%%%%%%%%%%%%%%%%%%%%%%%%%%%%%%%%%%%%%%%%%%%%%
 % The gyrotriangle circumgyrocenter in (+)E
%                       Fig.~\ref{fig333aenm}      Fig. 10.8
%%%%%%%%%%%%%%%%%%%%%%%%%%%%%%%%%%%%%%%%%%%%%%%%%%%%%%%%%%%%%%%%%%%%

Let $P\in\AAbt\cap\Rsn$ be a point given by its gyrobarycentric representation
\begin{equation} \label{tkner02s}
P = \frac{m_1\gAa A_1 + m_2\gAb A_2 + m_3\gAc A_3}
         {m_1\gAa + m_2\gAb + m_3\gAc}
\end{equation}
with respect to a gyrobarycentrically independent set $S=\{ A_1,A_2,A_3\}$
in an Einstein gyrovector space $(\Rsn,\op,\od)$. Additionally, let
$P^\prime\in\Rsn$ be a point that lies arbitrarily
on the circumgyrocircle $C(A_1A_2A_3)$ of gyrotriangle $A_1A_2A_3$,
as shown in Figs.~\ref{fig333enm}\,--\,\ref{fig333aenm},
given by its parametric gyrobarycentric representation
\begin{equation} \label{jkner02s}
P^\prime = P^\prime(t)
= \frac{m_1^\prime\gAa A_1 + m_2^\prime\gAb A_2 + m_3^\prime\gAc A_3}
         {m_1^\prime\gAa + m_2^\prime\gAb + m_3^\prime\gAc}
\end{equation}
with respect to the set $S$, where the
parametric gyrobarycentric coordinates of $P^\prime$ are presented
in \eqref{muramd07d1} below.

Furthermore, let $P^\pprime\in\AAbt\cap\Rsn$ be a point given by its gyrobarycentric representation
\begin{equation} \label{skner02s}
P^\pprime = \frac{m_1^\pprime\gAa A_1 + m_2^\pprime\gAb A_2 + m_3^\pprime\gAc A_3}
         {m_1^\pprime\gAa + m_2^\pprime\gAb + m_3^\pprime\gAc}
\end{equation}
with respect to the set $S$, such that
\begin{enumerate}
\item
$P^\pprime$ lies on the gyroline $L_{^{PP^\prime}}$ that passes through the points
$P$ and $P^\prime$; and
\item
$P^\pprime$ and $P^\prime$, $P^\pprime\ne P^\prime$,
lie on the circumgyrocircle $C(A_1A_2A_3)$ of gyrotriangle $A_1A_2A_3$,
as shown in Figs.~\ref{fig333enm}\,--\,\ref{fig333aenm}.
\end{enumerate}

We will determine the gyrobarycentric coordinates of $P^\pprime$
in terms of those of $P$ and $P^\prime$ and in terms of the
reference gyrotriangle $A_1A_2A_3$ in
\eqref{kfcn01}\,--\,\eqref{kfcn02} below.

Following the condition that $P^\pprime$ lies
on circumgyrocircle $C(A_1A_2A_3)$,
by the Gyrocircle Gyrobarycentric Representation Theorem \ref{thmsnvur}, p.~\pageref{thmsnvur},
the gyrobarycentric coordinates $m_k^\pprime$ of $P^\pprime$
satisfy the circumgyrocircle condition \eqref{muramd05}. That is,
\begin{equation} \label{muramd05s}
m_1^\pprime m_2^\pprime (\gammaab-1)+m_1^\pprime m_3^\pprime (\gammaac-1)
+m_2^\pprime m_3^\pprime (\gammabc-1) = 0
\,.
\end{equation}
%MATHEMATICA stam364

We wish to express, in \eqref{kfcn01} below,
the gyrobarycentric coordinates $m_k^\pprime$ of $P^\pprime$
in terms of the gyrobarycentric coordinates of $P$ and $P^\prime$.

Since the point $P^\pprime$ lies on gyroline $L_{^{PP^\prime}}$, its
gyrobarycentric coordinates $m_k^\pprime$, $k=1,2,3$, are given parametrically by
\cite[Sect.~4.10]{mybook05}, that is,
%%%%%%%%%%%%%%%%%%%%%%%%%%%%%%%%%%%%%%%%%%%%%%%%%%%%%%%%%%%%%%%%%%%%
\begin{equation} \label{skner05s}
\begin{split}
m_1^\pprime &=  m_1( m_1^\prime\gammaaa + m_2^\prime\gammaab +  m_3^\prime\gammaac)
- \left\{
 \left| \begin{matrix}
m_{1} & m_{2} \\[4pt]
m_{1}^\prime & m_{2}^\prime
\end{matrix} \right| \gammaab
+\left| \begin{matrix}
m_{1} & m_{3} \\[4pt]
m_{1}^\prime & m_{3}^\prime
\end{matrix} \right| \gammaac
\right\} t_0
\,,
\\[8pt]
m_2^\pprime &=  m_{2}( m_{1}^\prime\gammaaa +  m_{2}^\prime\gammaab +  m_{3}^\prime\gammaac)
+ \left\{
 \left| \begin{matrix}
m_{1} & m_{2} \\[4pt]
m_{1}^\prime & m_{2}^\prime
\end{matrix} \right| \gammaaa
-\left| \begin{matrix}
m_{2} & m_{3} \\[4pt]
m_{2}^\prime & m_{3}^\prime
\end{matrix} \right| \gammaac
\right\} t_0
\,,
\\[8pt]
m_3^\pprime &=  m_{3}( m_{1}^\prime\gammaaa +  m_{2}^\prime\gammaab +  m_{3}^\prime\gammaac)
+ \left\{
 \left| \begin{matrix}
m_{1} & m_{3} \\[4pt]
m_{1}^\prime & m_{3}^\prime
\end{matrix} \right| \gammaaa
+\left| \begin{matrix}
m_{2} & m_{3} \\[4pt]
m_{2}^\prime & m_{3}^\prime
\end{matrix} \right| \gammaab
\right\} t_0
\,,
\end{split}
\end{equation}
%MATHEMATICA  stam359
%%%%%%%%%%%%%%%%%%%%%%%%%%%%%%%%%%%%%%%%%%%%%%%%%%%%%%%%%%%%%%%%%%%%
$\gammaaa=1$, with the parameter $t_0\in\Rb$.

The point $P^\prime$ lies on circumgyrocircle $C(A_1A_2A_3)$
of gyrotriangle $A_1A_2A_3$.
Hence, by \eqref{muramd07r}, the gyrobarycentric coordinates $m_k^\prime$
of $P^\prime$ can be parametrized as
%%%%%%%%%%%%%%%%%%%%%%%%%%%%%%%%%%%%%%%%%%%%%%%%%%%%%%%%%%%%%%%%%%%%
\begin{equation} \label{muramd07d1}
\begin{split}
m_1^\prime &= -(\gammabc-1)t
\\
m_2^\prime &= (\gammaab-1)t^2 + (\gammaac-1)t = m_3^\prime t
\\
m_3^\prime &= (\gammaab-1)t + (\gammaac-1)
\,,
\end{split}
\end{equation}
%%%%%%%%%%%%%%%%%%%%%%%%%%%%%%%%%%%%%%%%%%%%%%%%%%%%%%%%%%%%%%%%%%%%
with the parameter $t\in\Rb\cup\{-\infty,\infty\}$.

Inserting \eqref{muramd07d1} in \eqref{skner05s},
and, successively, inserting the resulting gyrobarycentric coordinates
$m_k^\pprime$ in the circumgyrocircle condition \eqref{muramd05s},
we obtain a linear equation for the unknown $t_0$.
The resulting solution, $T_0$ of $t_0$ is too involved and hence is not
presented here.
Inserting $t_0=T_0$ in \eqref{skner05s}, we obtain the desired
gyrobarycentric coordinates $m_k^\pprime$, $k=1,2,3$, of $P^\pprime$.
The resulting expressions of $m_k^\pprime$ are, initially, involved.
Fortunately, however,
these involved expressions can be factorized and, owing to the
homogeneity of gyrobarycentric coordinates, nonzero common factors
are irrelevant and, hence, can be omitted.
The resulting
gyrobarycentric coordinates $m_k^\pprime$ in the gyrobarycentric representation
of $P^\pprime$ in \eqref{skner02s} turn out to be simple and elegant.
These are
%%%%%%%%%%%%%%%%%%%%%%%%%%%%%%%%%%%%%%%%%%%%%%%%%%%%%%%%%%%%%%%%%%%%
\begin{equation} \label{kfcn01}
\begin{split}
m_1^\pprime &= E_1E_2
\\
m_2^\pprime &= E_0E_1(\gammaac-1)
\\
m_3^\pprime &=-E_0E_2(\gammaab-1)
\,,
\end{split}
\end{equation}
%MATHEMATICA stam363.
%MATLAB fig333en
%%%%%%%%%%%%%%%%%%%%%%%%%%%%%%%%%%%%%%%%%%%%%%%%%%%%%%%%%%%%%%%%%%%%
where
%%%%%%%%%%%%%%%%%%%%%%%%%%%%%%%%%%%%%%%%%%%%%%%%%%%%%%%%%%%%%%%%%%%%
\begin{equation} \label{kfcn02}
\begin{split}
E_0 &= m_2-m_3t
\\
E_1 &=
m_1(\gammaac-1) + m_2(\gammabc-1) + m_1(\gammaab-1) t
\\
E_2 &=
m_1(\gammaac-1) + m_1(\gammaab-1) t + m_3(\gammabc-1) t
\,.
\end{split}
\end{equation}
%%%%%%%%%%%%%%%%%%%%%%%%%%%%%%%%%%%%%%%%%%%%%%%%%%%%%%%%%%%%%%%%%%%%

Interestingly, $E_0$, $E_1$ and $E_2$ are related by the equation
\begin{equation} \label{kfcp02}
E_1-E_2=E_0(\gammabc-1)
\,.
\end{equation}

Formalizing, we obtain the following
Theorem:
%%%%%%%%%%%%%%%%%%%%%%%%%%%%%%%%%%%%%%%%%%%%%%%%%%%%%%%%%%%%%%%%%%%%
% THEOREM NUMBER 9.14
\index{gyroline gyrocircle intersection, theorem}
\begin{theorem}
{\bf (The Gyroline Gyrocircle Intersection Theorem).}
\label{gyintersecthm}
\begin{enumerate}
\item
Let $A_1A_2A_3$ be a gyrotriangle that possesses a circumgyrocircle
in an Einstein gyrovector space $(\Rsn,\op,\od)$, $n\ge2$,
\item
let $P\in\AAbt\cap\Rsn$ be a point given by its gyrobarycentric representation
\begin{equation} \label{tkner02sthm}
P = \frac{m_1\gAa A_1 + m_2\gAb A_2 + m_3\gAc A_3}
         {m_1\gAa + m_2\gAb + m_3\gAc}
\end{equation}
with respect to the gyrobarycentrically independent set $S=\{ A_1,A_2,A_3\}$
of the vertices of the reference gyrotriangle $A_1A_2A_3$,
\item
let $P^\prime$ be a point that lies arbitrarily on the
circumgyrocircle $C(A_1A_2A_3)$ of gyrotriangle $A_1A_2A_3$,
$P^\prime\ne P$, and
\item
let $P^\pprime$ be the unique point that lies simultaneously
on the circumgyrocircle $C(A_1A_2A_3)$ and on the
gyroline $L_{^{PP^\prime}}$ that passes through the points
$P$ and $P^\prime$ where, in general, $P^\pprime\ne P^\prime$,
as shown in Figs.~\ref{fig333enm}\,--\,\ref{fig333aenm}.
Finally,
\item
let $P^\prime$ be parametrized by the
parameter $t$,
\begin{equation} \label{hfge}
t\in\Rb\cup\{-\infty,\infty\}
\,, 
\end{equation}
so that following
the Gyrocircle Gyrobarycentric Representation Theorem \ref{thmsnvur},
\eqref{muramd07r}, p.~\pageref{muramd07r},
$P^\prime$ possesses the  gyrobarycentric representation
\begin{equation} \label{jkner02sthm}
P^\prime = P^\prime(t)
= \frac{m_1^\prime\gAa A_1 + m_2^\prime\gAb A_2 + m_3^\prime\gAc A_3}
         {m_1^\prime\gAa + m_2^\prime\gAb + m_3^\prime\gAc}
\end{equation}
with respect to the set $S$,
where
%%%%%%%%%%%%%%%%%%%%%%%%%%%%%%%%%%%%%%%%%%%%%%%%%%%%%%%%%%%%%%%%%%%%
\begin{equation} \label{muramd07d1thm}
\begin{split}
m_1^\prime &= -(\gammabc-1)t
\\
m_2^\prime &= (\gammaab-1)t^2 + (\gammaac-1)t = m_3^\prime t
\\
m_3^\prime &= (\gammaab-1)t + (\gammaac-1)
\,.
\end{split}
\end{equation}
%%%%%%%%%%%%%%%%%%%%%%%%%%%%%%%%%%%%%%%%%%%%%%%%%%%%%%%%%%%%%%%%%%%%
\end{enumerate}

Then, the point $P^\pprime$ possesses the  $t$-dependent
gyrobarycentric representation
%%%%%%%%%%%%%%%%%%%%%%%%%%%%%%%%%%%%%%%%%%%%%%%%%%%%%%%%%%%%%%%%%%%%
\begin{equation} \label{skner02sthm}
P^\pprime = P^\pprime(t)
= \frac{m_1^\pprime\gAa A_1 + m_2^\pprime\gAb A_2 + m_3^\pprime\gAc A_3}
         {m_1^\pprime\gAa + m_2^\pprime\gAb + m_3^\pprime\gAc}
\end{equation}
%%%%%%%%%%%%%%%%%%%%%%%%%%%%%%%%%%%%%%%%%%%%%%%%%%%%%%%%%%%%%%%%%%%%
with respect to the set $S$, where
%%%%%%%%%%%%%%%%%%%%%%%%%%%%%%%%%%%%%%%%%%%%%%%%%%%%%%%%%%%%%%%%%%%%
\begin{equation} \label{kfcn01thm}
\begin{split}
m_1^\pprime &= E_1E_2
\\
m_2^\pprime &= E_0E_1(\gammaac-1)
\\
m_3^\pprime &=-E_0E_2(\gammaab-1)
\,,
\end{split}
\end{equation}
%MATHEMATICA stam363.
%MATLAB fig333en
%%%%%%%%%%%%%%%%%%%%%%%%%%%%%%%%%%%%%%%%%%%%%%%%%%%%%%%%%%%%%%%%%%%%
and where
%%%%%%%%%%%%%%%%%%%%%%%%%%%%%%%%%%%%%%%%%%%%%%%%%%%%%%%%%%%%%%%%%%%%
\begin{equation} \label{kfcn02thm}
\begin{split}
E_0 &= m_2-m_3t
\\
E_1 &=
m_1(\gammaac-1) + m_2(\gammabc-1) + m_1(\gammaab-1) t
\\
E_2 &=
m_1(\gammaac-1) + m_1(\gammaab-1) t + m_3(\gammabc-1) t
\,.
\end{split}
\end{equation}
%MATLAB fig333aen
%%%%%%%%%%%%%%%%%%%%%%%%%%%%%%%%%%%%%%%%%%%%%%%%%%%%%%%%%%%%%%%%%%%%
\end{theorem}

% SECTION NUMBER 19
\section{Gyrocircle--Gyroline Tangency Points} \label{slila4d6}
\index{gyrocircle gyroline tangency points}

%%%%%%%%%%%%%%%%%%%%%%%%%%%%%%%%%%%%%%%%%%%%%%%%%%%%%%%%%%%%%%%%%%%%
% THEOREM NUMBER 9.15
\index{gyrocircle gyrotangents, theorem}
\begin{theorem}
{\bf (The Gyrocircle Gyrotangents Theorem).}
\label{gtanthm}
Let $P\in\AAbt\cap\Rsn$ be a point given by its gyrobarycentric representation
\begin{equation} \label{tkner02st}
P = \frac{m_1\gAa A_1 + m_2\gAb A_2 + m_3\gAc A_3}
         {m_1\gAa + m_2\gAb + m_3\gAc}
\end{equation}
with respect to a gyrobarycentrically independent set $S=\{ A_1,A_2,A_3\}$
in an Einstein gyrovector plane $(\Rstwo,\op,\od)$,
and let $C(A_1A_2A_3)$ be the circumgyrocircle of
gyrotriangle $A_1A_2A_3$.
Then the two tangency points $P_{^{\pm}}$ of the gyrotangent gyrolines of
circumgyrocircle $C(A_1A_2A_3)$ that pass through  the point $P$,
when exist as shown in Fig.~\ref{fig334enm},
are given by their gyrobarycentric representations
\begin{equation} \label{prvks}
P_{^{\pm}} = \frac{m_1^\prime\gAa A_1 + m_2^\prime\gAb A_2 + m_3^\prime\gAc A_3}
         {m_1^\prime\gAa + m_2^\prime\gAb + m_3^\prime\gAc}
\end{equation}
with respect to $S$.
The gyrobarycentric coordinates $m_k^\prime$, $k=1,2,3$, of $P_{^{\pm}}$ are
given in terms of the gyrobarycentric coordinates $m_k$ of $P$
and the sides of gyrotriangle $A_1A_2A_3$ by the equations
%%%%%%%%%%%%%%%%%%%%%%%%%%%%%%%%%%%%%%%%%%%%%%%%%%%%%%%%%%%%%%%%%%%%
\begin{equation} \label{hugd04thm}
\begin{split}
m_1^\prime &= F_0F_1(\gammabc-1)
\\[6pt]
m_2^\prime &= F_1F_2
\\[6pt]
m_3^\prime &=-F_0F_2(\gammaab-1)
\,,
\end{split}
\end{equation}
%%%%%%%%%%%%%%%%%%%%%%%%%%%%%%%%%%%%%%%%%%%%%%%%%%%%%%%%%%%%%%%%%%%%
where
%%%%%%%%%%%%%%%%%%%%%%%%%%%%%%%%%%%%%%%%%%%%%%%%%%%%%%%%%%%%%%%%%%%%
\begin{equation} \label{hugd05thm}
\begin{split}
F_0 &= \phantom{-}m_1(\gammaab-1) + m_3(\gammabc-1)
\\[8pt]
F_1 &= \phantom{-}m_1(\gammaab-1)(\gammaac-1) \pm\sqrt{-\Delta_1\Delta_2}
\\[8pt]
F_2 &=-m_3(\gammaac-1)(\gammabc-1) \pm\sqrt{-\Delta_1\Delta_2}
\,,
\end{split}
\end{equation}
%Mathematica stam364, fig334aen.m
%%%%%%%%%%%%%%%%%%%%%%%%%%%%%%%%%%%%%%%%%%%%%%%%%%%%%%%%%%%%%%%%%%%%
and where
%%%%%%%%%%%%%%%%%%%%%%%%%%%%%%%%%%%%%%%%%%%%%%%%%%%%%%%%%%%%%%%%%%%%
\begin{equation} \label{hugd03thm}
\begin{split}
\Delta_1 &= m_1m_2(\gammaab-1)+m_1m_3(\gammaac-1)+m_2m_3(\gammabc-1)
\\[6pt]
\Delta_2 &= (\gammaab-1)(\gammaac-1)(\gammabc-1)>0
\,.
\end{split}
\end{equation}
%%%%%%%%%%%%%%%%%%%%%%%%%%%%%%%%%%%%%%%%%%%%%%%%%%%%%%%%%%%%%%%%%%%%
%Matlab fig334en

%%%%%%%%%%%%%%%%%%%%%%%%%%%%%%%%%%%%%%%%%%%%%%%%%%%%%%%%%%%%%%%%%%%%
% FIGURE 20
 
%%%%%%%%%%%%%%%%%%%%%%%%%%%%%%%%%%%%%%%%%%%%%%%%%%%%%%%%%%%%%%%%%%%%%%
%%%%% The hyperbolic semi-circle theorem            %%%%%%%%%%%%%%%%%%
%\begin{figure}[htbp]
\begin{figure}[t]  % try to put this figure on the top of the page
 \centering         % center the figure
 \psfrag{A1}{$A_1$}
 \psfrag{A2}{$A_2$}
 \psfrag{A3}{$A_3$}
 \psfrag{P}{$P$}
 \psfrag{PTp}{$P_{^+}$}
 \psfrag{PTm}{$P_{^-}$}
 \includegraphics[width=9cm]{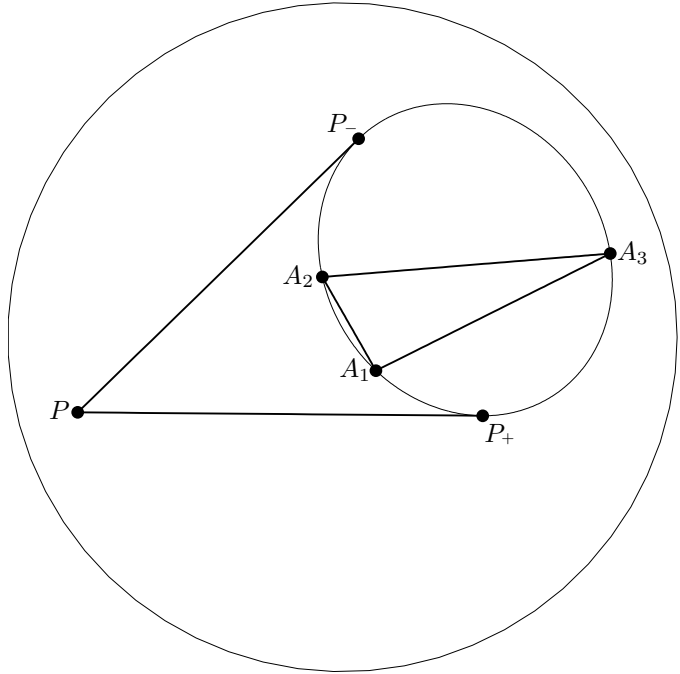}
\caption[Illustration of the Gyrocircle Tangents Theorem]{
Illustrating the Gyrocircle Tangents Theorem \ref{gtanthm}.
The two tangency points $P_{^{\pm}}$ of the circumgyrocircle $C(A_1A_2A_3)$ of
a gyrotriangle $A_1A_2A_3$ in an Einstein gyrovector plane $(\Rstwo,\op,\od)$,
which lie on the tangent lines that pass through a point $P$ that lies in the
exterior of the circumgyrocircle, are shown.
The Euclidean counterpart of this figure is presented in Fig.~\ref{fig336eucm}.
\label{fig334enm}}
\end{figure}
%%%%%%%%%%%%%%%%%%%%%%%%%%%%%%%%%%%%%%%%%%%%%%%%%%%%%%%%%%%%%%%%%%%%%%
 % The gyrotriangle circumgyrocenter in (+)E
%                    and tangency points.
%                       Fig.~\ref{fig334enm}      Fig. 10.9
%%%%%%%%%%%%%%%%%%%%%%%%%%%%%%%%%%%%%%%%%%%%%%%%%%%%%%%%%%%%%%%%%%%%

\begin{enumerate}
\item
Two distinct tangency points, $P_{^{\pm}}$,
exist if and only if $\Delta_1<0$ or,
equivalently, if and only if the point $P$ lies in the exterior
of the circumgyrocircle $C(A_1A_2A_3)$,
as shown in Fig.~\ref{fig334enm}.
\item
The two distinct tangency points degenerate to a single one,
$P_{^{\pm}}=P$ if and only if $\Delta_1=0$ or,
equivalently, if and only if the point $P$ lies
on the circumgyrocircle $C(A_1A_2A_3)$.
\item
There are no tangency points if and only if $\Delta_1>0$ or,
equivalently, if and only if the point $P$ lies in the interior
of the circumgyrocircle $C(A_1A_2A_3)$.
\end{enumerate}

The constant $\mPpm$ of the gyrobarycentric representation \eqref{prvks}
of $P_{\pm}$ is given by
%%%%%%%%%%%%%%%%%%%%%%%%%%%%%%%%%%%%%%%%%%%%%%%%%%%%%%%%%%%%%%%%%%%%
\begin{equation} \label{hilnad1}
\mPpm = (\gammaab-1) (\gammaac-1) (\gammabc-1) E_1 \pm \sqrt{-\Delta_1\Delta_2} E_2
\,,
\end{equation}
where
%%%%%%%%%%%%%%%%%%%%%%%%%%%%%%%%%%%%%%%%%%%%%%%%%%%%%%%%%%%%%%%%%%%%
\begin{equation} \label{hilnad2}
\begin{split}
E_1 &= (m_1-m_2+m_3) \{ m_1(\gammaab-1)+m_3(\gammabc-1)\} -2m_1m_3(\gammaac-1)
\\[6pt]
E_2 &= -m_1(\gammaab-1)^2+m_3(\gammabc-1)^2
+m_1(\gammaab-1) (\gammaac-1)
\\[6pt] & \phantom{=}
- m_3(\gammaac-1) (\gammabc-1)
+(m_1-m_3)(\gammaab-1) (\gammabc-1)
\,.
\end{split}
\end{equation}
%%%%%%%%%%%%%%%%%%%%%%%%%%%%%%%%%%%%%%%%%%%%%%%%%%%%%%%%%%%%%%%%%%%%
%%%%%%%%%%%%%%%%%%%%%%%%%%%%%%%%%%%%%%%%%%%%%%%%%%%%%%%%%%%%%%%%%%%%
\end{theorem}
%%%%%%%%%%%%%%%%%%%%%%%%%%%%%%%%%%%%%%%%%%%%%%%%%%%%%%%%%%%%%%%%%%%%
\begin{proof}

The  Gyroline Gyrocircle Intersection Theorem \ref{gyintersecthm}
enables gyrocircle--gyroline tangency points to be determined.
Indeed, in the special case when the two points $P^\prime$ and $P^\pprime$ in
Theorem \ref{gyintersecthm},
in \eqref{jkner02sthm} and in \eqref{skner02sthm},
shown in Fig.~\ref{fig333enm},
are coincident, the resulting point, $P_t:=P^\prime=P^\pprime$
is a point of tangency.
There are two distinct ways for the points $P^\prime$ and $P^\pprime$
to become coincident by approaching each other along the circumference of
the circumgyrocircle $C(A_1A_2A_3)$ of gyrotriangle $A_1A_2A_3$.
These give rise to two tangency points to the circumgyrocircle that lie
on the two gyrotangent gyrolines that pass through the point $P$,
shown in Fig.~\ref{fig334enm}.

The points $P^\prime$ and $P^\pprime$ in
\eqref{jkner02sthm} and \eqref{skner02sthm} possess
gyrobarycentric representations that involve the
parameter $t\in\Rb\cup\{-\infty,\infty\}$.
The condition $P^\prime=P^\pprime$ that gives rise to tangency points
determines the values of the parameter $t$ that correspond
to the two tangency points.

Indeed, if $P^\prime=P^\pprime$ then it follows from the
gyrobarycentric representations of $P^\prime$ and $P^\pprime$ in
\eqref{jkner02sthm} and \eqref{skner02sthm} that
%%%%%%%%%%%%%%%%%%%%%%%%%%%%%%%%%%%%%%%%%%%%%%%%%%%%%%%%%%%%%%%%%%%%
\begin{equation} \label{hugd01}
\begin{split}
\frac{m_2^\prime}{m_1^\prime} &= \frac{m_2^\pprime}{m_1^\pprime}
\\[6pt]
\frac{m_3^\prime}{m_1^\prime} &= \frac{m_3^\pprime}{m_1^\pprime}
\,,
\end{split}
\end{equation}
%%%%%%%%%%%%%%%%%%%%%%%%%%%%%%%%%%%%%%%%%%%%%%%%%%%%%%%%%%%%%%%%%%%%
thus obtaining two equations for the parameter $t$ that gives rise
to the tangency points.

Inserting \eqref{muramd07d1thm}
and \eqref{kfcn01thm}\,--\,\eqref{kfcn02thm}
in \eqref{hugd01}, we see that the two equations in \eqref{hugd01}
are equivalent to each other. Solving one of these equations for
the unknown parameter value $t$, we obtain
\begin{equation} \label{hugd02}
t = \frac{
-m_1 (\gammaab-1)(\gammaac-1) \pm \sqrt{-\Delta_1\Delta_2}
}{
\{m_1 (\gammaab-1) + m_3 (\gammabc-1)\} (\gammaab-1)
}
\,,
\end{equation}
% Mathematica stam364
% Matlab fig334en
where
%%%%%%%%%%%%%%%%%%%%%%%%%%%%%%%%%%%%%%%%%%%%%%%%%%%%%%%%%%%%%%%%%%%%
\begin{equation} \label{hugd03}
\begin{split}
\Delta_1 &= m_1m_2(\gammaab-1)+m_1m_3(\gammaac-1)+m_2m_3(\gammabc-1)
\\[6pt]
\Delta_2 &= (\gammaab-1)(\gammaac-1)(\gammabc-1)>0
\,.
\end{split}
\end{equation}
%%%%%%%%%%%%%%%%%%%%%%%%%%%%%%%%%%%%%%%%%%%%%%%%%%%%%%%%%%%%%%%%%%%%

We see from the Gyrocircle Gyrobarycentric Representation Theorem
\ref{thmsnvur}, p.~\pageref{thmsnvur}, that
\begin{enumerate}
\item
$\Delta_1=0$ if and only if $P$ lies on circumgyrocircle $C(A_1A_2A_3)$;
\item
$\Delta_1>0$
if and only if the point $P$ in \eqref{tkner02sthm} lies in the interior
of circumgyrocircle $C(A_1A_2A_3)$; and
\item
$\Delta_1<0$
if and only if the point $P$ lies in the exterior of
circumgyrocircle $C(A_1A_2A_3)$.
\end{enumerate}

We thus see from \eqref{hugd02} that a value of the real parameter $t$ that
corresponds to a tangency point $P_t$, $P_t:=P^\prime=P^\pprime$,
exists if and only if the point $P$
lies on circumgyrocircle $C(A_1A_2A_3)$ or in its exterior.
Indeed, this result is expected
since it is clear from Figs.~\ref{fig333enm}\,--\,\ref{fig333aenm} that
a gyrotangent gyroline to the circumgyrocircle that passes through the
point $P$ exists if and only if $P$ lies in the exterior
of the circumgyrocircle.
Clearly, the two distinct tangency points $P_{^{\pm}}$ of
circumgyrocircle $C(A_1A_2A_3)$ degenerate to a single point,
$P_{^{\pm}}=P$,
if and only if the point $P$ lies on the circumgyrocircle.

Inserting the parameter value of $t$ from \eqref{hugd02} into
the gyrobarycentric coordinates $m_k^\prime$, $k=1,2,3$, of
$P^\prime$ in \eqref{muramd07d1thm}, and omitting irrelevant
nonzero common factors, we obtain
the gyrobarycentric coordinates
in \eqref{hugd04thm}\,--\,\eqref{hugd03thm}, as desired.

By \cite[Eq.~(4.27), p.~90]{mybook05} and the
{\it Gyrobarycentric representation Gyrocovariance Theorem},
\cite[Theorem 4.6, pp.~90-91]{mybook05},
the constant $\mPpm$ of the
gyrobarycentric representation \eqref{prvks} of $P_{^{\pm}}$
is given by the equation
\begin{equation} \label{htetn1}
m_{P_{\pm}}^2 = (m_1^\prime)^2 + (m_2^\prime)^2 + (m_3^\prime)^2
+2(m_1^\prime m_2^\prime \gammaab+m_1^\prime m_3^\prime \gammaac
+m_2^\prime m_3^\prime \gammabc)
\,.
\end{equation}
The substitution of the gyrobarycentric coordinates $m_k^\prime$
from \eqref{hugd04thm} into \eqref{htetn1} yields
\eqref{hilnad1}, as desired.
\end{proof}
%MATLAB fig334en zerotg1,2
%MATLAB fig334cen zeromptm,p.
%MATHEMATICA stam373 calculated the constant of the b-rep of P_{+-}.
%   It is a perfect square.

% SECTION NUMBER 20
\section{Gyrocircle Gyrotangent Gyrolength} \label{slils}
\index{gyrocircle gyrotangent gyrolength}

Let $P_{^{\pm}}$ denote collectively
the two tangency points, $P_{+}$ and $P_{-}$, of the gyrotangent gyrolines
drawn from a point $P$ to the circumgyrocircle $C(A_1A_2A_3)$
of a gyrotriangle $A_1A_2A_3$ in an Einstein gyrovector space $(\Rsn,\op,\od)$,
as shown in Fig.~\ref{fig334enm} for $n=2$.
Accordingly, let
$P\in\rmspan\{\om A_1 \op A_2,~\om A_1 \op A_3\}\subset\Rn$
be a point on the gyroplane of gyrotriangle $A_1A_2A_3$
with gyrobarycentric representation
\begin{equation} \label{matos1}
P = \frac{
m_1 \gAa A_1 + m_2 \gAb A_2 + m_3 \gAc A_3
}{
m_1 \gAa + m_2 \gAb + m_3 \gAc
}
\end{equation}
so that, following Theorem \ref{gtanthm}, \eqref{prvks},
the tangency point $P_{^{\pm}}$ possesses the gyrobarycentric representation
\begin{equation} \label{matos2}
P_{^{\pm}} = \frac{m_1^\prime\gAa A_1 + m_2^\prime\gAb A_2 + m_3^\prime\gAc A_3}
         {m_1^\prime\gAa + m_2^\prime\gAb + m_3^\prime\gAc}
\end{equation}
with gyrobarycentric coordinates $m_k^\prime$, $k=1,2,3$, given by
\eqref{hugd04thm}\,--\,\eqref{hugd03thm}.

Hence, by \cite[Sect.~4.9]{mybook05} with $N=3$,
%%%%%%%%%%%%%%%%%%%%%%%%%%%%%%%%%%%%%%%%%%%%%%%%%%%%%%%%%%%%%%%%%%%%
\begin{equation} \label{matos3}
\begin{split}
 \gamma_{\om P\op P_{^{\pm}}}^{\phantom{O}} &= \frac{1}{\mP \mPpm}\left\{
(m_1m_2^\prime + m_1^\prime m_2) \gammaab +
(m_1m_3^\prime + m_1^\prime m_3) \gammaac \right.
\\[6pt] & \hspace{1.6cm}
+(m_2m_3^\prime + m_2^\prime m_3) \gammabc
+
\left.
m_1m_1^\prime + m_2m_2^\prime + m_3m_3^\prime \right\}
%\\[6pt] &
%\hspace{-1.2cm}
%= \frac{1}{\mP \mPpm}\left\{
%(m_1m_2^\prime + m_1^\prime m_2) \gammaab +
%m_1^\prime m_3 \gammaac +
%m_2^\prime m_3 \gammabc +
%m_1m_1^\prime + m_2m_2^\prime \right\}
\,,
\end{split}
\end{equation}
%%%%%%%%%%%%%%%%%%%%%%%%%%%%%%%%%%%%%%%%%%%%%%%%%%%%%%%%%%%%%%%%%%%%
where by \cite[Eq.~(4.27), p.~90]{mybook05} and the
{\it Gyrobarycentric representation Gyrocovariance Theorem},
\cite[Theorem 4.6, pp.~90-91]{mybook05},
the constant $\mP$ of the gyrobarycentric representation
\eqref{matos1} of $P$ is given by
\begin{equation} \label{skner06s}
m_P^2 = (m_1+m_2+m_3)^2 + 2(m_1m_2(\gammaab-1)+m_1m_3(\gammaac-1)+m_2m_3(\gammabc-1))
\end{equation}
and where $\mPpm$ is the constant of the gyrobarycentric representation
\eqref{matos2} of $P_{^{\pm}}$, given by \eqref{hilnad1}.

Inserting
$m_k^\prime$, $k=1,2,3$, from \eqref{hugd04thm}\,--\,\eqref{hugd03thm}
into \eqref{matos3}, we obtain the elegant equation
\begin{equation} \label{harfg}
\gamma_{\om P \op P_{\pm}}^{\phantom{O}} = \frac{
m_1+m_2+m_3 }{ m_P }
\,.
\end{equation}
%MATHEMATICA stam377 and stam373
%MATLAB fig334den

The gyrolength of the gyrotangent gyrosegment $PP_{\pm}$
can readily be obtained from \eqref{harfg}.
Indeed, following \eqref{harfg}
and \eqref{rugh1ds}, p.~\pageref{rugh1ds}, we have
\begin{equation} \label{harfh}
\|\om P \op P_{\pm}\| = s \frac{
\sqrt{\gamma_{\om P \op P_{\pm}}^2-1}
}{
\gamma_{\om P \op P_{\pm}}^{\phantom{O}}
}
\,.
\end{equation}

Unlike the gyrobarycentric coordinates $m_k^\prime$ and the
constant $\mPpm$ in \eqref{matos3}, the gamma factor in \eqref{harfg}
is free of the square root in \eqref{hugd05thm} and \eqref{hilnad1}.
This implies the gyrodistance equality
\begin{equation} \label{harfi}
\|\om P \op P_{+}\| = \|\om P \op P_{-}\|
\,.
\end{equation}

Following \eqref{skner06s}\,--\,\eqref{harfh} we have
\begin{equation} \label{harfha1}
\begin{split}
\|\om P \op P_{\pm}\|
&= s \frac{
\sqrt{\gamma_{\om P \op P_{\pm}}^2-1}
}{
\gamma_{\om P \op P_{\pm}}^{\phantom{O}}
}
\\[8pt] &
= s \frac{
\sqrt{
-2(m_1m_2(\gammaab-1) + m_1m_3(\gammaac-1) + m_2m_3(\gammabc-1))
}
}{
m_1+m_2+m_3
}
\end{split}
\end{equation}

Note that by Theorem \ref{thmsnvur}, p.~\pageref{thmsnvur},
the radicand on the extreme right-hand side of
\eqref{harfha1} is positive (zero) if and only if the point $P$
lies in the exterior of of circumgyrocircle $C(A_1A_2A_3)$
(if and only if the point $P$ lies on $C(A_1A_2A_3)$).

Formalizing results of this section, we obtain the
following theorem.
%%%%%%%%%%%%%%%%%%%%%%%%%%%%%%%%%%%%%%%%%%%%%%%%%%%%%%%%%%%%%%%%%%%%
\index{gyrocircle gyrotangent gyrosegment gyrolength}
% THEOREM NUMBER 9.16
\begin{theorem}\label{harfk} 
{\bf (Gyrocircle Gyrotangent Gyrolength).}
\begin{enumerate}
\item \label{itemgd1}
Let $A_1A_2A_3$ be a gyrotriangle that possesses a circumgyrocircle
$C(A_1A_2A_3)$
in an Einstein gyrovector space $(\Rsn,\op,\od)$,
\item \label{itemgd2}
let
$P\in\rmspan\{\om A_1 \op A_2,~\om A_1 \op A_3\}\subset\Rn$
be a point that lies in the exterior of $C(A_1A_2A_3)$,
given by its gyrobarycentric representation
\begin{equation} \label{tkner02w}
P = \frac{m_1\gAa A_1 + m_2\gAb A_2 + m_3\gAc A_3}
         {m_1\gAa + m_2\gAb + m_3\gAc}
\end{equation}
with respect to the set $S=\{ A_1,A_2,A_3\}$, and
\item \label{itemgd3}
let $P_{^{\pm}}$ represent collectively the two tangency points, $P_{+}$ and $P_{-}$,
of the two gyrotangent gyrolines
drawn from $P$ to circumgyrocircle $C(A_1A_2A_3)$, as shown in
Fig.~\ref{fig334enm}.
\end{enumerate}
Then, we have the following results:
\begin{enumerate}
\item \label{itemge1}
The tangency points $P_{^{\pm}}$ possess the gyrobarycentric representations
\begin{equation} \label{matos2s}
P_{^{\pm}} = \frac{m_1^\prime\gAa A_1 + m_2^\prime\gAb A_2 + m_3^\prime\gAc A_3}
         {m_1^\prime\gAa + m_2^\prime\gAb + m_3^\prime\gAc}
\end{equation}
with gyrobarycentric coordinates $m_k^\prime$, $k=1,2,3$, given by
\eqref{hugd04thm}\,--\,\eqref{hugd03thm}.
\item \label{itemge2}
The gyrolengths of the two gyroline gyrosegments $PP_{+}$ and $PP_{-}$
are equal,
\begin{equation} \label{wkdm1}
\|\om P \op P_{+}\| = \|\om P \op P_{-}\|
\,.
\end{equation}
\item \label{itemge3}
The gamma factor of each of the two gyrosegments $PP_{+}$ and $PP_{-}$
is given by the equations
\begin{equation} \label{harfkm}
\gamma_{\om P \op P_{\pm}}^{\phantom{O}}
=
\frac{m_1+m_2+m_3 }{\mP}
\,,
\end{equation}
%MATLAB fig334den
%MATHEMATICA stam377 and stam373
where 
$\mP$ is the constant of the gyrobarycentric representation of $P$
with respect to $S$.
\item \label{itemge4}
The gyrolength of each of the 
two gyrotangent gyrosegments $PP_{+}$ and $PP_{-}$ is given by
%%%%%%%%%%%%%%%%%%%%%%%%%%%%%%%%%%%%%%%%%%%%%%%%%%%%%%%%%%%%%%%%%%%%
\begin{equation} \label{harfdkb}
\begin{split}
\|\om P \op P_{\pm}\|
&= s \frac{
\sqrt{\gamma_{\om P \op P_{\pm}}^2-1}
}{
\gamma_{\om P \op P_{\pm}}^{\phantom{O}}
}
\\[8pt] & \hspace{-0.8cm}
= s \frac{
\sqrt{
-2\{ m_1m_2(\gammaab-1) + m_1m_3(\gammaac-1) + m_2m_3(\gammabc-1) \}
}
}{
m_1+m_2+m_3
}
\,.
\end{split}
\end{equation}
%%%%%%%%%%%%%%%%%%%%%%%%%%%%%%%%%%%%%%%%%%%%%%%%%%%%%%%%%%%%%%%%%%%%
\end{enumerate}
\end{theorem}
%%%%%%%%%%%%%%%%%%%%%%%%%%%%%%%%%%%%%%%%%%%%%%%%%%%%%%%%%%%%%%%%%%%%

% SECTION NUMBER 21
\section{Circle--Line Tangency Points} \label{slila4d7}
\index{circle-line tangency points}

The Euclidean counterpart of the
Gyrocircle Gyrotangents Theorem \ref{gtanthm}
is the following elegant theorem, illustrated in Fig.~\ref{fig336eucm}:

%%%%%%%%%%%%%%%%%%%%%%%%%%%%%%%%%%%%%%%%%%%%%%%%%%%%%%%%%%%%%%%%%%%%
% FIGURE 21
 
%%%%%%%%%%%%%%%%%%%%%%%%%%%%%%%%%%%%%%%%%%%%%%%%%%%%%%%%%%%%%%%%%%%%%%
%%%%% The hyperbolic semi-circle theorem            %%%%%%%%%%%%%%%%%%
%\begin{figure}[htbp]
\begin{figure}[t]  % try to put this figure on the top of the page
 \centering         % center the figure
 \psfrag{A1}{$A_1$}
 \psfrag{A2}{$A_2$}
 \psfrag{A3}{$A_3$}
 \psfrag{P}{$P$}
 \psfrag{Pp}{$P_{^+}$}
 \psfrag{Pm}{$P_{^-}$}
 \includegraphics[width=9cm]{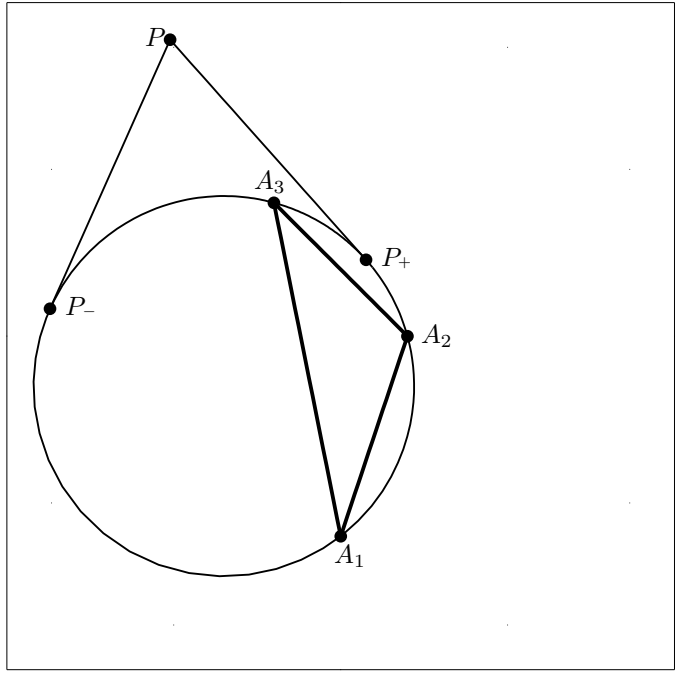}
\caption[Illustration of the Circle Tangents Theorem]{
Illustrating the Circle Tangents Theorem \ref{gtanthmeuc}.
The two tangency points $P_{^{\pm}}$ of the circumcircle $C(A_1A_2A_3)$ of
a triangle $A_1A_2A_3$ in a Euclidean plane $\Rtwo$,
which lie on the tangent lines that pass through a point $P$ that lies
in the exterior of the circumcircle, are shown.
This figure is the Euclidean counterpart of Fig.~\ref{fig334enm}.
\label{fig336eucm}}
\end{figure}
%%%%%%%%%%%%%%%%%%%%%%%%%%%%%%%%%%%%%%%%%%%%%%%%%%%%%%%%%%%%%%%%%%%%%%
 % The triangle circumcenter in (+)E
%                    and tangency points.
%                       Fig.~\ref{fig336eucm}      Fig. 10.9
%%%%%%%%%%%%%%%%%%%%%%%%%%%%%%%%%%%%%%%%%%%%%%%%%%%%%%%%%%%%%%%%%%%%

%%%%%%%%%%%%%%%%%%%%%%%%%%%%%%%%%%%%%%%%%%%%%%%%%%%%%%%%%%%%%%%%%%%%
% THEOREM NUMBER 9.17
\index{circle tangents, theorem}
\begin{theorem}
{\bf (The Circle Tangents Theorem).}
\label{gtanthmeuc}
Let $P\in\Rtwo$ be a point given by its barycentric representation
\begin{equation} \label{tkner02steuc}
P = \frac{m_1 A_1 + m_2 A_2 + m_3 A_3}
         {m_1 + m_2 + m_3}
\end{equation}
with respect to a barycentrically independent set $S=\{ A_1,A_2,A_3\}$
in a Euclidean plane $\Rtwo$,
and let $C(A_1A_2A_3)$ be the circumcircle of
triangle $A_1A_2A_3$.
Then the two tangency points $P_{^{\pm}}$ of the tangent lines of
circumcircle $C(A_1A_2A_3)$ that pass through  the point $P$,
when exist, are given by their barycentric representations
\begin{equation} \label{prvkseuc}
P_{^{\pm}} = \frac{m_1^\prime A_1 + m_2^\prime A_2 + m_3^\prime A_3}
         {m_1^\prime + m_2^\prime + m_3^\prime}
\,,
\end{equation}
shown in Fig.~\ref{fig336eucm}.
The barycentric coordinates of $P_{^{\pm}}$ are
given by the equations
%%%%%%%%%%%%%%%%%%%%%%%%%%%%%%%%%%%%%%%%%%%%%%%%%%%%%%%%%%%%%%%%%%%%
\begin{equation} \label{hugd04thmeuc}
\begin{split}
m_1^\prime &= \phantom{-}F_0F_1 \sin^2\alpha_1
\\[6pt]
m_2^\prime &= \phantom{-}F_1F_2
\\[6pt]
m_3^\prime &=-F_0F_2\sin^2\alpha_3
\,,
\end{split}
\end{equation}
%%%%%%%%%%%%%%%%%%%%%%%%%%%%%%%%%%%%%%%%%%%%%%%%%%%%%%%%%%%%%%%%%%%%
where
%%%%%%%%%%%%%%%%%%%%%%%%%%%%%%%%%%%%%%%%%%%%%%%%%%%%%%%%%%%%%%%%%%%%
\begin{equation} \label{hugd05thmeuc}
\begin{split}
F_0 &= m_1 \sin^2\alpha_3 + m_3 \sin^2\alpha_1
\\[8pt]
F_1 &= \{m_1 \sin^2\alpha_2 \sin^2\alpha_3
\pm\sin\alpha_1\sin\alpha_2\sin\alpha_3 \sqrt{-\Delta}\} \sin^2\alpha_1
\\[8pt]
F_2 &= \{m_3 \sin^2\alpha_1 \sin^2\alpha_2
\mp\sin\alpha_1\sin\alpha_2\sin\alpha_3 \sqrt{-\Delta}\} \sin^2\alpha_3
\,,
\end{split}
\end{equation}
%%%%%%%%%%%%%%%%%%%%%%%%%%%%%%%%%%%%%%%%%%%%%%%%%%%%%%%%%%%%%%%%%%%%
and where
%%%%%%%%%%%%%%%%%%%%%%%%%%%%%%%%%%%%%%%%%%%%%%%%%%%%%%%%%%%%%%%%%%%%
\begin{equation} \label{hugd03thmeuc}
\Delta = 
m_1m_2 \sin^2\alpha_3 + m_1m_3 \sin^2\alpha_2 + m_2m_3 \sin^2\alpha_1
\,.
\end{equation}
%%%%%%%%%%%%%%%%%%%%%%%%%%%%%%%%%%%%%%%%%%%%%%%%%%%%%%%%%%%%%%%%%%%%
%Matlab fig336eucn.m
\begin{enumerate}
\item
Two distinct tangency points, $P_{^{\pm}}$,
exist if and only if $\Delta<0$ or,
equivalently, if and only if the point $P$ lies in the exterior
of the circumcircle $C(A_1A_2A_3)$,
as shown in Fig.~\ref{fig336eucm}.
\item
The two distinct tangency points degenerate to a single one,
$P_{^{\pm}}=P$ if and only if $\Delta=0$ or,
equivalently, if and only if the point $P$ lies
on the circumcircle $C(A_1A_2A_3)$.
\item
There are no tangency points if and only if $\Delta>0$ or,
equivalently, if and only if the point $P$ lies in the interior
of the circumcircle $C(A_1A_2A_3)$.
\item
The lengths of the two tangent segments
$\|-P+P_{+}\|$ and $\|-P+P_{+}\|$
are equal, given by
\begin{equation} \label{rkdn1}
\|-P+P_{^{\pm}}\| = \frac{ \sqrt{
-\{ m_1m_2 a_{12}^2 + m_1m_3 a_{13}^2 + m_2m_3 a_{23}^2\}
}}{
m_1+m_2+m_3
}
\,.
\end{equation}
%MATLAB fig336eucE
\end{enumerate}
\end{theorem}
%%%%%%%%%%%%%%%%%%%%%%%%%%%%%%%%%%%%%%%%%%%%%%%%%%%%%%%%%%%%%%%%%%%
\begin{proof}
We will show that each result of this theorem is the
Euclidean limit of a corresponding result of its
hyperbolic counterpart, Theorem \ref{gtanthm}.
Noting that in the Euclidean limit, $s\rightarrow\infty$,
gamma factors tend to $1$, the Euclidean limit of
the gyrobarycentric coordinates
$m_k^\prime$, $k=1,2,3$, in \eqref{prvks}\,--\,\eqref{hugd03thm}
is trivial. The resulting barycentric coordinates
vanish, giving  rise to an indeterminate barycentric representation.
Thus, a straightforward application of the Euclidean limit
to the gyrobarycentric coordinates $m_k^\prime$ 
in \eqref{prvks}\,--\,\eqref{hugd03thm} in an attempt
to recover the corresponding barycentric coordinates $m_k^\prime$
in \eqref{prvkseuc}\,--\,\eqref{hugd03thmeuc}
is destined to fail.

However, owing to their homogeneity,
the gyrobarycentric coordinates $m_k^\prime$
in \eqref{prvks}\,--\,\eqref{hugd03thm}
can be written in a form that admits a Euclidean limit
to a viable barycentric coordinate system. Specifically,
the following equivalent form, in which
each gyrobarycentric coordinate is multiplied by a
common nonzero factor, is what we need, in which we replace
$m_k^\prime$ by $m_k^\pprime$:
%%%%%%%%%%%%%%%%%%%%%%%%%%%%%%%%%%%%%%%%%%%%%%%%%%%%%%%%%%%%%%%%%%%%
\begin{equation} \label{matok01}
\begin{split}
m_1^\pprime &= \frac{m_1^\prime}{(\gammaab-1)^4}
=
\frac{F_0}{\gammaab-1}
\frac{F_1}{(\gammaab-1)^2}
\frac{\gammabc-1}{\gammaab-1}
\\[8pt]
m_2^\pprime &= \frac{m_2^\prime}{(\gammaab-1)^4}
=
\frac{F_1}{(\gammaab-1)^2}
\frac{F_2}{(\gammaab-1)^2}
\\[8pt]
m_3^\pprime &= \frac{m_3^\prime}{(\gammaab-1)^4}
=
-\frac{F_0}{\gammaab-1}
\frac{F_1}{(\gammaab-1)^2}
\frac{\gammaab-1}{\gammaab-1}
\,,
\end{split}
\end{equation}
%%%%%%%%%%%%%%%%%%%%%%%%%%%%%%%%%%%%%%%%%%%%%%%%%%%%%%%%%%%%%%%%%%%%
where, by \eqref{hugd05thm},
%%%%%%%%%%%%%%%%%%%%%%%%%%%%%%%%%%%%%%%%%%%%%%%%%%%%%%%%%%%%%%%%%%%%
\begin{equation} \label{matok02}
\begin{split}
\frac{F_0}{\gammaab-1} &= m_1 + m_3 \frac{\gammabc-1}{\gammaab-1}
:= F_0^\prime
\\[8pt]
\frac{F_1}{(\gammaab-1)^2} &= m_1 \frac{\gammaac-1}{\gammaab-1}
\pm \left\{
\frac{-\Delta_1}{\gammaab-1}
\frac{ \Delta_2}{(\gammaab-1)^3}
\right\}^{\half}
:= F_1^\prime
\\[8pt]
\frac{F_2}{(\gammaab-1)^2} &= - m_3 \frac{\gammaac-1}{\gammaab-1}
\frac{\gammabc-1}{\gammaab-1}
\pm \left\{
\frac{-\Delta_1}{\gammaab-1}
\frac{ \Delta_2}{(\gammaab-1)^3}
\right\}^{\half}
:= F_2^\prime
\,,
\end{split}
\end{equation}
%%%%%%%%%%%%%%%%%%%%%%%%%%%%%%%%%%%%%%%%%%%%%%%%%%%%%%%%%%%%%%%%%%%%
and where
%%%%%%%%%%%%%%%%%%%%%%%%%%%%%%%%%%%%%%%%%%%%%%%%%%%%%%%%%%%%%%%%%%%%
\begin{equation} \label{matok03}
\begin{split}
\frac{\Delta_1}{\gammaab-1} &= m_1m_2
+ m_1m_3 \frac{\gammaac-1}{\gammaab-1}
+ m_2m_3 \frac{\gammabc-1}{\gammaab-1}
:= \Delta_1^\prime
\\[8pt]
\frac{\Delta_2}{(\gammaab-1)^3} &=
\frac{\gammaac-1}{\gammaab-1} \frac{\gammabc-1}{\gammaab-1}
:= \Delta_2^\prime
\,.
\end{split}
\end{equation}
%%%%%%%%%%%%%%%%%%%%%%%%%%%%%%%%%%%%%%%%%%%%%%%%%%%%%%%%%%%%%%%%%%%%

\phantom{O}

Hence, by \eqref{matok01}\,--\,\eqref{matok03},
%%%%%%%%%%%%%%%%%%%%%%%%%%%%%%%%%%%%%%%%%%%%%%%%%%%%%%%%%%%%%%%%%%%%
\begin{equation} \label{matok04}
\begin{split}
m_1^\pprime &= F_0^\prime F_1^\prime  \frac{\gammabc-1}{\gammaab-1}
\\[8pt]
m_2^\pprime &= F_1^\prime F_2^\prime
\\[8pt]
m_3^\pprime &=-F_0^\prime F_2^\prime
\,.
\end{split}
\end{equation}
%%%%%%%%%%%%%%%%%%%%%%%%%%%%%%%%%%%%%%%%%%%%%%%%%%%%%%%%%%%%%%%%%%%%

The gyrobarycentric coordinates $m_k^\pprime$, $k=1,2,3$,
in \eqref{matok01}\,--\,\eqref{matok04}
appear in a form that admits a nontrivial Euclidean limit, $s\rightarrow\infty$,
by means of Lemma \ref{lemwindk}.
Indeed, one can readily show that the Euclidean limit of the
gyrobarycentric coordinate system $m_k^\pprime$ in \eqref{matok04}
turns out to be the
barycentric coordinate system
\ref{hugd04thmeuc}\,--\,\eqref{hugd03thmeuc},
in which we rename $m_k^\pprime$ as $m_k^\prime$ $(k=1,2,3)$,
and $F_k^\prime$ as $F_k$ $(k=0,1,2)$.

Finally, it remains to prove \eqref{rkdn1}.
The length $\|-P+P_{^{\pm}}\|$ of the tangent segment
$PP_{^{\pm}}$, shown in Fig.~\ref{fig336eucm},
is the Euclidean limit, $s\rightarrow\infty$,
of the gyrolength $\|\om P \op P_{^{\pm}}\|$ of the gyrotangent gyrosegment
$PP_{^{\pm}}$, shown in Fig.~\ref{fig334enm},
\begin{equation} \label{mkdem1}
\lim_{s\rightarrow \infty} \|\om P \op P_{^{\pm}}\| = \|-P+P_{^{\pm}}\|
\,.
\end{equation}

Indeed, by \eqref{mkdem1}, \eqref{harfdkb} and
Lemma \ref{lemwindk}, p.~\pageref{lemwindk},
we have
\begin{equation} \label{mkdem3}
\begin{split}
\| -P+P_{\pm}\| &= \lim_{s\rightarrow\infty} \|\om P \op P_{\pm}\|
\\[8pt] & \hspace{-1.2cm}=
\lim_{s\rightarrow\infty} \frac{ \sqrt{
-2s^2\{m_1m_2(\gammaab-1) + m_1m_3(\gammaac-1) + m_2m_3(\gammabc-1)\}
}}{
m_1+m_2+m_3
}
\\[8pt] &=
\frac{ \sqrt{
-\{m_1m_2 a_{12}^2 + m_1m_3 a_{13}^2 + m_2m_3 a_{23}^2 \}
}}{
m_1+m_2+m_3
}
\,,
\end{split}
\end{equation}
where $a_{ij}=\|-A_i+A_j\|$,
as desired.
\end{proof}

% SECTION NUMBER 22
\section{Circumgyrocevians} \label{slila4}
\index{circumgyrocevian}

%%%%%%%%%%%%%%%%%%%%%%%%%%%%%%%%%%%%%%%%%%%%%%%%%%%%%%%%%%%%%%%%%%%%
% FIGURE 22
 
%%%%%%%%%%%%%%%%%%%%%%%%%%%%%%%%%%%%%%%%%%%%%%%%%%%%%%%%%%%%%%%%%%%%%%
%%%%% The hyperbolic semi-circle theorem            %%%%%%%%%%%%%%%%%%
%\begin{figure}[htbp]
\begin{figure}[t]  % try to put this figure on the top of the page
 \centering         % center the figure
 \psfrag{O}{$O$}
 \psfrag{P1}{$P_1$}
 \psfrag{P2}{$P_2$}
 \psfrag{P3}{$P_3$}
 \psfrag{P4}{$P_4$}
 \psfrag{P5}{$P_5$}
 \psfrag{P6}{$P_6$}
 \psfrag{P7}{$P_7$}
 \psfrag{P8}{\lower 0.60ex \hbox {$\hspace{-0.08cm}P_8$}}
 \psfrag{Q1}{$Q_1$}
 \psfrag{Q2}{$Q_2$}
 \psfrag{Q3}{$Q_3$}
 \psfrag{Q4}{$Q_4$}
 \psfrag{Q5}{$Q_5$}
 \psfrag{Q6}{$Q_6$}
 \psfrag{Q7}{$Q_7$}
 \psfrag{Q8}{$\hspace{-0.2cm}Q_8$}
 \psfrag{A1}{$A_1$}
 \psfrag{A2}{$A_2$}
 \psfrag{A3}{$A_3$}
 \psfrag{A4}{$A_4$}
%\psfrag{d1}{$d_1$}
%\psfrag{d2}{$d_2$}
%\psfrag{d3}{$d_3$}
%\psfrag{d23}{$d_{23}$}
%\psfrag{rr}{$R$}
%\psfrag{formula01}[]{$d_1=\|\om A_1\op A\|$}
%\psfrag{formula02}[]{$d_2=\|\om A_2\op A\|$}
%\psfrag{formula03}[]{$d_3=\|\om A_3\op A\|$}
%\psfrag{formula04}[]{$d_{23}=\|\om A_2\op A_3\|=d_3\om d_2$}
%\psfrag{formula05}[]{$R=\|\om A_k\op O\|,~k=1,2,3$}
%
%\includegraphics[width=9cm]{/home/ungar/dir_amy/dir_papers/dir_mybook01/dir_figs/fig317en.eps}
 \includegraphics[width=9cm]{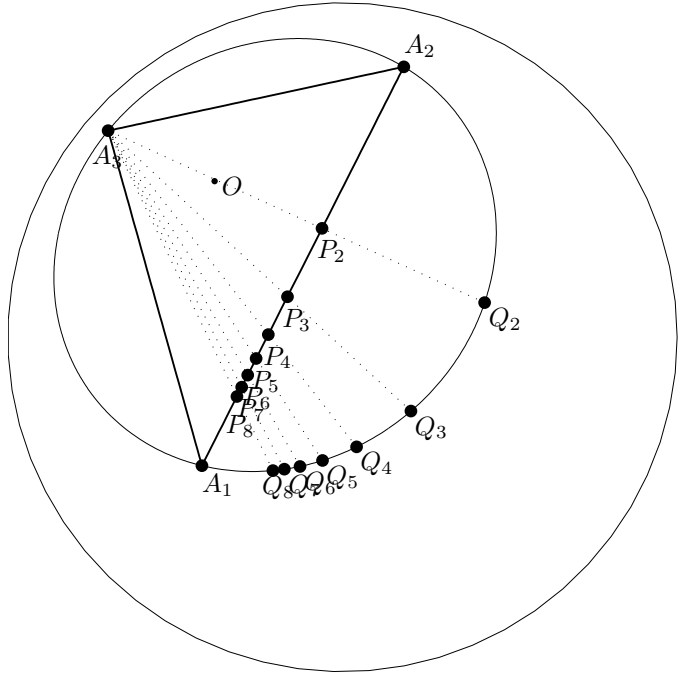}
\caption[circumgyrocevians]{
Circumgyrocevians.
A circumgyrocevian, $A_3Q_k$, $k=2,3,\ldots,8$, is a gyrocevian $A_3P_k$ of a
gyrotriangle $A_1A_2A_3$ in an Einstein gyrovector space,
extended to the point $Q_k$ where it meets the gyrotriangle circumgyrocircle.
The $k$th circumgyrocevian $A_3Q_k$ in this figure corresponds
to the parameter value
$t_1=1/k$ of the parameter $0<t_1<1$ in \eqref{yaglm01}.
The point $A_1$ corresponds to $t_1=0$, and
the point $A_2$ corresponds to $t_1=1$.
\label{fig317enm}}
\end{figure}
%%%%%%%%%%%%%%%%%%%%%%%%%%%%%%%%%%%%%%%%%%%%%%%%%%%%%%%%%%%%%%%%%%%%%%
 % A circumgyrocevian in (+)E
%                       Fig.~\ref{fig317enm}      Fig. 10.11
%%%%%%%%%%%%%%%%%%%%%%%%%%%%%%%%%%%%%%%%%%%%%%%%%%%%%%%%%%%%%%%%%%%%

%%%%%%%%%%%%%%%%%%%%%%%%%%%%%%%%%%%%%%%%%%%%%%%%%%%%%%%%%%%%%%%%%%%%
% DEFINITON NUMBER 9.18
\begin{definition}\label{circevian}
{\bf (Circumgyrocevians).}
\index{circumgyrocevian, def.}
{\it
A circumgyrocevian of a gyrotriangle is a gyrocevian of the
gyrotriangle extended to the point where it meets the gyrotriangle circumgyrocircle,
as shown in Fig.~\ref{fig317enm}.
}
\end{definition}
%%%%%%%%%%%%%%%%%%%%%%%%%%%%%%%%%%%%%%%%%%%%%%%%%%%%%%%%%%%%%%%%%%%%

The concept of the circumgyrocevian, together with its associated
Theorem \ref{gcircev}, p.~\pageref{gcircev}, is employed in Sect.~\ref{slila5}
and, consequently, proves useful in the study of the
Intersecting Gyrochords Theorem \ref{thmchord}, p.~\pageref{thmchord}.

In this section we determine the circumgyrocevian associated with
a given gyrocevian in an Einstein gyrovector space $(\Rsn,\op,\od)$.

Let $P$ be a generic point of a gyrosegment $A_1A_2$ in an
Einstein gyrovector space $(\Rsn,\op,\od)$, as shown in Fig.~\ref{fig318enm}.
Then, $P$ possesses the gyrobarycentric representation
\begin{equation} \label{yaglm01}
P =\frac{
(1-t_1)\gAa A_1 + t_1\gAb A_2
}{
(1-t_1)\gAa + t_1\gAb
}
\end{equation}
for any value of the parameter $t_1$, $0 \le t_1 \le 1$.
The parameter $t_1$ determines the location of the point $P$ on gyrosegment $A_1A_2$,
as shown in Fig.~\ref{fig317enm}.
Thus, in particular, $t_1=0$ gives $P=A_1$, $t_1=1$ gives $P=A_2$, and
$t=1/2$ gives the gyromidpoint $M_{^{A_1A_2}}$
of $A_1$ and $A_2$.

Applying the Gyrobarycentric Representation Gyrocovariance Theorem
\cite[Theorem 4.6, pp.~90-91]{mybook05}
with $X=\om A_1$, we have
%%%%%%%%%%%%%%%%%%%%%%%%%%%%%%%%%%%%%%%%%%%%%%%%%%%%%%%%%%%%%%%%%%%%
\begin{equation} \label{yaglm02}
\begin{split}
\om A_1 \op P &= \om A_1 \op \frac{
(1-t_1)\gAa A_1 + t_1\gAb A_2
}{
(1-t_1)\gAa + t_1\gAb
}
\\[8pt] & =
\frac{
(1-t_1)\gamma_{\om A_1\op A_1}^{\phantom{O}} (\om A_1\op A_1)
+
t_1\gamma_{\om A_1\op A_2}^{\phantom{O}} (\om A_1\op A_2)
}{
(1-t_1)\gamma_{\om A_1\op A_1}^{\phantom{O}}
+
t_1\gamma_{\om A_1\op A_2}^{\phantom{O}}
}
\\[8pt] & = \frac{
t_1\gammaab\ab_{12}
}{
1+(\gammaab-1)t_1
}
\,,
\end{split}
\end{equation}
%%%%%%%%%%%%%%%%%%%%%%%%%%%%%%%%%%%%%%%%%%%%%%%%%%%%%%%%%%%%%%%%%%%%
where we use the gyrotriangle index notation \eqref{indexnotation}.

Let $A_3P$ be the gyroray that emanates from the point $A_3$ and
passes through the point $P$ in the Einstein gyrovector space,
as shown in Fig.~\ref{fig318enm}.
Gyrolines and gyrorays in Einstein gyrovector spaces coincide with
Euclidean segments, enabling methods of linear algebra to be
applied for solving intersection problems for gyrolines and gyrorays.

Let us consider the gyroray
\begin{equation} \label{yaglm02d6}
(\om A_1\op A_3)(\om A_1\op P) = \ab_{13}(\om A_1\op P)
\end{equation}
that emanates from $(\om A_1\op A_3) = \ab_{13}$ and passes through the
point $\om A_1\op P$. This ray is the left gyrotranslation by $\om A_1$
of the gyroray $A_3P$.
By methods of linear algebra, the points $A(t)$ of the gyroray
\eqref{yaglm02d6}, parametrized by $t$, $t\ge0$,
are given by
\begin{equation} \label{yaglm03}
A(t) = \ab_{13}(1-t) + (\om A_1 \op P) t
\,.
\end{equation}

Equivalently, by means of \eqref{yaglm02}, \eqref{yaglm03} can be
written as
%%%%%%%%%%%%%%%%%%%%%%%%%%%%%%%%%%%%%%%%%%%%%%%%%%%%%%%%%%%%%%%%%%%%
\begin{equation} \label{yaglm04}
\begin{split}
A(t) &= \ab_{13}(1-t) + \frac{
t_1\gammaab\ab_{12}
}{
1+(\gammaab-1)t_1
} t
\\[8pt] &=
m_2\gammaab\ab_{12} + m_3\gammaac\ab_{13}
\,,
\end{split}
\end{equation}
%%%%%%%%%%%%%%%%%%%%%%%%%%%%%%%%%%%%%%%%%%%%%%%%%%%%%%%%%%%%%%%%%%%%
where
%%%%%%%%%%%%%%%%%%%%%%%%%%%%%%%%%%%%%%%%%%%%%%%%%%%%%%%%%%%%%%%%%%%%
\begin{equation} \label{yaglm05}
\begin{split}
m_2 &= \frac{t_1t}{1+(\gammaab-1)t_1}
\\[8pt]
m_3 &= \frac{1-t}{\gammaac}
\,.
\end{split}
\end{equation}
%%%%%%%%%%%%%%%%%%%%%%%%%%%%%%%%%%%%%%%%%%%%%%%%%%%%%%%%%%%%%%%%%%%%

The set of points $A(t)$, $0\le t<\infty$, is the set of all points
of the left gyrotranslated gyroray \eqref{yaglm02d6} by $\om A_1$.
Hence, the set of points
$A_1\op A(t)$, $0\le t<\infty$, is the set of all points
of the original gyroray
\begin{equation} \label{yaglm5d5}
A_3P = \{ A_1\op A(t):~0\le t<\infty\}
\,,
\end{equation}
which is of interest as indicated in
Figs.~\ref{fig317enm}\,--\,\ref{fig318eum}.

Noting that $\om A_1\op A_1=\zerb$ and $\gamma_{\zerb}^{\ph}=1$,
\eqref{yaglm04} can be written as
\begin{equation} \label{yaglm06}
A(t)=\frac{
m_1\gamma_{\om A_1\op A_1}^{\phantom{O}}(\om A_1\op A_1)
+
m_2\gamma_{\om A_1\op A_2}^{\phantom{O}}(\om A_1\op A_2)
+
m_3\gamma_{\om A_1\op A_3}^{\phantom{O}}(\om A_1\op A_3)
}{
m_1\gamma_{\om A_1\op A_1}^{\phantom{O}}
+
m_2\gamma_{\om A_1\op A_2}^{\phantom{O}}
+
m_3\gamma_{\om A_1\op A_3}^{\phantom{O}}
}
\,,
\end{equation}
where $m_1$ is determined by the condition
\begin{equation} \label{yaglm07}
m_1+m_2\gammaab+m_3\gammaac = 1
\end{equation}
that insures that \eqref{yaglm06} and \eqref{yaglm04} are identically
equal.

Solving \eqref{yaglm07} for $m_1$, where $m_2$ and $m_3$ are given
by \eqref{yaglm05}, we have
\begin{equation} \label{yaglm08}
m_1 = \frac{(1-t_1)t}{1+(\gammaab-1)t_1}
\,.
\end{equation}

By means of the Gyrobarycentric Representation Gyrocovariance Theorem
\cite[Theorem 4.6, pp.~90-91]{mybook05}
with $X=\om A_1$, \eqref{yaglm06} can be written as
\begin{equation} \label{yaglm09}
A(t) = \om A_1 \op \frac{
m_1\gAa A_1 + m_2\gAb A_2 + m_3\gAc A_3
}{
m_1\gAa + m_2\gAb + m_3\gAc
}
\,,
\end{equation}
where $m_1,m_2$ and $m_3$ are given by \eqref{yaglm08} and \eqref{yaglm05}.

Following \eqref{yaglm09} we have, by a left cancellation,
\begin{equation} \label{yaglm10}
Q(t):= A_1\op A(t) = \frac{
m_1\gAa A_1 + m_2\gAb A_2 + m_3\gAc A_3
}{
m_1\gAa + m_2\gAb + m_3\gAc
}
\,,
\end{equation}
where the set of points $Q(t)=A_1\op A(t)$, $0\le t<\infty$, forms the gyroray
$A_3P$, \eqref{yaglm5d5}, shown in Fig.~\ref{fig318enm}.

We have thus obtained in \eqref{yaglm10} the gyrobarycentric representation
with respect to $\{A_1,A_2,A_3\}$
of each point $Q(t)$ of the gyroray $A_3P$.
We now wish to determine the unique point of the gyroray $A_3P$, other than
$A_3$, that lies on the circumgyrocircle of gyrotriangle $A_1A_2A_3$
in Fig.~\ref{fig318enm}.

\index{gyrocircle gyrobarycentric representation theorem}
By the Gyrocircle Gyrobarycentric Representation Theorem \ref{thmsnvur},
p.~\pageref{thmsnvur},
a point $Q(t)$ of the gyroray $A_3P$ in \eqref{yaglm10} lies on the
circumgyrocircle of a gyrotriangle $A_1A_2A_3$ if and only if the
gyrobarycentric coordinates $m_k$, $k=1,2,3$, of $Q(t)$ in \eqref{yaglm10}
satisfy the circumgyrocircle condition \eqref{muramd05}.

Accordingly, we substitute $m_1.m_2.m_3$ from \eqref{yaglm08} and \eqref{yaglm05}
into \eqref{muramd05}, obtaining a linear equation for the unknown $t$, the
solution $t=t_0$ of which is substituted into \eqref{yaglm08} and \eqref{yaglm05},
obtaining gyrobarycentric coordinates $m_1.m_2.m_3$ for the point $Q=Q(t_0)$ that
lies on the circumgyrocircle of gyrotriangle $A_1A_2A_3$,
shown in
Figs.~\ref{fig317enm}\,--\,\ref{fig318eum}.
Finally, by removing
an irrelevant nonzero common factor of the gyrobarycentric coordinates, we obtain
the gyrobarycentric coordinates of $Q$,
%%%%%%%%%%%%%%%%%%%%%%%%%%%%%%%%%%%%%%%%%%%%%%%%%%%%%%%%%%%%%%%%%%%%
\begin{equation} \label{muramd11}
\begin{split}
m_1 &= \{\gammaac-1+(\gammabc-\gammaac)t_1\}(1-t_1)
\\
m_2 &= \{\gammaac-1+(\gammabc-\gammaac)t_1\}t_1
\\
m_3 &=-(\gammaab-1)(1-t_1)t_1
\,,
\end{split}
\end{equation}
%MATHEMATICA stam284
%%%%%%%%%%%%%%%%%%%%%%%%%%%%%%%%%%%%%%%%%%%%%%%%%%%%%%%%%%%%%%%%%%%%
where $t_1$, $0<t_1<1$, is a parameter that determines by \eqref{yaglm01} the
location of the point $P$ on the gyrochord $A_1A_2$ of
the circumgyrocircle of gyrotriangle $A_1A_2A_3$ in Fig.~\ref{fig318enm}.

The point $Q$, given by \eqref{yaglm10} with gyrobarycentric coordinates
given by \eqref{muramd11} is thus the only point of gyroray $A_3P$,
other than $A_3$, that lies on
the circumgyrocircle of gyrotriangle $A_1A_2A_3$ in Fig.~\ref{fig318enm}.
Formalizing, we thus obtain the following theorem:

%%%%%%%%%%%%%%%%%%%%%%%%%%%%%%%%%%%%%%%%%%%%%%%%%%%%%%%%%%%%%%%%%%%%
% FIGURE 23
 
%%%%%%%%%%%%%%%%%%%%%%%%%%%%%%%%%%%%%%%%%%%%%%%%%%%%%%%%%%%%%%%%%%%%%%
%%%%% The hyperbolic semi-circle theorem            %%%%%%%%%%%%%%%%%%
%\begin{figure}[htbp]
\begin{figure}[t]  % try to put this figure on the top of the page
 \centering         % center the figure
 \psfrag{O}{$O$}
 \psfrag{P}{$P$}
 \psfrag{Q}{$Q=A_4=B_2$}
 \psfrag{A1}{$A_1$}
 \psfrag{A2}{$A_2$}
 \psfrag{A3}{$\hspace{0.2cm}A_3=B_1$}
 \psfrag{A4}{$A_4$}
%\psfrag{d1}{$d_1$}
%\psfrag{d2}{$d_2$}
%\psfrag{d3}{$d_3$}
%\psfrag{d23}{$d_{23}$}
%\psfrag{rr}{$R$}
%\psfrag{formula01}[]{$d_1=\|\om A_1\op A\|$}
%\psfrag{formula02}[]{$d_2=\|\om A_2\op A\|$}
%\psfrag{formula03}[]{$d_3=\|\om A_3\op A\|$}
%\psfrag{formula04}[]{$d_{23}=\|\om A_2\op A_3\|=d_3\om d_2$}
%\psfrag{formula05}[]{$R=\|\om A_k\op O\|,~k=1,2,3$}
%
%\includegraphics[width=9cm]{/home/ungar/dir_amy/dir_papers/dir_mybook01/dir_figs/fig318en.eps}
 \includegraphics[width=9cm]{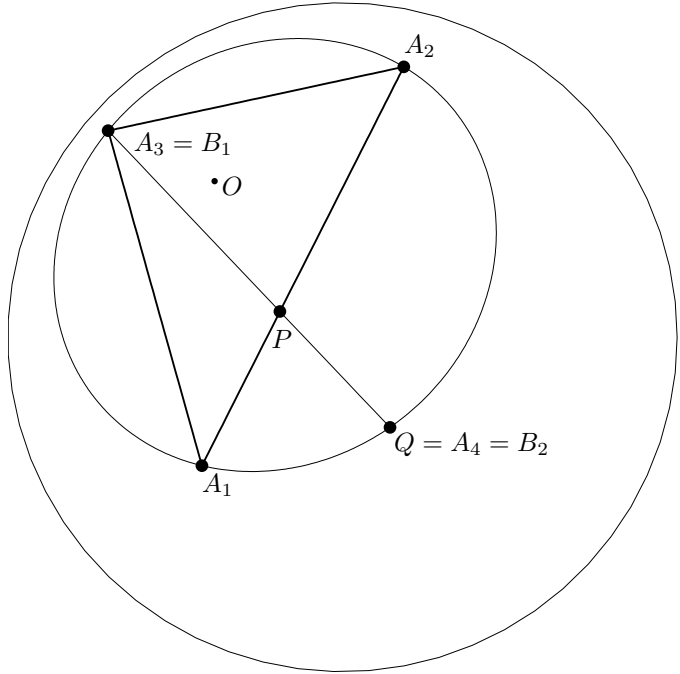}
\caption[Illustration of the Circumgyrocevian Theorem]{
Illustrating the Circumgyrocevian Theorem \ref{gcircev}, p.~\pageref{gcircev},
and
the Intersecting Gyrochords Theorem \ref{thmchord}, p.~\pageref{thmchord}.
A circumgyrocevian, $A_3Q$, is shown. It is a gyrocevian $A_3P$ of a
gyrotriangle $A_1A_2A_3$ in an Einstein gyrovector space $(\Rsn,\op,\od)$, $n\ge2$,
extended to the point $Q$ where it meets the
gyrotriangle circumgyrocircle.
With $A_4=Q$, this figure presents two gyrochords, $A_1A_2$ and
$A_3A_4=B_1B_2$, of
a gyrocircle, intersecting at a point $P$ inside the gyrocircle.
Here, the point $Q\in\Rsn$ is determined by its
gyrobarycentric representation \eqref{yagln03} in Theorem \ref{gcircev}.
This figure is the hyperbolic counterpart of Fig.~\ref{fig318eum}.
\label{fig318enm}}
\end{figure}
%%%%%%%%%%%%%%%%%%%%%%%%%%%%%%%%%%%%%%%%%%%%%%%%%%%%%%%%%%%%%%%%%%%%%%
 % A circumgyrocevian in (+)E
%                       Fig.~\ref{fig318enm}      Fig. 10.12
%%%%%%%%%%%%%%%%%%%%%%%%%%%%%%%%%%%%%%%%%%%%%%%%%%%%%%%%%%%%%%%%%%%%

%%%%%%%%%%%%%%%%%%%%%%%%%%%%%%%%%%%%%%%%%%%%%%%%%%%%%%%%%%%%%%%%%%%%
% THEOREM NUMBER 9.19
\index{circumgyrocevian theorem}
\begin{theorem}
{\bf (The Circumgyrocevian Theorem).}
\label{gcircev}
Let $A_3P$ be a gyrocevian of a gyrotriangle $A_1A_2A_3$
in an Einstein gyrovector space $(\Rsn,\op,\od)$, $n\ge2$,
and let $A_3Q$ be its corresponding circumgyrocevian,
as shown in Fig.~\ref{fig318enm}.
Furthermore, let
\begin{equation} \label{yagln01}
P =\frac{
(1-t_1)\gAa A_1 + t_1\gAb A_2
}{
(1-t_1)\gAa + t_1\gAb
}
\end{equation}
be the gyrobarycentric representation of $P$
with respect to the gyrobarycentrically independent set $\{A_1,A_2\}$
where the location of $P$ on the circumgyrocircle gyrochord $A_1A_2$ is determined
by the parameter $t_1$,
\begin{equation} \label{yagln02}
0 \le t_1 \le 1
\,.
\end{equation}
Then, the point $Q$ possesses the gyrobarycentric representation
\begin{equation} \label{yagln03}
Q = \frac{
m_1 \gAa A_1 + m_2 \gAb A_2 + m_3 \gAc A_3
}{
m_1 \gAa + m_2 \gAb + m_3 \gAc
}
\end{equation}
with respect to the gyrobarycentrically independent set $\{A_1,A_2,A_3\}$, where
the gyrobarycentric coordinates $m_1,m_2,m_3$ are given by
%%%%%%%%%%%%%%%%%%%%%%%%%%%%%%%%%%%%%%%%%%%%%%%%%%%%%%%%%%%%%%%%%%%%
\begin{equation} \label{yagln04}
\begin{split}
m_1 &= \{\gammaac-1+(\gammabc-\gammaac)t_1\}(1-t_1)
\\
m_2 &= \{\gammaac-1+(\gammabc-\gammaac)t_1\}t_1
\\
m_3 &=-(\gammaab-1)(1-t_1)t_1
\,.
\end{split}
\end{equation}
%MATHEMATICA stam284
%%%%%%%%%%%%%%%%%%%%%%%%%%%%%%%%%%%%%%%%%%%%%%%%%%%%%%%%%%%%%%%%%%%%
\end{theorem}

% EXAMPLE NUMBER 9.20
\begin{example}\label{ekhkd1}
For $t_1=0$ the gyrobarycentric coordinates in \eqref{yagln04} specialize
to $m_1=\gammaac-1>0$ and $m_2=m_3=0$ implying, by \eqref{yagln03},
$Q=A_1$.
\end{example}

% EXAMPLE NUMBER 9.21
\begin{example}\label{ekhkd2}
For $t_1=1$ the gyrobarycentric coordinates in \eqref{yagln04} specialize
to $m_2=\gammabc-1>0$ and $m_1=m_3=0$ implying, by \eqref{yagln03},
$Q=A_2$.
\end{example}

Exploring the Euclidean limit of
the Circumgyrocevian Theorem \ref{gcircev}, we obtain the
following Circumcevian Theorem:

%%%%%%%%%%%%%%%%%%%%%%%%%%%%%%%%%%%%%%%%%%%%%%%%%%%%%%%%%%%
% FIGURE 24
  
%%%%%%%%%%%%%%%%%%%%%%%%%%%%%%%%%%%%%%%%%%%%%%%%%%%%%%%%%%%%%%%%%%%%%%
%%%%% The hyperbolic semi-circle theorem            %%%%%%%%%%%%%%%%%%
%\begin{figure}[htbp]
\begin{figure}[t]  % try to put this figure on the top of the page
 \centering         % center the figure
\psfrag{A1}{$A_1$}
\psfrag{A2}{$A_2$}
\psfrag{B1}{$A_3$}
\psfrag{B2}{$Q$}
\psfrag{P}{$P$}
 \includegraphics[width=9cm]{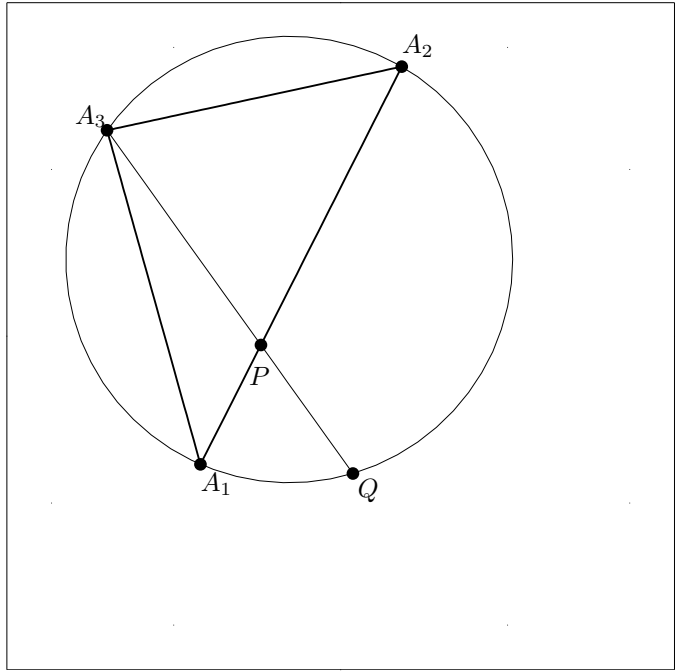}
\caption[Illustration of the Circumcevian Theorem]{
Illustrating the Circumcevian Theorem \ref{gcirceveuc}, p.~\pageref{gcirceveuc},
and the Intersecting Chords Theorem \ref{thmchordeuc}, p.~\pageref{thmchordeuc}.
A circular cevian, $A_3Q$, is shown. It is a cevian $A_3P$ of a
triangle $A_1A_2A_3$ in a Euclidean space $\Rn$, $n\ge2$,
extended to the point $Q$ where it meets the
triangle circumcircle.
This figure presents two chords, $A_1A_2$ and
$A_3Q$, of
a circle, intersecting at a point $P$ inside the circle.
Here, the point $Q\in\Rn$ is determined by its
barycentric representation \eqref{yagln03euc} in Theorem \ref{gcirceveuc}.
This figure is the Euclidean counterpart of Fig.~\ref{fig318enm}.
\label{fig318eum}}
\end{figure}
%%%%%%%%%%%%%%%%%%%%%%%%%%%%%%%%%%%%%%%%%%%%%%%%%%%%%%%%%%%%%%%%%%%%%%
                 % Fig. 10.13
% Point of concurrency of two gyrocevians in (+)E
%          Ceva,    Fig.~\ref{fig318eum}
%%%%%%%%%%%%%%%%%%%%%%%%%%%%%%%%%%%%%%%%%%%%%%%%%%%%%%%%%%%

%%%%%%%%%%%%%%%%%%%%%%%%%%%%%%%%%%%%%%%%%%%%%%%%%%%%%%%%%%%%%%%%%%%%
% THEOREM NUMBER 9.22
\index{circumcevian theorem}
\begin{theorem}
{\bf (The Circumcevian Theorem).}
\label{gcirceveuc}
Let $A_3P$ be a Cevian of a triangle $A_1A_2A_3$
in a Euclidean space $\Rn$, $n\ge2$,
and let $A_3Q$ be its corresponding circumcevian,
as shown in Fig.~\ref{fig318eum}.
% Fig.~\ref{fig318eum} is the Euclidean counterpart of Fig.~\ref{fig318enm}.
Furthermore, let
\begin{equation} \label{yagln01euc}
P = (1-t_1) A_1 + t_1 A_2
\end{equation}
be the barycentric representation of $P$
with respect to the barycentrically independent set $\{A_1,A_2\}$,
so that the location of $P$ on the circumcircle chord $A_1A_2$ is determined
by the parameter $t_1$,
\begin{equation} \label{yagln02euc}
0 \le t_1 \le 1
\,.
\end{equation}
Then, the point $Q$ possesses the barycentric representation
\begin{equation} \label{yagln03euc}
Q = \frac{
m_1 A_1 + m_2 A_2 + m_3 A_3
}{
m_1 + m_2 + m_3
}
\end{equation}
with respect to the barycentrically independent set $\{A_1,A_2,A_3\}$, where
the barycentric coordinates $m_1,m_2,m_3$ are given by
%%%%%%%%%%%%%%%%%%%%%%%%%%%%%%%%%%%%%%%%%%%%%%%%%%%%%%%%%%%%%%%%%%%%
\begin{equation} \label{yagln04uec}
\begin{split}
m_1 &= \{\sin^2\alpha_2 + t_1\sin(\alpha_1-\alpha_2)\sin(\alpha_3\}(1-t_1)
\\
m_2 &= \{\sin^2\alpha_2 + t_1\sin(\alpha_1-\alpha_2)\sin(\alpha_3\}t_1
\\
m_3 &=-(1-t_1)t_1 \sin^2\alpha_3
\,,
\end{split}
\end{equation}
%MATHEMATICA stam284
where $\alpha_k$, $k=1,2,3$, are the corresponding angles of
triangle $A_1A_2A_3$.
%%%%%%%%%%%%%%%%%%%%%%%%%%%%%%%%%%%%%%%%%%%%%%%%%%%%%%%%%%%%%%%%%%%%
\end{theorem}
\begin{proof}
In the Euclidean limit, $s\rightarrow\infty$, the
gyrobarycentric representation of $P\in\Rsn$ in \eqref{yagln01}\,--\,\eqref{yagln02}
reduces to the corresponding barycentric representations of $P\in\Rn$
in \eqref{yagln01euc}\,--\,\eqref{yagln02euc}.

Obviously, in that Euclidean limit, the
gyrobarycentric representation of $Q\in\Rsn$ in \eqref{yagln03} reduces to its
corresponding barycentric representation of $Q\in\Rn$ in \eqref{yagln03euc}.

Hence, it remains to show that the gyrobarycentric coordinates
$m_k$, $k=1,2,3$, of $Q\in\Rsn$ in \eqref{yagln04} tend to their
corresponding barycentric coordinates $m_k$, $k=1,2,3$, of $Q\in\Rn$
in \eqref{yagln04uec}.

In the Euclidean limit, $s\rightarrow\infty$, gamma factors tend to 1.
Hence,
straightforward Euclidean limits of $m_k$, $k=1,2,3$, in \eqref{yagln04},
as $s\rightarrow\infty$, give vanishing barycentric coordinates,
$m_k=0$, $k=1,2,3$, resulting in an
indeterminate barycentric representation of type $0/0$.

However, owing to the homogeneity of gyrobarycentric coordinates,
the gyrobarycentric coordinates $m_k=0$, $k=1,2,3$, of $Q$
in \eqref{yagln03}\,--\,\eqref{yagln04}
can be written in the following form, which admits the
Euclidean limit as $s\rightarrow\infty$,
%%%%%%%%%%%%%%%%%%%%%%%%%%%%%%%%%%%%%%%%%%%%%%%%%%%%%%%%%%%%%%%%%%%%
\begin{equation} \label{yagln04s}
\begin{split}
m_1 &= \left\{ \frac{\gammaac-1}{\gammaab-1}
+ \frac{\gammabc-\gammaac}{\gammaab-1}t_1
\right\} (1-t_1)
\\
m_2 &= \left\{ \frac{\gammaac-1}{\gammaab-1}
+ \frac{\gammabc-\gammaac}{\gammaab-1}t_1
\right\} t_1
\\
m_3 &=-(1-t_1)t_1
\,.
\end{split}
\end{equation}
%MATHEMATICA stam284
%%%%%%%%%%%%%%%%%%%%%%%%%%%%%%%%%%%%%%%%%%%%%%%%%%%%%%%%%%%%%%%%%%%%

Indeed, following Euclidean limits that can be derived from
Lemma \ref{lemwindk},
we have the Euclidean limits
%%%%%%%%%%%%%%%%%%%%%%%%%%%%%%%%%%%%%%%%%%%%%%%%%%%%%%%%%%%%%%%%%%%%
\begin{equation} \label{trkdn1}
\begin{split}
\lim_{s\rightarrow\infty}
\frac{\gammaac-1}{\gammaab-1}
&=
\frac{\sin^2\alpha_2}{\sin^2\alpha_3}
\\[8pt]
\lim_{s\rightarrow\infty}
\frac{\gammabc-\gammaac}{\gammaab-1}
&=
\frac{ \sin(\alpha_1-\alpha_2)}{\sin\alpha_3}
\,.
\end{split}
\end{equation}
%%%%%%%%%%%%%%%%%%%%%%%%%%%%%%%%%%%%%%%%%%%%%%%%%%%%%%%%%%%%%%%%%%%%

Hence, in the Euclidean limit, $s\rightarrow\infty$, the
gyrobarycentric coordinates $m_k$, $k=1,2,3$, of $Q\in\Rsn$
in \eqref{yagln04s} tend to the following barycentric coordinates
of $Q\in\Rn$,
%%%%%%%%%%%%%%%%%%%%%%%%%%%%%%%%%%%%%%%%%%%%%%%%%%%%%%%%%%%%%%%%%%%%
\begin{equation} \label{yagln04s1}
\begin{split}
m_1 &= \left\{ \frac{\sin^2\alpha_2}{\sin^2\alpha_3} +
 \frac{\sin(\alpha_1-\alpha_2)}{\sin\alpha_3} t_1
\right\} (1-t_1)
\\[8pt]
m_2 &= \left\{ \frac{\sin^2\alpha_2}{\sin^2\alpha_3} +
 \frac{\sin(\alpha_1-\alpha_2)}{\sin\alpha_3} t_1
\right\} t_1
\\[8pt]
m_3 &=-(1-t_1)t_1
\,.
\end{split}
\end{equation}
%%%%%%%%%%%%%%%%%%%%%%%%%%%%%%%%%%%%%%%%%%%%%%%%%%%%%%%%%%%%%%%%%%%%

Owing to their homogeneity, the barycentric coordinates \eqref{yagln04s1}
of $Q\in\Rn$ can be written as \eqref{yagln04uec},
as desired.
\end{proof}

% SECTION NUMBER 23
\section{Gyrodistances Related to the Gyrocevian} \label{slila5}

In order to prepare the stage for the
Intersecting Gyrochords Theorem \ref{thmchord}, p.~\pageref{thmchord},
in this section we calculate the gyrodistance
between $A_k$ and the gyrocevian foot $P$, $k=1,2,3$,
of a gyrotriangle $A_1A_2A_3$, shown in Fig.~\ref{fig318enm}.

Let $A_1A_2A_3$ be a gyrotriangle that possesses a circumgyrocircle
in an Einstein gyrovector space $(\Rsn,\op,\od)$,
and let $P$, given by \eqref{yagln01}, be a generic point of the interior
of chord $A_1A_2$ of the circumgyrocircle, as shown in Fig.~\ref{fig318enm}.

Applying the Gyrobarycentric Representation Gyrocovariance Theorem
\cite[Theorem 4.6, pp.~90-91]{mybook05}
to the gyrobarycentric representation \eqref{yagln01} of $P$,
using the gyrotriangle index notation \eqref{indexnotation}, p.~\pageref{indexnotation},
we have
%%%%%%%%%%%%%%%%%%%%%%%%%%%%%%%%%%%%%%%%%%%%%%%%%%%%%%%%%%%%%%%%%%%%
\begin{equation} \label{halof01}
\begin{split}
\om A_1 \op P &= \frac{t_1\gammaab\ab_{12}}{(1-t_1)+t_1\gammaab}
\\[4pt]
\om A_2 \op P &= \frac{(1-t_1)\gammaab\ab_{21}}{(1-t_1)\gammaab+t_1}
\\[4pt]
\om A_3 \op P &= \frac{
(1-t_1)\gammaac\ab_{31} + t_1\gammabc\ab_{23}
}{
(1-t_1)\gammaac + t_1\gammabc
}
\end{split}
\end{equation}
%%%%%%%%%%%%%%%%%%%%%%%%%%%%%%%%%%%%%%%%%%%%%%%%%%%%%%%%%%%%%%%%%%%%
and
%%%%%%%%%%%%%%%%%%%%%%%%%%%%%%%%%%%%%%%%%%%%%%%%%%%%%%%%%%%%%%%%%%%%
\begin{equation} \label{halof02}
\begin{split}
\gamma_{|A_1P|}^{\phantom{O}} &= \gamma_{\om A_1\op P}^{\phantom{O}}
=\frac{(1-t_1)+t_1\gammaab}{\mP}
\\[4pt] 
\gamma_{|A_2P|}^{\phantom{O}} &= \gamma_{\om A_2\op P}^{\phantom{O}}
=\frac{(1-t_1)\gammaab+t_1}{\mP}
\\[4pt]
\gamma_{|A_3P|}^{\phantom{O}} &= \gamma_{\om A_3\op P}^{\phantom{O}}
=\frac{(1-t_1)\gammaac+t_1\gammabc}{\mP}
\,,
\end{split}
\end{equation}
%%%%%%%%%%%%%%%%%%%%%%%%%%%%%%%%%%%%%%%%%%%%%%%%%%%%%%%%%%%%%%%%%%%%
where $\mP>0$ is the constant of the gyrobarycentric representation \eqref{yagln01}
of $P$ with respect to the set $\{A_1,A_2\}$, given by
\begin{equation} \label{halof03}
m_P^2 = (1-t_1)^2+t_1^2+2(1-t_1)t_1\gammaab
= 1+2(\gammaab-1)(1-t_1)t_1
\,.
\end{equation}

By Identity \eqref{rugh1ds}, p.~\pageref{rugh1ds}, and \eqref{halof02} and straightforward algebra,
we have
%%%%%%%%%%%%%%%%%%%%%%%%%%%%%%%%%%%%%%%%%%%%%%%%%%%%%%%%%%%%%%%%%%%%
\begin{equation} \label{halof04}
\begin{split}
\frac{1}{s^2} \gamma_{|A_1P|}^2 |A_1P|^2 &= \gamma_{|A_1P|}^2 - 1 = \frac{
(\gamma_{12}^2-1)t_1^2}{1+2(\gammaab-1)(1-t_1)t_1}
\\[8pt]
\frac{1}{s^2} \gamma_{|A_2P|}^2 |A_2P|^2 &= \gamma_{|A_2P|}^2 - 1 = \frac{
(\gamma_{12}^2-1)(1-t_1)^2}{1+2(\gammaab-1)(1-t_1)t_1}
\\[8pt]
\frac{1}{s^2} \gamma_{|A_3P|}^2 |A_3P|^2 &= \gamma_{|A_3P|}^2 - 1 = \frac{
\{\gamma_{13}(1-t_1) + \gammabc t_1\}^2}
{1+2(\gammaab-1)(1-t_1)t_1} - 1
\,.
\end{split}
\end{equation}
%MATHEMATICA stam285b
% MATLAB fig318enE.m  sgm1,2,3.
%%%%%%%%%%%%%%%%%%%%%%%%%%%%%%%%%%%%%%%%%%%%%%%%%%%%%%%%%%%%%%%%%%%%
The equations in \eqref{halof04} prove useful in the proof of
Theorem \ref{thmchord}, p.~\pageref{thmchord}.

Finally, by means of Identity \eqref{rugh1ds}, p.~\pageref{rugh1ds},
it follows from \eqref{halof04} that the squared gyrodistances
$|A_kP|^2=\|\om A_k\op P\|^2$ between $A_k$ and $P$,
$k=1,2,3$, are given by the equations
%%%%%%%%%%%%%%%%%%%%%%%%%%%%%%%%%%%%%%%%%%%%%%%%%%%%%%%%%%%%%%%%%%%%
\begin{equation} \label{halof05}
\begin{split}
\frac{1}{s^2} \|\om A_1\op P\|^2 &= \frac{
(\gamma_{12}^2-1)t_1^2
}{
\{(\gammaab-1)t_1+1\}^2
}
\\[8pt]
\frac{1}{s^2} \|\om A_2\op P\|^2 &= \frac{
(\gamma_{12}^2-1)(1-t_1)^2
}{
\{\gammaab(1-t_1)+t_1\}^2
}
\\[8pt]
\frac{1}{s^2} \|\om A_3\op P\|^2 &= 1 - \frac{
2(\gammaab-1)(1-t_1)t_1+1
}{
\{\gammaac(1-t_1)+\gammabc t_1\}^2
}
\,.
\end{split}
\end{equation}
%MATHEMATICA stam285b
% MATLAB fig318enE.m
%%%%%%%%%%%%%%%%%%%%%%%%%%%%%%%%%%%%%%%%%%%%%%%%%%%%%%%%%%%%%%%%%%%%

It is clear from the first and the second equations in \eqref{halof05} that
$|A_1P|=0$ when $t_1=0$ and $|A_2P|=0$ when $t_1=1$.
Hence, $P=A_1$ when $t_1=0$ and $P=A_2$ when $t_1=1$,
as indicated in Fig.~\ref{fig318enm}.

% SECTION NUMBER 24
\section{A Gyrodistance Related to the Circumgyrocevian} \label{slila6}

In order to prepare the stage for the
Intersecting Gyrochords Theorem \ref{thmchord}, p.~\pageref{thmchord},
in this section we calculate the gyrodistance $|PQ|=\|\om P\op Q\|$
between
the gyrocevian foot $P$ and the circumgyrocevian foot $Q$
of a gyrotriangle $A_1A_2A_3$,
shown in Fig.~\ref{fig318enm}, p.~\pageref{fig318enm}.

By Theorem \ref{gcircev}, the point $Q$, shown in Fig.~\ref{fig318enm},
is given by its  gyrobarycentric representation, \eqref{yagln03},
\begin{subequations} \label{yaglf01}
\begin{equation} \label{yaglf01a}
Q = \frac{
m_1 \gAa A_1 + m_2 \gAb A_2 + m_3 \gAc A_3
}{
m_1 \gAa + m_2 \gAb + m_3 \gAc
}
\end{equation}
with respect to the set $\{A_1,A_2,A_3\}$, where, \eqref{yagln04},
%%%%%%%%%%%%%%%%%%%%%%%%%%%%%%%%%%%%%%%%%%%%%%%%%%%%%%%%%%%%%%%%%%%%
\begin{equation} \label{yaglf01b}
\begin{split}
m_1 &= \{\gammaac-1+(\gammabc-\gammaac)t_1\}(1-t_1)
\\
m_2 &= \{\gammaac-1+(\gammabc-\gammaac)t_1\}t_1
\\
m_3 &=-(\gammaab-1)(1-t_1)t_1
\end{split}
\end{equation}
%%%%%%%%%%%%%%%%%%%%%%%%%%%%%%%%%%%%%%%%%%%%%%%%%%%%%%%%%%%%%%%%%%%%
\end{subequations}
and the point $P$ is given by its  gyrobarycentric representation,
\eqref{yagln01},
%%%%%%%%%%%%%%%%%%%%%%%%%%%%%%%%%%%%%%%%%%%%%%%%%%%%%%%%%%%%%%%%%%%%
\begin{subequations} \label{yaglf02}
\begin{equation} \label{yaglf02a}
P =\frac{
m_1^\prime \gAa A_1 + m_2^\prime \gAb A_2
}{
m_1^\prime \gAa + m_2^\prime \gAb
}
\,,
\end{equation}
where
\begin{equation} \label{yaglf02b}
\begin{split}
m_1^\prime &= 1-t_1
\\[4pt]
m_2^\prime &= t_1
\,,
\end{split}
\end{equation}
$0<t_1<1$.
\end{subequations}

Hence, by \cite[Sect.~4.9]{mybook05} with $N=3$,
%%%%%%%%%%%%%%%%%%%%%%%%%%%%%%%%%%%%%%%%%%%%%%%%%%%%%%%%%%%%%%%%%%%%
\begin{equation} \label{sadmon3}
\begin{split}
 \gamma_{\om P\op Q}^{\phantom{O}} &= \frac{1}{\mP \mQ}\left\{
(m_1m_2^\prime + m_1^\prime m_2) \gammaab +
(m_1m_3^\prime + m_1^\prime m_3) \gammaac \right.
\\[6pt] & \hspace{1.6cm}
+(m_2m_3^\prime + m_2^\prime m_3) \gammabc
+
\left.
m_1m_1^\prime + m_2m_2^\prime + m_3m_3^\prime \right\}
\\[6pt] &
\hspace{-1.2cm}
= \frac{1}{\mP \mQ}\left\{
(m_1m_2^\prime + m_1^\prime m_2) \gammaab +
m_1^\prime m_3 \gammaac +
m_2^\prime m_3 \gammabc +
m_1m_1^\prime + m_2m_2^\prime \right\}
\end{split}
\end{equation}
%%%%%%%%%%%%%%%%%%%%%%%%%%%%%%%%%%%%%%%%%%%%%%%%%%%%%%%%%%%%%%%%%%%%
noting that $m_3^\prime=0$ in the gyrobarycentric representation
of $P$ in \eqref{yaglf02}.

As we see from the
{\it gyrobarycentric representation constant} associated with
the Gyrobarycentric Representation Gyrocovariance Theorem
\cite[Theorem 4.6, pp.~90-91]{mybook05},
and by \eqref{yaglf01b} and \eqref{yaglf02b}, $\mQ>0$ and $\mP>0$
are given by
%%%%%%%%%%%%%%%%%%%%%%%%%%%%%%%%%%%%%%%%%%%%%%%%%%%%%%%%%%%%%%%%%%%%
\begin{equation} \label{sadmon4}
\begin{split}
m_Q^2 &= m_1^2 + m_2^2 + m_3^2 + 2(
m_1m_2\gammaab + m_1m_3\gammaac + m_2m_3\gammabc)
\\[6pt] 
&= \{(\gammaac-1)(1-t_1)+(\gammabc-1)t_1-(\gammaab-1)(1-t_1)t_1\}^2
\\[6pt] 
m_{P}^2 &= (m_1^\prime)^2 + (m_2^\prime)^2 +
2 m_1^\prime m_2^\prime \gammaab
\\[6pt]
&=1+2(\gammaab-1)(1-t_1)t_1
\,,
\end{split}
%MATHEMATICA stam285a zerop, zeroq
%MATLAB fig318enE.m zerop, zeroq
%MATLAB test0444
\end{equation}
%%%%%%%%%%%%%%%%%%%%%%%%%%%%%%%%%%%%%%%%%%%%%%%%%%%%%%%%%%%%%%%%%%%%
noting that $m_P^2>0$ for all $0<t_1<1$.
Also $m_Q^2>0$ for all $0<t_1<1$ according to the
Circumgyrocenter Theorem\index{circumgyrocenter theorem}
\ref{thmtivhvn}, p.~\pageref{thmtivhvn}, since
gyrotriangle $A_1A_2A_3$ possesses a circumgyrocircle on which the
point $Q$ lies, as shown in Fig.~\ref{fig318enm}.

In \eqref{sadmon4},
$m_Q^2$ is a {\it perfect square}\index{perfect square}
in the sense that it appears in \eqref{sadmon4} as a squared polynomial function.
Since $\mQ>0$, we have
\begin{equation} \label{sadmon4s}
\mQ = (\gammaac-1)(1-t_1)+(\gammabc-1)t_1-(\gammaab-1)(1-t_1)t_1
~>~0\,,
\end{equation}
where, indeed, $\mQ>0$ for all $0\le t_1\le1$.

Substituting \eqref{yaglf01b} and \eqref{yaglf02b}
into \eqref{sadmon3} and squaring, we obtain
$\gamma_{\om P\op Q}^{\,2}$ and, hence, $\gamma_{\om P\op Q}^{\,2}-1$
 expressed in terms of gamma factors
and the parameter $t_1$, $0<t_1<1$,
\begin{equation} \label{dfsvh1t}
\begin{split}
\frac{1}{s^2} & \gamma_{|PQ|}^{\,2} |PQ|^2 =
\gamma_{\om P\op Q}^{\,2} -1
\\ &=\frac{
(\gammaab-1)^2 \{[\gammaac(1-t_1)+\gammabc t_1]^2-2(\gammaab-1)(1-t_1)t_1-1\}
(1-t_1)^2t_1^2
}{
\{2(\gammaab-1)(1-t_1)t_1+1\}
\{(\gammaab-1)(1-t_1)t_1-\gammaac(1-t_1)-\gammabc t_1 +1\}^2
}
\,,
\end{split}
\end{equation}
%MATHEMATICA stam285a
%MATLAB fig318enE.m 
noting \eqref{rugh1ds}, p.~\pageref{rugh1ds}.
The extreme sides of \eqref{dfsvh1t} prove useful in the proof of
Theorem \ref{thmchord}, p.~\pageref{thmchord}.

Substituting $\gamma_{\om P\op Q}^{\,2}$ from \eqref{dfsvh1t}
into Identity, \eqref{rugh1ds}, p.~\pageref{rugh1ds},
\begin{equation} \label{sadmon5}
|PQ|^2 :=
\|\om P \op Q \|^2 = s^2 \frac{
\gamma_{\om P\op Q}^{\,2} - 1
}{
\gamma_{\om P\op Q}^{\,2}
}
\,,
\end{equation}
we obtain the desired gyrodistance,
\begin{subequations} \label{dfsvh}
\begin{equation} \label{dfsvh1}
\frac{1}{s^2} \|\om P \op Q \|^2
=\frac{
(\gammaab-1)^2 \{[\gammaac(1-t_1)+\gammabc t_1]^2-2(\gammaab-1)(1-t_1)t_1-1\}
(1-t_1)^2t_1^2
}{D^2}
\,,
\end{equation}
where
\begin{equation} \label{dfsvh2}
\begin{split}
D &= \gammaac-1 + (\gammaab-1)\gammaac t_1
-2(\gammaab-1)(1-t_1)t_1 - (\gammaab-1)\gammabc t_1^2
\\ &+
(\gammaac-\gammabc)\{-\gammaab+(\gammaab-1)(1-t_1)^2\} t_1
\,.
\end{split}
\end{equation}
%MATHEMATICA stam285a
%MATLAB fig318enE.m zerogdispq
\end{subequations}

% SECTION NUMBER 25
\section{Circumgyrocevian Gyrolength} \label{slila7}

In order to prepare the stage for the
Intersecting Gyrochords Theorem \ref{thmchord}, p.~\pageref{thmchord},
in this section we calculate the gyrodistance $|A_3Q|=\|\om A_3\op Q\|$
between
vertex $A_3$ and its opposing circumgyrocevian foot $Q$
of a gyrotriangle $A_1A_2A_3$,
shown in Fig.~\ref{fig318enm}, p.~\pageref{fig318enm}.

\index{circumgyrocevian}
Let $A_3Q$ be a circumgyrocevian of a gyrotriangle $A_1A_2A_3$
in an Einstein gyrovector space $(\Rsn,\op,\od)$,
as shown in Fig.~\ref{fig318enm}, so that the gyrobarycentric representation of $Q$
with respect to the gyrobarycentrically independent set $\{A_1,A_2,A_3\}$
is given by \eqref{yaglf01}.

Applying the Gyrobarycentric Representation Gyrocovariance Theorem
\cite[Theorem 4.6, pp.~90-91]{mybook05}
with $X=\om A_3$
to the gyrobarycentric representation \eqref{yaglf01a} of $Q$,
and using the gyrotriangle index notation, \eqref{indexnotation}, p.~\pageref{indexnotation},
we have
%%%%%%%%%%%%%%%%%%%%%%%%%%%%%%%%%%%%%%%%%%%%%%%%%%%%%%%%%%%%%%%%%%%%
\begin{equation} \label{trsum1}
\begin{split}
\om A_3 \op Q &= \frac{
m_1\gamma_{\om A_3\op A_1}^{\phantom{O}}(\om A_3\op A_1)
+
m_2\gamma_{\om A_3\op A_2}^{\phantom{O}}(\om A_3\op A_2)
+
m_3\gamma_{\om A_3\op A_3}^{\phantom{O}}(\om A_3\op A_3)
}{
m_1\gamma_{\om A_3\op A_1}^{\phantom{O}}
+
m_2\gamma_{\om A_3\op A_2}^{\phantom{O}}
+
m_3\gamma_{\om A_3\op A_3}^{\phantom{O}}
}
\\[8pt] &=
\frac{
m_1\gammaac\ab_{31} + m_2\gammabc\ab_{32}
}{
m_1\gammaac + m_2\gammabc + m_3
}
\end{split}
\end{equation}
%%%%%%%%%%%%%%%%%%%%%%%%%%%%%%%%%%%%%%%%%%%%%%%%%%%%%%%%%%%%%%%%%%%%
and
%%%%%%%%%%%%%%%%%%%%%%%%%%%%%%%%%%%%%%%%%%%%%%%%%%%%%%%%%%%%%%%%%%%%
\begin{equation} \label{trsum2}
\gamma_{\om A_3\op Q}^{\phantom{O}} = \frac{
m_1\gammaac + m_2\gammabc + m_3}{\mQ}
\,,
\end{equation}
%%%%%%%%%%%%%%%%%%%%%%%%%%%%%%%%%%%%%%%%%%%%%%%%%%%%%%%%%%%%%%%%%%%%
where $\mQ>0$ is the constant of the gyrobarycentric representation \eqref{yaglf01}
of $Q$, given by \eqref{sadmon4s}.

Substituting $m_1,m_2,m_3$ from \eqref{yaglf01b} and $\mQ$ from \eqref{sadmon4s}
into \eqref{trsum2}, we obtain the equation
%%%%%%%%%%%%%%%%%%%%%%%%%%%%%%%%%%%%%%%%%%%%%%%%%%%%%%%%%%%%%%%%%%%%
\begin{equation} \label{trsum3}
\gamma_{\om A_3 \op Q}^{\phantom{O}} = \frac{
\gammaac(\gammaac-1)(1-t_1) + \gammabc(\gammabc-1)t_1
-\{\gammaab-1 + (\gammaac-\gammabc)^2\}(1-t_1)t_1
}{
(\gammaac-1)(1-t_1) + (\gammabc-1)t_1 - (\gammaab-1)(1-t_1)t_1
}
\,.
\end{equation}
%MATHEMATICA stam285c zero5
%MATLAB fig318enE zero1
%%%%%%%%%%%%%%%%%%%%%%%%%%%%%%%%%%%%%%%%%%%%%%%%%%%%%%%%%%%%%%%%%%%%
Equation \eqref{trsum3} proves useful in the proof of
Theorem \ref{thmchord} in Sect.~\ref{slila8}.

\index{circumgyrocevian}
Finally, the gyrolength $\|\om A_3 \op Q\|$
of the circumgyrocevian $A_3Q$ of gyrotriangle $A_1A_2A_3$ in Fig.~\ref{fig318enm}
is obtained from \eqref{trsum3}
by means of \eqref{rugh1ds}, p.~\pageref{rugh1ds},
\begin{equation} \label{trsum4}
\begin{split}
\frac{1}{s^2}&\|\om A_3 \op Q \|^2 = \frac{
\gamma_{\om A_3\op Q}^{\,2} - 1
}{
\gamma_{\om A_3\op Q}^{\,2}
}
\\[8pt] &= \frac{
%\{\gamma_{13}^2(1-t_1)+\gamma_{23}^2t_1-1-[2(\gammaab-1)+(\gammabc-\gammaac)^2](1-t_1)t_1\}
\{\gamma_{13}^2t_2+\gamma_{23}^2t_1-1-[2(\gammaab-1)+(\gammabc-\gammaac)^2]t_2t_1\}
\{\gammaac-1 + (\gammabc-\gammaac)t_1\}^2
}{
\{
\gammaac(\gammaac-1)t_2 + \gammabc(\gammabc-1)t_1
-\{\gammaab-1 + (\gammaac-\gammabc)^2\}t_2t_1
\}^2
}
\,,
\end{split}
\end{equation}
%MATHEMATICA stam285c zero11
%MATLAB fig318enE zero2
where $t_2=1-t_1$, $0<t_1<1$.

% SECTION NUMBER 26
\section{The Intersecting Gyrochords Theorem} \label{slila8}

Following Sects.~\ref{slila5}\,--\,\ref{slila7}, we are now in the position
to prove the Intersecting Gyrochords Theorem.

%%%%%%%%%%%%%%%%%%%%%%%%%%%%%%%%%%%%%%%%%%%%%%%%%%%%%%%%%%%%%%%%%%%%
% THEOREM NUMBER 9.23
\index{intersecting gyrochords theorem}
\index{gyrochords, intersecting, theorem}
\begin{theorem}\label{thmchord}
{\bf (The Intersecting Gyrochords Theorem).}
If two gyrochords, $A_1A_2$ and $B_1B_2$, of a gyrocircle
in an Einstein gyrovector space $(\Rsn,\op,\od)$ intersect
at a point $P$, as shown in Fig.~\ref{fig318enam}.
then
\begin{equation} \label{manot01}
\frac{
\gamma_{|PA_1|}^{\phantom{O}} |PA_1|
\gamma_{|PA_2|}^{\phantom{O}} |PA_2|}
{\gamma_{|A_1A_2|}^{\phantom{O}} +1}
=
\frac{
\gamma_{|PB_1|}^{\phantom{O}} |PB_1|
\gamma_{|PB_2|}^{\phantom{O}} |PB_2|}
{\gamma_{|B_1B_2|}^{\phantom{O}} +1}
\,.
\end{equation}
% MATHEMATICA stam---
% MATLAB fig318enE zerothm
\end{theorem}
\begin{proof}
As a matter of notation we have
%%%%%%%%%%%%%%%%%%%%%%%%%%%%%%%%%%%%%%%%%%%%%%%%%%%%%%%%%%%%%%%%%%%%
\begin{equation} \label{hdru01}
\begin{split}
|A_1A_2| &= \|\om A_1 \op A_2\| = \gammaab
\\[6pt]
|B_1B_2| &= |A_3Q|
\,,
\end{split}
\end{equation}
%%%%%%%%%%%%%%%%%%%%%%%%%%%%%%%%%%%%%%%%%%%%%%%%%%%%%%%%%%%%%%%%%%%%
where we use the notation in Fig.~\ref{fig318enm} in which, in particular,
$B_1=A_3$ and $B_2=Q$.

Hence, noting that each side of \eqref{manot01} is positive,
it can be written equivalently as
\begin{equation} \label{hdru02}
\left(
\frac{
\gamma_{|PA_1|}^{\phantom{O}} |PA_1|
\gamma_{|PA_2|}^{\phantom{O}} |PA_2|}
{\gammaab+1}
\right)^2
=
\left(
\frac{
\gamma_{|PA_3|}^{\phantom{O}} |PA_3|
\gamma_{|P  Q|}^{\phantom{O}} |P  Q|}
{\gamma_{|A_3Q|}^{\phantom{O}} +1}
\right)^2
\,,
\end{equation}
where $|PA_1|=|A_1P|$ is the gyrolength $\|\om P \op A_1\|$
of gyrosegment $PA_1$, {\it etc.}

Substituting into \eqref{hdru02}
\begin{enumerate}
\item
$(\gamma_{|PA_k|}^{\phantom{O}} |PA_k|)^2$, $k=1,2,3$,
from \eqref{halof04} in Sect.~\ref{slila5}, and
\item
$(\gamma_{|PQ|}^{\phantom{O}} |PQ|)^2$ from \eqref{dfsvh1t}
in Sect.~\ref{slila6}, and
\item
$\gamma_{|A_3Q|}^{\phantom{O}}$ from \eqref{trsum3}
in Sect.~\ref{slila7},
\end{enumerate}
we  find that Identity \eqref{hdru02} is valid, each side of which
being equal to the right-hand side of each of
the two equations in \eqref{hdru02s} below,
%%%%%%%%%%%%%%%%%%%%%%%%%%%%%%%%%%%%%%%%%%%%%
\begin{equation} \label{hdru02s}
\begin{split}
\left(
\frac{
\gamma_{|PA_1|}^{\phantom{O}} |PA_1|
\gamma_{|PA_2|}^{\phantom{O}} |PA_2|}
{\gammaab+1}
\right)^2
&=
s^4
\left(
\frac{
(\gamma_{12}^2-1)(1-t_1)t_1
}{
2(\gamma_{12}^2-1)(1-t_1)t_1 +1
}
\right)^2
\\[8pt]
\left(
\frac{
\gamma_{|PA_3|}^{\phantom{O}} |PA_3|
\gamma_{|P  Q|}^{\phantom{O}} |P  Q|}
{\gamma_{|A_3Q|}^{\phantom{O}} +1}
\right)^2
&= s^4
\left(
\frac{
(\gamma_{12}^2-1)(1-t_1)t_1
}{
2(\gamma_{12}^2-1)(1-t_1)t_1 +1
}
\right)^2
\,.
\end{split}
\end{equation}
%MATHEMATICA stam288
% MATLAB fig318enE.m zerothm and zero3
%%%%%%%%%%%%%%%%%%%%%%%%%%%%%%%%%%%%%%%%%%%%%

The two equations in \eqref{hdru02s} share the right-hand sides, so that the
proof of \eqref{hdru02} and, hence, of \eqref{manot01}, is complete.
\end{proof}

%%%%%%%%%%%%%%%%%%%%%%%%%%%%%%%%%%%%%%%%%%%%%%%%%%%%%%%%%%%%%%%%%%%%
%sidebyside                    Gyroellipse in (+)Ees
% FIGURE 25 and 26
 %%%  Double Figs.  %%%%%%%%%%%%%%%%%%%%%%%%%%%%%%%%%%%%%%%%%%%%%%%%%%%%%%%%%%
%%%%%%%%%%%%%%%%%%%%%%%%%%%%%%%%%%%%%%%%%%%%%%%%%%%%%%%%%%%%%%%%%%%%%%
%\begin{figure}[ht]
\begin{figure}[t]  % try to put this figure on the top of the page
              % [h] tries to place the figure here
              % [b] tries to place the figure on the bottom of the page
              % [t] tries to place the figure on the top of the page
              % [P] tries to place the figure floatingly on the page
 \sidebyside {       % center two figures
 \psfrag{O}{$C$}
 \psfrag{P}{$P$}
 \psfrag{A1}{$A_1$}
 \psfrag{A2}{$A_2$}
 \psfrag{B1}{$\hspace{0.2cm}B_1$}
 \psfrag{B2}{$B_2$}
 \includegraphics[width=0.5\textwidth]{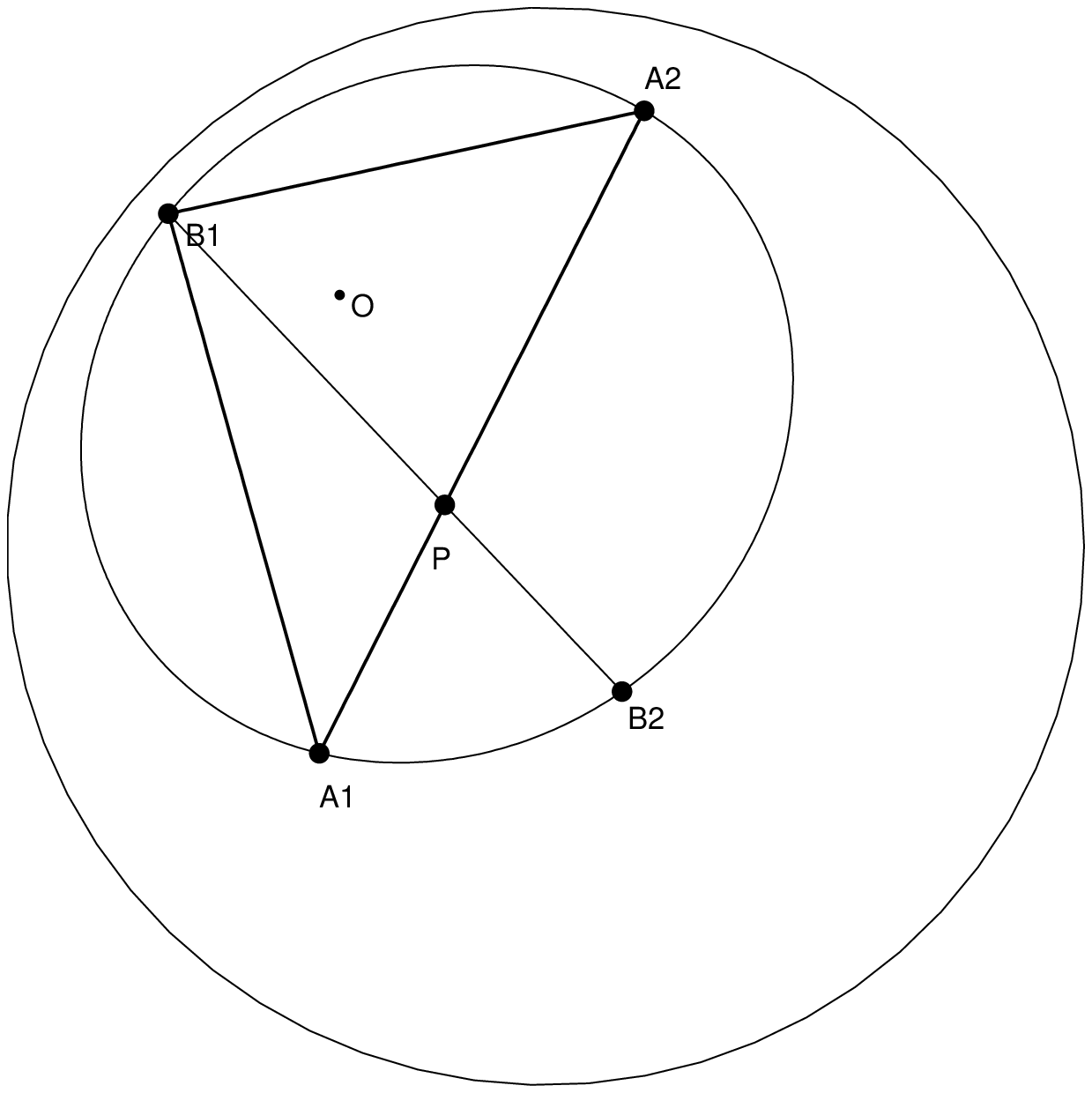}
\caption[Intersecting gyrochords]{
The Intersecting Gyrochords Theorem \ref{thmchord}
in the hyperbolic plane $(\Rstwo,\op,\od)$.
Two intersecting gyrochords of a gyrocircle, $A_1A_2$ and $B_1B_2$, are shown
in this hyperbolic counterpart of Fig.~\ref{fig318euam}.
\label{fig318enam}}}
  {
 \psfrag{O}{$C$}
 \psfrag{P}{$P$}
 \psfrag{A1}{$A_1$}
 \psfrag{A2}{$A_2$}
 \psfrag{B1}{$\hspace{0.2cm}B_1$}
 \psfrag{B2}{$B_2$}
 \includegraphics[width=0.5\textwidth]{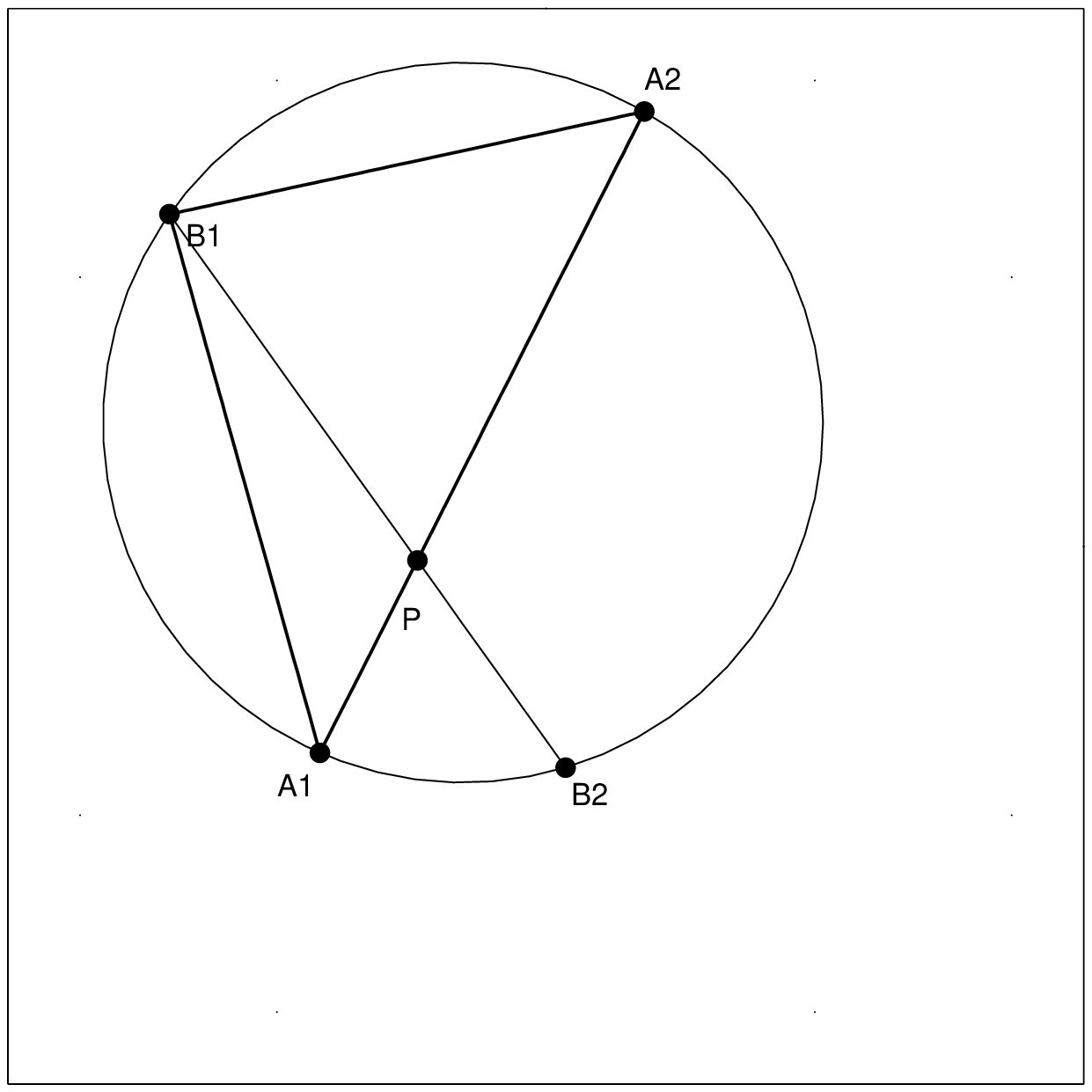}
\caption[Intersecting chords]{
The Intersecting Chords Theorem \ref{thmchordeuc} in the Euclidean plane $\Rtwo$.
Two intersecting chords of a circle, $A_1A_2$ and $B_1B_2$, of a circle are shown.
This figure is the Euclidean counterpart of Fig.~\ref{fig318enam}.
\label{fig318euam}} }
\end{figure}
%%%%%%%%%%%%%%%%%%%%%%%%%%%%%%%%%%%%%%%%%%%%%%%%%%%%%%%%%%%%%%%%%%%%%%
 % The Intersecting (Gyro)Chords
%   Fig.~\ref{fig318enam} and Fig.~\ref{fig318euam}
%%%%%%%%%%%%%%%%%%%%%%%%%%%%%%%%%%%%%%%%%%%%%%%%%%%%%%%%%%%%%%%%%%%n%

In the Euclidean limit, $s\rightarrow\infty$,
gyrolengths of gyrosegments tend to lengths of corresponding segments
and gamma factors tend to 1. Hence, in that limit,
the Intersecting Gyrochords Theorem \ref{thmchord}
reduces to the following well-known
Intersecting Chords Theorem of Euclidean geometry:

%%%%%%%%%%%%%%%%%%%%%%%%%%%%%%%%%%%%%%%%%%%%%%%%%%%%%%%%%%%%%%%%%%%%
% THEOREM NUMBER 9.24
\index{intersecting chords theorem}
\index{chords, intersecting, theorem}
\begin{theorem}\label{thmchordeuc}
{\bf (The Intersecting Chords Theorem).}
If two chords, $A_1A_2$ and $B_1B_2$, of a circle
in a Euclidean space $\Rn$ intersect
at a point $P$, as shown in Fig.~\ref{fig318euam}, then
\begin{equation} \label{manot01euc}
|PA_1| |PA_2| = |PB_1| |PB_2|
\,.
\end{equation}
\end{theorem}
%%%%%%%%%%%%%%%%%%%%%%%%%%%%%%%%%%%%%%%%%%%%%%%%%%%%%%%%%%%%%%%%%%%%

Our ambiguous notation that emphasizes analogies between
hyperbolic and Euclidean geometry should raise no confusion.
Thus, in particular, one should note that
\begin{enumerate}
\item
in the Intersecting Gyrochords Theorem \ref{thmchord},
$|PA_1|$ is the gyrolength,
\begin{equation} \label{hdnkw01}
|PA_1|=\|\om P\op A_1\|
\,,
\end{equation}
of gyrosegment $PA_1$, {\it etc.},
while, in contrast,
\item
in the Intersecting Chords Theorem \ref{thmchordeuc},
$|PA_1|$ is the length,
\begin{equation} \label{hdnkw02}
|PA_1|=\|-P+A_1\|
\,,
\end{equation}
of segment $PA_1$, {\it etc}.
\end{enumerate}

\end{document}